\documentclass[a4paper,sigplan]{acmart}
\pdfoutput=1

\usepackage{ifthen}
\newboolean{longVersion}
\setboolean{longVersion}{true}

\usepackage{standalone}
\usepackage{tikz}
\usetikzlibrary{positioning,arrows,automata,calc,shapes}

\definecolor{darkgreen}{rgb}{0.0, 0.5, 0.0}
\definecolor{darkred}{rgb}{0.9, 0.0, 0.0}
\definecolor{darkblue}{rgb}{0.0, 0.0, 0.9}
\newcommand{\posw}[1]{\textcolor{darkgreen}{+#1}}
\newcommand{\negw}[1]{\textcolor{darkred}{-#1}}
\newcommand{\nilw}{\textcolor{darkblue}{0}}
\newcommand{\neuw}[1]{\textcolor{brown}{#1}}

\tikzset{
	state/.style={rounded rectangle,draw=black,inner sep=1.5mm,minimum width=9mm,minimum height=5mm},
	rstate/.style={rectangle,draw=black,inner sep=1ex,minimum width=9mm,minimum height=5mm},
	istate/.style={minimum width=5mm, minimum height=6mm},
	bullet/.style={circle,draw=black,fill=black,inner sep=0cm, minimum size=0.7mm},
	initial text={},
	every initial by arrow/.style={->,>=stealth'},
	ptran/.style={rounded corners, ->,>=stealth',auto},
	ntran/.style={rounded corners, -,auto}
} 

\usepackage{hyperref}
\usepackage{microtype}%

\usepackage{natbib}
\usepackage{multibib}
\newcites{apx}{Literature Referred to in the Appendix}

\usepackage{thm-restate}%
\declaretheorem{theorem}

\usepackage{amsmath,amsfonts,amssymb}

\usepackage{multirow}

\newtheorem{remark}[theorem]{Remark}

\usepackage[utf8]{inputenc}

\usepackage{microtype}
\usepackage{ellipsis}

\usepackage[linesnumbered,ruled,noend]{algorithm2e}
\let\Until\undefined %

\usepackage{listings}

\usepackage{url}

\usepackage[rflt]{floatflt}

\usepackage{datatool}

\usepackage{intcalc}

\usepackage{numprint}
\npdecimalsign{.}
\npthousandsep{,}

\usepackage[textwidth=3cm,textsize=tiny]{todonotes}

\usepackage{hyperref}
\usepackage{paralist}

\usepackage{amssymb}
\usepackage{tikz}

\newcommand{\accdiaplus}{%
\ensuremath{\diamondplus}}

\newcommand{\rawdiaplus}{%
  \begin{tikzpicture}
    \useasboundingbox (-0.7ex, -0.7ex) rectangle (0.7ex, 0.7ex);
    \node (w) at (-0.7ex,0) {};
    \node (e) at (+0.7ex,0) {};
    \node (s) at (0,-0.7ex) {};
    \node (n) at (0,+0.7ex) {};
    \draw[line width=.2mm] (n.center) -- (e.center) -- (s.center) -- (w.center) -- cycle;
    \draw (n.center) -- (s.center);
    \draw (e.center) -- (w.center);
  \end{tikzpicture}}

\newsavebox{\diamondplusbox}
\savebox{\diamondplusbox}{\rawdiaplus}
\DeclareRobustCommand{\diamondplus}{\texorpdfstring{\mathop{\raisebox{-0.25ex}{\resizebox{1.45ex}{!}{\usebox{\diamondplusbox}}}}\nolimits}{diamondplus}}

\newcommand{\rawdiaminus}{%
  \begin{tikzpicture}
    \useasboundingbox (-0.7ex, -0.9ex) rectangle (0.7ex, 0.9ex);
    \node (w) at (-0.7ex,0) {};
    \node (e) at (+0.7ex,0) {};
    \node (s) at (0,-0.9ex) {};
    \node (n) at (0,+0.9ex) {};
    \draw (n.center) -- (e.center) -- (s.center) -- (w.center) -- (n.center);
    \draw (e.center) -- (w.center);
  \end{tikzpicture}}

\newsavebox{\diamondminusbox}
\savebox{\diamondminusbox}{\rawdiaminus}

\newcommand{\onlyLong}[1]{\ifthenelse{\boolean{longVersion}}{#1}{}}%
\newcommand{\onlyShort}[1]{\ifthenelse{\boolean{longVersion}}{}{#1}}%

\newcommand{\ie}{\textit{i.e.}, }
\newcommand{\eg}{\textit{e.g.}, }

\newcommand{\tudparagraph}[2]{%
\vspace*{#1}

\noindent
{\bf #2}
}

\newcommand{\fE}{\mathfrak{E}}

\newcommand{\cB}{\mathcal{B}}
\newcommand{\cC}{\mathcal{C}}

\newcommand{\cE}{\mathcal{E}}
\newcommand{\cF}{\mathcal{F}}
\newcommand{\cG}{\mathcal{G}}
\newcommand{\cH}{\mathcal{H}}

\newcommand{\cM}{\mathcal{M}}
\newcommand{\cN}{\mathcal{N}}

\newcommand{\cZ}{\mathcal{Z}}

\newcommand{\eqdef}{\ensuremath{\stackrel{\text{\tiny def}}{=}}}

\newcommand{\Ende}{\hfill ${\scriptscriptstyle \blacksquare}$}

\newcommand{\Distr}{\mathit{Distr}}

\renewcommand{\Pr}{\mathrm{Pr}}

\newcommand{\Expected}{\mathbb{E}}

\newcommand{\Exp}[2]{\mathbb{E}^{#1}_{#2}}

\newcommand{\ExpRew}[2]{\mathbb{E}^{#1}_{#2}}

\newcommand{\MP}{\mathrm{MP}} %

\newcommand{\ocM}{\mathfrak{M}}

\newcommand{\M}{\cM}

\newcommand{\sinit}{s_{\mathit{\scriptscriptstyle init}}}

\newcommand{\stateN}[2]{{#1}_{\scriptscriptstyle #2}}

\newcommand{\entrystate}[1]{{#1}_{\mathit{\scriptscriptstyle in}}}
\newcommand{\exitstate}[1]{{#1}_{\mathit{\scriptscriptstyle out}}}

\newcommand{\init}{{\mathit{\scriptscriptstyle init}}}

\newcommand{\Act}{\mathit{Act}}

\newcommand{\Post}{\mathit{Post}}

\newcommand{\mdpP}{P} %

\newcommand{\fpath}{\pi}

\newcommand{\cycle}{\xi}

\newcommand{\infpath}{\varsigma}

\newcommand{\first}{\mathit{first}}
\newcommand{\last}{\mathit{last}}

\renewcommand{\state}[2]{#1[#2]}
\newcommand{\fragment}[3]{#1[#2 \ldots #3]}
\newcommand{\prefix}[2]{\mathit{pref}(#1,#2)}

\newcommand{\FPaths}{\mathit{FPaths}}
\newcommand{\IPaths}{\mathit{IPaths}}

\newcommand{\MPaths}{\mathit{MPaths}}
\newcommand{\Paths}{\MPaths}

\newcommand{\InfPaths}{\IPaths}

\newcommand{\Cyl}{\mathit{Cyl}}

\newcommand{\sched}{\mathfrak{S}}
\newcommand{\tsched}{\mathfrak{T}}
\newcommand{\usched}{\mathfrak{U}}
\newcommand{\vsched}{\mathfrak{V}}

\newcommand{\residual}[2]{#1 {\uparrow} {#2}}

\newcommand{\stratone}{\sigma}
\newcommand{\strattwo}{\tau}
\newcommand{\play}[3]{\mathsf{Play}_{#1}(#2,#3)}

\newcommand{\MEC}{\mathit{MEC}}
\newcommand{\WDMEC}{\mathit{WDMEC}}

\newcommand{\wgt}{\mathit{wgt}}

\newcommand{\accwgt}{\textsf{wgt}}

\newcommand{\wmax}{w^{\max}}

\newcommand{\mincredit}{\mathit{mincredit}}
\newcommand{\lgr}{\mathit{lgr}}
\newcommand{\rec}{\mathit{rec}}

\newcommand{\maxwgt}[3]{\mathit{maxwgt}^{#1}_{#2}[#3]}

\newcommand{\ActEC}{\mathfrak{A}}

\newcommand{\spider}[2]{\mathsf{Spider}_{#2}(#1)}

\newcommand{\purge}[1]{\mathit{purge}(#1)}
\newcommand{\Purge}[2]{\mathit{purge}_{#1}(#2)}

\newcommand{\Limit}[1]{\mathit{Limit}_{#1}}

\newcommand{\good}{\mathit{good}}
\newcommand{\goal}{\mathit{goal}}
\newcommand{\fail}{\mathit{fail}}
\newcommand{\final}{\mathit{final}}

\newcommand{\ZeroEC}{\mathit{ZeroEC}}

\newcommand{\GoodEC}{\mathit{GoodEC}}
\newcommand{\PumpEC}{\mathit{PumpEC}}
\newcommand{\GambEC}{\mathit{GambEC}}

\newcommand{\divsched}{\mathfrak{D}}

\newcommand{\Good}{\mathit{Good}}

\newcommand{\neXt}{\bigcirc}

\DeclareMathOperator{\Until}{\ensuremath{\pmb{\mathrm{U}}}}

\newcommand{\true}{\mathit{true}}

\newcommand{\Nat}{\mathbb{N}}
\newcommand{\Rational}{\mathbb{Q}}
\newcommand{\Real}{\mathbb{R}}
\newcommand{\Integer}{\mathbb{Z}}

\newcommand{\NP}{\textsf{NP}}
\newcommand{\coNP}{\textsf{coNP}}
\newcommand{\PTIME}{\textsf{P}}
\newcommand{\PSPACE}{\textsf{PSPACE}}
\newcommand{\EXPTIME}{\textsf{EXPTIME}}
\newcommand{\UP}{\textsf{UP}}
\newcommand{\coUP}{\textsf{coUP}}

\newcommand{\move}[1]{\stackrel{#1}{\longrightarrow}}

\newcommand{\CiteAppendix}[1]{}

\colorlet{christelColor}{orange!30!white}
\colorlet{clemensColor}{red!30!white}
\colorlet{nathalieColor}{blue!30!white}
\colorlet{ocanColor}{green!30!white}
\colorlet{danielColor}{cyan!30!white}

\newcommand{\dwr}{\textsf{DWR}} %

\newcommand{\Easdwr}{$\dwr^{\exists,=1}$} %
\newcommand{\Eposdwr}{$\dwr^{\exists,>0}$}
\newcommand{\Uasdwr}{$\dwr^{\forall,=1}$}
\newcommand{\Uposdwr}{$\dwr{}^{\forall,>0}$}

\newcommand{\valueEas}[1]{K_{#1}^{\exists,=1}} %
\newcommand{\valueEpos}[1]{K_{#1}^{\exists,>0}}
\newcommand{\valueUas}[1]{K_{#1}^{\forall,=1}}
\newcommand{\valueUpos}[1]{K_{#1}^{\forall,>0}}

\newcommand{\wB}{\textsf{WB}} %
\newcommand{\wcoB}{\textsf{WcoB}} %

\newcommand{\EaswB}{$\wB^{\exists,=1}$} %
\newcommand{\EaswcoB}{$\wcoB^{\exists,=1}$} %
\newcommand{\EposwB}{$\wB^{\exists,>0}$} %
\newcommand{\EposwcoB}{$\wcoB^{\exists,>0}$} %

\newcommand{\valueEposB}[1]{B_{#1}^{\exists,>0}}

\newcommand{\UaswB}{$\wB^{\forall,=1}$} %
\newcommand{\UposwB}{$\wB^{\forall,>0}$} %

\newcommand{\valueUasB}[1]{B_{#1}^{\forall,=1}}

\title{Stochastic Shortest Paths and Weight-Bounded Properties 
       in Markov Decision Processes}

\author{Christel Baier}
\affiliation{Technische Universit\"at Dresden, Germany}
\author{Nathalie Bertrand}
\affiliation{Univ Rennes, Inria, CNRS, IRISA, France}
\author{Clemens Dubslaff}
\author{Daniel Gburek}
\affiliation{Technische Universit\"at Dresden, Germany}
\author{Ocan Sankur}
\affiliation{Univ Rennes, Inria, CNRS, IRISA, France}
\thanks{The authors are partly supported by the DFG through
        the collaborative research centre HAEC (SFB 912),
        the Excellence Initiative by the German Federal
            and State Governments (cluster of excellence cfAED),
        the Research Training Group QuantLA (GRK 1763), and
        the DFG-project BA-1679/11-1. The collaboration is supported by Inria associate team programme.}

\begin{abstract}
  The paper deals with finite-state Markov decision processes (MDPs) 
  with integer weights assigned to each state-action pair.
  New algorithms are presented to classify end components according to their 
  limiting behavior with respect to the accumulated weights.
  These algorithms are used to provide solutions for
  two types of fundamental problems for integer-weighted MDPs.
  First, a polynomial-time algorithm for the classical 
  stochastic shortest path problem is presented, generalizing
  known results for special classes of weighted MDPs.
  Second, qualitative probability constraints
  for weight-bounded (repeated) reachability conditions are addressed. 
  Among others, it is shown that the problem to decide whether a
  disjunction of weight-bounded reachability conditions
  holds almost surely under some scheduler
  belongs to $\textrm{NP}\cap \textrm{coNP}$,
  is solvable in pseudo-polynomial time and is at least as hard as solving
  two-player mean-payoff games, while the corresponding problem
  for universal quantification over schedulers is solvable in polynomial time.
\end{abstract}
 
\begin{document}

\maketitle

\renewcommand{\thetheorem}{\arabic{section}.\arabic{theorem}}
\renewcommand{\thelemma}{\arabic{section}.\arabic{lemma}}
\renewcommand{\thecorollary}{\arabic{section}.\arabic{corollary}}
\renewcommand{\theremark}{\arabic{section}.\arabic{remark}}
\renewcommand{\theexample}{\arabic{section}.\arabic{example}}
\renewcommand{\thedefinition}{\arabic{section}.\arabic{definition}}

\section{Introduction}

Markov decision processes (MDPs) are a prominent model used, \eg
in operations research, artificial intelligence, robotics
and the formal analysis of probabilistic nondeterministic programs.
Various types of stochastic shortest (or longest) path problems can be
formalized as an optimization problem for MDPs with integer or
rational weights for the transitions where the task is to
determine an optimal scheduling policy for the MDP until
reaching a target.
Here, optimality is understood with respect to %
the expected accumulated weight or the probability of reaching the target
under weight constraints.
Such problems can be seen as a
control-synthesis problem that, \eg asks to implement a
decision-making routine for a robot so that the robot
eventually reaches a safe state almost surely, while providing
guarantees on the achieved utility.  

Stochastic shortest (or longest) path problems are well understood and
supported by various tools for finite-state MDPs with nonnegative
weights only, for which the algorithms can rely on the monotonicity of
accumulated weights along the prefixes of paths.  In this case,
schedulers that maximize or minimize the expected accumulated weight
until reaching the target can be determined in polynomial time based
on a preprocessing of end components (\ie strongly connected
sub-MDPs) and linear programs~\cite{BerTsi91,deAlf99}.  One can
compute schedulers maximizing the probability for reaching the target
within a given cost in pseudo-polynomial time using an iterative
approach that successively increases the weight bound and treats
zero-weight loops by linear-programming
techniques~\cite{UB13,BDDKK14}. The corresponding decision problem is
\PSPACE-hard, even for acyclic MDPs~\cite{HaaseKiefer15}.

For MDPs with arbitrary integer weights, the lack of monotonicity of
accumulated weights makes analogous questions much harder. 
Even for finite-state
Markov chains with integer weights, 
the set of relevant configurations (\ie states augmented with
the weight that has been accumulated so far) 
can be infinite and, in MDPs with integer weights
optimal or $\varepsilon$-optimal
schedulers might require an infinite amount of memory.  
The latter is known from
energy-MDPs~\cite{ChatDoy11,BKN16,MaySchTozWoj17} 
where one aims at
finding a scheduler under which the system never runs out of energy
(\ie the accumulated weight plus some initial credit is
always positive) and satisfies an $\omega$-regular property
(\eg a parity condition) with probability 1 or maximizes the
expected mean payoff. 
Another indication for the additional difficulties that arise when
switching from nonnegative weights to integers is given by the work on
one-counter MDPs \cite{BBEKW10}, which can be seen as MDPs where all
weights are in $\{-1,0,+1\}$ and that terminate as soon as the counter
value is 0.  Among others, \cite{BBEKW10} establishes \PSPACE-hardness
and an $\EXPTIME$ upper bound for the almost-sure termination problem
under some scheduler, while the corresponding weight-bounded
(control-state) reachability problem in nonnegative MDPs is in
\PTIME~\cite{UB13}.

This paper addresses several fundamental problems 
for MDPs with integer weights.
Our main contributions 
are as follows.
First, we show that the classical stochastic shortest path problem,
where the task is to \emph{minimize the expected weight} until
reaching a target, is solvable in polynomial time for arbitrary
integer-weighted MDPs.  We hereby extend previous results for
restricted classes of MDPs \cite{BerTsi91,deAlf99}, while the general
case was open.
Second,
  we study disjunctions of
  \emph{weight-bounded reachability conditions} with qualitative probability
  bounds and existential or universal scheduler quantification.
  The problem to check the existence of a scheduler satisfying a
  disjunction of weight-bounded reachability conditions almost surely
  (referred to as decision problem \Easdwr) is shown to be in
  $\NP{\cap}\coNP$, solvable in pseudo-polynomial time, and as hard as
  non-stochastic two-player mean-payoff games (and therefore not known
  to be in \PTIME).
  The same complexity results are achieved for checking whether a
  disjunction of 
  weight-bounded reachability conditions holds with positive
  probability under all schedulers (problem \Uposdwr).
  In contrast, problem \Uasdwr{} that asks
  whether a disjunctive weight-bounded reachability
  condition holds almost surely under all schedulers 
  is shown to be in \PTIME.
 We also present algorithms for computing optimal weight-bounds
 with analogous time complexities:
 pseudo-polynomial  for the optimization variants of \Easdwr{} and
 \Uposdwr{} and polynomial  for \Uasdwr.
These results should be contrasted with the 
polynomial-time decidability of \Easdwr{} and \Uposdwr{} 
for MDPs where all
weights are nonnegative \cite{UB13}.

Although several other problems for integer-weighted MDPs
are known to be in $\NP\cap \coNP$
and as hard as nonstochastic two-player mean-payoff
games (see, \eg \cite{ChatDoy11,MaySchTozWoj17,BFRR17}
and the discussion on related work in Section~\ref{sec:discussion}),
our techniques crucially depart from previous work
by heavily relying on new algorithms to classify end
components (ECs) of MDPs.
We see these results on the \emph{classification of ECs} as a further
main contribution as it provides a useful vehicle for reasoning about
different problems for integer-weighted MDPs. An indication for the
latter is that we use these classification algorithms not
only to establish the results listed above for \Easdwr{} and \Uasdwr,
but also to prove the polynomial-time solvability of the classical
shortest path problem in general integer-weighted MDPs and to deal
with weight-bounded B\"uchi conditions.

Our classification of ECs is according to the existence of schedulers that
increase the weight to infinity (\emph{pumping ECs}), or ensure that the
weight eventually exceeds any threshold possibly without converging to
$+\infty$ (\emph{weight-divergent ECs}), or have oscillating behavior
(\emph{gambling ECs}), or keep the accumulated weights within a compact
interval (\emph{bounded ECs}).  
A sufficient and necessary criterion for the pumping
property is that the maximal expected mean payoff is positive, which
is decidable in polynomial time by computing the maximal expected mean
payoff using linear-programming techniques~\cite{Puterman,Kallenberg}.
While this observation has been made by several other authors, we are
not aware of earlier algorithms for checking 
the gambling or boundedness property. 
For checking weight-divergence, the results of 
\cite{BBEKW10} 
for one-counter MDPs without boundary 
yield a polynomial time bound
for the special case of MDPs where all weights are in 
$\{+1,0,-1\}$ and a pseudo-polynomial time bound in the general case.
We improve this result by presenting
a polynomial-time algorithm for deciding weight-divergence for
MDPs with arbitrary integer weights.
Moreover, in case that the given MDP $\cM$ is not weight-divergent,
the algorithm generates a new MDP $\cN$ with the same state space that
has no 0-ECs (\ie end components where the accumulated weight
of all cycles is 0) and that is equivalent to $\cM$ for all
properties that are invariant with respect to behaviors inside 0-ECs.
The generation of such an MDP $\cN$ relies on an iterative technique
to flatten 0-ECs. This new technique, called \emph{spider
  construction}, can be seen as a generalization of the method
proposed in \cite{Alfaro98Thesis,deAlf99} to eliminate 0-ECs in
nonnegative MDPs. There, all states that belong to some
maximal end component of the sub-MDP built by state-action pairs with weight 0
are collapsed. This technique obviously fails for integer-weighted MDPs as 0-ECs can
contain state-action pairs with negative and positive weights.  The
spider construction maintains the state space, but turns the graph
structure of maximal 0-ECs into an acyclic graph with a single sink
state that captures the original behavior of all other states in the
same maximal 0-EC.
Besides deciding weight-divergence, 
the spider construction 
will be the key to solve the classical shortest path problem
for arbitrary integer-weighted MDPs.

Checking the gambling property is \NP-complete in the general case,
but can be decided in polynomial time using the spider construction,
provided that the maximal expected mean payoff is 0. The latter is the
relevant case for solving problems \Easdwr{} and \Uasdwr{} as well as
corresponding problems for weight-bounded B\"uchi conditions.
We establish 
an analogous result for the boundedness property,
shown to be equivalent to the existence of 0-ECs
in cases where the given end component has maximal expected mean payoff 0.

\textbf{Outline.}
Section~\ref{sec:classification} presents 
the classification of end components and
corresponding algorithms.
Our results on
the stochastic shortest path problem
and weight-bounded (repeated) reachability properties
will be presented in Sections~\ref{sec:SSP} and \ref{sec:DWR},
respectively.
For full proofs we refer to the \onlyLong{appendix}\onlyShort{extended version of this paper~\cite{ext}}.

\section{Preliminaries}

\label{sec:prelim}

We briefly define our notations; for details see, \eg\cite{Puterman,BaierKatoen08}.

\begin{definition}[Markov decision processes (MDPs).]
An \emph{MDP} is a tuple
$\cM = (S,\Act,P,\wgt)$ where $S$ is a finite set of states, $\Act$
is a finite set of actions, 
$P\colon S \times \Act \times S \to [0,1]\cap \Rational$
is a probabilistic transition function satisfying 
$\sum_{t \in S} P(s,\alpha,t) \in \{0,1\}$
for all
$(s,\alpha) \in S \times \Act$,
and 
$\wgt\colon S \times \Act \to \Integer$ is a weight function.
\end{definition}
Action $\alpha$ is \emph{enabled} in $s$ 
if $\sum_{t\in S} P(s,\alpha,t)=1$, 
in which case $(s,\alpha)$ is called a \emph{state-action pair} of $\cM$.
$\Act(s)$ denotes the set of actions enabled in $s$.
State $s$ is called a \emph{trap} if $\Act(s)=\varnothing$.

Let $\|\cM\|$ denote the number of state-action pairs in $\cM$.
The \emph{size} of MDP $\cM$ is $\|\cM\|$ plus %
the sum of the logarithmic lengths of the probabilities
and weights in $\cM$.
~\\
\indent A \emph{path} in an MDP $\M = (S,\Act,P,\wgt)$ is an
alternating sequence of states and actions, that can be finite
$\fpath = s_0\, \alpha_0 \, s_1 \, \alpha_1 \, s_2 \, \alpha_2 \ldots
s_n$ or infinite
$\infpath = s_0\, \alpha_0 \, s_1 \, \alpha_1 \, s_2 \, \alpha_2
\ldots$, such that for every index $i$, $\alpha_i \in
\Act(s_i)$ and $P(s_i,\alpha_i,s_{i+1})>0$. 
A path is called \emph{maximal} if it is infinite or ends in a trap.
$\FPaths$, $\InfPaths$ and $\Paths$ denote the set of finite, infinite and
maximal paths, respectively. 
The \emph{weight} of a finite path
$\fpath = s_0\, \alpha_0 \, s_1 \, \alpha_1 \, \ldots \, \alpha_{n-1}\,
s_n$ is $\wgt(\fpath) = \sum_{i=0}^{n-1} \wgt(s_i,\alpha_i)$. For any
path $\fpath = s_0\, \alpha_0 \, s_1 \, \alpha_1 \, s_2 \, \alpha_2 \ldots$,
we write $\prefix{\fpath}{i}$ for its prefix up to state $s_i$.  The
first (resp.~last) state of a finite path $\fpath$ is denoted
$\first(\fpath)$ (resp.~$\last(\fpath)$).  If $\infpath$ is infinite,
$\lim(\infpath)$ is the set of state-action pairs occurring
infinitely often in $\infpath$.

A \emph{scheduler} resolves nondeterminism in MDPs. 
Formally, a scheduler for $\M$ is a partial function
$\sched\colon \FPaths \to \Distr(\Act)$ that maps every finite path 
$\fpath$ where $t=\last(\fpath)$ is not a trap to a
distribution over $\Act(t)$. 
Given a scheduler $\sched$ and a state $s$, the behavior of $\cM$ under 
$\sched$ with starting state $s$ can be formalized by a 
(possibly infinite-state)
Markov chain. $\Pr^{\sched}_{\cM,s}$ denotes the induced probability
measure.
We use standard notions for 
deterministic, memoryless, finite- and infinite-memory schedulers.
Thus, memoryless deterministic (MD) schedulers
can be viewed as functions assigning actions to non-trap states 
and the induced Markov chain is finite.

The analysis of the behaviors in MDPs often relies
on their end components. An \emph{end component} of
$\cM$ %
is a pair $\cE = (T,\ActEC)$ consisting of
a set of states $T \subseteq S$ and a function
$\ActEC\colon T \to 2^\Act$ such that
(1) $\emptyset \neq \ActEC(s) \subseteq \Act(s)$  for each $s \in T$,
(2) $\{t\in S : P(s,\alpha,t)>0\}\subseteq T$
   for each $s \in T$ and $\alpha \in \ActEC(s)$, 
and (3) the sub-MDP
induced by $(T,\ActEC)$ is strongly connected.
We often identify end components 
with their sets of state-action pairs.
That is, if $\cE=(T,\ActEC)$ is as above, we identify
$\cE$ with the set $\{(t,\alpha) : t\in T, \alpha \in \ActEC(t)\}$
and  rely on the fact that for each scheduler
the limit $\lim(\infpath)$ of almost all infinite $\sched$-paths
$\infpath$ constitutes an end component  \cite{Alfaro98Thesis}.
$\cE$ is a \emph{maximal end component} (MEC) if
there is no end component $\cF$ such that $\cE$ is strictly contained
in $\cF$. MECs of an MDP are computable 
in polynomial time \cite{Alfaro98Thesis,ChatHen11}.
All notations introduced
for MDPs can be used for end components, which are themselves strongly
connected MDPs.

\tudparagraph{1ex}{Specifying properties.} 
We use the term \emph{properties} to denote measurable subsets of
$(S \times \Integer)^{\omega} \cup (S \times \Integer)^*\times S$
with respect to the standard cylindrical sigma-algebra.
To reason about probabilities of properties
concerning the measure $\Pr^{\sched}_{\cM,s}$
where $\sched$ is a scheduler and $s$ is a starting state,
every path (state-action sequence) in $\cM$ 
is naturally mapped to a state-integer sequence.
Temporal properties with weight
constraints will be described by LTL-like formulas. 
The atoms of such
formulas are (sets of) states or weight expressions of the form
$\mathord{\accwgt} \bowtie w$ where 
$\mathord{\bowtie} \in \{\leqslant,<,\geqslant,>,=\}$
is a comparison operator and $w\in \Integer$ is a threshold.  Such
formulas are interpreted over path-position pairs. More precisely,
given a path
$\infpath = s_0\, \alpha_0 \, s_1 \, \alpha_1 \, s_2 \, \alpha_2
\ldots$ in $\cM$ and $i\in \Nat$
$ (\infpath,i) \, \models \, \accwgt \bowtie w \text{ iff }
\wgt(\prefix{\infpath}{i}) \bowtie w $, and as usual,
$\infpath \models \varphi$ is a shortcut for
$(\infpath,0)\models \varphi$.  Towards an example, 
let $\goal$ be a state in $\cM$. Then
$\infpath \models \Diamond (\goal \wedge (\accwgt \geqslant w))$ iff
$\infpath$ has a finite prefix $\fpath$ such that
$\last(\fpath)=\goal$ and $\wgt(\fpath) \geqslant w$.

To reason about optimal probabilities of a property $\varphi$, let 
$\Pr^{\sup}_{\cM,s}(\varphi) = 
   \sup_{\sched} \Pr^{\sched}_{\cM,s}(\varphi)$ and 
  $ \Pr^{\inf}_{\cM,s}(\varphi)  = 
   \inf_{\sched} \Pr^{\sched}_{\cM,s}(\varphi)$
where $\sched$ ranges over all schedulers for $\cM$.
We write $\Pr^{\max}_{\cM,s}(\varphi)$ rather than
$\Pr^{\sup}_{\cM,s}(\varphi)$ if the supremum is indeed a maximum,
which is the case, \eg if $\varphi$ is an ordinary LTL formula 
(without weight constraints).  
Note that the maximum/minimum might not exist for
weight-bounded properties. In any case,
$\Pr^{\max}_{\cM,s}(\varphi) =1$ (resp.
$\Pr^{\max}_{\cM,s}(\varphi) >0$) indicates the existence of a
scheduler $\sched$ with $\Pr^{\sched}_{\cM,s}(\varphi) =1$ (resp.
$\Pr^{\sched}_{\cM,s}(\varphi) >0$).

 Given  a random variable $f$,
 $\ExpRew{\sup}{\cM,s}(f) = \sup_{\sched} \ExpRew{\sched}{\cM,s}(f)$
 and
 $\ExpRew{\inf}{\cM,s}(f)= \inf_{\sched} \ExpRew{\sched}{\cM,s}(f)$
 denote the extremal expectations of  $f$,
 where sup and inf take values in
 $\Real \cup \{-\infty,+\infty\}$,
while, for instance, $\ExpRew{\max}{\cM,s}(f)$ 
will be used when the maximum exists.
 In particular, we will use the random variable associated with the
 mean payoff, defined on infinite paths by
 $ \MP(\infpath) = \limsup_{n \to \infty}
 \frac{\wgt(\prefix{\infpath}{n})}{n} $.
Recall that the maximal expected mean payoff in strongly connected
MDPs does not depend on the starting state and that there exist
MD-schedulers with a single \emph{bottom strongly connected component} (BSCC) 
maximizing the expected mean payoff. When $\cM$ is strongly connected, we omit
the starting state and write $\ExpRew{\max}{\cM}(\MP)$.
\section{Classification of End Components}
\setcounter{theorem}{0}
\label{sec:classification}

As basic building blocks of our algorithms, we define four types of schedulers and end components
of MDPs.
The \emph{pumping} end components have a scheduler 
that let the accumulated weight almost surely diverge to infinity;
positively (resp. negatively) \emph{weight-divergent} ones have a scheduler
where almost surely the limsup (resp. liminf) of the accumulated sum 
is infinity (resp. minus infinity);
the \emph{gambling} ones have schedulers with expected mean payoff 0 
and where the accumulated 
weight approaches both plus and 
minus infinity with probability~$1$;
while the \emph{zero end components} only have~$0$ cycles,
so the weight stays bounded with probability~$1$.
\begin{definition}
An infinite path $\infpath$ in an MDP $\cM$ is called\\[.3em]
\begin{tabular}{l@{\hspace*{0.15cm}}l}
 \footnotesize{$\bullet$} &
 \emph{pumping} if
  $\liminf\limits_{n \to \infty} \ \wgt(\prefix{\infpath}{n}) = +\infty$,
  \\[0.5ex]

  \footnotesize{$\bullet$} &
  \emph{positively weight-divergent}, or briefly \emph{weight-divergent},
  \\ 
  &
  if
   $\limsup\limits_{n \to \infty}  \wgt(\prefix{\infpath}{n}) = +\infty$,
  \\ [0.5ex]
  \footnotesize{$\bullet$} &
  \emph{negatively weight-divergent} 
  if
   ${\liminf\limits_{n \to \infty}  \wgt(\prefix{\infpath}{n}) = -\infty}$,
  \\[0.5ex]
  \footnotesize{$\bullet$} &
  \emph{gambling} if $\infpath$ is 
  positively and negatively weight-divergent,
  \\[0.5ex]
  \footnotesize{$\bullet$} &
   \emph{bounded from below}
   if
   $\liminf\limits_{n \to \infty} \wgt(\prefix{\infpath}{n}) \in \Integer$.
\end{tabular}
\end{definition}
\noindent
A scheduler $\sched$ for $\cM$ is called \emph{pumping from state $s$}
if\linebreak
\(\Pr^{\sched}_{\cM,s}\{\infpath \in \InfPaths : \text{$\infpath$ is
  pumping}\}=1\), \ie almost all $\sched$-paths from $s$ are pumping.
$\sched$ is called \emph{pumping} if it is pumping from all states
$s$.  The MDP $\cM$ itself is said to be \emph{pumping} if it has at
least one pumping scheduler.  $\cM$ is called \emph{universally
  pumping} if all schedulers of $\cM$ are pumping.

The notions of weight-divergent (or negatively weight-divergent or
bounded from below) schedulers and MDPs are defined analogously.
\emph{Gambling} schedulers are those where almost all paths are
gambling and where the expected mean payoff is 0.  A strongly
connected MDP $\cM$ is called \emph{gambling} if
$\Exp{\max}{\cM}(\MP)=0$ and $\cM$ has a gambling scheduler
(see Fig.~\ref{fig:gambling-ec}).
Obviously, a strongly connected MDP $\cM$ is pumping (universal
pumping or weight-divergent or gambling, respectively) from some state iff $\cM$ is
pumping (universal pumping or weight-divergent or
gambling, respectively).  %
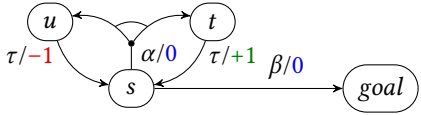
\begin{figure}[ht]
    \centering%

\begin{tikzpicture}
	\node[state] (s) {$s$};
	\node[bullet, above=3mm of s] (b) {};
	\node[state, above right= 0mm and 10mm of b] (t) {$t$};
	\node[state, above left= 0mm and 10mm of b] (u) {$u$};
	\node[state,right=25mm of s] (goal) {$\goal$};
	\path%
		(s) edge[ntran] node[right, pos=0.7]{$\alpha$/\nilw} (b)
		(b) edge[ptran, bend left] coordinate[pos=0.3] (bt) (t)
		(b) edge[ptran, bend right] coordinate[pos=0.3] (bu) (u)
		(t) edge[ptran, bend left] node[right, pos=0.2]{$\tau$/\posw{1}} (s)
		(u) edge[ptran, bend right] node[left, pos=0.2]{$\tau$/\negw{-1}} (s)
		(s) edge[ptran] node[above,pos=0.7]{$\beta$/\nilw} (goal);
	\draw%
		(bt) to[bend right] (bu);
\end{tikzpicture}

     \caption{EC $\cE = \{(s,\alpha), (u,\tau), (t,\tau)\}$ is gambling in
    case all distributions are uniform.
      The MD scheduler that always takes~$(s,\alpha)$ is gambling.
      Moreover,
      $\goal$ can be reached almost surely for any weight threshold, using the
      infinite-memory scheduler that takes~$(s,\alpha)$ if below the threshold,
      and $(s,\beta)$ otherwise. 
      One can show that this cannot be achieved with a finite-memory scheduler.
      \vspace{-0.2cm}
    }
    \label{fig:gambling-ec}
  \end{figure}

A \emph{zero end component} (0-EC) is an end component
$\cE$ where $\wgt(\cycle)=0$ for each cycle $\cycle$ in $\cE$ and 
use the term \emph{0-BSCC} when $\cE$ contains at most one state-action pair
$(s,\alpha)$ for each state $s$ in $\cE$. Thus, each 0-BSCC is a bottom strongly connected
component of an MD-scheduler.  A cycle $\cycle$ in $\cM$ is called
positive if $\wgt(\cycle)>0$, and negative if $\wgt(\cycle)<0$.

Recall characterizations of these notions for Markov chains:\vspace*{-1.4em}
\begin{lemma}[Folklore -- see, \eg \cite{KSBD15}] %
 \label{weight-div-MC}
 \label{gambling-MC}
  Let $\cC$ be a strongly connected finite Markov chain.
  \begin{enumerate}
  \item [(a)] 
     $\cC$ is pumping iff\ \ $\Exp{}{\cC}(\MP)>0$.
  \item [(b)]
     $\ExpRew{}{\cC}(\MP)=0$ iff  $\cC$ is a 0-BSCC
     or $\cC$ is gambling.
  \item [(c)]
     If\ \ $\ExpRew{}{\cC}(\MP)=0$
     then the following statements are equivalent:
  \begin{inparaenum}
   \item [(1)]  
      $\cC$ is gambling,
   \item [(2)]  
      $\cC$ is positively weight-divergent,
   \item [(3)]  
      $\cC$ is negatively weight-divergent,
   \item [(4)]
      $\cC$ has a positive cycle,
   \item [(5)]
      $\cC$ has a negative cycle.
  \end{inparaenum}
  \item [(d)]
     If\ \ $\ExpRew{}{\cC}(\MP)=0$
     then the following are equivalent:
  \begin{inparaenum}
   \item [(1)]  
      $\cC$ is a 0-BSCC,
   \item [(2)]
      $\cC$ is bounded from below,
   \item [(3)]
      the set of paths bounded from below has positive measure.
   \end{inparaenum}
  \end{enumerate}
\end{lemma}

The goal of this section is to provide an analogous characterization
for strongly connected MDPs and efficient algorithms to decide whether
an MDP is of a given type. 

This is simple for the existential and universal 
pumping property, checkable in polynomial time
\onlyLong{(see Lemmas~\ref{appendix:pumping-ecs} 
and~\ref{universally-pumping} in the appendix)}:

\begin{restatable}{lemma}{lempumpingec}%
  \label{lemma:pumping-ecs}
  Let $\cM$ be a strongly connected MDP.
  Then, $\cM$ is pumping iff
  $\cM$ has a pumping MD-scheduler
  iff\ \ $\ExpRew{\max}{\cM}(\MP)>0$.
  Likewise, $\cM$ is universally pumping iff
  all MD-schedulers are pumping iff\ \
  $\ExpRew{\min}{\cM}(\MP)>0$. 
\end{restatable}

The remainder of this section addresses the tasks to
check weight-divergence, the gambling property and the 
computation of all states belonging to a 0-EC.%
\footnote{%
  We focus here on results for (positive) weight-divergence.
The negative case can be obtained analogously
  by multiplying all weights with $-1$.}
We start with an observation on weight-divergence
\onlyLong{(see App.~\ref{appendix:pumping} for a proof)\vspace*{-1.4em}}:
\begin{restatable}{lemma}{poswgtdiv}%
\label{main:mp-pos-wgt-div}%
\label{wgt-div-implies-non-neg-mp} \
Let $\cM$ be a strongly connected MDP. %
 If $\cM$ is positively weight-divergent 
    then $\ExpRew{\max}{\cM}(\MP)\geqslant 0$. Conversely,    
 if\ \ $\ExpRew{\max}{\cM}(\MP)>0$, then $\cM$
    is positively weight-divergent.
\end{restatable}

\subsection{Spider Construction for Flattening 0-ECs}
\begin{figure*}[hbtp]
	\begin{minipage}{.32\textwidth}
\pgfdeclarelayer{background}%
\pgfsetlayers{background,main}%
\begin{tikzpicture}[x=11mm,y=10mm,font=\scriptsize]
	\node (s) at (0,0) [state] {$s$};
	\node (stu) at (1,0) [bullet] {};
	\node (t) at (2,0) [state] {$t$};
	\node (tsrw) at (2,.7) [bullet] {};
	\node (twv) at (3,0) [bullet] {};
	\node (w) at (4,0) [state] {$w$};
	\node (r) at (4,.7) [state] {$r$};
	\node (u) at (2,-1) [state] {$u$};
	\node (v) at (3,-1) [state] {$v$};
	
    \draw [ntran] (s) -- node[above=-.05,pos=.5]{$\alpha$/\posw{1}} (stu);
    \draw [ptran] (stu) -- coordinate[pos=.4] (bt) (t);
    \draw [ptran] (stu) -- coordinate[pos=.25] (bu) (u);
    \draw (bt) to[bend left] (bu);
    \draw [ptran] (t) -- node[right,pos=.4]{$\alpha$/\nilw} (u);
    \draw [ntran] (t) -- node[above=-.05,pos=.5]{$\beta$/\posw{2}} (twv);
    \draw [ptran] (twv) -- coordinate[pos=.3] (bww) (w);
    \draw [ptran] (twv) -- coordinate[pos=.33] (bvv) (v);
    \draw (bvv) to[bend right=40] (bww);
    \draw [ptran] (v) -- node[below,pos=.5]{$\beta$/\negw{2}} (u);
    \draw [ptran] (u) -- node[below,pos=.4]{$\alpha$/\negw{1}} +(-1.6,0) -- (s);
    \draw [ntran] (t) -- node[left,pos=.4]{$\gamma$/\posw{4}} (tsrw);
    \draw [ptran] (tsrw) -- coordinate[pos=.1] (bs) +(-2,0) -- (s);
    \draw [ptran] (tsrw) -- coordinate[pos=.12] (bw) (w);
    \draw [ptran] (tsrw) -- (r);
    \draw (bs) to[bend left=80] (bw);
    \draw [ptran] (w) -- 
    	node[left,pos=.5]{$\beta$/\negw{3}} +(0,-1.5) -- (0,-1.5) -- (s);
    \begin{pgfonlayer}{background}%
    	\node[above left=.3 and 0 of s]{$\mathcal{M}\colon$};
        \draw[rounded corners=1em,line width=2em,black!10,fill=black!10]
        	(-.2,.05) -- (2.15, .05) -- (2.15,-1.05) -- (.6, -1.05) --
			node[right,pos=.6,white]{$\mathcal{E}$} cycle;
    \end{pgfonlayer}
\end{tikzpicture}

 	\end{minipage}
	\begin{minipage}{.32\textwidth}
		\vspace*{-.05em}
\pgfdeclarelayer{background}%
\pgfsetlayers{background,main}%
\begin{tikzpicture}[x=11mm,y=10mm,font=\scriptsize]
	\node (s) at (0,0) [state] {$s$};
	\node (t) at (2,0) [state] {$t$};
	\node (tsrw) at (2,.7) [bullet] {};
	\node (twv) at (3,0) [bullet] {};
	\node (w) at (4,0) [state] {$w$};
	\node (r) at (4,.7) [state] {$r$};
	\node (u) at (2,-1) [state] {$u$};
	\node (v) at (3,-1) [state] {$v$};
	
    \draw [ptran] (s) -- node[above=-.05,pos=.4]{$\tau$/\posw{1}} (t);
    \draw [ptran] (u) -- node[left,pos=.4]{$\tau$/\nilw} (t);
    \draw [ntran] (t) -- node[above=-.05,pos=.5]{$\beta$/\posw{2}} (twv);
    \draw [ptran] (twv) -- coordinate[pos=.3] (bww) (w);
    \draw [ptran] (twv) -- coordinate[pos=.33] (bvv) (v);
    \draw (bvv) to[bend right=40] (bww);
    \draw [ptran] (v) -- node[below,pos=.5]{$\beta$/\negw{2}} (u);
    \draw [ntran] (t) -- node[left,pos=.4]{$\gamma$/\posw{4}} (tsrw);
    \draw [ptran] (tsrw) -- coordinate[pos=.1] (bs) +(-2,0) -- (s);
    \draw [ptran] (tsrw) -- coordinate[pos=.12] (bw) (w);
    \draw [ptran] (tsrw) -- (r);
    \draw (bs) to[bend left=80] (bw);
    \draw [ptran] (w) -- 
    	node[left,pos=.5]{$\beta$/\negw{3}} +(0,-1.5) -- (0,-1.5) -- (s);
    \begin{pgfonlayer}{background}%
    	\node[above left=.3 and 0 of s]{$\mathcal{M}_1\colon$};
        \draw[rounded corners=1em,line width=.8em,black!10,fill=black!10]
        	(3.8,.25) -- (4.4,.1) --
			(4.1, -1.5) --  node[above,pos=.75,white]{$\mathcal{F}$}
			(-.1, -1.5) -- (-.3,.2) -- cycle;
    \end{pgfonlayer}
\end{tikzpicture}

 	\end{minipage}
	\begin{minipage}{.32\textwidth}
		\vspace*{-.76em}
\pgfdeclarelayer{background}%
\pgfsetlayers{background,main}%
\begin{tikzpicture}[x=11mm,y=10mm,font=\scriptsize]
	\node (s) at (0,0) [state] {$s$};
	\node (t) at (2,0) [state] {$t$};
	\node (tsrw) at (2,.7) [bullet] {};
	\node (w) at (4,0) [state] {$w$};
	\node (r) at (4,.7) [state] {$r$};
	\node (u) at (2,-1) [state] {$u$};
	\node (v) at (3,-1) [state] {$v$};
    \draw [ptran] (s) -- node[above=-.05,pos=.4]{$\tau$/\posw{1}} (t);
    \draw [ptran] (u) -- node[left,pos=.4]{$\tau$/\nilw} (t);
    \draw [ntran] (t) -- node[left,pos=.4]{$\gamma$/\posw{4}} (tsrw);
    \draw [ptran] (w) -- node[above=-.05,pos=.7]{$\tau$/\negw{2}} (t);
    \draw [ptran] (v) -- node[right,pos=.5]{$\tau$/\negw{2}} (t);
    \draw [ptran] (tsrw) -- coordinate[pos=.1] (bs) +(-2,0) -- (s);
    \draw [ptran] (tsrw) -- coordinate[pos=.12] (bw) (w);
    \draw [ptran] (tsrw) -- (r);
    \draw (bs) to[bend left=80] (bw);
    \node[above left=.3 and 0 of s]{$\mathcal{M}_2\colon$};
\end{tikzpicture}
 	\end{minipage}
	\caption{\label{fig:example-spider}Illustration of the spider construction:
$\cM_1=\spider{\cM}{\cE,t}$ and $\cM_2=\spider{\cM_1}{\cF,t}$. %
}
  \end{figure*}
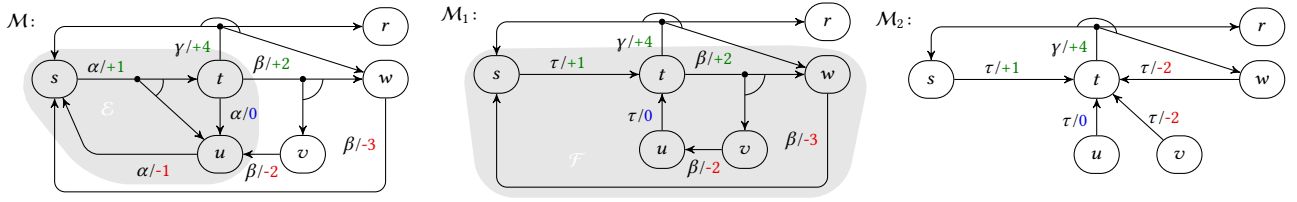
  
\label{sec:spider-construction}

In this section, we present a method to eliminate a given 0-EC from an MDP 
by ``flattening'' it, crucial for our algorithms. 
This so-called \emph{spider construction} preserves the state space and
all properties of interest, in particular, those that are invariant 
by adding or removing path segments of weight~$0$.
It will be used for checking weight-divergence
(Section \ref{sec:wgt-div}) and for the stochastic 
shortest path algorithm (Section \ref{sec:min-exp-acc-weight}).
  
Let $\cM$ be an MDP and $\cE$ a 0-BSCC of $\cM$, \ie
for each state $s$ in $\cE$ there is a unique action
$\alpha_s\in \Act(s)$ such that $(s,\alpha_s)\in \cE$.
The spider construction for $\cM$ and $\cE$ works as follows.
As $\cE$ is a 0-EC, all paths in $\cE$ from $s$ to some state
$t$ in $\cE$ have the same weight, say $w(s,t)$.
Note that then each path from $t$ to $s$ has weight $w(t,s)=-w(s,t)$.

\begin{definition}%
  Let $\cM$ be an MDP, $\cE$ a $0$-BSCC of $\cM$, and $s_0$ a
  reference state in $\cE$.  The \emph{spider MDP}
  $\cN=\spider{\cM}{\cE,s_0}$ (or shortly $\spider{\cM}{\cE}$) results
  from $\cM$ by
  \\[.3em]
 (i) removing the
  state-action pairs $(s,\alpha_s)$ for all states $s$ in $\cE$;
  \\[0.5ex]
 (ii)
  adding state-action pairs $(s,\tau)$ for each state $s$ in $\cE$
  with $s\not= s_0$ where $P_{\cN}(s,\tau,s_0)=1$ and
  $\wgt_{\cN}(s,\tau)=w(s,s_0)$; and\\[0.5ex]
 (iii)  for each state $s\neq s_0$ in $\cE$ and action
  $\beta \in \Act_{\cM}(s){\setminus}\{\alpha_s\}$, replacing
  $(s,\beta)$ with $(s_0,\beta)$ s.t.
  $P_{\cN}(s_0,\beta,u) = P_{\cM}(s,\beta,u)$ for all states $u$ in
  $\cM$ and $\wgt_{\cN}(s_0,\beta)=w(s_0,s)+\wgt_{\cM}(s,\beta)$.
\end{definition}
\begin{example}
	\label{example:spider}
	We exemplify the spider construction in 
	Figure~\ref{fig:example-spider}: Starting with an
  MDP $\cM$, we apply the spider construction twice, 
  each with reference state $s_0=t$. First,
  the 0-BSCC $\cE = \{(s,\alpha), (t,\alpha), (u,\alpha)\}$ of $\cM$
  is chosen, obtaining $\cM_1=\spider{\cM}{\cE,t}$.
Then, taking the 0-BSCC $\cF=\{(s,\tau), (t,\beta), (u,\tau), (v,\beta),(w,\beta)\}$ of $\cM_1$,
we obtain $\cM_2=\spider{\cM_1}{\cF,t}$.
In each step, the chosen 0-EC turns into a sub-MDP where the reference
state is the only sink.
\Ende
\end{example}

To formally state the equivalence of $\cM$ and $\spider{\cM}{\cE}$,
we define the notion of \emph{$\cE$-invariant properties}.
Given a path 
$\infpath = t_0 \, \alpha_0 \, t_1 \, \alpha_1\ldots$, let
$\Purge{\cE}{\infpath}$ $\in$ 
$(S \times \Integer)^\omega \cup (S \times \Integer)^* \times S$
be obtained from $\infpath$ by 
(1) replacing each fragment
$t_i \, \alpha_i\, \ldots \, \alpha_j \, t_{j+1}$
of $\infpath$ such that
(a) either $i=0$ or $(t_{i-1},\alpha_{i-1})\notin \cE$,
(b) $(t_j,\alpha_j)\notin \cE$, and
(c) $(t_{\ell},\alpha_{\ell})\in \cE$ for $\ell=i,i{+}1,\ldots,j{-}1$
with $t_i \, w \, t_{j+1}$ where $w=w(t_i,t_j) +\wgt(t_j,\alpha_j)$ 
and 
(2) replacing each action $\alpha_i$ 
in the resulting sequence with $\wgt(t_i,\alpha_i)$.
A property $\varphi$ is called $\cE$-invariant if
for all maximal paths $\infpath$ we have:
\begin{inparaenum}
\item [(I1)] 
  if $\infpath$ has an infinite suffix of state-action pairs in
  $\cE$, then $\infpath \not\models \varphi$ and
\item [(I2)] 
  if $\infpath \models \varphi$ and
  $\infpath'$ is a maximal path with 
  $\Purge{\cE}{\infpath} = \Purge{\cE}{\infpath'}$ then
  $\infpath' \models \varphi$.
\end{inparaenum}
Weight-divergence and the pumping property
are $\cE$-invariant properties, and so are properties 
of the form $\Diamond (t \wedge (\accwgt \bowtie K))$
where  $t$ is a trap,
$\bowtie$ a comparison operator (\eg $=$ or $\geqslant$) 
and $K\in \Integer$.

\begin{lemma}%
  \label{spider-construction}
  The spider construction generates an
  MDP \linebreak $\spider{\cM}{\cE}$ that satisfies the following properties:
  \begin{enumerate}
  \item [(S1)]
     $\cM$ and $\spider{\cM}{\cE}$ have the same state space
     and $\|\spider{\cM}{\cE}\|=\|\cM\|{-}1$.
  \item [(S2)]
     If\ \ $\cE\not=\cM$ and $\cM$ is strongly connected then 
     $\spider{\cM}{\cE}$ 
     has a single MEC that is reachable from all states.
  \item [(S3)] 
     $\cM$ and $\spider{\cM}{\cE}$ are equivalent for 
     $\cE$-invariant properties in the following sense:
     \begin{enumerate}
     \item [(S3.1)]
        For each scheduler $\tsched$ for $\spider{\cM}{\cE}$ there is a 
        scheduler $\sched$ for $\cM$ with
        $\Pr^{\sched}_{\cM,s}(\varphi)=
         \Pr^{\tsched}_{\spider{\cM}{\cE},s}(\varphi)$
        for all states $s$ and all 
        $\cE$-invariant properties $\varphi$.
        If\ \ $\tsched$ is MD, then $\sched$ can be chosen MD.
     \item [(S3.2)]
        For each scheduler $\sched$ for $\cM$ 
        there exists a 
        scheduler $\tsched$ for $\spider{\cM}{\cE}$ such that
        \begin{center}
         $\Pr^{\sched}_{\cM,s}(\varphi) \leqslant
          \Pr^{\tsched}_{\spider{\cM}{\cE},s}(\varphi) \leqslant
          \Pr^{\sched}_{\cM,s}(\varphi) + p^{\sched}_{s}$
        \end{center}
        for all states $s$ and all 
        $\cE$-invariant properties $\varphi$.
       Here, $p^{\sched}_s=
         \Pr^{\sched}_{\cM,s}
           \{\infpath \in \InfPaths : \lim(\infpath)=\cE\}$.
     \end{enumerate}
  \item [(S4)]
     Suppose that $\cE$ is contained in an MEC
     $\cG$ of $\cM$ with $\Exp{\max}{\cG}(\MP)=0$. Then
     for each state $s$ with $s\notin \cE$:
     $s$ belongs to a 0-EC of $\cM$ iff\ \
     $s$ belongs to a 0-EC of $\spider{\cM}{\cE}$.
     Likewise, for each state-action pair $(s,\alpha)$ of $\cM$:
     $(s,\alpha)$ belongs to a 0-EC of $\cM$ iff
     $(s,\alpha)\in \cE$ or
     $(s_0,\alpha)$ belongs to a 0-EC of $\spider{\cM}{\cE}$.
  \end{enumerate}
\end{lemma}
\onlyLong{The proof is given in Appendix~\ref{sec:prop-spider}.}
The main property of the spider construction 
is that it eliminates the given 0-BSCC while maintaining all other 0-EC, 
as stated in (S4). 
(S3) states an equivalence between
$\cM$ and $\spider{\cM}{\cE}$ with respect to $\cE$-invariant properties. 
While any scheduler for $\spider{\cM}{\cE}$ 
can be transformed to an equivalent scheduler for~$\cM$ (case (S3.1)), 
the converse direction (case (S3.2))
is more involved and requires restrictions, which are, however,
sufficient for our applications.

As a consequence of the equivalence stated in (S3) we obtain
that weight-divergent and pumping end components are preserved by the spider construction:

\begin{corollary}%
 \label{spider-preserves-wgt-div}
  If $\cM$ is strongly connected and $\cE$ is a 0-BSCC of $\cM$ then
  $\cM$ is weight-divergent (resp.~pumping) 
  iff\ \ $\spider{\cM}{\cE}$ is weight-divergent (resp.~pumping).
\end{corollary}

\subsection{Checking Weight-Divergence}
\label{sec:wgt-div}
We present an algorithm to check the weight-divergence 
of an end component~(see Algorithm~\ref{alg:check-wd}). 
Such end components will be useful, \eg
when solving weight-bounded reachability problems that require the accumulated weight
to be above a threshold.
\begin{algorithm}[h]
	\small
	\SetAlgoLined
	\DontPrintSemicolon
	\SetKwInOut{Input}{input}\SetKwInOut{Output}{output}
	\SetKwData{n}{n}\SetKwData{f}{f}\SetKwData{g}{g}
	\SetKwData{Low}{l}\SetKwData{x}{x}
	\Input{strongly connected MDP $\cM$} 
	\Output{``yes'' if $\cM$ is weight divergent and ``no'' otherwise}
	\BlankLine
	Compute $e:=\Exp{\max}{\cM}(\MP)$ and $\sched$ with $\Exp{\sched}{\cM}(\MP)=e$\;
	\lIf{$e<0$}{\Return ``no''}
	\lIf{$e>0$ or $\sched$ has a gambling BSCC}{\Return ``yes''}
	Pick a 0-BSCC $\cE$ of $\sched$\;
	\lIf{$\cM=\cE$}{\Return ``no''}
	Compute the MEC $\cF$ of $\spider{\cM}{\cE}$ that is reachable from all states and
	\Return $\textsf{Wgtdiv}(\cF)$
	\caption{$\textsf{Wgtdiv}(\cdot)$}
	\label{alg:check-wd}
\end{algorithm}
Given a strongly connected MDP $\cM$ 
we first compute $\Exp{\max}{\cM}(\MP)$ and  
an MD-scheduler $\sched$ maximizing the expected mean payoff. 
If $\Exp{\max}{\cM}(\MP)>0$ then $\cM$ is pumping
(Lemma \ref{lemma:pumping-ecs})
and therefore positively weight-divergent.
If $\Exp{\max}{\cM}(\MP)<0$ then  all schedulers for $\cM$ are
negatively weight-divergent
(Lemma \ref{lemma:pumping-ecs}
with weights multiplied by $-1$), 
and hence, $\cM$ is not positively weight-divergent.
If $\Exp{\max}{\cM}(\MP) = 0$ and $\sched$ has a gambling BSCC then
$\cM$ is gambling and therefore positively weight-divergent.
Otherwise, each BSCC of the Markov chain induced by $\sched$ is a 0-BSCC
(Lemma \ref{gambling-MC}) and we pick such a 0-BSCC $\cE$ of $\sched$. 
In case $\cM=\cE$ then $\cM$ is a 0-EC, hence not weight-divergent,
and the algorithm terminates.
If $\cM\not=\cE$, we apply the spider construction to generate
the MDP $\spider{\cM}{\cE}$ that contains a unique maximal end component
$\cF$ ((S2) in Lemma \ref{spider-construction}).
Repeating the procedure recursively on $\cF$ etc. thus generates
a sequence of MDPs $\cM_0=\cM$, $\cM_1,\ldots,\cM_{\ell}$
with $\cM_{i+1}=\spider{\cM_i}{\cE_i}$
for some 0-BSCC $\cE_i$ of $\cM_i$.
All $\cM_i$'s have the same state space and
the number of state-action pairs is strictly decreasing, \ie
we have $\|\cM_0\| > \|\cM_1\| > \ldots > \|\cM_{\ell}\|$
by property (S1) in Lemma \ref{spider-construction}.
Moreover, $\cM_i$ is weight-divergent iff 
$\cM$ is  weight-divergent
(Corollary \ref{spider-preserves-wgt-div}).

As each iteration takes polynomial time
and the size of each $\cM_i$ is polynomially bounded by the size of $\cM$
\onlyLong{~(see Lemma \ref{size-iterative-spider-construction})},
the algorithm runs in polynomial time.
Using an inductive argument and Lemma \ref{spider-construction}
\onlyLong{(see Appendix \ref{appendix:wgt-div})}, we obtain:

\begin{theorem}%
 \label{weight-divergence-algorithm}
  The algorithm for checking weight-divergence of a strongly connected
  MDP $\cM$ runs in polynomial time. %
  If $\cM$ is weight-divergent then it either finds a pumping or
  a gambling MD-scheduler.
  If $\cM$ is not weight-divergent, then 
  it generates  an
  MDP $\cN$ without 0-ECs
  on the same state space as~$\cM$,
  and is equivalent to~$\cM$ w.r.t. all  properties that are 
  $\cE$-invariant for
  all 0-ECs $\cE$ of $\cM$ in the sense of (S3) in
  Lemma~\ref{spider-construction}.
\end{theorem}

Observe the following consequence of this theorem:

\begin{corollary}
  Let $\cM$ be a strongly connected MDP with 
  $\Exp{\max}{\cM}(\MP)=0$.
  Then, $\cM$ is weight-divergent iff $\cM$ is gambling  iff 
  $\cM$ has a gambling MD-scheduler.
\end{corollary}

However, an MDP can have gambling schedulers, but no 
gambling MD-scheduler:
Consider the MDP with state-action
pairs $(s,\alpha)$, $(s,\beta)$,
where $P(s,\alpha,s) = 1, P(s,\beta,s) = 1$,
$\wgt(s,\alpha)=-\wgt(s,\beta) = 1$.
Then, $\Exp{\max}{\cM}(\MP)=+\infty$ and there is no gambling MD-scheduler, 
while the randomized memoryless scheduler $\sched$ 
with~$\sched(s)(\alpha)=\sched(s)(\beta)=\frac{1}{2}$ 
is gambling.

Given a strongly connected MDP $\cM$ with $\ExpRew{\max}{\cM}(\MP)=0$, 
$\cM$ is gambling 
iff $\cM$ is weight-divergent.
Thus, the gambling property for strongly connected MDPs with maximal expected
mean payoff 0 can be checked in polynomial time using Theorem~\ref{weight-divergence-algorithm},
which yields part (a) of the next theorem. \onlyLong{For the remaining part of the proof, see Appendix~\ref{sec:checking-gambling}.}

\begin{restatable}{theorem}{gambling}
  \label{thm:checking-gambling}
  Given a strongly connected MDP $\cM$,
  the existence
  of a gambling MD-scheduler
  is (a) decidable in polynomial time if $\ExpRew{\max}{\cM}(\MP)=0$,
  and (b) \NP-complete in general.
\end{restatable}

One can compute an MD-scheduler in polynomial time that maximizes the probability
of weight-divergence. In fact, one can compute weight-divergent MECs 
(and corresponding weight-divergent MD-schedulers)
and  maximize the probability of reaching one of these components.
Likewise,
the minimal probability of weight-divergence equals the 
maximal probability to reach the set $V$ of states of all
trap states and all states belonging to an
MEC $\cE$ where either $\Exp{\min}{\cE}(\MP)<0$ or
$\Exp{\min}{\cE}(\MP)=0$ and $\cE$ has a 0-EC.
Theorem~\ref{thm:checking-zero-EC} below shows
that set $V$ is computable in polynomial time.
This yields a polynomial-time algorithm for finding an MD-scheduler minimizing
the weight-divergence probability.
Previous work established the polynomial-time computability
of maximal weight-divergence probabilities 
in special cases. In fact,
\cite[Theorem 3.1]{BBEKW10} 
presents an algorithm to compute
an MD-scheduler maximizing the probability for weight-divergent paths
in a given MDP where the weights belong to~$\{-1,0,1\}$.
Thus, \cite{BBEKW10} yields a \emph{pseudo-polynomial} time bound for
deciding weight-divergence or computing the maximal
weight-divergence probabilities in MDPs with integer weights.
Theorem \ref{weight-divergence-algorithm} and the previous paragraph
improve this result by establishing a polynomial time bound.
Moreover, our algorithm is different; while \cite{BBEKW10} uses
transformations to incorporate accumulated weights in the state space
(up to some threshold), our algorithm uses the spider construction
and maintains the state space.

\subsection{Reasoning about 0-ECs}

\label{sec:0-ECs}

We are now interested in checking the existence of 0-ECs
and computing all state-action pairs inside some 0-EC, useful, \eg 
to deal with weight-bounded constraints (see Section \ref{sec:DWR}).

In MDPs without weight-divergent end components,
the weight-divergence algorithm can be used to determine all state-action
pairs belonging to a 0-EC in polynomial time.
However, this does not work in general
as the algorithm stops as soon as a
weight-divergent end component is found.

To check whether a given 
strongly connected MDP $\cM$ with $\ExpRew{\max}{\cM}(\MP)=0$
contains a 0-EC, we use an iterative approach:
we apply standard algorithms to compute an MD-scheduler
$\sched$ with a single BSCC $\cB$ maximizing the expected mean payoff
(in particular, $\Exp{}{\cB}(\MP)=0$) and checks whether $\cB$ is a 0-BSCC.
If yes, $\cB$ is a 0-EC of $\cM$. 
Otherwise, $\cB$ is gambling (see Lemma \ref{gambling-MC}). 
In this case, we give a transformation that modifies the transition probabilities in $\cB$ 
to obtain an MDP $\cM'$ with the same structure as $\cM$
(in particular, with the same 0-ECs) 
such that $\cM'$ has fewer gambling MD-schedulers than $\cM$.
Thus, if $\ExpRew{\max}{\cM'}(\MP)<0$ then $\cM$ has no 0-EC.
Otherwise, we repeat the procedure on $\cM'$.

This transformation is crucial in several results that follow. 
\onlyLong{Detailed construction
and the proof of the following theorem are given in Appendix 
Sections~\ref{sec:algo-checking-0-EC} and \ref{sec:complexity-existence-0-ECs}.}

\begin{restatable}{theorem}{thmcompZec}%
  \label{thm:checking-zero-EC}
  Given a strongly connected MDP $\cM$,
  the existence of 0-ECs 
  is (a) decidable in polynomial time if\ \ $\ExpRew{\max}{\cM}(\MP)=0$,
  and (b) \NP-complete in the general case.
\end{restatable}

Combining the above decision algorithm and the iterative elimination 
of~0-ECs,
we can also compute the set of all 0-ECs in polynomial time.
An important notion in our algorithms is the \emph{recurrence value} 
defined as follows.
For a state $s$ of a 0-EC in a strongly connected
MDP $\cM$ with $\Exp{\max}{\cM}(\MP)=0$,
$\rec(s)$ is the maximal integer $K$
s.t.
$\Pr^{\sched}_{\cM,s}\big(\Box (\accwgt \geqslant K) \wedge \Box \Diamond s\big)=1$
for some $\sched$ that only uses
actions belonging to some 0-EC.
In fact, to ensure that the accumulated weight stays above~0, 
it does not suffice
to enter a 0-EC with nonnegative weight,
as 0-ECs can contain state-action pairs with negative weight.

\begin{lemma}%
\label{mincredit-ZeroEC}
If $\cM$ is strongly connected and
$\ExpRew{\max}{\cM}(\MP)=0$ then the set $\ZeroEC$ consisting of all
states $s$ that belong to some $0$-EC,
as well as the recurrence values $\rec(s)$
for the states $s \in \ZeroEC$
are computable in polynomial time.
\end{lemma}

\subsection{Universal Negative Weight-Divergence and Boundedness}

\label{sec:uni-neg-wgt-div}

We now show how to determine end components that are bounded from below
and those that are universally negatively weight-divergent.
Part (a) of the following theorem 
\onlyLong{(see Appendix~\ref{sec:boundedness})}
is the MDP-analogue of
part (d) of Lemma~\ref{gambling-MC}.

\begin{restatable}{theorem}{thmnegwdZec}%
\label{universal-neg-wgt-div}
  Let $\cM$ be a strongly connected MDP with
 $\Exp{\max}{\cM}(\MP)=0$.
  Then,
  \begin{inparaenum}
  \item [(a)]
    $\cM$ contains a 0-EC iff
    $\cM$ has a scheduler where the measure of infinite paths that are
    bounded from below is positive
    iff
    $\cM$ has a scheduler that is bounded from below;
  \item [(b)]
  $\cM$ has no 0-EC iff 
    each scheduler for $\cM$ is negatively weight-divergent.
  \end{inparaenum}
\end{restatable}

Given a  strongly connected 
MDP $\cM$, universal (positive) weight-divergence of $\cM$
can be checked in polynomial time.
In fact, if $\Exp{\min}{\cE}(\MP)>0$, then 
$\cM$ is universally weight-divergent, 
and if $\Exp{\min}{\cE}(\MP)<0$, it is not.
If $\ExpRew{\max}{\cM}(\MP)=0$, we use Theorem \ref{universal-neg-wgt-div}
(by multiplying the weights by $-1$) and
check the nonexistence of 0-ECs 
by Theorem \ref{thm:checking-zero-EC}.
We get:

\begin{corollary}
  \label{checking-universal-wgt-div}
  Universal (positive) weight-divergence of an MDP
  can be checked in polynomial time.
\end{corollary}

\begin{remark}%
\label{remark:computing-bounded-from-below}
{\rm
The set of states $s$ of an arbitrary
 MDP $\cM$ that belongs to an end component
bounded from below can be computed in polynomial time as follows.
We first determine the MECs of $\cM$
and their maximal expected mean payoff.
MECs $\cE$ with $\Exp{\max}{\cE}(\MP)>0$ are pumping and
therefore bounded from below.
MECs $\cE$ with either $\Exp{\max}{\cE}(\MP)<0$ or
$\Exp{\max}{\cE}(\MP)=0$ and $\cE$  has no 0-EC are
universally negatively weight-divergent 
(Theorem \ref{universal-neg-wgt-div}).
Hence, none of their states
belongs to an end component that is bounded from below.
Otherwise, \ie if $\Exp{\max}{\cE}(\MP)=0$ and $\cE$ has 0-ECs,
we compute the maximal 0-ECs 
using the techniques presented in Section \ref{sec:0-ECs}
(see Lemma \ref{mincredit-ZeroEC}).
 }
\end{remark}

\section{Stochastic Shortest Paths}
\setcounter{theorem}{0}
\label{sec:SSP}
\label{sec:min-exp-acc-weight}
We present an algorithm to solve the stochastic shortest path problem
that relies on the classification of end components presented above.
The classical shortest path problem for MDPs
is to compute the \emph{minimal expected accumulated weight} until
reaching a goal state $\goal$. Here, the infimum is taken over all
\emph{proper} schedulers. These are schedulers $\sched$ that reach
$\goal$ almost surely, \ie
$\Pr^{\sched}_{\cM,s}(\Diamond \goal)=1$
for all states $s\in S$.

We assume, w.l.o.g., that $\goal$ is a trap, and that all states $s$
are reachable from an initial state $\sinit$ and can reach $\goal$. We write
$\accdiaplus \goal$ for the random variable that represents the
accumulated weight until reaching $\goal$: it assigns to each path
reaching $\goal$ its accumulated weight, and is undefined otherwise.
Formally, $(\accdiaplus \goal)(\infpath) = \wgt(\infpath)$
if $\infpath \models \Diamond \goal$ 
and undefined if
$\infpath \not\models \Diamond \goal$. 
The \emph{stochastic
  shortest path problem} aims at computing the minimal expected
accumulated weight until reaching $\goal$:
\[
  \Exp{\inf}{\M,\sinit}(\accdiaplus \goal)\enskip = \enskip\inf\nolimits_{\sched
    \textrm{ proper}} \Expected^\sched_{\M,\sinit}(\accdiaplus \goal)
  \enspace.
\]
Although for each proper scheduler this quantity is finite, the infimum may be $-\infty$. We
describe a polynomial-time algorithm to check
whether $\Exp{\inf}{\M,\sinit}(\accdiaplus
\goal)$ is finite and to compute it, both
using our classification of end components.

It is well known (see, \eg \cite{Kallenberg})
that if $\M$ is \emph{contracting}, \ie if all schedulers are
proper, then
$\Exp{\inf}{\M,\sinit}(\accdiaplus \goal)  > -\infty$ and one
can compute $\Exp{\inf}{\M,\sinit}(\accdiaplus \goal)$ using
linear-programming techniques. To relax the assumption of
$\M$ being contracting, Bertsekas and Tsitsiklis \cite{BerTsi91} 
identified conditions
that guarantee the finiteness of the values
$\Exp{\inf}{\M,\sinit}(\accdiaplus \goal)$, the
existence of a minimizing MD-scheduler, and the computability
of the vector 
$(\Expected^{\inf}_{\M,s}(\accdiaplus \goal))_{s \in S}$ 
as the unique solution of a linear program 
(or using value and policy iteration).
The assumptions of \cite{BerTsi91},
written $\mathsf{(BT)}$ in the sequel, are: (i) existence of a
proper scheduler, and (ii) under each non-proper
scheduler the expected accumulated weight is $+\infty$ from at least one state. 
While these assumptions are sound, they are incomplete in the sense
that there are MDPs where
$\Exp{\inf}{\M,s}(\accdiaplus \goal)$ is finite for all states
$s$, but $\mathsf{(BT)}$ does not hold.

Orthogonally, De Alfaro~\cite{deAlf99} showed that in MDPs where the
weights are either all nonnegative or all nonpositive, one can
decide in polynomial time whether
$\Exp{\inf}{\M,\sinit}(\accdiaplus \goal)$ is finite. Moreover,
when this is the case, $\M$ can be transformed into another MDP
that has proper schedulers, 
satisfies $\mathsf{(BT)}$ and preserves the minimal expected
accumulated weight.
Using the classification of end components,
we generalize De Alfaro's
result and provide a characterization of finiteness of the minimal
expected accumulated weight.

\begin{restatable}{lemma}{ssp}
\label{th:finiteness-min-exp-accwgt}
  Let $\cM$ be an MDP with a distinguished initial state $\sinit$ and
  a trap state $\goal$ such that
  all states are reachable from $\sinit$ and can reach
  $\goal$. Then, $\Exp{\inf}{\cM,\sinit}(\accdiaplus \goal)$ is
  finite iff $\cM$ has no negatively weight-divergent end
  component.
  If so, then $\cM$ satisfies $\mathsf{(BT)}$ iff $\cM$ has no 0-EC.
\end{restatable}

The above lemma allows us to derive our algorithm by first determining if
$\Exp{\inf}{\M,\sinit}(\accdiaplus \goal)$ is finite, and then using
the iterative spider construction to transform $\cM$ into an equivalent
new MDP satisfying~$\mathsf{BT}$.

More precisely, one can check in polynomial time whether
$\Exp{\inf}{\M,\sinit}(\accdiaplus \goal) > -\infty$ by applying
Theorem~\ref{weight-divergence-algorithm} to the maximal end
components of $\cM$ (in fact, checking negative
weight-divergence reduces to checking positive weight-divergence after
multiplication of all weights by $-1$).
If so, by the iterative spider construction to
flatten 0-ECs\onlyShort{~(see~Section~\ref{sec:spider-construction})}\onlyLong{~(see Appendix~\ref{sec:iterative-application-spider})},
we obtain in polynomial time an MDP $\cN$ such that $\cN$ satisfies
condition $\mathsf{(BT)}$ and
$\Exp{\inf}{\cN,s}(\accdiaplus \goal) = \Exp{\inf}{\cM,s}(\accdiaplus
\goal)$ for each state $s$.  
To establish this result, we rely on the equivalence of $\cM$ and
$\cN$ w.r.t. properties that are $\cE$-invariant
for each 0-EC $\cE$ ((S3) in Lemma \ref{spider-construction}).
This yields:

\begin{theorem}
  \label{thm:shortest-paths}
  Given an arbitrary MDP~$\cM$, one can compute
  $\Exp{\inf}{\cM,\sinit}(\accdiaplus \goal)$ in
  polynomial time.
\end{theorem}
Analogous results are obtained for maximal expected accumulated
weights $\Exp{\sup}{\cM,s}(\accdiaplus \goal)$ by multiplying all
weights in $\cM$ with $-1$. \onlyLong{Details of this section are given in 
Appendix \ref{sec:appendix-min-exp-wgt}.}

\section{Qualitative Weight-Bounded Properties}
\setcounter{theorem}{0}
\label{sec:DWR}

\subsection{Disjunctive Weight-Bounded Reachability}

\label{sec:eventually}
We consider properties that combine reachability objectives with
quantitative constraints on the accumulated weight when reaching the targets.

\begin{definition}%
A \emph{disjunctive weight-bounded reachability} property, 
\dwr-property for short, 
is defined by a set $T \subseteq S$ of target states, 
and for each $t \in T$ a weight 
threshold $K_t \in \Integer \cup \{-\infty\}$ as 
\( \varphi 
= 
\bigvee_{t\in T} \Diamond \big(t \wedge (\accwgt \geqslant K_t)\big).
\)
\end{definition}

Our objective is to study the following decision problems: 
Given an MDP $\cM$, a state $s$ in $\cM$ and
a \dwr-property $\varphi$
\begin{center}
\begin{tabular}{l@{\hspace*{0.2cm}}l}
  \Easdwr: &
 $\exists \sched$ s.t. $\Pr^{\sched}_{\cM,s}(\varphi)=1$?
  \\[1ex]
  \Eposdwr: &
 $\exists \sched$ s.t. $\Pr^{\sched}_{\cM,s}(\varphi)>0$?
\end{tabular}
\end{center}
as well as their variants \Uasdwr{} and \Uposdwr{} with universal
quantification over schedulers.
Let $T^*=\{t\in T : K_t =-\infty\}$ denote the set of states for which no
accumulated weight constraint is specified.
For corresponding optimization problems, we assume
$T{\setminus}T^*=\{\goal\}$ to be a singleton, write $\varphi_K$ for
$\varphi$ with $K=K_{\goal}$, and ask to compute
\begin{center}
 \begin{tabular}{l@{\hspace*{0.2cm}}c@{\hspace*{0.2cm}}l}
   $\valueEas{\M,s}$ & = & 
   $\sup \, 
    \{ \, K \in \Integer \, \mid \,
             \exists \sched \text{ s.t. } 
             \Pr^{\sched}_{\cM,s}(\varphi_K)=1  \, \}$,
   \\[0.5ex]
   $\valueEpos{\M,s}$ & = &
   $\sup \, \{  \, K \in \Integer \, \mid \,
             \exists \sched \text{ s.t. } 
             \Pr^{\sched}_{\cM,s}(\varphi_K)>0  \, \}$,
 \end{tabular}
\end{center}
and the analogous values $\valueUas{\cM,s}$ and $\valueUpos{\cM,s}$
where the supremum belongs to $\Integer \cup \{\pm \infty\}$.

Deciding \Eposdwr{} and computing $\valueEpos{\M,s}$ can be done using
standard shortest-path algorithms in weighted graphs.
Thus, \Eposdwr{} belongs to \PTIME\ and the value $\valueEpos{\cM,s}$ 
is computable in polynomial time.
\onlyLong{See Appendix \ref{Eposdwr}.}

In contrast, we do not know if \Uposdwr{} is in \PTIME, but show that 
it is as hard as mean-payoff games, and is polynomially reducible to
mean-payoff B\"uchi games\onlyLong{~(Appendix~\ref{sec:Uposdwr-incoNP})}.

\begin{restatable}{theorem}{dwrupos}
\label{thm:DWR-U0}
The decision problem \Uposdwr{} is in $\NP \cap \coNP$, and at least as hard as
(non-stochastic) mean-payoff games.  The value $\valueUpos{\cM,s}$ is
computable in pseudo-polynomial time.
\end{restatable}
We now give a polynomial-time algorithm for \Uasdwr.
In the case where all states of~$T$ are traps, we show that
$\Pr^{\sched}_{\cM,s}(\varphi)=1$ for all schedulers $\sched$
iff (i) $\Pr^{\min}_{\cM,s}(\Diamond T)=1$ and (ii)
$\wgt(\fpath) \geqslant K_t$ for each path $\fpath$ from $s$ to some
state $t \in T{\setminus}T^*$.  (In particular, (ii) implies that
the paths from $s$ to some state in $T{\setminus}T^*$ do not contain
negative cycles.) Thus, this case can be solved with standard MDP and
shortest-path algorithms in graphs. The general case requires an 
analysis of end components. 
  If each end component containing $t\in T{\setminus}T^*$
  is weight-divergent, then the weight-constraint is useless and we may
  set $K_t=+\infty$.
  Otherwise we show that $t$ can be treated as a trap.
  To check whether all end components containing $t$ are weight-divergent
  we consider the MECs $\cE$ containing $t$ and
  distinguish cases where $\Exp{\min}{\cM,\cE}(\MP)>0$
  or $\Exp{\min}{\cM,\cE}(\MP)=0$ and $\cE$ does not have a 0-EC
  containing $t$.
\begin{restatable}{theorem}{dwruas}
 \label{thm:DWRU1}
 The decision problem \Uasdwr{} belongs to \PTIME\ and the
 value $\valueUas{\cM,s}$ is computable in polynomial time.
\end{restatable}

The remaining case~\Easdwr{} is perhaps the most interesting case; it is
also our main and most technical result. First, we observe that
infinite memory can be necessary.
\begin{example}
 \label{example:dwr}
Let $\cM$ be the MDP depicted left in Figure~\ref{fig:fixedpoint}.
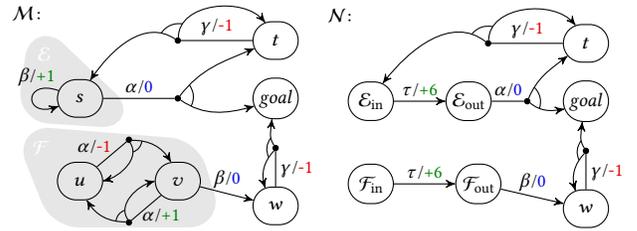
\begin{figure}[h]%
	\begin{minipage}{.2\textwidth}%
		\hspace*{-2em}
\pgfdeclarelayer{background}%
\pgfsetlayers{background,main}%
\begin{tikzpicture}[x=13mm,y=11mm,font=\scriptsize]
	\node (s) at (0,0) [state] {$s$};
	\node (stg) at (1,0) [bullet] {};
	\node (t) at (2,.7) [state] {$t$};
	\node (g) at (2,0) [state] {$\mathit{goal}$};
	\node (tst) at (1,.7) [bullet] {};
	\node (u) at (0,-1) [state] {$u$};
	\node (uvu) at (.5,-.5) [bullet] {};
	\node (v) at (1,-1) [state] {$v$};
	\node (vuv) at (.5,-1.5) [bullet] {};
	\node (w) at (2,-1.3) [state] {$w$};
	\node (wgw) at (2,-.6) [bullet] {};

	\draw (s) edge[ptran,loop, min distance=5mm, out=160, in=200] node[above=-.05,pos=.3]{$\beta$/\posw{1}} (s);
    \draw [ntran] (s) -- node[above=-.05,pos=.5]{$\alpha$/\nilw} (stg);
    \draw (stg) edge[ptran, bend left, out=20, in=170] coordinate[pos=.15] (bst) (t);
    \draw (stg) edge[ptran, bend right, in=200] coordinate[pos=.15] (bsg) (g);
    \draw (bst) to[bend left] (bsg);
    \draw [ntran] (t) -- node[above=-.05,pos=.5]{$\gamma$/\negw{1}} (tst);
    \draw (tst) edge[ptran, bend right, out=-60] coordinate[pos=.15] (bts) (s);
    \draw (tst) edge[ptran, bend left, out=80, in=120] coordinate[pos=.15] (btt) (t);
    \draw (bts) to[bend left] (btt);
    \draw [ntran] (u) -- node[left=0,pos=.8]{$\alpha$/\negw{1}} (uvu);
    \draw (uvu) edge[ptran, bend left] coordinate[pos=.2] (buv) (v);
    \draw (uvu) edge[ptran, bend right,out=70, in=140] coordinate[pos=.2] (buu) (u);
    \draw (buu) to[bend right] (buv);
    \draw [ntran] (v) -- node[right=-.05,pos=.8]{$\alpha$/\posw{1}} (vuv);
    \draw (vuv) edge[ptran, bend left] coordinate[pos=.2] (bvu) (u);
    \draw (vuv) edge[ptran, bend left, out=70, in=140] coordinate[pos=.2] (bvv) (v);
    \draw (bvu) to[bend left] (bvv);
    \draw [ptran] (v) -- node[above=-.05,pos=.5]{$\beta$/\nilw} (w);
    \draw [ntran] (w) -- node[right=-.05,pos=.5]{$\gamma$/\negw{1}} (wgw);
    \draw (wgw) edge[ptran, bend left, in=170] coordinate[pos=.3] (bwg) (g);
    \draw (wgw) edge[ptran, bend right] coordinate[pos=.2] (bww) (w);
    \draw (bwg) to[bend right] (bww);

    \begin{pgfonlayer}{background}%
    	\node[above left=.6 and .2 of s]{$\mathcal{M}\colon$};
		\draw[rounded corners=.7em,line width=1em,black!10,fill=black!10]
        	(-.4,.7) -- (.4,-.2) -- (-.5, -.13) --
			node[right=-.25,pos=.8,white]{$\mathcal{E}$} cycle;
        \draw[rounded corners=.7em,line width=1em,black!10,fill=black!10]
        	(-.5,-.5) -- (1.1, -.65) -- (1.3,-1.4) -- (-.2, -1.4) --
			node[right=-.25,pos=.9,white]{$\mathcal{F}$} cycle;
    \end{pgfonlayer}
\end{tikzpicture}

 	\end{minipage}%
 	\begin{minipage}{.2\textwidth}%
\pgfdeclarelayer{background}%
\pgfsetlayers{background,main}%
\begin{tikzpicture}[x=13mm,y=11mm,font=\scriptsize]
	\node (s) at (-.2,0) [state] {$\mathcal{E}_\mathrm{in}$};
	\node (eout) at (.8,0) [state] {$\mathcal{E}_\mathrm{out}$};
	\node (stg) at (1.4,0) [bullet] {};
	\node (t) at (2,.7) [state] {$t$};
	\node (g) at (2,0) [state] {$\mathit{goal}$};
	\node (tst) at (1,.7) [bullet] {};
	\node (u) at (-.2,-1) [state] {$\mathcal{F}_\mathrm{in}$};
	\node (v) at (.9,-1) [state] {$\mathcal{F}_\mathrm{out}$};
	\node (w) at (2,-1.3) [state] {$w$};
	\node (wgw) at (2,-.6) [bullet] {};

	\draw [ptran] (s) -- node[above=-.05,pos=.5]{$\tau$/\posw{6}} (eout);
    \draw [ntran] (eout) -- node[above=-.05,pos=.5]{$\alpha$/\nilw} (stg);
    \draw (stg) edge[ptran, bend left, out=20, in=170] coordinate[pos=.2] (bst) (t);
    \draw (stg) edge[ptran, bend right, in=200] coordinate[pos=.3] (bsg) (g);
    \draw (bst) to[bend left] (bsg);
    \draw [ntran] (t) -- node[above=-.05,pos=.5]{$\gamma$/\negw{1}} (tst);
    \draw (tst) edge[ptran, bend right, out=-60] coordinate[pos=.13] (bts) (s);
    \draw (tst) edge[ptran, bend left, out=80, in=120] coordinate[pos=.15] (btt) (t);
    \draw (bts) to[bend left] (btt);
    \draw [ptran] (u) -- node[above=-.05,pos=.5]{$\tau$/\posw{6}} (v);
    \draw [ptran] (v) -- node[above=-.05,pos=.5]{$\beta$/\nilw} (w);
    \draw [ntran] (w) -- node[right=-.05,pos=.5]{$\gamma$/\negw{1}} (wgw);
    \draw (wgw) edge[ptran, bend left, in=170] coordinate[pos=.3] (bwg) (g);
    \draw (wgw) edge[ptran, bend right] coordinate[pos=.2] (bww) (w);
    \draw (bwg) to[bend right] (bww);

    \begin{pgfonlayer}{background}%
    	\node[above left=.6 and 0 of s]{$\mathcal{N}\colon$};
    \end{pgfonlayer}
\end{tikzpicture}

 	\end{minipage}
 	\caption{\label{fig:fixedpoint} Resolution of \Easdwr{} on an
          example.}
\end{figure}%
~\\Consider the weight-bounded reachability property 
$\varphi_K = \Diamond (\goal \wedge (\accwgt \geqslant K))$. 
Given $K\in \Integer$, a scheduler $\sched_K$ ensuring
$\Pr^{\sched_K}_{\cM,s}(\varphi_K)=1$ acts as follows: 
for a finite path $\fpath$ ending in state $s$ with accumulated weight $k$,
$\sched_K$ schedules $K{-}k$ times action $\beta$, followed by
$\alpha$. Thus, all $\sched_K$-paths from $s$ ending in state $t$ or
$\goal$ have weight at least $K$ and $K_{\cM,s}^{\exists,=1}=+\infty$. 
However, for every finite-memory scheduler $\sched$, there is no $K\in \Integer$ 
with $\Pr^{\sched}_{\cM,s}(\varphi_K)=1$.
\Ende
\end{example}

\begin{restatable}{theorem}{dwreas}
\label{thm:DWR-E1}
The decision problem \Easdwr{} is in $\textrm{NP}\cap \textrm{coNP}$,
and at least as hard as (non-stochastic) mean-payoff games. The value
$\valueEas{\cM,s}$ is computable in pseudo-polynomial time.
\end{restatable}  

\begin{proof}[Proof sketch]
  We sketch the proof for the upper bound.
  The general case easily reduces to the same problem for
  $T{\setminus}T^*=\{\goal\}$ is a singleton; so we make this assumption.

  First, in the case where $\cM$ has no positively
  weight-divergent end components,
  we give a polynomial-time reduction to mean payoff games which can be solved in $\NP\cap\coNP$.

  For the general case, let us write $\cE_1,\ldots,\cE_k$ for the
  maximal positively weight-divergent end components of $\cM$. They
  can be computed by first determining the MECs 
  and checking weight-divergence for each of them by Theorem~\ref{weight-divergence-algorithm}.
  We then show that
  there exists $K_i\in \{+\infty,-\infty\}$ such that for all states
  $s$ in $\cE_i$ we have $\valueEas{\cM,s}=K_i$.
  This observation follows from the fact that any scheduler can be modified
  to have a first phase where the weight is increased by a desired constant
  inside a weight-divergent end component.

  We compute the set $\GoodEC=\{ \cE_i : K_i=+\infty\}$ using the
  greatest fixed point of a monotonic operator
  $\Omega\colon 2^{\fE}\to 2^{\fE}$ where $\fE=\{\cE_1,\ldots,\cE_k\}$
  using the techniques for MDPs without positively weight-divergent
  end components. To define this operator $\Omega$, we switch from
  $\cM$ to a new MDP $\cN$ obtained from $\cM$ by replacing each
  $\cE\in \fE$ with two fresh states $\entrystate{\cE}$ and
  $\exitstate{\cE}$.  The actions enabled in $\exitstate{\cE}$ serve
  to mimic $\cM$'s state-action pairs $(s,\alpha)$ where $s$ is a
  state of $\cE$ and $P_{\cM}(s,\alpha,s')>0$ for at least one state
  $s'$ outside $\cE$.  A single action $\tau$ is enabled in
  $\entrystate{\cE}$ with
  $P_{\cN}(\entrystate{\cE},\tau,\exitstate{\cE})=1$ whose weight is
  chosen large enough to ensure that $\entrystate{\cE}$ and
  $\exitstate{\cE}$ do not belong to a negative simple cycle.  The
  construction is illustrated in Fig.~\ref{fig:fixedpoint}.  $\cN$ has
  no positively weight-divergent end components by construction.
  However, the values in~$\cN$ can be used as lower bounds of those in~$\cM$.
  In particular, we may have
  $\valueEas{\cN,r}={-}\infty$ and $\valueEas{\cM,r'}={+}\infty$ 
  where~$r$ and~$r'$ are corresponding states in~$\cM$ and~$\cN$
  (\eg state~{$s$} in
  Fig.~\ref{fig:fixedpoint} has value~$+\infty$ in~$\cM$ but
  $\exitstate{\cE}$ has value~$-\infty$ in $\cN$). 

  Despite this, we can identify end components in $\GoodEC$, \ie with
  value $+\infty$, using~$\cN$ via a fixed-point computation. Namely,
  we define the operator $\Omega$ that assigns to each
  $X \subseteq \fE$ the set of end components $\cE \in \fE$ for which
  there is $K\in \Integer$ with
  $\Pr^{\max}_{\cN,\exitstate{\cE}}\big(\varphi_K[X]\big)=1$ where
  \[
    \varphi_K[X]\ =\ \Diamond \big(T^* \cup \{\entrystate{\cE} :\cE\in X\}\big) \vee \Diamond
  \big(\goal \wedge (\accwgt \geqslant K)\big).
\]
  Intuitively, these are states from which almost surely we either satisfy~$\varphi$,
  or reach another weight-divergent end component that allows to increase the 
  weight and start again.
  This fixed-point computation applied to $\cN$ in Fig.~\ref{fig:fixedpoint} yields, \eg  
  $X_0 = \{\cE,\cF\}, \Omega(X_0)=\{\cE\}, \Omega(\Omega(X_0))=\{\cE\}$.
  In fact, from~$\cE$ one can either immediately reach~$\goal$ or go back to~$\cE$;
  while from~$\cF$ there is no bound on the accumulated weight towards reaching~$\goal$.

  The above computation yields the values of the states of weight-divergent end components; in fact, we show that $\valueEas{\cM,s}=+\infty$
  iff $\Pr^{\max}_{\cM,s}(\Diamond (T^*\cup \GoodEC))=1$.
  For other states, we show that the maximal~$K$ such that~$\Pr_{\cN,s}^{\max}(\phi_K[\GoodEC])=1$ 
  corresponds to $\valueEas{\cM,s'}$ where~$s$ and~$s'$ are corresponding states.
  Here, $\phi_K[\GoodEC]$ is an instance
  of \Easdwr and~$\cN$ has no weight-divergent end components, so we can use the $\NP\cap\coNP$ algorithm described at the beginning.
\onlyLong{Details are given in Appendix~\ref{sec:Easdwr-inNP}.}
\end{proof}

\subsection{Weight-Bounded Repeated Reachability}%

\label{sec:Buechi}

Beyond weight-bounded reachability, we address a
B\"uchi weight condition in conjunction with a standard 
B\"uchi condition. 
Given an MDP $\cM$ without traps, 
a set $F \cup \{s\}$ of states in $\cM$ 
and $K \in \Integer$, we consider the problems
\begin{center}
\begin{tabular}{ll}
  \EaswB: &
  $\exists \sched$ s.t. $\Pr^{\sched}_{\cM,s}
             (\Box \Diamond (\accwgt \geqslant K)\wedge \Box \Diamond
              F)=1$?
  \\[0.5ex]
  \EposwB: &
  $\exists \sched$ s.t. $\Pr^{\sched}_{\cM,s}
            (\Box \Diamond (\accwgt \geqslant K)\wedge \Box \Diamond
              F)>0$?
\end{tabular}
\end{center}
and the corresponding problems
\UaswB{} and \UposwB{} 
with universal quantification over schedulers.
The two existential
problems are polynomially reducible to the 
respective existential \dwr{} problems, maintaining the same complexity
classes. 
The universal problems can be solved using techniques to treat existential
problems for coB\"uchi weight constraints, which again are
polynomially reducible to  \Eposdwr{} and \Easdwr, respectively.
\onlyLong{For details see Appendix \ref{sec:appendix-buechi}.}

\begin{restatable}{theorem}{thmEBcoB}
\label{almost-sure-Buechi-and-coBuechi}
\EposwB{} and \UaswB{} are decidable in polynomial time.  \EaswB{} and
\UposwB{} are in $\NP\cap\coNP$, decidable in pseudo-polynomial time,
and at least as hard as mean-payoff games. %
\end{restatable}

The proof of Theorem \ref{almost-sure-Buechi-and-coBuechi}
heavily uses the concepts of Section \ref{sec:classification}.
Let us briefly describe the reduction of \EaswB{} and \EposwB{} to \Easdwr{}
and \Eposdwr{} for some \dwr{} formula
$\varphi = \bigvee_{t\in T} \Diamond \big(t \wedge (\accwgt \geqslant
K_t)\big)$. We define $T^*$ as the set of all states in maximal
weight-divergent end components containing at least one state in
$F$ and $T{\setminus} T^*$ as the set of states belonging to a
maximal 0-EC $\cZ$ of a maximal end component $\cE$ with
$\Exp{\max}{\cE}(\MP)=0$ and $\cZ{\cap}F\not= \varnothing$.  Note that
both $T^*$ and $T{\setminus}T^*$ are computable in polynomial time
(due to Theorem~\ref{weight-divergence-algorithm} and
Lemma~\ref{mincredit-ZeroEC}). For the states in $T{\setminus}T^*$, we
let $K_t{=}K$, where $K$ is taken from the input of \EaswB{} or \EposwB.

To solve problem \UaswB{} we rely on the observation that
\UaswB{} holds iff (i) $\Pr^{\min}_{\cM,s}(\Box \Diamond F)=1$
and (ii) there is no scheduler $\sched$ with
$\Pr^{\sched}_{\cM^-,s}\big(\Diamond \Box (\accwgt \geqslant L)\big)>0$
where $\cM^-$ results from $\cM$ by multiplying all weights with $-1$
and $L=-(K{-}1)$. While (i) can be checked in polynomial time, (ii) is equivalent to the complement of
\Eposdwr{} for $\cM^-$ and %
$\bigvee_{t \in T} \Diamond \big(t \wedge (\accwgt \geqslant K_t)\big)$
where $T^*$ denotes the set of states belonging to 
a pumping end component of $\cM^-$
and $T{\setminus}T^*$ is the set of states belonging to the set $\ZeroEC$ 
and $K_t=L{-}\rec(t)$. 
Here $\ZeroEC$ is the set of states that belong to a maximal 0-EC $\cZ$ 
of a maximal end component $\cE$ of $\cM$ or $\cM^-$ 
with $\Exp{\max}{\cE}(\MP)=0$
and moreover, $\rec(t)$ refers to this maximal end component $\cE$. 

For problem \UaswB{} we transform $\cM^-$ into a new MDP $\cN$ such that
\UaswB{} holds for $\cM$ iff there is no scheduler for $\cN$
where the coB\"uchi weight constraint $\Diamond \Box (\accwgt \geqslant L)$
holds almost surely, which
can be checked applying the algorithm for
\Easdwr{} for $\cN$ and the same DWR property as for \Uposdwr.
Here $L$ is as above and
$\cN$ arises from $\cM^-$ by identifying all states that belong to an
end component not containing an $F$-state and replacing their enabled actions
with a self-loop of weight 0.

The  optimization problems of \EaswB{} and \UposwB{}
are computable in pseudo-polynomial time, 
and optimal weight bounds for \EposwB{} and \UaswB
in polynomial time.

\subsection{Discussion on Related Work}

\label{sec:discussion}

To the best of our knowledge, problems \Easdwr, \Uposdwr{} and \Uasdwr{}
or the variants for B\"uchi weight constraints
have not been studied before for general integer-weighted MDPs.
Qualitative weight-bounded reachability
properties in MDPs with only nonnegative weights are
decidable in polynomial time~\cite{UB13}. 
This result %
relies on the monotonicity of accumulated weights along all paths.
The lack of 
monotonicity in the general case rules out analogous algorithms.

For Markov chains, 
qualitative weight-bounded reachability properties can be treated 
in polynomial time \cite{KSBD15}. 
This result uses
expected mean payoff in BSCCs,
variants of shortest-path algorithms and the continued-fraction
method.
In MDPs, however, optimal schedulers might need infinite memory (see
Example \ref{example:dwr}) so these algorithms cannot
be adapted. 
In fact, our algorithms crucially 
rely on the classification of end components.

Let us point out the similarities and differences between the problems
we considered and the ones for energy
MDPs~\cite{ChatDoy11,MaySchTozWoj17}.  Rephrased
for our notations, the energy-MDP problem is to check whether
$\Pr^{\max}_{\cM,s}\big(\Box (\accwgt \geqslant K) \wedge \phi\big)=1$ where
$\varphi$ is a parity condition and $K\in \Integer$. This problem is in
$\text{NP}\cap \text{coNP}$ and at least as hard as two-player mean-payoff games,
even if $\phi=\true$. The complement of the energy-MDP problem asks
whether $\Pr^{\min}_{\cM,s}\big(\Diamond (\accwgt < K)\vee \neg \phi\big)>0$,
which corresponds to
$\Pr^{\min}_{\cM,s}\big(\Diamond (\accwgt \geqslant K) \vee \neg \phi\big)>0$
when switching from $\wgt$ to $-\wgt$ and from $K$ to $-(K{-}1)$.
However, although in the spirit of this problem, \Uposdwr{} asks whether
$\Pr^{\min}_{\cM,s}\big(\Diamond (\goal \wedge (\accwgt \geqslant K))\big)>0$,
in the case  $T^*=\varnothing$ and 
$T\setminus T^*=\{\goal\}$.
Given the similarities of these questions, 
and our decision procedure that reduces \Uposdwr{}
to mean-payoff B\"uchi games,
it is no surprise that the problem \Uposdwr{}
is at least as hard as mean-payoff games.

Nevertheless, the instances \Easdwr{} and \Uasdwr{} are of different
nature than energy-MDPs. These
can rather be seen as variants of the
\emph{termination problem} for \emph{one-counter MDPs} \cite{EWY08,BBEKW10}.  
One-counter MDPs have their weights
in $\{-1,0,+1\}$, while we allow
arbitrary weights.
Moreover, a one-counter MDPs halts whenever the counter reaches 0,
but there is no lower bound on the accumulated weight in our setting.
Following \cite{BBEKW10}, we refer to these one-counter MDPs 
as \emph{one-counter MDP with boundary}
and to MDPs in our setting with weights in $\{-1,0,+1\}$
as \emph{boundaryless one-counter MDPs}.

We commented on \cite{BBEKW10} in the paragraph following Theorem~\ref{thm:checking-gambling}.
For one-counter MDPs $\ocM$ with boundary, \cite{BBEKW10} 
also provides an exponential-time algorithm for checking
$\Pr^{\max}_{\ocM,s}\big(\bigvee_{t\in T} \Diamond (t \wedge (\accwgt=0))\big)=1$
and shows PSPACE-hardness.
This contrasts with our $\textrm{NP}\cap \textrm{coNP}$ upper bound for
\Easdwr{} with arbitrary integer weights
(Theorem \ref{thm:DWR-E1}).
Besides the differences ``boundary vs boundaryless''
and ``integer vs unit weights'',
we consider
objectives imposing lower bounds on the accumulated weights.
Considering $\Diamond (t \wedge (\accwgt = K_t))$ would raise the complexity
in our setting at least to EXPTIME-hardness,
by
\cite{HaaseKiefer15} which shows that
for MDPs $\cM$ with non-negative integer weights 
and $\Pr^{\min}_{\cM,s}(\Diamond \goal)=1$,
checking whether
$\Pr^{\max}_{\cM,s}(\Diamond (\goal \wedge (\accwgt=K)))=1$
for some given $K\in \Nat$ is EXPTIME-complete.

Nondeterministic and probabilistic models for vector addition systems
(VASS-MDPs) can be seen as boundary MDPs with multiple weight
functions. Decidable results on VASS-MDPs include the existence of a
scheduler that almost surely ensures some property expressible in
$\mu$-calculus (with no constraint on the accumulated
weights)~\cite{ACMSS16}. The decision algorithms rely on the
termination of fixed-point computations thanks to well-quasi orderings,
thus yielding much higher complexity than our techniques.
\section{Conclusion}

\label{sec:conclusion}
We provided a classification of end components according to their behaviors with respect to the accumulated weight. This allowed us to solve the general stochastic shortest path problem and to derive algorithms for weight-bounded properties. We believe our classification helps
better understanding the accumulated weights in MDPs, and can be helpful for other problems and
perhaps simplify existing results.

An interesting future work is to
address analogous questions for quantitative probability
thresholds. %
This appears to be challenging as the
probabilities for weight-bounded properties can be irrational,
even in Markov chains \cite{EWY08,BBEKW10}.
\bibliographystyle{myalpha}  %
\bibliography{lit} 
\onecolumn
\appendix
~\\\textbf{\huge Appendix}\\[1em]
In this appendix, we provide the proofs of the results of
the main paper that had to be omitted due to space constraints.\\
~\\\textbf{\Large Contents - Overview of the Paper}\vspace*{-4em}
\setcounter{tocdepth}{2}
\renewcommand*{\contentsname}{}
\tableofcontents

\renewcommand{\thetheorem}{\Alph{section}.\arabic{theorem}}
\renewcommand{\thelemma}{\Alph{section}.\arabic{lemma}}
\renewcommand{\thecorollary}{\Alph{section}.\arabic{corollary}}
\renewcommand{\theremark}{\Alph{section}.\arabic{remark}}
\renewcommand{\theexample}{\Alph{section}.\arabic{example}}
\renewcommand{\thedefinition}{\Alph{section}.\arabic{definition}}

\bibliographystyleapx{myalpha}
\bibliographyapx{lit}

\section{Additional Notations}
\setcounter{theorem}{0}

\label{appendix-notations}

\tudparagraph{1ex}{Notations for paths (concatenation, length, fragments).}
Given a finite path $\fpath'=s_0\alpha_0\ldots \alpha_n s_n$ 
and a (finite or infinite) path $\fpath=t_0\beta_0 t_1 \beta_1  \ldots$ 
with $s_n=t_m$ then
$\fpath';\fpath$ denotes the path
$s_0\alpha_0\ldots \alpha_n s_n \beta_0 t_1 \beta_1  \ldots$.
The length of a path $\fpath$, denoted $|\fpath|$,
is defined as the number of state-action pairs in
$\fpath$, \ie if $\fpath'$ is a finite path as above then 
$|\fpath'|=n$, while $|\fpath|=\infty$ for each infinite path.
If $\fpath$ is a path as above and $i,n\in \Nat$ with
$i \leqslant n \leqslant |\fpath|$ then 
$\state{\fpath}{n}$ denotes $t_n$ (the $(n+1)$-st state of $\fpath$)
and $\fragment{\fpath}{i}{n}$
the finite path
$t_i \, \beta_i \, t_{i+1}\, \beta_{i+1} \, \ldots \, \beta_{n-1} \, t_n$.
Thus, $\fragment{\fpath}{0}{n}=\prefix{\fpath}{n}$,
$\fragment{\fpath}{n}{n}=\state{\fpath}{n}$ and
$\first(\fpath)=\state{\fpath}{0}$.
Similarly, $\last(\fpath)=\state{\fpath}{n}$ if $n=|\fpath|$ is finite.

\tudparagraph{1ex}{Residual schedulers.}
Let $\sched$ be a scheduler and $\fpath$ a finite path.
The residual scheduler $\sched{\uparrow}{\fpath}$ is defined by
$(\sched{\uparrow}{\fpath})(\fpath') = 
 \sched(\fpath; \fpath')$
if~$\first(\fpath') = last(\fpath)$, and
$(\sched{\uparrow}{\fpath})(\fpath') = \sched(\fpath')$ otherwise.

\tudparagraph{1ex}{Finite-memory scheduler.}
A finite memory scheduler can defined as follows. Given a finite set~$M$
of memory elements, $\sched_u(s,m) = m'$ is an \emph{update function} that determines
the new memory element given current state~$s$ and current memory element~$m$;
and~$\sched_n(s,m) = \alpha$ determines the action to be played at state~$s$ and if the memory contains~$m$.

\tudparagraph{1ex}{Markov chain.} A \emph{Markov chain} is an MDP $\cM=(S,\Act,P,\wgt)$
where $\Act$ is a singleton. We occasionally use the notation $\cC=(S,P',\wgt')$ for $\cM$,
where $P'\colon S\times S\rightarrow [0,1]$ and $\wgt'\colon S\rightarrow\Integer$ are
defined as $P$ and $\wgt$ but omitting the uniquely defined action.

\tudparagraph{1ex}{Limit of infinite paths.}
The limit of an infinite path $\infpath$, denoted $\lim(\infpath)$,
is the set of state-action pairs that occur infintely often in $\infpath$.
If $\cE$ is an end component then we often write $\Limit{\cE}$ 
for $\{\infpath \in \InfPaths : \lim(\infpath)=\cE\}$.
At various places, we rely on De Alfaro's result \cite{Alfaro98Thesis}
stating that for each scheduler $\sched$,
the limit of almost all infinite $\sched$-paths is an end component.
Formally, for each scheduler $\sched$ and each state $s$, we have
$\Pr^{\sched}_{\cM,s}\big(\bigcup_{\cE}\Limit{\cE}\big)=1$
when $\cE$ ranges over all (possibly nonmaximal) end components.

\tudparagraph{1ex}{Probabilities.}
Recall that we use the notation $\Pr^{\max}_{\cM,s}(\varphi)$ when
there exists a scheduler $\sched$ for $\cM$ such that
$\Pr^{\sched}_{\cM,s}(\varphi)=\Pr^{\sup}_{\cM,s}(\varphi)$.

If $\fpath$ is a finite path starting in $s$ and $\sched$ a scheduler
then $\Pr^{\sched}_{\cM,s}(\fpath)$ is used as a shortform notation
for $\Pr^{\sched}_{\cM,s}\big(\Cyl(\fpath)\big)$ where $\Cyl(\fpath)$ denotes
the cylinder set of $\fpath$, \ie the set of all maximal paths
$\infpath$ where $\fpath$ is a prefix of $\infpath$.

\tudparagraph{1ex}{Properties.}
Let $\cM$ be an MDP with state space $S$.
Define $\Lambda_\cM=(S \times \Integer)^{\omega} \cup (S \times \Integer)^*\times S$.
The set $\Lambda_\cM$ is equipped with the (standard) sigma-algebra
generated by the cylinder sets of the finite sequences in $(S \times \Integer)^*$.
A property is a measurable subset of $\Lambda_\cM$.
Of course, each path $t_0\, \alpha_0 \, t_1 \, \alpha_1 \ldots$ in $\cM$
naturally induces such a sequence in $\Lambda_\cM$ 
by replacing $\alpha_i$ with $\wgt(t_i,\alpha_i)$.
Denote the function mapping every $t_0\, \alpha_0 \, t_1 \, \alpha_1 \ldots$
to $t_0\, \wgt(t_0,\alpha_0) \, t_1 \, \wgt(t_1,\alpha_1) \ldots$ by $f$.
Hence, for each scheduler $\sched$ and each state $s$
the probability measure $\Pr^{\sched}_{\cM,s}$ on 
induces a probability measure $\Pr^{\sched}_{\cM,s,\sharp}$ on $\Lambda_\cM$.
Formally, for every property $\varphi$ we have 
$\Pr^{\sched}_{\cM,s,\sharp}(\varphi)=\Pr^{\sched}_{\cM,s}\big(\{\fpath \ : \ f(\fpath)\in\varphi\}\big)$.
To simplify notions, we identify the two probability measures 
$\Pr^{\sched}_{\cM,s}$ and $\Pr^{\sched}_{\cM,s,\sharp}$,
\ie for every property $\varphi$ 
we write $\Pr^{\sched}_{\cM,s}(\varphi)$ 
rather than $\Pr^{\sched}_{\cM,s,\sharp}(\varphi)$.

\section{Proofs and Complements for Section~\ref{sec:classification}}
\setcounter{theorem}{0}
\label{appendix:sec-classification}

\subsection{Illustration of the Notions on End Components}

\begin{figure}[ht]
	\begin{minipage}[b]{0.39\textwidth}
		\centering%

\begin{tikzpicture}
	\node[state] (s) {$s$};
	\node[state,right=15mm of s] (goal) {$\goal$};
	\path%
		(s) edge[ptran,loop, min distance=9mm, out=60, in=120] node[above]{$\alpha$/\posw{1}} (s)
		(s) edge[ptran] node[above]{$\beta$/\negw{2}} (goal);
\end{tikzpicture}
 		\caption{MDP with pumping EC $\cE = \{(s,\alpha)\}$.}
		\label{fig:pumping-ec}
	\end{minipage}
\end{figure}
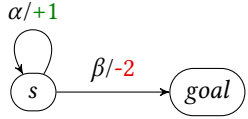

\begin{example}[Pumping EC] {\rm Let $\cM$ be the simple MDP
    depicted in Figure~\ref{fig:pumping-ec} consisting of the two
    states $s$ and $\goal$, probabilistic transitions $\Pr(s,\alpha,s)=1$ and
    $\Pr(s,\beta,\goal)=1$ with weights $\wgt(s,\alpha)=+1$ and $\wgt(s,\beta)=-2$. 
    The pair $(s,\alpha)$
    constitute a maximal end component of $\cM$ that is trivially pumping.

    This example also illustrates that no MD-scheduler can ensure
    $\Diamond (\goal \wedge (\accwgt \geqslant 0))$
    almost surely. Indeed, if $\sched$ is the scheduler that takes
    action $\alpha$ twice in $s$ and then action $\beta$ to move to $\goal$, then
    $\Pr^{\sched}_{\cM,s}\big(\Diamond (\goal \wedge (\accwgt
    \geqslant 0))\big)=1$ .  However, there is no MD-scheduler
    $\tsched$ satisfying 
    $\Pr^{\tsched}_{\cM,s}\big(\Diamond (\goal \wedge (\accwgt
    \geqslant 0))\big)=1$.
\Ende
}
\end{example}

\subsection{Mean Payoff in Strongly Connected Markov Chains}

\label{sec:appendix-meanpayoff-MC}

Let us start by recalling some simple observations on Markov chains.
\begin{lemma}[Folkore]
\label{MC-weight-div-expected-mp-pos-or-neg}
  For each finite strongly connected Markov chain $\cC$:
  \begin{enumerate}
  \item [(a)]
     If\ \ $\ExpRew{}{\cC}(\MP)>0$ then $\cC$ is positively pumping.
  \item [(b)]
     If\ \ $\ExpRew{}{\cC}(\MP)<0$ then $\cC$ is negatively pumping.
  \end{enumerate}
\end{lemma}

In what follows, if $s\in S$ then we write ``$\wgt$ until $s$''
to denote the random variable that assigns to each infinite path
$\infpath$ the accumulated weight $\wgt(\fpath)$ of the shortest
prefix $\fpath$ of $\infpath$ with $\last(\fpath) = s$, provided that
$\infpath \models \Diamond s$.  Note that ``$\wgt$ until $s$'' agrees
with the random variable $\accdiaplus s$ defined in the core of the
paper; we use here ``$\wgt$ until $s$'' instead, and also ``steps
until $s$'' defined below.  For the infinite paths $\infpath$ with
$\infpath \not\models \Diamond s$, ``$\wgt$ until $s$'' is undefined.
Thus, if $s,t\in S$ and $\sched$ is a scheduler with
$\Pr^{\sched}_{\cM,t}(\Diamond s)=1$ then
$\ExpRew{\sched}{t}(\textrm{``$\wgt$ until $s$''})$ stands for the
expected accumulated weight from $t$ until reaching $s$.
Similarly, ``steps until $s$'' 
denotes the random variable 
counting the number of steps until reaching state $s$.
Thus, if $\Pr^{\sched}_{\cM,t}(\Diamond s)=1$ then
$\ExpRew{\sched}{t}(\textrm{``steps until $s$''})$ stands for the
expected number of steps from $t$ until reaching $s$ with respect to
scheduler $\sched$.

\begin{lemma}[Quotient representation of expected mean payoff in MCs]
\label{lemma:scMC-exp-mp-quotient}
  Let $\cC=(S,P,\wgt)$ 
  be a strongly connected Markov chain. Then, for each state
  $s$ in $\cC$ we have:
  $$
    \ExpRew{}{\cC}(\MP) \ \ = \ \ 
    \begin{array}{rcl}
        \ \wgt(s) & + &
          \sum\limits_{t\in S} 
             P(s,t) \cdot \ExpRew{}{\cC,t}(\textrm{``$\wgt$ until $s$''}) \
        \\[0.5ex]
        \hline
        \\[-2.2ex] 
        1 \ & + & 
          \sum\limits_{t\in S} 
             P(s,t) \cdot \ExpRew{}{\cC,t}(\textrm{``steps until $s$''}) \
    \end{array} 
  $$
\end{lemma}

\begin{proof}
  The statement is a consequence of 
  the well-known fact that for 
  (finite-state) strongly connected Markov chains,
  the long-run frequencies of almost all paths 
  converge to the steady-state probabilities.
  (We suppose here the definition
  of the steady-state probability of state $s$ 
  as the Ces\`aro limit
  $\lim_{n \to \infty} \frac{1}{n+1} \cdot 
  \sum_{i=0}^n \Pr_{\cC,\iota}(\neXt^i s)$
  where $\iota$ is an arbitrary initial distribution.)
\end{proof}

\subsection{The Pumping Property in MDPs}
\label{appendix:pumping}
We now provide the proofs for Lemma \ref{lemma:pumping-ecs} and 
Lemma \ref{main:mp-pos-wgt-div}.

\begin{lemma}
\label{mp-pos-wgt-div}
    Each strongly connected MDP $\cM$
    with $\ExpRew{\max}{\cM}(\MP)>0$ 
    has a pumping MD-scheduler.
\end{lemma}

\begin{proof}
  Let $\sched$ be an MD-scheduler for $\cM$
  that maximizes the expected mean payoff and where the induced
  Markov chain has a single BSCC.
  Then:
  $$
    \lim_{n \to \infty} \ \frac{1}{n} \wgt(\prefix{\infpath}{n}) \ \ = \ \ 
    \ExpRew{\sched}{\cM}(\MP) \ \ > \ \ 0
  $$
  for almost all infinite $\sched$-paths $\infpath$.
  Hence, $\sched$ is pumping.
\end{proof}

Recall that that no scheduler can exceed the maximum expected mean payoff in MDPs:
\begin{lemma}[Folklore -- see, \eg \cite{Puterman}]
\label{folklore-mean-payoff}
  Let $\cM$ be a strongly connected MDP. Then   for each scheduler $\sched$ and each state $s$
  \[
   \Pr^{\sched}_{\cM,s}
     \bigl\{ \ \infpath \in \InfPaths \ : \ 
               \limsup_{n \to \infty} \ \frac{1}{n} \wgt(\prefix{\infpath}{n})
               \ \leqslant \ \ExpRew{\max}{\cM}(\MP) \ 
     \bigr\} \ \ = \ \ 1 \enspace.
  \]
\end{lemma}

By the results of de Alfaro \cite{Alfaro98Thesis}, for each scheduler $\sched$,
the limit of almost all infinite $\sched$-paths is an end component.
Recall that $\lim(\infpath)$, the limit of an infinite path $\infpath$,
is the set of state-action pairs that occur infinitely often in $\infpath$.
Hence, we get (see, \eg \citeapx{RanRasSan15}):

\begin{lemma}
\label{mp-negative-not-wgt-div} \
     Let $\cM$ be a strongly connected MDP  such that for some state $s_0$ in $\cM$
     \[
       \Pr^{\max}_{\cM,s_0}
       \bigl\{ \, \infpath \in \InfPaths \, : \, 
                 \limsup_{n \to \infty} \,
                      \wgt(\prefix{\infpath}{n}) =+\infty \,
       \bigr\} \ >  \ 0 \enspace.
     \]
      Then $\ExpRew{\max}{\cM}(\MP)\geqslant 0$.
\end{lemma}

\begin{corollary}
  \label{MP-pos-or-neg-implies-strict-wgt-div}
    Let $\cM$ be a strongly connected MDP $\cM$.
    Then: 
    \begin{enumerate}
    \item [(a)] 
      If\ \ $\ExpRew{\max}{\cM}(\MP)<0$ then $\cM$ is universally
      negatively pumping, \ie for each scheduler
      $\sched$ and each state~$s$:
      $$
       \Pr^{\sched}_{\cM,s}
       \bigl\{ \ \infpath \in \InfPaths \ : \ 
             \limsup_{n \to \infty} \ \wgt(\prefix{\infpath}{n})
             \ = \ -\infty
       \bigr\} \ \ = \ \ 1
      $$
    \item [(b)] 
      If\ \ $\ExpRew{\min}{\cM}(\MP)>0$ then $\cM$ is universally pumping, \ie
      for each scheduler
      $\sched$ and each state~$s$:
      $$
       \Pr^{\sched}_{\cM,s}
       \bigl\{ \ \infpath \in \InfPaths \ : \ 
             \liminf_{n \to \infty} \ \wgt(\prefix{\infpath}{n})
             \ = \ +\infty
       \bigr\} \ \ = \ \ 1
      $$
    \end{enumerate}
\end{corollary}

\poswgtdiv*

\begin{proof}
  Part (a) follows from Lemma~\ref{mp-pos-wgt-div} as each pumping
  scheduler is weight-divergent.
  Part (b) is an immediate consequence of 
  Corollary \ref{MP-pos-or-neg-implies-strict-wgt-div}.
\end{proof}

\lempumpingec*

For the proof, the statement of Lemma \ref{lemma:pumping-ecs} is split into the following two
Lemmas \ref{appendix:pumping-ecs} and \ref{universally-pumping}:

\begin{lemma}%
  \label{appendix:pumping-ecs}
  Let $\cM$ be a strongly connected MDP.
  Then, the following three statements are equivalent:
  \begin{enumerate}
  \item [(i)]
     $\cM$ is pumping.
  \item [(ii)] 
     $\cM$ has a pumping MD-scheduler.
  \item [(iii)]
     $\ExpRew{\max}{\cM}(\MP)>0$.
  \end{enumerate}
\end{lemma}

\begin{proof}
 The implication ``(ii) $\Longrightarrow$ (i)'' is trivial,
 while ``(iii) $\Longrightarrow$ (ii)'' has been shown
 in Lemma \ref{mp-pos-wgt-div}.
 It remains to prove the implication 
 ``(i) $\Longrightarrow$ (iii)''.
 Suppose $\cM$ is pumping and let $\sched$ be a pumping scheduler.
 Then Corollary~\ref{MP-pos-or-neg-implies-strict-wgt-div}~(a)
 implies $\ExpRew{\max}{\cM}(\MP)\geqslant 0$.
 Let $\Gamma$ denote the set of pumping paths in $\cM$, \ie
 \[
     \Gamma \ \ = \ \
     \bigl\{ \ \infpath \in \InfPaths \ : \
         \liminf_{n \to \infty} \ \wgt(\prefix{\infpath}{n}) = +\infty \
   \ \bigr\} \enspace.
 \]
 If $\cE$ is an end component then
 let $\Gamma_{\cE}=\{\infpath \in \Gamma : \lim(\infpath)=\cE\}$.
 For each end component $\cE$ where $\Pr^{\sched}_{\cM,s_0}(\Gamma_{\cE})>0$
 for some state $s_0$ we have:
 $
   \sup_{\sched'} \Pr^{\sched'}_{\cM,s}(\Gamma_{\cE}) = 1,
 $
 where $s$ is an arbitrary state in $\cE$ and 
 $\sched'$ ranges over all residual schedulers 
 $\residual{\sched}{\fpath}$ of $\sched$ with 
 $\first(\fpath)=s_0$ and $\last(\fpath)=s$.

 We pick an end component $\cE$ such that
 $\Pr^{\sched}_{\cM,s_0}(\Gamma_{\cE})>0$ for some state $s_0$ and
 an MD-scheduler $\usched$ for $\cE$ with a single BSCC $\cB$
 such that $\ExpRew{\usched}{\cE}(\MP)\geqslant 0$.
 (Note that $\ExpRew{\max}{\cE}(\MP)\geqslant 0$ 
  by Lemma \ref{folklore-mean-payoff}.)
 Let $s$ be a state of $\cB$ and
 $
    E = 
    \min_{t\in \cE} 
       \ExpRew{\usched}{\cE,t}(\text{``$\wgt$ until $s$''}).
 $
 We choose a positive integer $\Delta \in \Nat$, 
 $p\in ]0,1[$ and a residual scheduler $\sched'$ of $\sched$ 
 such that:
\[%
   p\Delta  +  (1{-}p)E \ > \ 0 
   \qquad \textrm{and} \qquad
   \Pr^{\sched'}_{\cM,s}
    \big(\, 
      \neXt (\Diamond (s \wedge (\accwgt \geqslant \Delta))) \, 
    \big)
   \ > \ 
   p \enspace.
\]
 Let $\Pi$ denote the set of $\sched'$-paths $\fpath$ from $s$ to $s$
 such that $\wgt(\fpath) \geqslant \Delta$. For $n\in \Nat$, 
 let $\Pi_n$ be the set of
 paths $\fpath \in \Pi$ such that $|\fpath|\leqslant n$.
 We pick some $n \in \Nat$ such that
 $\sum_{\fpath\in \Pi_n} \Pr^{\sched'}_{\cM,s}(\fpath) \ >  \ p$
 (possible due to the choice of $\Delta$, $p$, and $\sched'$).
 
 Let $\tsched$ be the following scheduler operating in two modes:
 normal mode and recovery mode.
 In its normal mode, $\tsched$ attempts to generate a path in $\Pi_n$
 by mimicking $\sched'$.
 If it fails, \ie if the path $\fpath$ that has been generated since the
 last switch from recovery to normal mode is not a prefix of some
 path $\fpath \in \Pi_n$, then $\tsched$ switches to recovery mode 
 where it behaves as $\usched$ until
 state $s$ has been reached. As soon as $s$ has been reached in recovery mode, 
 $\tsched$ switches back to normal mode and attempts to generate a path
 $\fpath \in \Pi_n$.
 If $\tsched$ in normal mode has generated a path $\fpath \in \Pi_n$ then it
 keeps in normal mode and restarts to attempt to generate a path in $\Pi_n$.

 This scheduler $\tsched$ is pumping and the memory requirements are 
 finite (as $\Pi_n$ is finite).
 According to Lemma \ref{lemma:scMC-exp-mp-quotient} applied to the
 Markov chain induced by $\tsched$, we obtain:
 $$
   \ExpRew{\tsched}{\cM,s}(\MP) 
    \ \ \geqslant \ \ \frac{p \Delta + (1{-}p)E}{c} \ \ > \ \ 0,
 $$
 where $c$ is the expected number of steps under $\tsched$ to return to $s$
 from $s$ in normal mode.
 Hence, $\ExpRew{\max}{\cM}(\MP) >0$.
\end{proof}

\begin{lemma}
\label{universally-pumping}
  Let $\cM$ be a strongly connected MDP.
  Then, $\cM$ is universally pumping
    iff\ \ $\ExpRew{\min}{\cM}(\MP) > 0$.
\end{lemma}
    
\begin{proof}
The implication ``$\Longleftarrow$'' has been stated in
Corollary \ref{MP-pos-or-neg-implies-strict-wgt-div}~(b).

To prove ``$\Longrightarrow$'', we
assume that 
the minimal expected mean payoff in $\cM$ is nonpositive, 
(\ie $\ExpRew{\min}{\cM}(\MP) \leqslant 0$)
and
show that $\cM$ is not universally pumping.
Pick an MD-scheduler $\sched$ minimizing the expected mean payoff in $\cM$.
Then, $\ExpRew{\sched}{\cM}(\MP) \leqslant 0$.
But then, the limit
$\lim_{n \to \infty} \ \wgt(\prefix{\infpath}{n})$
exists for almost all $\sched$-paths and equals 
$\ExpRew{\sched}{\cM}(\MP)$.
Hence, $\cM$ is not universally pumping.
\end{proof}

\subsection{Properties of 0-ECs}

\label{sec:sero-end-components}
\label{sec:zero-end-components}

We establish some properties of 0-ECs necessary to prove
Theorems~\ref{thm:checking-zero-EC}
and to establish an algorithm for computing all states
belonging to some 0-EC (Lemma \ref{mincredit-ZeroEC}).

\subsubsection{Maximal 0-ECs}

\label{sec:max-0-EC}
Let $\cE_1$ and $\cE_2$ be two 0-ECs.
If $\cE_1$ and $\cE_2$ are weight-divergent, 
then so is $\cE_1 \cup \cE_2$.
Thus, any weight-divergent end component is contained 
in a maximal end-component that is weight-divergent.
The same holds for pumping end components or for gambling end components.

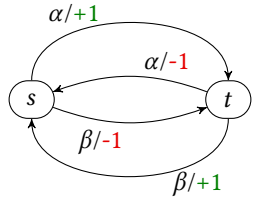
\begin{figure}[ht]
	\centering%

\begin{tikzpicture}
	\node[state] (s) {$s$};
	\node[state,right=20mm of s] (t) {$t$};
	\path%
		(s) edge[ptran, bend left=90] node[above, pos=0.3]{$\alpha$/\posw{1}} (t)
		(t) edge[ptran, bend right=20] node[above, pos=0.25]{$\alpha$/\negw{1}} (s)
		(s) edge[ptran, bend right=20] node[below, pos=0.3]{$\beta$/\negw{1}} (t)
		(t) edge[ptran, bend left=90] node[below, pos=0.25]{$\beta$/\posw{1}} (s);
\end{tikzpicture}
 	\caption{An example showing that the union of 0-ECs is not a 0-EC.}
\label{fig:union-of-zero-ecs}
\end{figure}

However, in general, the union of 0-ECs is not a 0-EC, and there are
MDPs with maximal end components that are not a 0-EC but contain
0-ECs. This is the case of the MDP $\cM$ depicted in
Figure~\ref{fig:union-of-zero-ecs}. %
The union of the 0-ECs $\cE_{\alpha}=\{(s,\alpha),(t,\alpha)\}$ and
 $\cE_{\beta}=\{(s,\beta),(t,\beta)\}$ is not a 0-EC.
However, we have:

\begin{lemma}[Union of 0-ECs]
 \label{lemma:union-of-0-EC}
  Let $\cM$ be a strongly connected MDP with $\ExpRew{\max}{\cM}(\MP)=0$,
  and let $\cE_1$ and $\cE_2$ be 0-ECs that have at least one common
  state. Then, $\cE_1 \cup \cE_2$ is a 0-EC too.
\end{lemma}

\begin{proof}
  Let $i\in \{1,2\}$. 
  As $\cE_i$ is a 0-EC, for each pair $(s,t)$ of states in $\cE_i$
  there exists an integer $w_i(s,t)$ such that all paths from $s$ to $t$ in
  $\cE_i$ have weight $w_i(s,t)$. Clearly, $w_i(t,s)=-w_i(s,t)$.

  The union $\cE = \cE_1 \cup \cE_2$ is an end component of $\cM$.
  Hence, $\ExpRew{\max}{\cE}(\MP)=0$.
  Suppose by contradiction that $\cE$ is not a 0-EC.
  Then, there exist two states $s$ and $t$ in $\cE$ such that
  $w_1(s,t) \not= w_2(s,t)$, say $w_1(s,t) > w_2(s,t)$.
  For $i=1,2$, let $\sched_i$ be an MD-scheduler for $\cE_i$ such that
  $s$ and $t$ belong to a BSCC $\cB_i$ of $\sched_i$.
  We now combine $\sched_1$ and $\sched_2$ to a scheduler $\sched$ for $\cE$.
  $\sched$ alternates between two modes. Starting in mode 1, $\sched$
  behaves as $\sched_1$ until state $t$ is reached. It then switches to 
  mode 2 where $\sched$ behaves as $\sched_2$ until state $s$ is reached,
  in which case it switches back to mode 1.

  $\sched$ is a finite-memory scheduler. Hence, it induces a finite 
  Markov chain $\cC_{\sched}$.
  We now apply Lemma \ref{lemma:scMC-exp-mp-quotient} 
  to $\cC_{\sched}$ and obtain:
  $$
    \ExpRew{\sched}{\cE}(\MP) \ \ = \ \ 
    \begin{array}{rcl}
        \ \wgt(s) & + &
          \sum\limits_{u\in S} 
             P(s,\alpha,u) \cdot 
             \ExpRew{\sched}{\cE,u}(\textrm{``$\wgt$ until $s$''}) \
        \\[0.5ex]
        \hline
        \\[-2.2ex] 
        1 \ & + & 
          \sum\limits_{u\in S} 
             P(s,\alpha,u) \cdot 
             \ExpRew{\sched}{\cE,u}(\textrm{``steps until $s$''}) \
    \end{array} 
  $$
  where $\alpha=\sched(s)$.
  By definition of $\sched$, for each state $u$ with
  $P(s,\alpha,u)>0$:
\[
    \ExpRew{\sched}{\cE,u}(\textrm{``$\wgt$ until $s$''}) 
     \ =  \
    \ExpRew{\sched_1}{\cE,u}(\textrm{``$\wgt$ until $t$''}) 
    \ + \ 
    \ExpRew{\sched_2}{\cE,t}(\textrm{``$\wgt$ until $s$''}) 
=
    w_1(s,t) \, + \, w_2(t,s) 
    \ \ =  \ \ 
    w_1(s,t) \, - \, w_2(s,t) 
    \ \  > \ \ 0 \enspace.
\]
  Hence, $\ExpRew{\sched}{\cE}(\MP) >0$, which contradicts 
$\ExpRew{\max}{\cM}(\MP) =0$.
\end{proof}

Thus, each MDP as in Lemma~\ref{lemma:union-of-0-EC} has finitely many
\emph{maximal 0-ECs}, \ie 0-ECs $\cE$ such that there is no
0-EC $\cE'$ with $\cE \subseteq \cE'$ and $\cE \not= \cE'$.
Maximal 0-ECs are non-overlapping in the sense that they do not share
any state.
A further consequence of the existence of maximal 0-ECs is that
whenever $\cE_1, \cE_2$ are 0-ECs that share two states $s$ and $t$
then the weights of all paths from $s$ to $t$ in $\cE_1$ and $\cE_2$
are the same.

\subsubsection{Criterion for 0-BSCCs}

\label{sec:criterion-0-BSCC}

We start by the following observation for strongly connected Markov chains.

\begin{lemma}[Criterion for 0-BSCCs]
\label{lemma:BSCC-max-weight}
  Let $\cC$ be a strongly connected Markov chain with
  $\ExpRew{}{\cC}(\MP)=0$ such that
  \begin{equation}
    \label{dagger}
    \ExpRew{}{\cC,t}(\textrm{``$\wgt$ until $s$''})
    \ \ = \ \
    \ExpRew{}{\cC,u}(\textrm{``$\wgt$ until $s$''})
    \tag{$\dagger$}
  \end{equation}
  for all states $s,t,u$ in $\cC$ with
  $P(s,t)>0$, $P(s,u)>0$.
  Then, $\wgt(\cycle)=0$ for all cycles $\cycle$ in $\cC$.
\end{lemma}

\begin{proof}
Using Lemma \ref{lemma:scMC-exp-mp-quotient}, 
condition \eqref{dagger} together with $\ExpRew{}{\cC}(\MP)=0$ implies that
\begin{equation}
  \label{minus-wgt}
  \ExpRew{}{\cC,t}(\textrm{``$\wgt$ until $s$''}) \ \ =  \ \ -\wgt(s) 
  \tag{C1}
\end{equation}
for all states $s$, $t\in \cC$ with $P(s,t)>0$.

Let $\Post(v)=\{v'\in S : P(v,v')>0\}$ denote the set of
direct successors of a state $v$.
States $v_1,v_2$ are called \emph{siblings} if there is some state $v\in S$
such that $v_1,v_2 \in \Post(v)$.
We now show that \eqref{dagger} propagates to arbitrary siblings, \ie
whenever $s,v_1,v_2$ are states in $\cC$ then:
\begin{equation}
  \label{sibling}
  \ExpRew{}{\cC,v_1}(\textrm{``$\wgt$ until $s$''}) \ = \
  \ExpRew{}{\cC,v_2}(\textrm{``$\wgt$ until $s$''})
  \quad 
  \text{if $v_1,v_2$ are siblings}\enskip.
  \tag{C2}
\end{equation}
Suppose by contradiction that $s,v,v_1,v_2\in S$ are states
such that $v_1,v_2\in \Post(v)$ and 
$$
  \ExpRew{}{\cC,v_1}(\textrm{``$\wgt$ until $s$''}) \ > \
  \ExpRew{}{\cC,v_2}(\textrm{``$\wgt$ until $s$''})\enskip.
$$
By assumption \eqref{dagger}, we have $s\not= v$.

Let $\cC'$ be the Markov chain
 that has the same graph and weight structure
as $\cC$, but $P'(v,v_1) =P(v,v_1) +\delta$, 
$P'(v,v_2)=P(v,v_2)-\delta$ for some $\delta$ with
$0< \delta < \min \{ 1{-}P(v,v_1), P(v,v_2)\}$ and $P'(\cdot)=P(\cdot)$ 
in all other cases. 
As only the probabilities of the transitions from $v$ are modified, we get
for all states $t$
$$
 \ExpRew{}{\cC',t}(\textrm{``$\wgt$ until $v$''})
 \ \ = \ \ 
 \ExpRew{}{\cC,t}(\textrm{``$\wgt$ until $v$''})\enskip.
$$ 
In particular, if $t\in \Post(v)$ then
$\ExpRew{}{\cC',t}(\textrm{``$\wgt$ until $v$''}) = -\wgt(v)$
where we use \eqref{minus-wgt}.
Hence:
\begin{eqnarray*}
      \sum\limits_{t\in S} 
             P'(v,t) \cdot 
             \ExpRew{}{\cC',t}(\textrm{``$\wgt$ until $v$''}) 
\ \ \ =
      \sum\limits_{t\in \Post(v)} \!\!
        P'(v,t) \cdot 
        \underbrace{\ExpRew{}{\cC,t}(\textrm{``$\wgt$ until $v$''})}_{
            =-\wgt(v)} 
    \ \ \ = \ \ \ -\wgt(v) \enskip.
\end{eqnarray*}
But then the quotient representation of the mean payoff in $\cC'$
applied to state $v$ 
(Lemma \ref{lemma:scMC-exp-mp-quotient})
yields $\ExpRew{}{\cC'}(\MP)=0$. 

We now show that for all states $t\in S$:
\begin{equation}
  \label{plus}
  \ExpRew{}{\cC',t}(\textrm{``$\wgt$ until $s$''}) 
  \ \ \geqslant \ \ 
  \ExpRew{}{\cC,t}(\textrm{``$\wgt$ until $s$''}) \enskip.
  \tag{C3}
\end{equation}
To prove \eqref{plus}, we consider 
  the function $\Upsilon\colon\Real^{S}\to \Real^S$ defined as follows.
  If $f=(f_t)_{t\in S}$ is a vector then
  $
    \Upsilon(f) \ = \ \bigl(\Upsilon_t(f) \bigr)_{t\in S} 
  $
  where $\Upsilon_s(f)=0$ and for $t\in S \setminus \{s\}$:
  $
    \Upsilon_t(f) =
    \wgt(t) \ + \ \sum_{u\in S} P'(t,u)\cdot f_u.
  $
  Let $e=(e_t)_{t\in S} $ denote the vector
  with $e_t=\ExpRew{}{\cC,t}(\textrm{``$\wgt$ until $s$''})$,
  and $e' = (e'_t)_{t\in S}$ the corresponding vector for $\cC'$, \ie
  $e'_t= \ExpRew{}{\cC',t}(\textrm{``$\wgt$ until $s$''})$ for all states $t$.
  It is well known that: 
  \begin{center}
   \begin{tabular}{ll}
       (i) & $e'$ is the unique fixed point of $\Upsilon$ \\[1ex]
       (ii) \ \ & 
     $f \leqslant \Upsilon(f)$ implies $f \leqslant \Upsilon(f) \leqslant e'$ 
   \end{tabular}
  \end{center}
  where we use the 
  element-wise natural order on vectors, \ie 
  $(f_t)_{t\in S} \leqslant (g_t)_{t\in S}$ iff $f_t \leqslant g_t$ for all
  $t\in S$.
  Using the assumption $e_{v_1} > e_{v_2}$ we obtain: 
  \begin{eqnarray*}
     \Upsilon_v(e) & \ \ = \ \ & 
     \wgt(v) \ +  \sum_{t\in \Post(v)} \!\! P'(v,t) \cdot e_t
     \\
     \\[0ex]
     & = & 
     \wgt(v) \ +  \sum_{\stackrel{t\in \Post(v)}{t\notin \{v_1,v_2\}}} 
         \!\! P(v,t) \cdot e_t
     \ \ + \ \
     (P(v_1,t)+\delta)\cdot e_{v_1} \ + \ \ (P(v_2,t)-\delta)\cdot e_{v_2} 
     \\
     \\[0ex]
     & = & 
     \underbrace{\wgt(v) \ + 
                 \!\!\!\sum_{t\in \Post(v)}\!\!\! P(v,t) \cdot e_t}_{=e_v}
     \ + \ 
     \delta \cdot \underbrace{(e_{v_1}-e_{v_2})}_{>0} 
      \ \ \ > \ \ \ e_v
  \end{eqnarray*}
  and $\Upsilon_t(e)=e_t$ for $t\in S \setminus \{v\}$. 
  Thus, $e \leqslant \Upsilon(e)$ and therefore $e \leqslant e'$.
  This yields statement \eqref{plus}. 

  Moreover, the above calculation shows 
  $e_{v} < \Upsilon_{v}(e)$. This yields $e_v < e_v'$  by statement (ii).
  Hence, for all states $u \in S$:
\begin{equation}
  \label{plusplus}
  \ExpRew{}{\cC',u}(\textrm{``$\wgt$ until $s$''}) 
  \ \ > \ \ 
  \ExpRew{}{\cC,u}(\textrm{``$\wgt$ until $s$''})
  \quad
  \text{if $u \models \exists ((\neg s) \Until v)$}
  \tag{C4}
\end{equation}
where we use the CTL-notation $\exists ((\neg s) \Until v)$ to denote
the existence of a finite path to $v$ that does not traverse $s$.
In particular, \eqref{plusplus} holds for $u=v$.
As $\cC$ is strongly connected, state $v$ is accessible from $s$.
Let $\fpath = s_0 \, s_1 \ldots s_k$ be a shortest path from
$s=s_0$ to $s_k=v$, where ``shortest'' refers to the standard length
(number of transitions) rather than the accumulated weight. 
As $s\not= v$ we have $k \geqslant 1$.
Moreover, $k=1$ iff $v=s_1 \in \Post(s)$. Otherwise, \ie if $k\geqslant 2$,
then $s_1 \ldots s_k$ is a path from state $s_1=u\in \Post(s)$ to $v$
that does not traverse $s$. 
In both cases, $u \models \exists ((\neg s) \Until v)$ for some state
$u\in \Post(s)$.
As $s \not= v$ we have $P'(s,t)=P(s,t)$ for all states $t$.
By \eqref{plus} and \eqref{plusplus} we obtain:
\begin{eqnarray*}
  & \ \ &
  \wgt(s) \ + \ 
  \sum_{t\in S} P'(s,t)\cdot \ExpRew{}{\cC',t}(\textrm{``$\wgt$ until $s$''}) 
  \\
  \\[0ex]
  = & & 
  \wgt(s) \ +  \!
  \sum_{\stackrel{t\in \Post(s)}{t\not= u}} \!\!\!
     P(s,t)\cdot 
     \underbrace{\ExpRew{}{\cC',t}(\textrm{``$\wgt$ until $s$''})}_{
       \geqslant \ExpRew{}{\cC,t}(\textrm{``$\wgt$ until $s$''})}
  \ \ + \ \ 
     P(s,u)\cdot 
     \underbrace{\ExpRew{}{\cC',u}(\textrm{``$\wgt$ until $s$''})}_{
        > \ExpRew{}{\cC,u}(\textrm{``$\wgt$ until $s$''})} 
  \\
  \\[0ex]
  > & & 
  \wgt(s) \ + \ 
  \sum_{t\in S} P(s,t)\cdot 
    \underbrace{\ExpRew{}{\cC,t}(\textrm{``$\wgt$ until $s$''})}_{=-\wgt(s)}
  \ \ 
  \stackrel{\text{\eqref{minus-wgt}}}{=} \ \ 0
\end{eqnarray*}
The quotient representation of the mean payoff 
(Lemma \ref{lemma:scMC-exp-mp-quotient})
yields $\ExpRew{}{\cC'}(\MP)>0$. Contradiction.
This completes the proof of statement \eqref{sibling}.

We now fix a state $s$ and 
show by induction on the length (number of transitions) of 
paths $\cycle$ starting from $s$ that
\begin{equation}
  \label{exp-wgt-last-cycle}
  \ExpRew{}{\cC,\last(\cycle)}(\textrm{``$\wgt$ until $s$''})
  \ \ = \ \ 
  -\wgt(\cycle)
  \tag{C5}
\end{equation}
In the basis of induction we consider a path of length 1, \ie $\cycle$
consists of a single transition from $s$ to some state $t\in \Post(s)$.
But then $\wgt(\cycle)=\wgt(s)$ and the claim follows directly from 
\eqref{minus-wgt}.
In the step of induction $k \Longrightarrow k{+}1$,
we regard a path $\cycle = s_0 \, s_1 \ldots s_k \, s_{k+1}$ 
of length $k{+}1$ starting in
$s_0=s$. By induction hypothesis we have: 
$$
  \ExpRew{}{\cC,s_k}(\textrm{``$\wgt$ until $s$''})
  \ \ = \ \ 
  -\wgt(s_0 \, s_1 \ldots s_k) 
$$
Recall that $\wgt(s_0 \, s_1 \ldots s_k) = \sum_{i=0}^{k-1} \wgt(s_i)$. 
Suppose first $s_k=s$. Then,
$\ExpRew{}{\cC,s_k}(\textrm{``$\wgt$ until $s$''}) =0=\wgt(s_0\ldots s_k)$,
in which case $\wgt(\cycle)=\wgt(s_k)$ and the claim follows directly from
assumption \eqref{dagger}.
Suppose now $s\not= s_k$. By \eqref{sibling}, we get:
\begin{center}
   $\ExpRew{}{\cC,s_{k+1}}(\textrm{``$\wgt$ until $s$''})
   \ = \ 
   \ExpRew{}{\cC,t}(\textrm{``$\wgt$ until $s$''})$
   \ for all $t\in \Post(s_k)$.
\end{center}
Hence:
\begin{eqnarray*}
  -\sum_{i=0}^{k-1} \wgt(s_i) & \ \ = \ \ &
  \ExpRew{}{\cC,s_k}(\textrm{``$\wgt$ until $s$''})
  \\
  \\[0ex]
  & = &
  \wgt(s_k) \ + \!\! \sum_{t\in \Post(s_k)}\!\!\! 
      P(s_k,t) \cdot 
      \underbrace{\ExpRew{}{\cC,t}(\textrm{``$\wgt$ until $s$''})}_{
        = \ExpRew{}{\cC,s_{k+1}}(\textrm{``$\wgt$ until $s$''})}
  \\
  \\[0ex]
  & = &
  \wgt(s_k) \ + 
  \ExpRew{}{\cC,s_{k+1}}(\textrm{``$\wgt$ until $s$''})
    \cdot \!\!\!
    \underbrace{\sum_{t\in \Post(s_k)}\!\!\! P(s_k,t)}_{=1}
  \\
  \\[0ex]
  & = &
  \wgt(s_k) \ + \ \ExpRew{}{\cC,s_{k+1}}(\textrm{``$\wgt$ until $s$''})
\end{eqnarray*} 
We conclude:
$$
  -\wgt(\cycle) \ \ = \ \ 
  -\sum_{i=0}^{k} \wgt(s_i) \ \ = \ \ 
  \ExpRew{}{\cC,s_{k+1}}(\textrm{``$\wgt$ until $s$''})
$$
This completes the proof of the induction step.  

We finally use statement \eqref{exp-wgt-last-cycle} to show that
$\cC$ is a 0-BSCC. Let $\cycle$ be a cycle in $\cC$ and $s$ an arbitrary
state on $\cycle$. Statement \eqref{exp-wgt-last-cycle} yields:
$
  -\wgt(\cycle)
  \ \ = \ \
  \ExpRew{}{\cC,s}(\textrm{``$\wgt$ until $s$''})
  \ \ = \ \ 0,
$
which completes the proof.
\end{proof}

The goal is to apply the observation 
above on the relation between expected mean payoff and expected accumulated
weights until reaching a target in Markov chains
for checking the existence of 0-EC in strongly connected MDPs with
$\ExpRew{\max}{\cM}(\MP)=0$. 
We start with the following observation stating that MD-schedulers
with expected mean payoff 0 maximize the expected accumulated
weight until reaching any state $s$ of its BSCCs. 
Here, the maximum is taken over 
all MD-schedulers $\sched$ where $s$ belongs to a BSCC of $\sched$.%
\footnote{%
  The supremum of the expectations of 
  ``$\wgt$ until $s$'' under all schedulers is infinite if $\cM$ has gambling
  end components that do not contain $s$. This is, however, irrelevant for our
  purposes as we are only interested in the maximal expectation of
  ``$\wgt$ until $s$'' when ranging over MD-schedulers under which 
  the long-run frequency of $s$ is positive.}
More precisely:

\begin{lemma}
\label{lemma:wgt-exp-equals-zero}
  Let $\cM$ be a strongly connected MDP with $\ExpRew{\max}{\cM}(\MP)=0$ and
  let $\sched$ be an MD-scheduler with
  $\ExpRew{\sched}{\cM,s}(\MP)=0$ for all states $s$. 
  Furthermore, let $\cB$ be a BSCC
  of the Markov chain induced by $\sched$ and $(s,\alpha)$ a state-action pair
  in $\cB$ (\ie $s$ is a state of $\cB$ and $\alpha=\sched(s)$). 
  Then:
  \begin{equation}
       \label{ddagger}
       \wgt(s,\alpha) \ + \ 
          \sum\limits_{t\in \cB} P(s,\alpha,t) \cdot 
                \ExpRew{\sched}{\cM,t}(\textrm{``$\wgt$ until $s$''})
       \ \ = \ \ 0.
       \tag{$\ddagger$}
  \end{equation}
  Moreover, whenever $\tsched$ is an MD-scheduler
  for $\cM$ with  $\tsched(s)=\alpha$ where
  $s$ belongs to a BSCC of the Markov chain $\cC_{\tsched}$ 
  induced by $\tsched$, then for all states $t$ with $P(s,\alpha,t)>0$
  $$
    \ExpRew{\sched}{\cM,t}(\textrm{``$\wgt$ until $s$''})
    \ \ \geqslant \ \ \ExpRew{\tsched}{\cM,t}(\textrm{``$\wgt$ until $s$''})\enskip.
  $$
\end{lemma}

\begin{proof}
  The first part (statement \eqref{ddagger})
  follows directly from Lemma \ref{lemma:scMC-exp-mp-quotient}
  applied to $\cB$ viewed as a strongly connected Markov chain.
  For the second part, we assume that
  $\tsched$ is some MD-scheduler
  with  $\tsched(s)=\alpha$  where
  $s$ belongs to a BSCC $\cB'$ of $\tsched$.
  Suppose by contradiction that there is some state $t$ with
  $P(s,\alpha,t)>0$ and
  $\ExpRew{\sched}{\cM,t}(\textrm{``$\wgt$ until $s$''})
   \ < \ \ExpRew{\tsched}{\cM,t}(\textrm{``$\wgt$ until $s$''})$.
  Let $\vsched$ denote the following scheduler operating in two modes:
  In its first mode, $\vsched$ behaves as $\sched$. 
  It switches to the second mode
  when entering $t$ via the $\alpha$-transition from $s$.
  In its second mode, $\vsched$ behaves as $\tsched$ until it visits $s$,
  in which case it switches back to its first mode 
  where it behaves as $\sched$.
  Note that $\vsched$ is not memoryless, but a finite-memory scheduler.
  Hence, the Markov chain $\cC_{\vsched}$ 
  induced by $\vsched$ is finite as well.
  Lemma \ref{lemma:scMC-exp-mp-quotient} applied to the $\cC_{\vsched}$
  yields:
  $$
     \ExpRew{\vsched}{\cM,s}(\MP)
     \ \ = \ \  
     \begin{array}{rclcl}
           \wgt(s,\alpha) & + & 
           P(s,\alpha,t)\cdot \ExpRew{\tsched}{\cE,t}(\textrm{``$\wgt$ until $s$''})
           & + & 
           \sum\limits_{u\not= t}     
            P(s,\alpha,u)\cdot \ExpRew{\sched}{\cE,u}(\textrm{``$\wgt$ until $s$''})
           \\[1ex]
           \hline 
           \\[-2.2ex]
           1 \ & + & 
           P(s,\alpha,t)\cdot \ExpRew{\tsched}{\cE,t}(\textrm{``steps until $s$''})
           & + &
           \sum\limits_{u\not= t}     
           P(s,\alpha,u)\cdot \ExpRew{\sched}{\cE,u}(\textrm{``steps until $s$''})
    \end{array}\enskip.
  $$
  Using $\ExpRew{\tsched}{\cE,t}(\textrm{``$\wgt$ until $s$''}) > 
  	\ExpRew{\sched}{\cE,t}(\textrm{``$\wgt$ until $s$''})$ and 
   \eqref{ddagger}, we get
  for the numerator of $\ExpRew{\vsched}{\cM,s}(\MP)$:
  \begin{eqnarray*}
    & \ \ \ &
    \wgt(s,\alpha) + 
    P(s,\alpha,t)\cdot \ExpRew{\tsched}{\cE,t}(\textrm{``$\wgt$ until $s$''})
    + 
           \sum_{u\not= t}     
            P(s,\alpha,u)\cdot \ExpRew{\sched}{\cE,u}(\textrm{``$\wgt$ until $s$''})
    \\
    \\[0ex]
    > & &
    \wgt(s,\alpha) + 
    P(s,\alpha,t)\cdot \ExpRew{\tsched}{\cE,t}(\textrm{``$\wgt$ until $s$''})
    + 
    \sum_{u\not= t}     
            P(s,\alpha,u)\cdot  \ExpRew{\sched}{\cE,u}(\textrm{``$\wgt$ until $s$''})
    \\
    \\[0ex]
    = & &
    \wgt(s,\alpha) + 
    \sum_{u\in S}     
            P(s,\alpha,u)\cdot  \ExpRew{\sched}{\cE,u}(\textrm{``$\wgt$ until $s$''})
    \ \ \ = \ \ \ 0
  \end{eqnarray*}
  This yields $\ExpRew{\vsched}{\cM,s}(\MP) >0$, which
  is impossible as
  $\ExpRew{\vsched}{\cM,s}(\MP)
    \leqslant \ExpRew{\max}{\cM}(\MP)=0$.
\end{proof}

\subsubsection{Algorithm to Check the Existence of 0-ECs}

\label{sec:algo-checking-0-EC}
\label{algo:existence-0-EC}

We show that given a strongly connected MDP $\cM$ with
$\ExpRew{\max}{\cM}(\MP)=0$, the task to 
decide the existence of 0-ECs is solvable by an algorithm 
that runs in time in polynomial in the size
of the given MDP. 

For this, we rely on the observation that the property of being a 0-EC does 
not depend on the precise transition probabilities, but only on the graph
structure and the weights.
The idea is now to modify the transition probabilities of a state-action pair
$(s,\alpha)$ in a gambling BSCC of $\cM$ such that
the transformed MDP $\cM'$ has the same graph structure and weights 
(thus, $\cM$ and $\cM'$ have the same 0-ECs) and
$\cM'$ enjoys the following property:
\begin{enumerate}
\item []
  Whenever $\sched$ is an MD-scheduler for $\cM$ and $\sched(s)=\alpha$
  such that $\sched$ is gambling in $\cM$, 
  then $\ExpRew{\sched}{\cM',s}(\MP) <0$.
\end{enumerate}
Thus, the gambling BSCCs of $\cM$ containing the state-action pair $(s,\alpha)$
are no longer gambling in $\cM'$.
Hence, the only end components in $\cM'$ 
with maximal mean payoff 0 are the 0-ECs.
This then ensures that
$\ExpRew{\max}{\cM'}(\MP) =0$ if and only if $\cM$ (and $\cM'$) have a 0-EC.

The algorithm for checking the existence of a 0-EC
in a strongly connected MDP $\cM$ with
$\ExpRew{\max}{\cM}(\MP)=0$  proceeds as follows. 
It first runs a standard polynomial-time 
algorithm to compute
an MD-scheduler $\sched$ with 
$\ExpRew{\sched}{\cM,s}(\MP)=0$ for all states $s$.
We may assume w.l.o.g. that 
the Markov chain $\cC_{\sched}$ induced by $\sched$ has a single BSCC
$\cB$.
We then check in polynomial time whether $\cB$ is a 0-BSCC.
(For this, we can rely on Lemma \ref{weight-div-MC}
and check the nonexistence of positive cycles
using standard graph algorithms.)
If so, then $\cB$ is a 0-EC of $\cM$ and 
the algorithm terminates by returning $\cB$.
Otherwise, there exist states $s, t, u$ in $\cB$ such that
$P(s,\alpha,t)>0$, $P(s,\alpha,u)>0$ where $\alpha = \sched(s)$
and (see Lemma \ref{lemma:BSCC-max-weight})
$$
  \ExpRew{\sched}{t}(\textrm{``$\wgt$ until $s$''})
  \ \ > \ \ 
  \ExpRew{\sched}{u}(\textrm{``$\wgt$ until $s$''})\enskip.
$$
Let now $\cM'$ be the MDP resulting from $\cM$ by changing the transition
probabilities for the state-action pair $(s,\alpha)$ as follows.
We pick a value $\delta >0$ such that $P(s,\alpha,t)-\delta >0$
and $P(s,\alpha,u)+\delta  < 1$ and define:
$$
   P'(s,\alpha,t) \ =  \ P(s,\alpha,t)-\delta, \qquad
   P'(s,\alpha,u) \ =  \ P(s,\alpha,u)+\delta
$$
and $P'(s,\alpha,s')=P(s,\alpha,s')$ for all other states
$s'\in S \setminus \{t,u\}$.
The transition probabilities for all other state-action pairs
as well as the weight function remain unchanged.
We then have $\ExpRew{\sched}{\cM',s'}(\MP) <0$ 
for all states $s'$ in $\cB$
and $\ExpRew{\tsched}{\cM',s'}(\MP) \leqslant \ExpRew{\sched}{\cM,s'}(\MP)$ 
for all states $s'$ and all MD-schedulers $\tsched$ for $\cM'$
(see Lemma~\ref{soundness-0-EC-algo} below).
In particular, $\ExpRew{\max}{\cM'}(\MP)\leqslant 0$.
We then call again an algorithm to compute an MD-scheduler
$\sched'$ for $\cM'$ that maximizes the expected mean payoff.
If $\ExpRew{\max}{\cM'}(\MP) < 0$ then the algorithm terminates with the
answer ``no, $\cM$ has no 0-EC''. 
Otherwise, the algorithm repeats to modify the transition probabilities
of a state-action pair $(s',\alpha')$ in some BSCC of $\sched'$, and so on.

The presented algorithm terminates after at most $|S|\cdot |\Act|$ steps 
as the transition probabilities of each state-action pair are perturbed 
at most once (see Lemma \ref{soundness-0-EC-algo} below). 
The cost of iteration are dominated by the cost for computing an MD-scheduler
$\sched$ maximizing the expected mean payoff and the values
$\Exp{\sched}{\cM,s'}(\text{``$\wgt$ until $s$'})$.
Thus, the time complexity is polynomial in the size of $\cM$.

\begin{lemma}[Soundness of the transformation]
  \label{soundness-0-EC-algo}
  Let $\cM$ be a strongly connected MDP with
  $\ExpRew{\max}{\cM}(\MP)=0$ and let $\sched$ and $\cM'$ be as above,
  modifying the transition probabilities for $(s,\alpha)$. Then,
  for all MD-schedulers $\tsched$ for $\cM$ (and $\cM'$):
  \begin{enumerate}
  \item [(a)]
     If $s$ belongs to a BSCC of $\cC_{\tsched}$ 
     and $\tsched(s) =\alpha$
     then $\ExpRew{\tsched}{\cM',s}(\MP) <0$.
  \item [(b)]
     If $\cB$ is a BSCC of $\tsched$ such that either 
     $s$ does not belong to $\cB$ 
     or $\tsched(s) \not= \alpha$
     then $\ExpRew{\tsched}{\cM',s'}(\MP) = \ExpRew{\tsched}{\cM,s'}(\MP)$
     for all states $s' \in \cB$.
  \end{enumerate}
  In particular, we have $\ExpRew{\max}{\cM'}(\MP)\leqslant 0$.
\end{lemma}

\begin{proof}
 Statement (b) is obvious, as in this case, $\cB$ is not affected by the 
 switch from $\cM$ to $\cM'$.
 For the proof of statement (a) we rely on the second part of
 Lemma \ref{lemma:wgt-exp-equals-zero} yielding for all states $s'$ with $P(s,\alpha,s')>0$ that
 $$
    \ExpRew{\tsched}{\cM',s'}(\textrm{``$\wgt$ until $s$''})
    \ \ = \ \ 
    \ExpRew{\tsched}{\cM,s'}(\textrm{``$\wgt$ until $s$''})
    \ \ \leqslant \ \
    \ExpRew{\sched}{\cM,s'}(\textrm{``$\wgt$ until $s$''})\enskip.
 $$
 But then:
 \begin{eqnarray*}
    & \ \ &
    \wgt(s,\alpha) \ + \
    \sum_{s'\in S} \, P'(s,\alpha,s') \cdot
    \ExpRew{\tsched}{\cM',s'}(\textrm{``$\wgt$ until $s$''})
    \\
    \\[0ex]
    \leqslant & \ \ &
    \wgt(s,\alpha) \ + \!\!
     \sum_{s'\in S\setminus \{t,u\}} \!\!\!\!
        P(s,\alpha,s') \cdot
        \ExpRew{\sched}{\cM,s'}(\textrm{``$\wgt$ until $s$''})
    \\
    \\[0ex]
    & & 
    \ \ \ + \ \
    (P(s,\alpha,t) -\delta) \cdot 
       \ExpRew{\sched}{\cM,t}(\textrm{``$\wgt$ until $s$''})
    \\
    \\[0ex]
    & & 
    \ \ \ + \ \
    (P(s,\alpha,u) +\delta) \cdot 
       \ExpRew{\sched}{\cM,u}(\textrm{``$\wgt$ until $s$''})
    \\
    \\[0ex]
    = & &
     \wgt(s,\alpha) \ + \
    \sum_{s'\in S} P(s,\alpha,s') \cdot
    \ExpRew{\sched}{\cM',s'}(\textrm{``$\wgt$ until $s$''})
    \\
    \\[-0.5ex]
    & & 
    \ \ \ - \ \
    \delta \cdot \bigl(\ 
                  \ExpRew{\sched}{\cM,t}(\textrm{``$\wgt$ until $s$''})
                  -
                  \ExpRew{\sched}{\cM,u}(\textrm{``$\wgt$ until $s$''}) \ 
                 \bigr)
    \\
    \\[0ex]
    < & &
     \wgt(s,\alpha) \ + \
    \sum_{s'\in S} P(s,\alpha,s') \cdot
    \ExpRew{\sched}{\cM',s'}(\textrm{``$\wgt$ until $s$''})
    \ \ \ = \ \ \ 0\enskip.
\end{eqnarray*}
Here, we use statement \eqref{ddagger} of 
Lemma \ref{lemma:wgt-exp-equals-zero} and the facts that $\delta >0$ and
$\ExpRew{\sched}{\cM,t}(\textrm{``$\wgt$ until $s$''})
 > \ExpRew{\sched}{\cM,u}(\textrm{``$\wgt$ until $s$''})$.
Hence, by Lemma~\ref{lemma:scMC-exp-mp-quotient} we obtain
$$
  \ExpRew{\tsched}{\cM'}(\MP)
  \enskip = \enskip
  \begin{array}{rcl}
    \wgt(s,\alpha) & + &
    \sum\limits_{s'\in S} P'(s,\alpha,s') \cdot
       \ExpRew{\tsched}{\cM',s'}(\textrm{``$\wgt$ until $s$''})
    \\[1ex]
    \hline
    \\[-2.2ex]
     1 \ & + &
    \sum\limits_{s'\in S} P'(s,\alpha,s') \cdot
      \ExpRew{\tsched}{\cM',s'}(\textrm{``steps until $s$''})
  \end{array}
  \ \ < \ \ 0 \enskip.
$$
\end{proof}

\subsubsection{Complexity of Checking the Existence of 0-ECs}

\label{sec:complexity-existence-0-ECs}

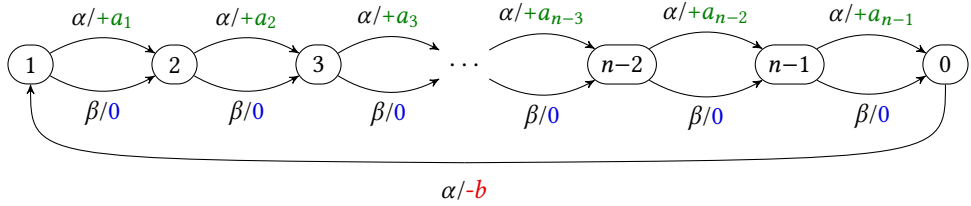
\begin{figure}[ht]
	\centering%

\begin{tikzpicture}
	\node[state] (s1) {$1$};
	\node[state, right=13mm of s1] (s2) {$2$};
	\node[state, right=13mm of s2] (s3) {$3$};
	\node[istate, right=13mm of s3] (dots) {$\ldots$};
	\node[state, right=13mm of dots] (sn2) {$n{-}2$};
	\node[state, right=13mm of sn2] (sn1) {$n{-}1$};
	\node[state, right=13mm of sn1] (s0) {$0$};
	\path%
		(s1) edge[ptran, bend left] node[above] {$\alpha$/\posw{$a_1$}} (s2)
		(s2) edge[ptran, bend left] node[above] {$\alpha$/\posw{$a_2$}} (s3)
		(s3) edge[ptran, bend left] node[above] {$\alpha$/\posw{$a_3$}} (dots)
		(dots) edge[ptran, bend left] node[above] {$\alpha$/\posw{$a_{n-3}$}} (sn2)
		(sn2) edge[ptran, bend left] node[above] {$\alpha$/\posw{$a_{n-2}$}} (sn1)
		(sn1) edge[ptran, bend left] node[above] {$\alpha$/\posw{$a_{n-1}$}} (s0)
		(s1) edge[ptran, bend right] node[below] {$\beta$/\nilw} (s2)
		(s2) edge[ptran, bend right] node[below] {$\beta$/\nilw} (s3)
		(s3) edge[ptran, bend right] node[below] {$\beta$/\nilw} (dots)
		(dots) edge[ptran, bend right] node[below] {$\beta$/\nilw} (sn2)
		(sn2) edge[ptran, bend right] node[below] {$\beta$/\nilw} (sn1)
		(sn1) edge[ptran, bend right] node[below] {$\beta$/\nilw} (s0);
	\coordinate[below=10mm of s1] (bs1);
	\coordinate[below=10mm of s0] (bs0);
	\coordinate[below=10mm of dots] (bdots);
	\node[below=1mm of bdots] {$\alpha$/\negw{$b$}};
	\draw[-] 
	  	(s0.south) .. controls (bs0) .. (bdots);
	\draw[ptran]	
		(bdots) .. controls (bs1) .. (s1.south);				
\end{tikzpicture}
 	\caption{Reduction from subset sum for the \NP-hardness of checking the existence of 0-ECs.}
\label{fig:complexity-checking-existence-zero-ecs}
\end{figure}

\thmcompZec*

\begin{proof}
  A polynomial-time procedure for
  checking the existence of a~0-EC in strongly connected MDPs 
  with $\ExpRew{\max}{\cM}(\MP)=0$
  has been presented in
  Section \ref{algo:existence-0-EC}.

  We now show the \NP-completeness of the general case.
  An \NP-algorithm is obtained by nondeterministically guessing 
  an MD-scheduler and checking whether the induced Markov chain has a 0-BSCC.
  \NP-hardness can be easily obtained via a reduction from the subset 
  sum problem:
  We are given a finite sequence of nonnegative 
  integers $a_1,\ldots a_{n-1},b$
  and the task is to find a subset $I$ of $\{1,\ldots,n{-}1\}$ 
  with $\sum_{i\in I} a_i=b$.
  For this, we regard the MDP $\cM$ with state space $\{0,1,\ldots,n{-}1\}$
  illustrated by Figure~\ref{fig:complexity-checking-existence-zero-ecs}.
  Each state $i\in \{1,\ldots,n{-}1\}$ has two actions $\alpha$ and $\beta$
  with $\wgt(i,\alpha)=a_i$, $\wgt(i,\beta)=0$,
  as well as $P(i,\alpha,(i{+}1)\mod n)=  P(i,\beta,(i{+}1)\mod n)= 1$.
  In state 0, only action $\alpha$ is enabled with $\wgt(0,\alpha)=-b$ and
  $P(0,\alpha,1)=1$. In all remaining cases, we have $P(\cdot)=0$.
  Then, $\cM$ is strongly connected and each subset $I$ of $\{1,\ldots,n{-}1\}$
  induces an end component $\cE_I$ 
  consisting of the state-action pairs $(0,\alpha)$ and
  $(i,\alpha)$ for $i\in I$ and $(i,\beta)$ 
  for $i\in \{1,\ldots,n{-}1\}\setminus I$. The BSCCs of MD-schedulers
  are exactly the end components $\cE_I$ for $I \subseteq \{1,\ldots,n{-}1\}$.
  Moreover, $\cE_I$ is a 0-EC iff $\sum_{i\in I} a_i=b$.
  Hence, $\cM$ has a 0-EC iff there exists a subset $I$ of $\{1,\ldots,n{-}1\}$
  with  $\sum_{i\in I} a_i=b$.
\end{proof}

\subsection{Spider Construction and Weight-Divergence Algorithm}

\label{sec:appendix-spider}

Recall that the purpose of the spider construction $\cM \leadsto \spider{\cM}{\cE}$
was to flatten a 0-BSCC $\cE$ in the MDP $\cM$. For a detailed presentation in this
appendix, we recall the construction from the main paper:

As $\cE$ is a 0-BSCC, for each state $s$ in $\cE$ there is a single
state-action pair $(s,\alpha_s)\in \cE$.
Given two states $s,t$ in $\cE$, recall that
$w(s,t)$ denotes the weight of every path from $s$ to $t$ in $\cE$.
The spider construction picks a reference state $s_0$ in $\cE$. 
Then, the MDP $\cN=\spider{\cM}{\cE,s_0}$ arises from $\cM$ by performing the
following steps for each state $s$ that appears in $\cE$:
\begin{description}
 \item [(i)]
   Remove the state-action pair $(s,\alpha_s)$
 \item [(ii)]
   In case $s\neq s_0$, add a new state-action pair $(s,\tau)$ with 
   $P_{\cN}(s,\tau,s_0)=1$ and $\wgt_{\cN}(s,\tau)=w(s,s_0)$
 \item [(iii)]
   In case $s\neq s_0$, replace every state-action pair $(s,\beta)\in \cM$ 
   with $\beta\neq\alpha_s$ by $(s_0,\beta)$. The transition probabilities and 
   weights of these state-action pairs are given by
    $P_{\cN}(s_0,\beta,u)=P_{\cM}(s,\beta,u)$ for all states $u$
    and $\wgt_{\cN}(s_0,\beta)=w(s_0,s) + \wgt_\cM(s,\beta)$
\end{description}
Recall that we often write $\spider{\cM}{\cE}$ rather than $\spider{\cM}{\cE,s_0}$ 
when the reference state $s_0$ is clear from the context or irrelevant, \eg 
in case the spider construction is used as a vehicle to reduce the MDP's
number of state-action pairs where the actual structure of the arising
graph is not of interest. 
As an example to illustrate how the choice of the reference state influences the
graph structure of the MDP, let us return to the MDP of Example~\ref{example:spider}. 
When taking state $s$ instead of state $t$ as first reference state and then state $u$
instead of $t$, we obtain the MDPs $\cM_1^s=\spider{\cM}{\cE,s}$ and
$\cM_2^u=\spider{\cM_1^s}{\cE,u}$ as depicted in Figure~\ref{fig:example-spider-s}. 
Note that by changing the reference state during an iterative application of the
spider construction, chains of $\tau$ transitions may arise.
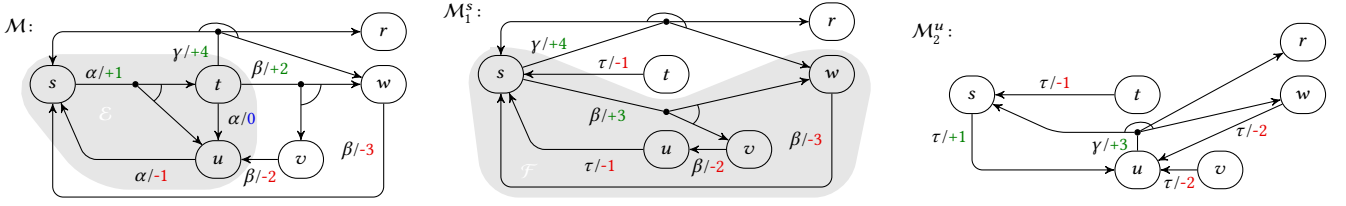
\begin{figure*}[h]
	\begin{tabular}{ccc}
\pgfdeclarelayer{background}%
\pgfsetlayers{background,main}%
\begin{tikzpicture}[x=11mm,y=10mm,font=\scriptsize]
	\node (s) at (0,0) [state] {$s$};
	\node (stu) at (1,0) [bullet] {};
	\node (t) at (2,0) [state] {$t$};
	\node (tsrw) at (2,.7) [bullet] {};
	\node (twv) at (3,0) [bullet] {};
	\node (w) at (4,0) [state] {$w$};
	\node (r) at (4,.7) [state] {$r$};
	\node (u) at (2,-1) [state] {$u$};
	\node (v) at (3,-1) [state] {$v$};
	
    \draw [ntran] (s) -- node[above=-.05,pos=.5]{$\alpha$/\posw{1}} (stu);
    \draw [ptran] (stu) -- coordinate[pos=.4] (bt) (t);
    \draw [ptran] (stu) -- coordinate[pos=.25] (bu) (u);
    \draw (bt) to[bend left] (bu);
    \draw [ptran] (t) -- node[right,pos=.4]{$\alpha$/\nilw} (u);
    \draw [ntran] (t) -- node[above=-.05,pos=.5]{$\beta$/\posw{2}} (twv);
    \draw [ptran] (twv) -- coordinate[pos=.3] (bww) (w);
    \draw [ptran] (twv) -- coordinate[pos=.33] (bvv) (v);
    \draw (bvv) to[bend right=40] (bww);
    \draw [ptran] (v) -- node[below,pos=.5]{$\beta$/\negw{2}} (u);
    \draw [ptran] (u) -- node[below,pos=.4]{$\alpha$/\negw{1}} +(-1.6,0) -- (s);
    \draw [ntran] (t) -- node[left,pos=.4]{$\gamma$/\posw{4}} (tsrw);
    \draw [ptran] (tsrw) -- coordinate[pos=.1] (bs) +(-2,0) -- (s);
    \draw [ptran] (tsrw) -- coordinate[pos=.12] (bw) (w);
    \draw [ptran] (tsrw) -- (r);
    \draw (bs) to[bend left=80] (bw);
    \draw [ptran] (w) -- 
    	node[left,pos=.5]{$\beta$/\negw{3}} +(0,-1.5) -- (0,-1.5) -- (s);
    \begin{pgfonlayer}{background}%
    	\node[above left=.3 and 0 of s]{$\mathcal{M}\colon$};
        \draw[rounded corners=1em,line width=2em,black!10,fill=black!10]
        	(-.2,.05) -- (2.15, .05) -- (2.15,-1.05) -- (.6, -1.05) --
			node[right,pos=.6,white]{$\mathcal{E}$} cycle;
    \end{pgfonlayer}
\end{tikzpicture}

 &
\pgfdeclarelayer{background}%
\pgfsetlayers{background,main}%
\begin{tikzpicture}[x=11mm,y=10mm,font=\scriptsize]
	\node (s) at (0,0) [state] {$s$};
	\node (t) at (2,0) [state] {$t$};
	\node (tsrw) at (2,.7) [bullet] {};
	\node (twv) at (2,-.5) [bullet] {};
	\node (w) at (4,0) [state] {$w$};
	\node (r) at (4,.7) [state] {$r$};
	\node (u) at (2,-1) [state] {$u$};
	\node (v) at (3,-1) [state] {$v$};
	
    \draw [ptran] (t) -- node[above=-.05,pos=.25]{$\tau$/\negw{1}} (s);
    \draw [ntran] (s) -- node[below,pos=.6]{$\beta$/\posw{3}} (twv);
    \draw [ptran] (twv) -- coordinate[pos=.2] (bww) (w);
    \draw [ptran] (twv) -- coordinate[pos=.4] (bvv) (v);
    \draw (bvv) to[bend right=40] (bww);
    \draw [ptran] (v) -- node[below,pos=.5]{$\beta$/\negw{2}} (u);
    \draw [ptran] (u) -- node[below,pos=.4]{$\tau$/\negw{1}} +(-1.6,0) -- (s);
    \draw [ntran] (s) -- node[left=.2,pos=.5]{$\gamma$/\posw{4}} (tsrw);
    \draw [ptran] (tsrw) -- coordinate[pos=.1] (bs) +(-2,0) -- (s);
    \draw [ptran] (tsrw) -- coordinate[pos=.12] (bw) (w);
    \draw [ptran] (tsrw) -- (r);
    \draw (bs) to[bend left=80] (bw);
    \draw [ptran] (w) -- 
    	node[left,pos=.5]{$\beta$/\negw{3}} +(0,-1.5) -- (0,-1.5) -- (s);
    \begin{pgfonlayer}{background}%
    	\node[above left=.3 and 0 of s]{$\mathcal{M}_1^s\colon$};
        \draw[rounded corners=1em,line width=.8em,black!10,fill=black!10]
        	(.2,.25) -- (2, -.55) -- (3.8,.25) -- (4.4,.1) -- (4.1, -1.5) -- 
			node[above=-.4em,pos=.9,white]{$\mathcal{F}$} (-.1, -1.5) -- (-.3,.2) -- cycle;
    \end{pgfonlayer}
\end{tikzpicture}

 &
\pgfdeclarelayer{background}%
\pgfsetlayers{background,main}%
\begin{tikzpicture}[x=11mm,y=10mm,font=\scriptsize]
	\node (s) at (0,0) [state] {$s$};
	\node (t) at (2,0) [state] {$t$};
	\node (tsrw) at (2,-.5) [bullet] {};
	\node (w) at (4,0) [state] {$w$};
	\node (r) at (4,.7) [state] {$r$};
	\node (u) at (2,-1) [state] {$u$};
	\node (v) at (3,-1) [state] {$v$};
	
    \draw [ptran] (t) -- node[above=-.05,pos=.5]{$\tau$/\negw{1}} (s);
    \draw [ptran] (w) -- node[below=-.05,pos=.25]{$\tau$/\negw{2}} (u);
    \draw [ptran] (v) -- node[below=-.05,pos=.5]{$\tau$/\negw{2}} (u);
    \draw [ptran] (s) -- node[left=-.05,pos=.4]{$\tau$/\posw{1}} (0,-1) -- (u);
    \draw [ntran] (u) -- node[left=0,pos=.3]{$\gamma$/\posw{3}} (tsrw);
    \draw [ptran] (tsrw) -- coordinate[pos=.12] (bs) (1,-.5) -- (s);
    \draw [ptran] (tsrw) -- coordinate[pos=.09] (bw) (w);
    \draw [ptran] (tsrw) -- (r);
    \draw (bs) to[bend left=80] (bw);
    	\node[above left=.3 and 0 of s]{$\mathcal{M}_2^u\colon$};
\end{tikzpicture}

 	\end{tabular}
	\caption{\label{fig:example-spider-s}Spider constructions with $s$ and $u$ as reference states.}
\end{figure*}

In the following, we suppose that in a preprocessing step the actions in $\cM$ are
renamed such that $\Act_{\cM}(s)\cap\Act_{\cM}(t)=\varnothing$
for all states $s,t$ in $\cM$ with $s \not= t$. This technical requirement
will be used in upcoming proofs and can be achieved by simply renaming actions.

\subsubsection{Properties of the Spider Construction}

\label{sec:prop-spider}

We now prove that the MDP $\spider{\cM}{\cE}$ enjoys the properties
as stated in Lemma \ref{spider-construction}.
The proofs for statements (S1), (S2) will be established in
Lemma \ref{spider-construction-S1S2}.
Property (S4) will be shown in Lemma \ref{spider-construction-S4}.
The equivalence of $\cM$ and
$\spider{\cM}{\cE}$ as stated in (S3) 
will be a consequence of
Lemma \ref{preservation-of-0-EC-invariant-properties}.

\begin{lemma}[See (S1) and (S2) in Lemma \ref{spider-construction}]
  \label{spider-construction-S1S2}
  Let $\cM$ be an MDP and $\cE$ a 0-BSCC of $\cM$.
  Then, the spider construction generates an
  MDP $\spider{\cM}{\cE,s_0}$ that satisfies the following properties:
  \begin{enumerate}
  \item [(S1)]
     $\cM$ and $\spider{\cM}{\cE,s_0}$ have the same state space
     and $\|\spider{\cM}{\cE,s_0}\|=\|\cM\|{-}1$ 
  \item [(S2)]
     If $\cM$ is strongly connected and $\cE\not= \cM$ then
     $\spider{\cM}{\cE,s_0}$ 
     has a single maximal end component
     $\cF$, reachable from all states and
    containing the reference state $s_0$.
  \end{enumerate}
\end{lemma}

\begin{proof}
 The first statement of (S1) is obvious:
 If $m$ is the number of states in $\cE$ then
 step \textbf{(i)} removes $m$ state-action pairs, while step \textbf{(ii)} introduces
 $m{-}1$ new state-action pairs.
 Step \textbf{(iii)} has no effect on the number of state-action pairs.
 Hence, $\|\spider{\cM}{\cE}\|=\|\cM\|{-}1$.

 To prove statement (S2), we suppose that $\cM\not= \cE$
 and that $\cM$ is strongly connected.
 Let $T$ denote the set of states $s$ in $\cE$ such that
 $s\not= s_0$ and $P_{\cM}(u,\alpha,s)>0$ for some state-action
 pair $(u,\alpha)\in \cM \setminus \cE$.
 Let $\cF$ denote the sub-MDP of $\spider{\cM}{\cE}$
 obtained by removing all state-action pairs $(s,\tau)$ where
 $s$ is a state of $\cE$ that is not contained in $T$.
 Thus, the state space of $\cF$ consists of the
 states of $\cM$ that do not belong to $\cE$ and the states in
 $T \cup \{s_0\}$.

 Obviously, $\cF$ is reachable from all states and $\cF$ contains
 the reference state $s_0$.
 We show that $\cF$ is strongly connected.
 For this, we prove that all states $s$ in $\cF$ with $s \not= s_0$
 are reachable from $s_0$ and can reach $s_0$. 
 \begin{itemize}
 \item
   We first show that $s_0$ is reachable from each state $s$ in 
   $\spider{\cM}{\cE}$.
   Let $\fpath=t_0\, \beta_0\, t_1 \, \beta_1 \ldots \beta_{n-1} \, t_n$
   be a simple path in $\cM$ from $t_0=s$ to $t_n=s_0$,
   where ``simple'' means that
   $t_i \not= t_j$ for $0 \leqslant i < j \leqslant n$.
   We pick the smallest index $i\in \{0,1,\ldots,n\}$
   such that $t_i$ belongs to $\cE$.
   Then,  $(t_j,\beta_j)\in \cF$  for $0\leqslant j < i$.
   Hence, if $t_i\not= s_0$ then
   $t_0\, \beta_0\, t_1 \, \beta_1 \ldots \beta_{i-1} \, t_i \, \tau \, s_0$
   is a path in $\cF$ from $s=t_0$ to $s_0$.
   Likewise, if $t_i = s_0$ then
   $t_0\, \beta_0\, t_1 \, \beta_1 \ldots \beta_{i-1} \, t_i$
   is a path in $\cF$ from $s=t_0$ to $s_0=t_i$.
 \item
   We next show that each state $s$ that is not contained in $\cE$
   is reachable from $s_0$ in $\cF$.
   Let $\fpath=t_0\, \beta_0\, t_1 \, \beta_1 \ldots \beta_{n-1} \, t_n$
   be a simple path in $\cM$ from $t_0=s_0$ to $t_n=s$.
   If none of the states $t_1,\ldots,t_n$ belongs to $\cE$
   then $(t_i,\beta_i)\in \cF$ for all $0 \leqslant i < n$
   and $\fpath$ is a path in $\cF$ from $s_0$ to $s$.
   Suppose now that at least one of the states $t_1,\ldots,t_n$
   is contained in $\cE$.
   We pick the largest index $i\in \{1,\ldots,n\}$
   where $t_i$ is a state of $\cE$.
   By assumption $s=t_n$ is not contained in $\cE$.
   This yields $i < n$.
     Then, $(t_i,\beta_i)\notin \cE$ 
     and therefore $(s_0,\beta_i)\in \cF$ due to step \textbf{(iii)}.
     As the states $t_{i+1},\ldots,t_n$ do not belong to $\cE$,
     we have
     $(t_j,\beta_j)\in \cF$ for $i< j < n$.
     Thus,
     $s_0\, \beta_i\, t_{i+1} \, \beta_{i+1} \ldots \beta_{n-1} \, t_n$
     is a path from $s_0$ to $s$ in $\cF$.
   \item
    It remains to show that each state $s\in T$ 
    is reachable from $s_0$ in $\cF$.
    By definition of $T$ there is state-action pair $(u,\alpha)$
    such that $P_{\cM}(u,\alpha,s)>0$ and
    $(u,\alpha)\in \cM \setminus \cE$.
    If $u$ belongs to $\cE$ then $(s_0,\alpha)\in \cF$ and hence,
    $s_0\, \alpha \, s$ is a path in $\cF$.
    If $u$ does not belong to $\cE$ then $s$ is reachable from
    $s_0$ in $\cF$ via a path of the form $\fpath \, \alpha \, s$
    where $\fpath$ is a path from $s_0$ to $u$ in $\cF$ (see above).
 \end{itemize}
 Thus, $\cF$ is an end component of $\spider{\cM}{\cE}$.
 As the states in $\cE$ that are not contained in $T$ do not have
 incoming edges in $\spider{\cM}{\cE}$, these states are not contained
 in any end component of $\spider{\cM}{\cE}$.
 Thus, $\cF$ subsumes each other end component of $\spider{\cM}{\cE}$.
 This shows that $\cF$ is the unique maximal end component 
 of $\spider{\cM}{\cE}$.
\end{proof}

\begin{remark}[Result of the spider construction if $\cM$ is a 0-BSCC]
{\rm
Note that the requirement $\cM\not= \cE$ in (S2) is necessary
as if $\cM$ is a 0-BSCC then $\spider{\cM}{\cM}$ is acyclic
and consists of $\tau$-transitions from all states $s$ with $s\not= s_0$
to $s_0$.
 }
\end{remark}

\begin{lemma}[See (S4) in Lemma \ref{spider-construction}]
  \label{spider-construction-S4}
  Let $\cM$ be an MDP and $\cE$ a 0-BSCC of $\cM$
  that is contained in a maximal end component of $\cM$
  with maximal expected mean payoff 0. Then:
  \begin{enumerate}
  \item [(a)]
     For each state $s$ that is not contained in $\cE$:
     $s$ belongs to a 0-EC of $\cM$ iff
     $s$ belongs to a 0-EC of $\spider{\cM}{\cE}$.
  \item [(b)]
     For each state-action pair $(s,\alpha)$ of $\cM$:
     $(s,\alpha)$ belongs to a 0-EC of $\cM$ iff
     $(s,\alpha)\in \cE$ or
     $(s_0,\alpha)$ belongs to a 0-EC of $\spider{\cM}{\cE,s_0}$.
  \end{enumerate}
 \end{lemma}

\begin{proof}
 We first consider a 0-EC $\cZ$ of $\cM$.
 The claim is obvious for $\cZ = \cE$.
 Hence, we may suppose $\cZ \not= \cE$.
 We now show that there is a 0-EC $\cF$ of $\spider{\cM}{\cE}$
 that contains all states in $\cF$ that are not contained in $\cE$
 and all actions $\alpha$ such that $(s,\alpha)\in \cZ \setminus \cE$.
 
 If $\cZ$ does not contain a state of $\cE$, then $\cZ$ is clearly a 0-EC of 
 $\spider{\cM}{\cE}$.
 If $\cZ$ contains the reference state $s_0$, then
 the set $\cF$ is a 0-EC of $\spider{\cM}{\cE}$,
 consisting of the state-action pairs 
 \begin{itemize}
 \item
    $(s,\alpha)\in \cZ$ where $s$ is a state not contained in $\cE$,
 \item 
    $(s,\tau)$
    where $s$ is a state of $\cE$ with $s\not= s_0$ and $(s,\alpha_s)\in \cZ$,
 \item 
    $(s_0,\alpha)$ where $s$ is a state of $\cE$ (possibly $s=s_0$) with 
    $\alpha\not= \alpha_s$ and
    $(s,\alpha)\in \cZ$.
 \end{itemize}
 Otherwise, \ie $\cZ$ does not contain $s_0$ but some state of $\cE$,
 then $\cZ \cup \cE$ is a 0-EC of $\cM$ 
 (see Lemma \ref{lemma:union-of-0-EC})
 and we can apply the argument before.

 Vice versa, we show that each 0-EC $\cF$ of $\spider{\cM}{\cE}$
 induces a 0-EC $\cZ$ of $\cM$ such that $\cZ$ contains all states
 of $\cF$ and all actions $\alpha$ of $\cM$ where $(s_0,\alpha)\in \cF$.
 If $\cF$ is a 0-EC of $\spider{\cM}{\cE}$ 
 that does not contain $s_0$, then $\cF$ is a 0-EC of $\cM$.
 (We use here the fact that the states $s$ in $\cE$ with $s\not= s_0$
  have a single $\tau$-transition to $s_0$ in $\spider{\cM}{\cE}$,
  but no other transition. Hence, each end component of $\spider{\cM}{\cE}$
  that contains $s$ must also contain $s_0$.)
 If $\cF$ is a 0-EC of $\spider{\cM}{\cE}$ containing $s_0$,
 then the set consisting of all state-action pairs in $\cE$ plus
 the state-action pairs $(s,\alpha_s)$
 where $s$ is a state in $\cF$ and $(s_0,\alpha_s)\in \cF$
 is a 0-EC of $\cM$.
\end{proof}

\begin{remark}[The values $w(s,t)$]
 \label{remark-w(s,t)}
{\rm
Recall from Section \ref{sec:max-0-EC} 
that for each maximal end component $\cG$ of $\cM$
with $\Exp{\max}{\cG}(\MP)=0$,
each 0-EC of $\cG$ is
contained in a \emph{maximal 0-ECs}, where maximality is understood with
respect to the property ``being a 0-EC''.
This implies the existence of integers
$w(s,t)\in \Integer$ 
for all states $s,t$ that belong to the same maximal 0-EC 
of $\cM$ that is contained in some maximal end component $\cG$ of
$\cM$ with $\Exp{\max}{\cG}(\MP)=0$
such that the weight of all paths from $s$ to $t$ inside
some 0-EC equals $w(s,t)$, no matter which 0-EC is chosen.

The transformation 
``0-EC $\cF$ of $\spider{\cM}{\cE}$ $\leadsto$ 0-EC $\cZ$ of $\cM$''
explained in the proof of Lemma \ref{spider-construction-S4}
preserves the $w$-values.
More precisely, if $\cF$ is a 0-EC of $\spider{\cM}{\cE}$
then $w(s,t)$ is the weight of each path from $s$ to $t$ in $\cF$.
Vice versa, if $\cZ$ is a 0-EC of $\cM$ and $s$, $t$ are states in $\cZ$
then $w(s,t)$ is the weight of
each path from $s$ to $t$ in $\spider{\cM}{\cE}$ that is built
by $\tau$-transitions and actions belonging to $\cZ$.
\Ende
  }
\end{remark}

We now turn to the proof of the scheduler transformations stated
in (S3) of Lemma \ref{spider-construction}.
Before doing so, let us recall the definition of the purge-function
and of $\cE$-invariant
properties:
Given a (finite or infinite) path 
$\infpath = t_0 \, \alpha_0 \, t_1 \, \alpha_1 \ldots$ in $\cM$, let
$\Purge{\cE}{\infpath} \in (S \times \Integer)^\omega \cup
 (S \times \Integer)^* \times S$
denote the sequence arising from $\infpath$ by 
\begin{enumerate}
\item [(1)] replacing each fragment
$t_i \, \alpha_i\, \ldots \, \alpha_{j-1}\, t_j \, \alpha_j \, t_{j+1}$
of $\infpath$ such that
\begin{itemize}
\item
  either $i=0$ or $(t_{i-1},\alpha_{i-1})\notin \cE$,
\item
  $(t_j,\alpha_j)\notin \cE$
\item
  $(t_{\ell},\alpha_{\ell})\in \cE$ for $\ell=i,i{+}1,\ldots,j$
\end{itemize}
with $t_i \, w \, t_{j+1}$ where $w=w(t_i,t_j) +\wgt(t_j,\alpha_j)$ 
and 
\item [(2)] replacing each action $\alpha_i$ 
in the resulting sequence with $\wgt(s_i,\alpha_i)$.
\end{enumerate}
Note that step (1) yields a (finite or infinite) sequence
$t_0'\, c_0 \, t_1' \, c_1 \, \ldots$ where the $t_0',t_1',\ldots$ 
are states and
$c_i\in \Act \cup \Integer$. Moreover, $c_i \in \Integer$ implies
$t_i' \in \cE$, while $c_i \in \Act$ implies that 
$(t_i',c_i)$ is a state-action
pair of $\cM$ that is not contained in $\cE$ (but still $t_i' \in \cE$
is possible).
For example, if 
$$
 \fpath\ \ = \ \ 
 t_0\, \beta_0 \, t_1 \, \beta_1 \, t_2 \, \gamma_2 \, t_3 \, 
  \beta_3 \, t_4 \, \gamma_4 \, t_5 \, \gamma_5 \, t_6 \, \beta_6 \, t_7
$$
where all state-action pairs of the form $(t_k,\gamma_k)$ belong to $\cE$,
while the state-action pairs of the form $(t_l,\beta_l)$ do not,
then after step (1) we obtain the sequence
$$
 t_0\, \beta_0 \, t_1 \, \beta_1 \, t_2 \, w_2 \, t_4 \, w_4 \, t_7
$$
where $w_2=\underbrace{\wgt(t_2,\gamma_2)}_{=w(t_2,t_3)}+\wgt(t_3,\beta_3)$
and $w_4=
     \underbrace{\wgt(t_4,\gamma_4)+\wgt(t_5,\gamma_5)}_{=w(t_4,t_6)} 
      +\wgt(t_6,\beta_3)$.
This yields:
$$
 \Purge{\cE}{\fpath}
 \ \ = \ \
 t_0\, w_0 \, t_1 \, w_1 \, t_2 \, w_2 \, t_4 \, w_4 \, t_7
$$
where
$w_0=\wgt(t_0,\beta_0)$ and $w_1 =\wgt(t_1,\beta_1)$.
Note that this implies $t_2,t_3,t_4,t_5,t_6$ to be states in $\cE$, 
while states $t_0,t_1$ and $t_7$ might or might not be contained in $\cE$.

A property $\varphi$ is called $\cE$-invariant if
for all maximal paths $\infpath$ the following conditions (I1) and (I2) hold:
\begin{enumerate}
\item [(I1)] 
  If $\infpath$ has an infinite suffix consisting of state-action pairs in
  $\cE$ then $\infpath \not\models \varphi$.
\item [(I2)] 
  If $\infpath \models \varphi$ and
  $\infpath'$ is a maximal path with 
  $\Purge{\cE}{\infpath} = \Purge{\cE}{\infpath'}$ then
  $\infpath' \models \varphi$.
\end{enumerate}
Examples for $\cE$-invariant properties are 
(positive or negative) weight-divergence or the pumping property,
and so are properties
of the form $\Diamond (t \wedge (\accwgt \bowtie w))$
where $t$ is a state not contained in $\cE$,
$\bowtie$ a comparison operator (\eg $=$ or $\geqslant$),
and $w\in \Integer$.

   Recall that $\lim(\infpath)$ denotes the set of state-action pairs that
   appear infinitely often in $\infpath$.
   In what follows, we write 
   $\Limit{\cE}$ to denote the set 
   $\{\infpath \in \InfPaths : \lim(\infpath)=\cE\}$.

\begin{lemma}[See (S3) in Lemma \ref{spider-construction}]
\label{preservation-of-0-EC-invariant-properties}
 Let $\cM$ and $\cE$ be as before.
 Then:
 \begin{enumerate}    
 \item [(S3.1)]
   For each scheduler $\tsched$ for $\spider{\cM}{\cE}$ 
   there is a scheduler
   $\sched$ for $\cM$ such that
   $\Pr^{\tsched}_{\spider{\cM}{\cE},s}(\varphi)=
    \Pr^{\sched}_{\cM,s}(\varphi)$
   for all states $s$ in $\cM$ and all
   $\cE$-invariant properties $\varphi$.
   If $\tsched$ is an MD-scheduler then
   $\sched$ can be chosen as an MD-scheduler.
 \item [(S3.2)]
        For each scheduler $\sched$ for $\cM$ 
        there exists a 
        (randomized) scheduler $\tsched$ for $\spider{\cM}{\cE}$ such that
        $$
         \Pr^{\sched}_{\cM,s}(\varphi) \ \ \leqslant \ \
         \Pr^{\tsched}_{\spider{\cM}{\cE},s}(\varphi) \ \ \leqslant \ \
         \Pr^{\sched}_{\cM,s}(\varphi) + 
         \Pr^{\sched}_{\cM,s}(\Limit{\cE})
       $$
        for all states $s$ and all 
        $\cE$-invariant properties $\varphi$.
 \end{enumerate}
\end{lemma}

Note that  (S3.2) implies that
$\Pr^{\sched}_{\cM,s}(\varphi)=\Pr^{\tsched}_{\spider{\cM}{\cE}}(\varphi)$
if $\Pr^{\sched}_{\cM,s}(\Limit{\cE})=0$.

\begin{proof}
 For statement (S3.1), we observe that
 $\cM$ can mimick $\spider{\cM}{\cE}$'s 
 $\tau$-transitions, followed by
 the state-action pair $(s_0,\beta)$ with $\beta \in \Act_{\cM}(s)$ for some
 state $s$ in $\cE$ with $s\not= s_0$: First, the
 MD-scheduler that realizes $\cE$ is simulated
 by choosing $\alpha_u$ for each state $u$ in $\cE$
 until state $s$ has been reached and then taking action $\beta$ in state $s$.
 Note that $w(s,s_0)=-w(s_0,s)$ and therefore
 $$
  \wgt_{\spider{\cM}{\cE}}(s,\tau)+\wgt_{\spider{\cM}{\cE}}(s_0,\beta) 
  \ \ = \ \
  w(s,s_0) + w(s_0,s)+\wgt_{\cM}(s,\beta) 
  \ \ = \ \ \wgt_{\cM}(s,\beta)\enskip.
 $$
 This yields a scheduler transformation
 ``scheduler $\tsched$ for $\spider{\cM}{\cE}$ $\leadsto$
   scheduler $\sched$ for $\cM$''
 preserving the probabilities of all $\cE$-invariant properties.
 Moreover, $\sched$ is MD if so is $\tsched$.

 The idea to provide scheduler transformations as stated in (S3.2)
 relies on the observation that
 $\spider{\cM}{\cE}$ 
 can mimic $\cM$'s behavior inside $\cE$, 
 followed by a state-action pair $(s,\beta)$, where $s$ is a state in $\cE$
 with $s\not= s_0$ and $\beta \not= \alpha_s$, 
 by taking the $\tau$-transition from $s$ to $s_0$, 
 followed by the state-action pair $(s_0,\beta)$.

 Statement (S3.2) is obvious if $\cM=\cE$ in which case
 $\Pr^{\sched}_{\cM,s}(\psi)=0$ for all
 $\cE$-invariant properties $\psi$. (Recall that all paths with an 
 infinite suffix in $\cE$ violate $\psi$; see (I1) in the definition of
 $\cE$-invariance.)
 Suppose now that $\cM\not= \cE$.

 Let $\cH$ be the MDP resulting from $\spider{\cM}{\cE,s_0}$ 
 by adding a fresh state
 $\final$, a fresh action name $\iota$ 
 and a deterministic transition from the reference state
 $s_0$ of the spider construction to $\final$
 and a deterministic self-loop at state $\final$, both with action label
 $\iota$ and weight 0.
 Let $S'=S \cup \{\final\}$.
 Obviously, $\cH$ and $\cM$ have the same traps.
 
 We introduce a purge-function, called $\Purge{\cH}{\cdot}$, 
 similar to the purge-function for paths in $\cM$.
 Recall that the idea to define
 $\Purge{\cE}{\fpath'}$ for a path $\fpath'$ in $\cM$
 was to abstract away from the behavior inside $\cE$ and
 just representing the state where $\cE$ is entered and the action
 where $\cE$ is left, and then replacing the action with the corresponding
 weight. Thus, for a path fragment of $\fpath'$ in which $\cE$ is entered
 via state $s$ and left via action $\alpha \in Act(t)$ 
 (where $t$ is a state of $\cE$, while the state-action pair $(t,\alpha)$
 does not belong to $\cE$) 
 leading to state $u$,
 the corresponding path fragment is replaced with
 $s w u$ where $w=w(s,t)+\wgt_{\cM}(t,\alpha)$.
 With the switch from $\cM$ to $\spider{\cM}{\cE,s_0}$ resp. $\cH$
 the corresponding behavior of such a path fragment from $s$ to $u$
 consists of two or one transition, namely
  $s \, \tau \, s_0 \, \alpha \, u$ if $s\not= s_0$
 or $s_0 \, \alpha \, u$ if $s=s_0$.
 The idea is now to replace these path fragments with $s \, w' \, u$ where
 $w'= w(s,s_0)+\wgt_{\cH}(s_0,\alpha)$.
 Note that $\wgt_{\cH}(s_0,\alpha)=\wgt_{\spider{\cM}{\cE,s_0}}(s_0,\alpha)=
  \wgt(s_0,t)+\wgt_{\cM}(t,\alpha)$. As $w(s,s_0)+w(s_0,t)=w(s,t)$
 we get $w=w'$.

 The formal definition of $\Purge{\cH}{\fpath}$ for the paths in $\cH$
 is as follows.
 Let
 $\fpath=t_0\alpha_0 t_1 \alpha_1 t_2 \alpha_2 \ldots$
 be a (finite or infinite) path in $\cH$.
 The sequence $\Purge{\cH}{\fpath}\in
 (S' \times \Integer)^{\omega} \cup (S' \times \Integer)^* \times S'$ 
 results from $\fpath$ as follows: First, 
 each $\alpha_i$ with $\wgt_{\cH}(t_i,\alpha_i)$ is replaced, 
 provided $t_i\not= s_0$ and $(t_i,\alpha_i)$ is a state-action pair
 of $\cM$. Then, each path fragment of the form 
 $t_i \, \tau \, t_{i+1} \, \alpha_{i+1} \, t_{i+2}$
 where $t_i$ belongs to $\cE$, $t_i\not= s_0$, and 
 $t_{i+1}$ is the reference state $s_0$ of the spider construction
 is replaced with $t_i \, w \, t_{i+2}$ where 
 $w=\wgt_{\cH}(t_i,\tau)+\wgt_{\cH}(s_0,\alpha_{i+1})$.
 Finally, each action $\alpha_i$ where $t_i=s_0$ is replaced
 with $\wgt_{\cH}(s_0,\alpha_i)$.
 Note that $\Purge{\cH}{\fpath}$ is finite iff $\fpath$ is finite.

 We now define the randomized scheduler $\usched$ for $\cH$ that mimics
 the given scheduler $\sched$ for $\cM$ as follows.
 Suppose $\fpath=t_0\alpha_0 t_1 \alpha_1  \ldots \alpha_{n-1} t_n$
 is a finite path in $\cH$ where $t_n$ is
 not a trap and different from the auxiliary state $\final$.
 \begin{itemize}
 \item
 If $t_n$ is not a state of $\cE$ and $(t_n,\alpha)$ a state-action pair
 of $\cM$ (and $\cH$) then $\usched(\fpath)(\alpha)$ equals the 
 conditional probability
 for $\sched$ to generate a 
 path of the form $\fpath' \alpha u$ where
 $\Purge{\cE}{\fpath'}=\Purge{\cH}{\fpath}$ 
 and $u$ is an arbitrary
 $\alpha$-successor of $t_n$ under the condition that $\sched$ indeed schedules
 such a path $\fpath'$.
 That is, if $t_0=s$ and $\Pi_{\cM,s}$ denotes the set of all finite 
 paths $\fpath'$ with $\first(\fpath')=s$ and
 $\Purge{\cE}{\fpath'}=\Purge{\cH}{\fpath}$ 
 then
 $$
  \usched(\fpath)(\alpha) \ \ = \ \ 
  \frac{\sum\limits_{\fpath'\in \Pi_{\cM,s}} 
        \sum\limits_{u\in S} \Pr^{\sched}_{\cM,s}(\fpath' \alpha u)}
       {\sum\limits_{\fpath'\in \Pi_{\cM,s}} \Pr^{\sched}_{\cM,s}(\fpath')}
 $$
 where we assume that $\Pi_{\cM,s}$ contains at least one $\sched$-path and use
 $\Pr^{\sched}_{\cM,s}(\fpath')$ as a short-form notation
 for the probability for the cylinder set of $\fpath'$ under $\sched$.%
 \footnote{%
  The cylinder set of a finite path $\fpath'$ denotes the
  set of maximal paths $\infpath$ where $\fpath'$ is a prefix of $\infpath$.}
 If there is no $\sched$-path $\fpath' \in \Pi_{\cM,s}$ then
 $\usched(\fpath)$ is irrelevant for our purposes.
 \item
 Suppose now $t_n$ is a state of $\cE$, $t_n\not= s_0$ 
 and either $n=0$ or
 $n \geqslant 1$ and $\alpha_{n-1}\not= \tau$.
 Then, $\usched(\fpath)(\tau)=1$.
 Moreover, if $\alpha$ is an action in $\Act_{\cH}(s_0)\setminus \{\iota\}$
 then
 $\usched(\fpath \tau s_0)(\alpha)$ equals the conditional probability
 for $\sched$ to generate a
 path of the form $\fpath'$ with
 $\Purge{\cE}{\fpath'}=\Purge{\cH}{\fpath \tau s_0 \alpha u}$ 
 where $u$ is an arbitrary
 $\alpha$-successor of $s_0$ in $\cH$.
 With the remaining probability, $\usched$ moves
 from $s_0$ to $\final$ via the fresh action $\iota$.
 That is, $\usched(\fpath \tau s_0)(\iota)=
   1-\sum_{\alpha} \usched(\fpath \tau s_0)(\alpha)$.
 This value equals the conditional 
 probability for $\sched$ to generate an infinite
 path $\infpath = \fpath''; \infpath'$ where 
 $\Purge{\cE}{\fpath''}=\Purge{\cH}{\fpath}$ and $\infpath'$ is an
 infinite path in $\cE$.
 In both cases, the condition is that $\sched$ generates at least one 
 path $\fpath''$ with $\Purge{\cE}{\fpath''}=\Purge{\cH}{\fpath}$.
 If no such path $\fpath''$ exists then $\usched(\fpath)$ is irrelevant
 for our purposes.
 \item
 The definition of $\usched(\fpath)$ is analogous when $\last(\fpath)=s_0$
 and either 
 $n=0$ or $(t_{n-1},\alpha_{n-1})$ is a state-action pair of $\cM$.
 \end{itemize}
 Note that $\usched(\fpath_1)=\usched(\fpath_2)$ whenever $\fpath_1,\fpath_2$
 are finite paths with $\Purge{\cE}{\fpath_1}=\Purge{\cE}{\fpath_2}$.

 For simplicity, let us assume that $\cM$ has no traps, in which case
 all maximal paths in $\cM$ (and $\cH$) are infinite.
 Given an $\cE$-invariant property $\psi$ for $\cM$,
 let $\purge{\psi}$ denote the set of all  words
 $\Purge{\cE}{\infpath}$ where $\infpath$ is an infinite path
 in $\cM$ with $\infpath \models \psi$.
 (Recall that $\infpath \not\models \psi$ for 
 all infinite paths $\infpath$ with
 $\lim(\infpath) =\cE$ by definition of $\cE$-invariance.) 
 Then, $\purge{\psi} \subseteq (S \times \Integer)^{\omega}$
 and $\purge{\psi}$ can be viewed as 
 a measurable subset of $(S' \times \Integer)^{\omega}$
 where measurability is understood with respect to the sigma-algebra 
 $\Sigma_{\cH}$
 generated by the cylinder sets spanned by the finite strings
 $(S' \times \Integer)^* \times S'$. 
 (Recall that $S'=S \cup \{\final\}$.)

 The probability measure $\Pr^{\sched}_{\cM,s}$  on the maximal paths
 of $\cM$ induces a probability measure $\mu^{\sched}_{\cM,s}$ 
  on the sigma-algebra $\Sigma_{\cH}$ over
 $(S' \times \Integer)^{\omega}$
 using the embedding
 $e\colon \InfPaths_{\cM}\to (S'\times \Integer)^{\omega}$ given by:
 \begin{itemize}
 \item
   $e(\infpath) = \Purge{\cE}{\infpath}$ if $\infpath$ contains
   infinitely many actions not contained in $\cE$ and
 \item  
   $e(\infpath) = \Purge{\cE}{\fpath'} \, 0 \, \final \, 0 \, \final \ldots$
   if $\infpath = \fpath';\infpath'$
   where $\infpath'$ is an infinite paths consisting of state-action pairs
   in $\cE$ and either $\fpath'$ consists of a single state or
   the last state-action pair of $\fpath'$ does not belong to $\cE$.
 \end{itemize}
 Then, $\mu^{\sched}_{\cM,s}$ is the unique probability measure 
 given by 
 $\mu^{\sched}_{\cM,s}(\Cyl(t_0\, w_0\, \ldots w_{n-1}\, t_n))
  = \sum_{\fpath'} \Pr^{\sched}_{\cM,s}(\Cyl(\fpath'))$
 where $\fpath'$ ranges over all finite paths in $\cM$
 with $\Purge{\cE}{\fpath'}= t_0\, w_0\, \ldots w_{n-1}\, t_n$.
 Given a 0-EC-invariant property $\psi$,
 the image of the set of infinite $\infpath$ in $\cM$ satisfying
 $\psi$ under $e$ equals $\purge{\psi}$.
 This yields
 $\Pr^{\sched}_{\cM,s}(\psi)=\mu^{\sched}_{\cM,s}(\purge{\psi})$.
        
 Likewise, scheduler $\usched$ for $\cH$ induces a probability
 measure $\mu^{\usched}_{\cH,s}$ over this sigma-algebra
 such that $\Pr^{\usched}_{\cH,s}(\psi)=\mu^{\usched}_{\cH,s}(\purge{\psi})$
 for each 0-EC-invariant property $\psi$.

 By construction, $\mu^{\sched}_{\cM,s}$ and $\mu^{\usched}_{\cH,s}$
 agree on the cylinder sets of the finite strings in
 $(S' \times \Integer)^* \times S'$. 
 \begin{enumerate}
 \item []
  This can be shown by induction on the length of strings in
  $(S' \times \Integer)^* \times S'$.
  In the step of induction, we consider a string of the form
  $t_0\, w_0\, t_1\, w_1 \ldots w_{n-1}\, t_{n}\, w_{n}\, t_{n+1}
   \in (S'\times \Integer)^*\times S'$.
  Let $s=t_0$.
  Suppose for simplicity that $t_n$ is a state in $\cM$ that does not belong
  to $\cE$. Let $A$ denote the set of actions $\alpha$ where
  $(t_n,\alpha)$ is a state-action pair in $\cM$ (and $\cH$) 
  and $\wgt(t_n,\alpha)=w_{n}$.
  Let $\Pi_{\cM,s}$ denote the set of finite paths $\fpath'$
  in $\cM$
  with $\Purge{\cE}{\fpath'}=t_0\, w_0\, t_1\, w_1 \ldots w_{n-1}\, t_{n}$.
  Likewise, we write 
  $\Pi_{\cH,s}$ to denote the set of finite paths $\fpath$ in $\cH$
  with $\Purge{\cH}{\fpath}=t_0\, w_0\, t_1\, w_1 \ldots w_{n-1}\, t_{n}$.
  By induction hypothesis we have
  $$
   \sum_{\fpath\in \Pi_{\cH,s}} \Pr^{\usched}_{\cH,s}(\fpath)
   \ \ = \ \
   \mu^{\usched}_{\cH,s}(t_0\, w_0\, t_1\, w_1 \ldots w_{n-1}\, t_{n})
   \ \ = \ \ 
   \mu^{\sched}_{\cM,s}(t_0\, w_0\, t_1\, w_1 \ldots w_{n-1}\, t_{n})
   \ \ = \ \ 
   \sum_{\fpath'\in \Pi_{\cM,s}} \Pr^{\usched}_{\cH,s}(\fpath')
  $$
  Let $p$ denote this value and suppose $p>0$.
  Then:
  \begin{eqnarray*}
    \mu^{\usched}_{\cH,s}
      (t_0\, w_0\, t_1\, w_1 \ldots w_{n-1}\, t_{n} \, w_{n} \, t_{n+1})
     & \ = \ &
     \sum_{\fpath \in \Pi_{\cH,s}} 
     \sum_{\beta \in A}
       \Pr^{\usched}_{\cH,s}(\fpath) \cdot
       \usched(\fpath)(\beta) \cdot P(t_n,\beta,t_{n+1})
     \\
     \\[-1ex]
     & = &
     \sum_{\fpath \in \Pi_{\cH,s}} 
     \sum_{\beta \in A}
       \Pr^{\usched}_{\cH,s}(\fpath) \cdot
       \frac{\sum\limits_{\fpath'\in \Pi_{\cM,s}} 
             \sum\limits_{u\in S} \Pr^{\sched}_{\cM,s}(\fpath' \beta u)}
            {\sum\limits_{\fpath'\in \Pi_{\cM,s}}
                                          \Pr^{\sched}_{\cM,s}(\fpath')}
       \cdot P(t_n,\beta,t_{n+1})
     \\
     \\[-1ex]
     & = &
     \sum_{\fpath \in \Pi_{\cH,s}} 
     \sum_{\beta \in A}
       \Pr^{\usched}_{\cH,s}(\fpath) \cdot
       \frac{1}{p} \cdot \sum\limits_{\fpath'\in \Pi_{\cM,s}} 
             \sum\limits_{u\in S} \Pr^{\sched}_{\cM,s}(\fpath' \beta u)
       \cdot P(t_n,\beta,t_{n+1})
     \\
     \\[-1ex]
     & = &
     \frac{1}{p} \cdot
     \sum_{\beta \in A}
       P(t_n,\beta,t_{n+1}) \cdot
       \sum\limits_{\fpath'\in \Pi_{\cM,s}} 
             \sum\limits_{u\in S} \Pr^{\sched}_{\cM,s}(\fpath' \beta u)
     \cdot
       \sum_{\fpath \in \Pi_{\cH,s}} \Pr^{\usched}_{\cH,s}(\fpath) 
     \\
     \\[-1ex]
     & = &
     \sum_{\beta \in A}
       P(t_n,\beta,t_{n+1}) \cdot
       \sum\limits_{\fpath'\in \Pi_{\cM,s}} 
           \Pr^{\sched}_{\cM,s}(\fpath')\cdot \sched(\fpath')(\beta)
           \cdot 
            \sum\limits_{u\in S} 
               P(t_n,\beta, u)
     \\
     \\[-1ex]
     & = &
     \sum\limits_{\fpath'\in \Pi_{\cM,s}} 
     \sum_{\beta \in A}  
           \Pr^{\sched}_{\cM,s}(\fpath')\cdot \sched(\fpath')(\beta)
           \cdot P(t_n,\beta,t_{n+1}) 
     \\
     \\[-1ex]
     & = &     
     \mu^{\sched}_{\cM,s}
      (t_0\, w_0\, t_1\, w_1 \ldots w_{n-1}\, t_{n} \, w_{n} \, t_{n+1})
  \end{eqnarray*}
  The calculation for the other cases is similar.
 \end{enumerate}
 By Caratheodory's measure-extension theorem,
 $\Pr^{\sched}_{\cM,s}$ and $\Pr^{\usched}_{\cH,s}$ agree 
 when viewed as measures over $(S' \times \Integer)^{\omega}$.
 For all $\cE$-invariant properties $\psi$
 it holds $\psi=e^{-1}(\Purge{\cE}{\psi})$
 and hence, we obtain
 $$
   \Pr^{\sched}_{\cM,s}(\psi)
   \ \ = \ \ 
   \mu^{\sched}_{\cM,s}(\Purge{\cE}{\psi})
   \ \ = \ \
   \mu^{\sched}_{\cH,s}(\Purge{\cE}{\psi})
   \ \ = \ \ 
   \Pr^{\usched}_{\cH,s}(\psi)\enskip.
 $$

 We finally switch from scheduler $\usched$ for $\cH$ to a scheduler
 $\tsched$ for $\spider{\cM}{\cE,s_0}$. For this, we pick arbitrary
 actions $\alpha_s$ enabled in state $s$ of $\cM$ and define
 $\tsched$ by $\tsched(\fpath)(\beta)=\usched(\fpath)(\beta)$ for
 each action $\beta \not= \alpha_s$ that is enabled in $s$ as a state of 
 $\spider{\cM}{\cE,s_0}$ and
 $\tsched(\fpath)(\alpha_s)=\usched(\fpath)(\alpha_s)+\usched(\fpath)(\iota)$.
 We then have:
 $$
   \Pr^{\tsched}_{\spider{\cM}{\cE,s_0},s}(\psi)\ \ \geqslant \ \
   \Pr^{\usched}_{\cH,s}(\psi)
 $$
 for all $\cE$-invariant properties $\psi$.
 Clearly, $\Pr^{\usched}_{\cH,s}(\psi) + \Pr^{\sched}_{\cM,s}(\Limit{\cE})$
 is an upper bound for $\Pr^{\tsched}_{\spider{\cM}{\cE,s_0},s}(\psi)$.
\end{proof}

As a consequence of Lemma \ref{preservation-of-0-EC-invariant-properties}
we get:

\begin{corollary}
  For each $\cE$-invariant property $\varphi$ and each state
  $s$ in $\cM$, $\Pr^{\sup}_{\cM,s}(\varphi)=\Pr^{\sup}_{\spider{\cM}{\cE},s}(\varphi)$.
  Furthermore, the existence of a scheduler $\sched$ for $\cM$ with
  $\Pr^{\sup}_{\cM,s}(\varphi)=\Pr^{\sched}_{\cM,s}(\varphi)$
  implies the existence of a scheduler $\tsched$ for $\spider{\cM}{\cE}$
  with
  $\Pr^{\sup}_{\cM,s}(\varphi)=\Pr^{\tsched}_{\spider{\cM}{\cE},s}(\varphi)$,
  and vice versa.
\end{corollary}

As weight-divergence and the gambling condition
are $\cE$-invariant properties, we obtain that the
spider construction preserves weight-divergence and
the pumping property as stated in Corollary \ref{spider-preserves-wgt-div}.
Moreover, we get:

\begin{corollary}
 \label{spider-MP-at-most-0}
   Suppose $\cM$ is strongly connected and
   $\Exp{\max}{\cM}(\MP)=0$.
   If $\cE$ is a 0-BSCC of $\cM$ then either
   $\spider{\cM}{\cE}$ has no maximal end component
   or $\Exp{\max}{\cF}(\MP)\leqslant 0$ for the unique
   maximal end component $\cF$ of $\spider{\cM}{\cE}$.
\end{corollary} 

\begin{proof}
 Suppose $\spider{\cM}{\cE}$ has end components.
 Let $\cF$ be the unique maximal end component of
 $\spider{\cM}{\cE}$.
 But then $\Exp{\max}{\cF}(\MP)\leqslant 0$ 
 as otherwise $\cF$ would be pumping (see Lemma \ref{lemma:pumping-ecs}), 
 in which case
 $\cM$ would be pumping (by Corollary \ref{spider-preserves-wgt-div}).
 This, however, is impossible (again by Lemma \ref{lemma:pumping-ecs})
 as $\Exp{\max}{\cM}(\MP)=0$.
\end{proof}

When the spider construction is applied to an
MDP $\cM$ that is not strongly connected,
then $\spider{\cM}{\cE}$ is obtained from $\cM$ by replacing $\cF$ with
$\spider{\cF}{\cE}$ where $\cF$ is the unique maximal end component
of $\cM$ that contains the given 0-BSCC $\cE$.
Moreover, state-action pairs $(s,\alpha)\in \cM\setminus \cF$
with $s$ being a state of $\cF$ that is different from the
reference state $s_0$ are replaced with $(s_0,\alpha)$
where $P_{\spider{\cM}{\cE}}(s_0,\alpha,u) =P_{\cM}(s,\alpha,u)$
for all states $u$ and 
$\wgt_{\spider{\cM}{\cE}}(s_0,\alpha)=w(s_0,s)+\wgt_{\cM}(s,\alpha)$.
Obviously, there is no end component $\cG$ of $\spider{\cM}{\cE}$ 
that subsumes $\spider{\cF}{\cE}$ 
and $\cG \not= \spider{\cF}{\cE}$.
Note that otherwise there would be a corresponding end component
of $\cM$ that strictly subsumes $\cF$, which is impossible by the
maximality of $\cF$.
Hence, by Corollary \ref{spider-preserves-wgt-div}:

\begin{corollary}[Generalization of Corollary \ref{spider-preserves-wgt-div}
     for possibly not strongly connected MDPs]
 \label{app:spider-preserves-wgt-div} 
   Let $\cM$ be a (possibly not strongly connected) MDP and $\cE$
   a 0-BSCC of $\cM$.
   Then, $\cM$ has a weight-divergent (resp.~pumping) end component iff
   $\spider{\cM}{\cE}$ has a weight-divergent (resp.~pumping) end component.
\end{corollary}

\subsubsection{Iterative Application of the Spider Construction}

\label{sec:iterative-application-spider}

Let $\cM$ be a (possibly not strongly connected) MDP
such that $\Exp{\max}{\cF}(\MP)\leqslant 0$ for all maximal end components
$\cF$ of $\cM$.
The iterative application of the spider construction 
generates a sequence $\cM_0=\cM$, $\cM_1,\ldots,\cM_{\ell}=\cN$
of MDPs with the same state space and where 
$\cM_{i+1}=\spider{\cM_i}{\cE_i,s_{0,i}}$ 
arises from $\cM_i$ by flattening some 0-BSCC $\cE_i$
of $\cM_i$ through the spider construction.
More precisely, $\cE_i$ is a 0-BSCC contained in a maximal end component
$\cF_i$ of $\cM_i$ where $\Exp{\max}{\cF_i}(\MP)=0$
and $\cM_{i+1}$ is obtained from $\cM_i$ by 
\begin{itemize}
\item
  replacing $\cF_i$ with
  $\spider{\cF_i}{\cE_i,s_{0,i}}$ and
\item
  replacing all state-action pairs $(s,\alpha)\in \cM_i \setminus \cF_i$ 
  with $(s_{0,i},\alpha)$, provided
  $s$ is a state of $\cF_i$ different from the 
  reference state~$s_{0,i}$.
\end{itemize}
The transition probabilities of the new state-action pairs $(s_{0,i},\alpha)$
are the same as in $\cM_i$, \ie
$P_{\cM_i}(s,\alpha,u)=P_{\cM_{i+1}}(s_{0,i},\alpha,u)$
for all states $u$, and the weight is given by
$\wgt_{\cM_{i+1}}(s_{0,i},\alpha)=w(s_{0,i},s)+\wgt_{\cM_i}(s,\alpha)$.
Here, $w(s,t)$ is the weight of some/each path $s$ to $t$ inside some
0-EC of $\cF_i$ (see Remark \ref{remark-w(s,t)}).

The algorithm terminates with the final MDP $\cM_{\ell}$ 
if either there is no maximal end component $\cF$
of $\cM_{\ell}$ where $\Exp{\max}{\cF}(\MP)=0$ or 
for each maximal end component $\cF$ of $\cM_{\ell}$
with $\Exp{\max}{\cF}(\MP)=0$,
the constructed  MD-scheduler $\sched_{\cF}$ for $\cF$ maximizing
the expected mean payoff has a gambling BSCC.

\begin{figure}[ht]
	\centering%
\pgfdeclarelayer{background}%
\pgfsetlayers{background,main}%
\begin{tikzpicture}
	\node[state]  (t0) {$t$};
	\node[istate,above left=2mm and 3mm of t0] {$\cM$:};
	\node[state, right= of t0]  (u0) {$u$};
	\node[state, right= of u0]  (s0) {$s$};
	\path%
		(t0) edge[ptran, bend left] node[above]{$\beta$/\negw{-1}} (u0)
		(u0) edge[ptran, bend left] node[below]{$\alpha$/\posw{1}} (t0)
		(u0) edge[ptran, bend left] node[above]{$\gamma$/\negw{2}} (s0)
		(s0) edge[ptran, bend left] node[below]{$\delta$/\posw{2}} (u0);
	\coordinate[above left=3mm and 1mm of t0] (g1);
	\coordinate[above right=3mm and 1mm of u0] (g2);
	\coordinate[below right=4mm and 1mm of u0] (g3);
	\coordinate[below left=4mm and 1mm of t0] (g4);
	\begin{pgfonlayer}{background}%
        		\draw[rounded corners=1em,line width=2em,black!10,fill=black!10]
        			(g1) -- (g2) -- (g3) -- (g4) -- cycle;
		\node[istate, below=1mm of t0,white]{$\mathcal{E}_0$};
    	\end{pgfonlayer}
	\node[state,right=20mm of s0]  (t1) {$t$};
	\node[istate,above left=2mm and 3mm of t1] {$\cM_1$:};
	\node[state, right= of t1]  (u1) {$u$};
	\node[state, right= of u1]  (s1) {$s$};
	\path%
		(t1) edge[ptran, bend left] node[above]{$\tau$/\negw{1}} (u1)
		(u1) edge[ptran, bend left] node[above]{$\gamma$/\negw{2}} (s1)
		(s1) edge[ptran, bend left] node[below]{$\delta$/\posw{2}} (u1);
	\coordinate[above left=3mm and 1mm of u1] (gg1);
	\coordinate[above right=3mm and 1mm of s1] (gg2);
	\coordinate[below right=4mm and 1mm of s1] (gg3);
	\coordinate[below left=4mm and 1mm of u1] (gg4);
	\begin{pgfonlayer}{background}%
        		\draw[rounded corners=1em,line width=2em,black!10,fill=black!10]
        			(gg1) -- (gg2) -- (gg3) -- (gg4) -- cycle;
		\node[istate, below=1mm of u1,white]{$\mathcal{E}_1$};
    	\end{pgfonlayer}
	\node[state,right=20mm of s1]  (t2) {$t$};
	\node[istate,above left=2mm and 3mm of t2] {$\cM_2$:};
	\node[state, right= of t2]  (u2) {$u$};
	\node[state, right= of u2]  (s2) {$s$};
	\path%
		(t2) edge[ptran, bend left] node[above]{$\tau$/\negw{1}} (u2)
		(u2) edge[ptran, bend left] node[above]{$\tau$/\negw{2}} (s2);
\end{tikzpicture}
 	\caption{Iterative application of the spider construction.}
\label{fig:iterative-spider}
\end{figure}
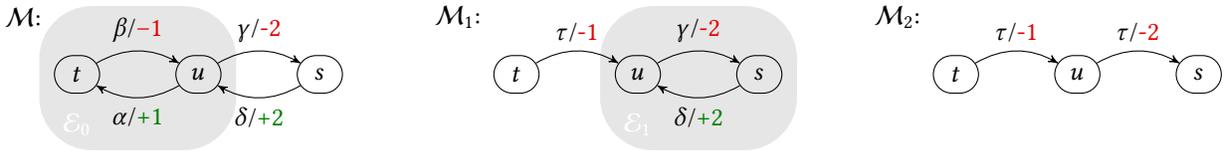

\begin{example}
  The MDPs occurring within the following example
  are illustrated by Figure~\ref{fig:iterative-spider}.
  Let $\cM_0=\cM$ be the strongly connected 
  MDP on the left of the figure.
  The iterative spider construction might first detect the 0-BSCC
  $\cE_0=\{(t,\beta), (u,\alpha)\}$.
  It then generates $\cM_1=\spider{\cM}{\cE_0,u}$ shown in the center,
  where~$u$ is the reference state.
  Then, $\cE_1=\{(u,\gamma), (s,\delta)\}$ is the unique maximal end component
  of $\cM_1$, and $\cE_1$ is even a 0-BSCC of $\cM_1$. 
  The next iteration is
   $\cM_2=\spider{\cM_1}{\cE_1,s}$ shown on the right.
  As $\cM_2$ does not have any end component, the iterative
  spider construction terminates with the MDP $\cM_2$.
\Ende
\end{example}

We get by Lemma \ref{spider-construction} and 
Corollary \ref{spider-preserves-wgt-div}:

\begin{lemma}[Maximal end components of $\cM_i$]
\label{max-EC-of-Mi}
 For each $i\in \{1,\ldots,\ell\}$:
\begin{enumerate}
\item [(a)]
  There is an injection $\iota$ that maps each maximal
  end component $\cF$ of $\cM_i$ to one of the maximal end components
  of the original MDP $\cM$
  such that the state space of $\cF$ is contained in the state space
  of $\iota(\cF)$ and $\cF$ is weight-divergent (resp.~pumping) 
  iff $\iota(\cF)$ is weight-divergent (resp.~pumping).
\item [(b)]
  The states and actions that are contained in a 0-EC of $\cM$
  are exactly the states and non-$\tau$ actions that 
  belong to one of the 0-BSCCs of $\cE_1,\ldots,\cE_{i-1}$
  or are contained in a 0-EC of $\cM_i$.
\end{enumerate}
\end{lemma}

As the spider construction preserves
the weight-divergence and pumping property
(see Corollary \ref{spider-preserves-wgt-div}),
Lemma \ref{max-EC-of-Mi} yields:

\begin{corollary}
  If $\cM$ has no weight-divergent end component
  then $\Exp{\max}{\cF}(\MP)<0$ for all end components $\cF$ 
  of the final MDP~$\cM_{\ell}$.
\end{corollary}

By property (S2) of Lemma \ref{spider-construction},
the number of state-action pairs is strictly decreasing, \ie
$\|\cM_0\| > \|\cM_1\| > \ldots > \|\cM_{\ell}\|$.
Hence, the number $\ell$ 
of recursive calls of the spider construction is bounded
by $\|\cM\|$.

Analogous to (S3) in Lemma~\ref{spider-construction} (see also
Lemma~\ref{preservation-of-0-EC-invariant-properties}) we obtain the
equivalence of $\cM$ and the MDPs $\cM_1,\ldots,\cM_{\ell}$ with
respect to the class of 0-EC-invariant properties.  These are
properties that are $\cE$-invariant for each 0-EC $\cE$ of $\cM$.

\begin{lemma}[Equivalence of $\cM$ and $\cM_i$ 
                   w.r.t. 0-EC-invariant properties]
 \label{lemma:equivalence-M-and-Mi}
   The original MDP $\cM$ and the MDP $\cM_i$ for
   $i\in \{1,\ldots,\ell\}$
   are equivalent in the following sense:
 \begin{enumerate}
 \item [(a)]
   For each scheduler $\tsched$ for $\cM_i$ 
   there is a scheduler
   $\sched$ for $\cM$ such that
   $\Pr^{\tsched}_{\cM_i,s}(\varphi)=
    \Pr^{\sched}_{\cM,s}(\varphi)$
   for all states $s$ in $\cM$ and all
   0-EC-invariant properties $\varphi$.
   If $\tsched$ is an MD-scheduler then
   $\sched$ can be chosen as an MD-scheduler.
 \item [(b)]
   For each scheduler $\sched$ for $\cM$ there is a scheduler $\tsched$ for
   $\cM_i$ such that
   $$
    \Pr^{\sched}_{\cM,s}(\varphi) \ \ \leqslant \ \
    \Pr^{\tsched}_{\cM_i,s}(\varphi) \ \ \leqslant \ \
    \Pr^{\sched}_{\cM,s}(\varphi) + 
    \Pr^{\sched}_{\cM,s}\bigl\{
          \infpath \in \InfPaths :
          \lim(\infpath) \in \{\cE_1,\ldots,\cE_{i-1}\} \bigr\}
   $$
   for all states $s$ in $\cM$ and all
   0-EC-invariant properties $\varphi$.
   In particular:
   \begin{center}
   $\Pr^{\sched}_{\cM,s}(\varphi) =
    \Pr^{\tsched}_{\cM_i,s}(\varphi)$
   \ \
   if \ $\Pr^{\sched}_{\cM,s}\{
          \infpath \in \InfPaths :
          \textrm{$\lim(\infpath)$ is a 0-EC}\}=0$.
   \end{center}
 \end{enumerate}
  Thus, $\Pr^{\sup}_{\cM,s}(\varphi)=\Pr^{\sup}_{\cM_i,s}(\varphi)$
  for each 0-EC-invariant property $\varphi$ and each state
  $s$ in $\cM$.
  Furthermore, the existence of a scheduler $\sched$ for $\cM$ with
  $\Pr^{\sup}_{\cM,s}(\varphi)=\Pr^{\sched}_{\cM,s}(\varphi)$
  implies the existence of a scheduler $\tsched$ for $\cM_i$
  with
  $\Pr^{\sup}_{\cM,s}(\varphi)=\Pr^{\tsched}_{\cM_,s}(\varphi)$,
  and vice versa.
\end{lemma}
The proof follows from Lemma~\ref{preservation-of-0-EC-invariant-properties}
  using an inductive argument.

\tudparagraph{1ex}{Graph structure of the MDPs $\cM_i$.}
  With the renaming of the actions -- prior to the
  application of the spider construction -- to ensure that the action sets
   for each pair of distinct states are disjoint, the action label
  $\tau$ of transitions in $\cM_{j-1}$ will be renamed
  when constructing $\cM_i$ by applying the spider construction towards 
  $\cM_{i-1}$. This, however, is irrelevant for the following
  arguments and we will simply refer to them as $\tau$-transitions.

\begin{lemma}[$\tau$-transitions in the MDPs $\cM_i$] 
 \label{tau-transitions}
 For $i\in \{1,\ldots,\ell\}$:
 \begin{enumerate}
 \item [(a)]
  The $\tau$-transitions are always deterministic, \ie
  have a single target state that will be reached with probability 1. 
  The weight of a $\tau$-transition
  from $s$ to $t$ in $\cM_i$ is $w(s,t)$.
 \item [(b)]
  Each state $s$ has at most one
  $\tau$-transition in $\cM_i$, 
  and if $s$ has a $\tau$-transition in $\cM_i$ then no other
  action is enabled in $s$ as a state of $\cM_i$.
 \item [(c)]
   The graph built by the $\tau$-transitions in  $\cM_i$ is acyclic.
 \end{enumerate}
  In particular, $n$ is an upper bound for the total number of
  $\tau$-transitions in each of the MDPs $\cM_i$.
\end{lemma}

\begin{proof}
  Statement (a) is clear by the spider construction.
  The proof for statement (b) is by induction on $i$.
  For $i=1$ (basis of induction) this is clear by the spider construction.
  For the step of induction $i\Longrightarrow i{+}1$,
  we use the fact that if $s$ does not belong to the selected 0-BSCC 
  $\cE_i$ of the unique maximal end component $\cF_i$ of $\cM_i$
  then $(s,\alpha)\in \cM_i$ iff $(s,\alpha)\in \cM_{i+1}$.
  (Recall that $\cM_{i+1}$ arises from $\cM_i$ by
   replacing $\cF_i$ with $\spider{\cF_i}{\cE_i}$ for some 0-BSCC
   $\cE_i$ of $\cF_i$ and 
   replacing all state-action pairs $(s,\alpha)\in \cM_i \setminus \cE_i$ where
   $s$ is a state of $\cF_i$ different from the reference state
   $s_{0,i}$ with $(s_{0,i},\alpha)$.)
  If $s$ is contained in $\cE_i$, but not the reference state $s_{0,i}$ of the
  spider construction, then the spider construction replaces all
  state-action pairs $(s,\alpha)\in \cM_i\setminus \cE_i$ with
  $(s_{0,i},\alpha)$, 
  discards the unique state-action pair $(s,\alpha)\in \cE_i$
  and creates a $\tau$-transition from $s$ to $s_{0,i}$.
  Thus, if $s$ belongs to $\cE_i$ and $(s,\tau)\in \cM_i$ then 
  $(s,\tau)\in \cE_i$ (as $s$ has no other actions in $\cM_i$) and
  $(s,\tau)$ will be discarded. 
  In particular, the reference state $s_{0,i}$ 
  does not get ``new'' $\tau$-transitions, \ie
  \begin{enumerate}
  \item []
    State $t$ is a $\tau$-successor of $s_{0,i}$ in $\cM_i$ iff
    $t$ is a $\tau$-successor of $s_{0,i}$ in $\cM_{i+1}$.
    \hfill (*)
  \end{enumerate}
  This yields the claim for (b).

  Statement (c) follows by induction on $i$.
  The claim is obvious for $i=0$.
  In the step of induction, we assume
  that the MDPs $\cM_0,\ldots,\cM_i$ do not have any $\tau$-cycle,
  \ie a cycle built by $\tau$-transitions.
  Suppose by contradiction that $\cM_{i+1}=\spider{\cM_i}{\cE}$ 
  has a $\tau$-cycle
  $\cycle=t_0 \, \tau \, t_1 \, \tau \ldots \tau \, t_n$.
  At least one of the states $t_i$ must be contained in
  the 0-BSCC $\cE_i$ of $\cM_i$ as otherwise $\cycle$ would be
  a $\tau$-cycle in $\cM_i$.
  All states $s$ in $\cE_i$ that are different from the reference
  state $s_{0,i}$ have a $\tau$-transition to $s_{0,i}$ in $\cM_{i+1}$
  and no other action is enabled in $s$ as a state of $\cM_{i+1}$.
  Hence, with some state of $\cE_i$ also the reference state
   $s_{0,i}$ must be contained in $\cycle$. 
  Now let $s_{0,i}=t_0$ without loss of generality.
  But then the transition from $t_0=s_{0,i}$ to $t_1$ must have been
  generated in an earlier application of the spider construction
  (see (*) above).
  That is, $t_1=s_{0,i_1}$ for some $i_1 < i$.
  We repeat this argument and get natural numbers
  $i_1,\ldots,i_n$ with $0\leqslant i_n < i_{n-1} < \ldots < i_1$
  such that $t_m=s_{0,i_m}$.
  But then $\cycle$   is a $\tau$-cycle in $\cM_i$.
  Contradiction.
\end{proof}

Towards establishing the polynomial-time complexity of the iterative application
of the spider construction, let us discuss the sizes of the MDPs $\cM_i$. 
Although a single application $\cM_i\leadsto\cM_{i+1}$
of the spider construction can be done in polynomial time depending of the size
of $\cM_i$, this is not as obvious for the sequence $\cM_0\leadsto \ldots\leadsto \cM_{l}$:
As the number $l$ of applications required depends on the input $\cM_0$, the size of $\cM_l$
could be exponential in the size of $\cM_0$. Fortunately, this is not the
case, as the following lemma shows.
 
\begin{lemma}[Size of the MDPs $\cM_i$]
  \label{size-iterative-spider-construction}
  Notations as before. 
  The size of $\cM_i$ is polynomially bounded by the size of $\cM$.
\end{lemma}

\begin{proof}
 Recall that the size of an MDP has been defined as the number
 of states plus the total sum of the logarithmic lengths of the weights
 of all state-action pairs and the transition probabilities.
 Let $S$ be the state space of $\cM$ (and $\cM_i$), and let
 $n=|S|$ denote the number of states in $\cM$.
 Furthermore, let 
 $\wmax = \{|\wgt(s,\alpha)| : s\in S, \alpha \in \Act_{\cM}(s)\}$.

 The state-action pairs of $\cM_i$ have the form
 $(s,\beta)$ where $\beta$ is an action of $\cM$ 
 or stand for a $\tau$-transition.
 All $\tau$-actions are deterministic (\ie probability 1 for a single
 target state and 0 for all other states).
 The weight of a $\tau$-transition from $s$ to $t$ in $\cM_i$
 is $w(s,t)$, 
 which is the weight of all paths from $s$ to $t$
 in each 0-EC containing $s$ and $t$.
 The latter statement is a consequence of the fact that the union
  of 0-ECs in MDPs with maximal expected mean payoff 0 is again a 0-EC.
  See Lemma \ref{lemma:union-of-0-EC}.
 Moreover, $|w(s,t)| \leqslant (n{-}1) \cdot \wmax$.

 The total number of $\tau$-transitions in $\cM_i$ 
 is bounded by $n$ (see Lemma \ref{tau-transitions}).
 Hence, the total logarithmic length of the weights of all
 $\tau$-transitions in $\cM_i$
 is polynomially bounded by
 the size of the original MDP $\cM$.
 Likewise, for the state-action pairs
 $(s,\beta)$ where $\beta$ is an action of $\cM$, 
 say $\beta \in \Act_{\cM}(t)$,
 the logarithmic length of the 
 transition probabilities in $\cM_i$ are the same as in
 $\cM$ as we have $P_{\cM_i}(s,\beta,u)=P_{\cM}(s,\beta,u)$.
 Moreover, we have $\wgt_{\cM_i}(s,\beta)=w(s,t) + \wgt_{\cM}(t,\beta)$.
 But then again the total logarithmic length of the
 weights for these state-action pairs in $\cM_i$ is 
 polynomially bounded by
 the size of $\cM$.
\end{proof}

\begin{lemma}
\label{spider-construction-preserves-proper-actions}
   If $\alpha$ is an action of $\cM$ that does not belong to a 0-EC
   then $\alpha$ is an action of each of the MDPs $\cM_i$.
\end{lemma}

\begin{proof}
  By induction on $i$.
\end{proof}

As a consequence of part (b) of Lemma \ref{max-EC-of-Mi} and
Lemma \ref{spider-construction-preserves-proper-actions} we get:

\begin{remark}[Actions in the final MDP $\cN$]
\label{action-in-N-wgt-div-algo}
{\rm
   Let $\cN=\cM_{\ell}$ be the MDP that has been generated by the
   iterative application of the spider construction to an MDP $\cM$
   that has no weight-divergent end component.
   Recall from Lemma \ref{tau-transitions} 
   that each state $s$  has at most one
   $\tau$-transition in $\cN$, and if so, then no other action
   is enabled in $s$ as a state of $\cN$.

   The actions in $\cN$ are either actions of $\cM$
   that do not belong to a 0-EC of $\cM$ or $\tau$.
   In more detail, this means the following.
   \begin{itemize} 
   \item
     If $s$ is a state that does not belong to a 0-EC of $\cM$
     then $(s,\alpha)\in \cM$ iff $(s,\alpha)\in \cM_i$ 
     with the same weight and the same transition probabilities.
   \item
     Suppose now that $s$ is a state that belongs to some 0-EC of $\cM$.
     Then, for each state-action pair $(s,\alpha)\in \cN$:
     \begin{itemize}
     \item
        Either $\alpha=\tau$, 
        in which case $P_{\cN}(s,\alpha,t)=1$ and $\wgt(s,\alpha)=w(s,t)$
        for some state $t$ that belongs to the same maximal 0-EC as $s$,
     \item
        or there is a state $t$ that belongs 
        to the same maximal 0-EC $\cE$ of $\cM$ 
        and $(t,\alpha)\in \cM$.
        In this case, $P_{\cM}(t,\alpha,u)=P_{\cN}(s,\alpha,u)$ for all states
        $u$ and $\wgt_{\cN}(s,\alpha)=w(s,t)+\wgt_{\cM}(t,\alpha)$
        and $(t,\alpha)$ does not belong to any 0-EC of $\cM$.
     \end{itemize}
  \end{itemize}
  In particular, $\cN$ does not contain any action of $\cM$ that belongs to
  a 0-EC of $\cM$.
\Ende
 }
\end{remark}

\begin{lemma}
\label{proper-actions-in-M_i}
  Let $\cZ$ be a maximal 0-EC of $\cM$ and $i\in \{0,1,\ldots,\ell\}$.
  Let $T_i$ denote the set states $t$ that belong to 
  $\cZ$ such that either $t=s_{0,j}$ for the largest index
  $j\in \{0,1,\ldots,i{-}1\}$ where all states of $\cE_j$ are contained in
  $\cZ$ or
  $\cM_i$ contains a state-action pair $(t,\alpha)$ where $\alpha$ is an
  action of $\cM$. 
  Then, for each state $s$ in $\cZ$ and each state $t\in T$
  there is a path 
  $\fpath = t_0\, \alpha_0 \, t_1 \, \alpha_1 \ldots \alpha_{m-1} \, t_m$ 
  from $t_0=s$ to $t=t_m$ in $\cM_i$ 
  such that for each $j\in \{0,1,\ldots,m{-}1\}$
  either $\alpha_j=\tau$ or $\alpha_j$ is an action of $\cZ$.
  Moreover, the weight of each such path is $w(s,t)$.
\end{lemma}

\begin{proof}
  We first observe that if $s$ is a state of $\cZ$ and
  $\alpha$  an action of $\cM$ such that
  $(s,\alpha)\in \cM_i \setminus \cZ$ then 
  either $s$ is the reference state $s_{0,j}$ of $\cE_j$ for some $j < i$
  such that $s$ is not contained in $\cE_{j+1}\cup \ldots \cE_{i-1}$
  or there is some action of $\cZ$ such that 
  $(s,\beta) \in \cM_i$.
  This is a consequence of (S4) in Lemma \ref{spider-construction}
  (see also Lemma \ref{spider-construction-S4}).
  Hence, it suffices to consider for the case where $\cM$ is a 0-EC,
  in which case $\cZ=\cM$.
  The claim then follows by induction on $i$. 
  The basis of induction $i=0$ is trivial and
  the step of induction follows from 
  statement (S2) of Lemma~\ref{spider-construction-S1S2}.
\end{proof}

\begin{corollary}
\label{M-0-EC-spider}
  If $\cM$ is a 0-EC then the final MDP $\cN=\cM_{\ell}$ 
  generated by the weight-divergence algorithm
  can be viewed as an acyclic graph built by $\tau$-transitions
  with a single trap state that is reachable from all other states.
\end{corollary}

\begin{lemma}[Properties of the final MDP]
 \label{properties-final-MDP-wgt-div-algo}
   If  $\cM=\cM_0$ is not weight-divergent then
   the final MDP $\cN=\cM_{\ell}$ 
   generated through the iterative application of the
   spider construction
   enjoys the following properties:
   \begin{enumerate}
   \item [(a)]
     Let $\fpath$ be a path in $\cN$ from $s$ to $t$ 
     built by $\tau$-transitions.
     Then, $\wgt_{\cN}(\fpath)=w(s,t)$.
   \item [(b)]

     For each maximal 0-EC $\cZ$ of $\cM$ there is a state
     $t_{\cZ}$ such that: 
     \begin{itemize}
     \item
       For each state $s$ in $\cZ$ there is a path $\fpath_s$ from
       $s$ to $t_{\cZ}$ in $\cN$ built by $\tau$-transitions.
       Moreover, $\fpath_s$ is a prefix of 
       each maximal path $\infpath$ in $\cN$ with
       $s=\first(\infpath)$.
     \item 
       If $(s,\alpha)$ is a state-action pair in $\cN$ where $s$ belongs
       to $\cZ$ and $\alpha$ is an action of $\cM$ then
       $s=t_{\cZ}$.
     \item
       Whenever $s$ is a state of $\cZ$ and $(s,\alpha)$ a state-action
       pair of $\cM$ that does not belong any 0-EC then 
       $(t_{\cZ},\alpha)$ is a state-action pair of $\cN$.
     \end{itemize}
   \item [(c)]
     If $\cZ$ is a maximal 0-EC in $\cM$ and
     $\cN$ does not contain a state-action pair $(t,\alpha)$
     where $t$ belongs to $\cZ$ and $\alpha$ is an action of $\cM$
     then $\cZ=\cM$.
  \end{enumerate} 
\end{lemma}

\begin{proof}
Statement (a) is clear from Lemma \ref{tau-transitions}.
Statement (b) follows from Lemma \ref{proper-actions-in-M_i}
and the observation that all actions in $\cN$ are either actions
of $\cM$ or $\tau$ (see Remark \ref{action-in-N-wgt-div-algo}).
This yields that
whenever $s$ and $t$ belong to the same maximal 0-EC $\cZ$ of $\cM$
     and $\cN$ contains a state-action pair $(t,\alpha)$ where
     $\alpha$ is an action of $\cM$
     then $t$ is reachable from $s$ in $\cN$ via $\tau$-transitions.
As $\cN$ has no 0-EC, each maximal 0-EC of $\cM$ can contain only
one such a state $t$.
Statement (c) follows from
Lemma \ref{spider-construction-preserves-proper-actions}.
\end{proof}

\subsubsection{Soundness of the Weight-Divergence Algorithm}

\label{appendix:wgt-div}

We are now ready to prove the soundness of the weight-divergence
algorithm for a given strongly connected MDP $\cM$ presented in
Section \ref{sec:wgt-div}. We rephrase now
Theorem~\ref{weight-divergence-algorithm} to make the connection between $\cM$ 
and $\cN$ obtained from the algorithm explicit to check weight-divergence.

\begin{theorem}%
\label{appendix:wgt-div-algo}
  The algorithm for checking weight-divergence of a strongly connected
  MDP $\cM$ runs in time polynomial in the size of $\cM$.
  If $\cM$ is weight-divergent then it either finds a pumping or
  gambling MD-scheduler. If $\cM$ is not weight-divergent, then 
  it generates  a new MDP $\cN$ such that
  \begin{enumerate}
  \item [(W1)]
     $\cN$ and $\cM$ have the same state space
     and $\|\cN\|\leqslant \|\cM\|$.
  \item [(W2)]
     $\cN$ has at most one maximal end component, and if so,
     $\Exp{\max}{\cF}(\MP)<0$ for the unique maximal end component
     $\cF$ of $\cN$ and $\cF$ is reachable from all states in $\cN$.
  \item [(W3)]
     $\cN$ and $\cM$ are equivalent with respect to the class of
     0-EC-invariant properties in the sense that the statement 
     of Lemma~\ref{lemma:equivalence-M-and-Mi} holds.
 \end{enumerate}
\end{theorem}

\begin{proof}
The weight-divergence algorithm generates a sequence 
$\cM_0=\cM,\cM_1,\ldots,\cM_{\ell}=\cN$ of MDPs as stated
at the beginning of Section \ref{sec:iterative-application-spider}.
Hence, each of the $\cM_i$'s for $i<\ell$ 
has a unique maximal end component $\cF_i$, 
$\Exp{\max}{\cF_i}(\MP)=0$ and
$\cM_{i+1}=\spider{\cM_i}{\cE_i}$
where $\cE_i$ is a 0-BSCC of $\cF_i$.
  Then, $\cM_0,\ldots,\cM_{\ell}$ have the same state space
  and $\|\cM_{i+1}\|=\|\cM_i\|{-}1$.
  This yields (W1) and $\ell \leqslant \|\cM\|$.
  By Lemma \ref{max-EC-of-Mi}, for each $i\in \{0,1,\ldots,\ell\}$
  we have that
  $\cM$ is weight-divergent iff
  $\cM$ is weight-divergent.
Lemma \ref{lemma:equivalence-M-and-Mi} 
yields the equivalence of $\cM$ and $\cM_i$
as stated in (W3).

In case that $\cN=\cM_{\ell}$ 
has end components, the reachability of the
unique maximal end component $\cF$ of $\cN$ follows from the observation that
the reference state $s_{0,i}$ is accessible from all states
in $\cM_{i+1}=\spider{\cM_i}{\cE_i,s_{0,i}}$ 
(by induction on $i$).

For $i=0,1,\ldots,\ell$,
the weight-divergence algorithm computes $\Exp{\max}{\cF_i}(\MP)$ and  
an MD-scheduler $\tsched$ maximizing the expected mean payoff. 
The case $\Exp{\max}{\cF_i}(\MP)\not=0$ is only possible if $i=\ell$
as we then have:
\begin{itemize}
\item
  If $\Exp{\max}{\cF_i}(\MP)<0$ then all schedulers for $\cM_i$ are
  negatively pumping
  (Lemma \ref{universally-pumping}
  with all weights multiplied by $-1$), 
  and hence, $\cM_i$ and $\cM$ are not positively weight-divergent.
  In this case, the final MDP $\cM_i$ enjoys the properties
  (W1), (W2) and (W3).
\item
  If $\Exp{\max}{\cF_i}(\MP)>0$ then $\cM$ is pumping
  (Lemma \ref{lemma:pumping-ecs})
  and therefore positively weight-divergent.
  In this case, $\tsched$ is a pumping MD-scheduler for $\cM_i$.
  We now can 
  rely on the scheduler transformation presented in
  part (a) of Lemma \ref{lemma:equivalence-M-and-Mi}
  to obtain a pumping MD-scheduler $\sched$ for $\cM$.
\end{itemize}
Suppose now $\Exp{\max}{\cF_i}(\MP) = 0$. 
If $\tsched$ has a gambling BSCC then
$\cM$ is gambling and therefore positively weight-divergent.
In this case, we can again
  rely on the scheduler transformation presented in
  part (a) of Lemma \ref{lemma:equivalence-M-and-Mi}
  to obtain a gambling MD-scheduler $\sched$ for $\cM$.

Otherwise, each BSCC of the Markov chain induced by $\tsched$ is a 0-BSCC
(Lemma \ref{gambling-MC}).
In this case, $i<\ell$ and $\cE_i$ is one of the  0-BSCCs  of $\tsched$.
The weight-divergence algorithm generates the MDP
$\cM_{i+1}=\spider{\cM_i}{\cE_i}$.
If $\cF_i=\cE_i$ then $\ell=i{+}1$ and 
$\cM_i$ (and therefore $\cM$) are not weight-divergent
and the final MDP $\cM_{i+1}$ has no end components.
Thus, the conditions (W1), (W2) and (W3)
are fulfilled.
Suppose now that $\cF_i \not= \cE_i$. In this case, the
procedure will be repeated with the MDP $\cM_{i+1}$.

Lemma \ref{size-iterative-spider-construction} yields that the size
of the MDPs $\cM_i$ is polynomially bounded in the size of 
the original MDP $\cM$.
Thus, the cost per iteration (the computation of the 
maximal end component $\cF_i$ of $\cM_i$
and its maximal expected mean payoff as well as the new MDP
$\cM_{i+1}=\spider{\cM_i}{\cE_i}$)
are polynomially bounded in the size of $\cM$.
This yields a polynomial-time bound for the weight-divergence algorithm.
\end{proof}

  Statement (W2) in Theorem \ref{appendix:wgt-div-algo}
  implies that if the final MDP 
  $\cN$ has end components then
  $\cN$ is \emph{universally negatively pumping} in the sense that
  $\limsup\limits_{n \to \infty} \ \wgt(\prefix{\infpath}{n}) = -\infty$
  holds for almost all paths $\infpath$ under each scheduler
  for $\cN$.

\subsection{Universal Negative Weight-Divergence and Boundedness}

\label{sec:univ-neg-wgt-div}
\label{sec:boundedness}

This section provides the proof for Theorem \ref{universal-neg-wgt-div}
stating that for strongly connected MDPs with maximal mean payoff 0,
the absence of 0-ECs is equivalent to the universal negative weight-divergence
property: 
\thmnegwdZec*

The proof for Theorem \ref{universal-neg-wgt-div} is split in the
proof of each statement (a) and (b). The essential argument for the
proof is that if $\cM$ has schedulers where the set of paths that are
bounded from below then $\cM$ has positive measure then $\cM$ must
have a 0-EC.

Let us first recall the definition of boundedness from below and introduce
related notions.
  Let $L,U\in \Integer \cup\{\pm\infty\}$ with $L \leqslant U$. 
  An infinite path $\infpath$ is said to be $(L,U)$-bounded
  iff
  $$
     \stackrel{\infty}{\forall}n \in \Nat. \ 
     L \ \leqslant \ \wgt(\prefix{\infpath}{n}) \ \leqslant \ U
  $$
  where $\stackrel{\infty}{\forall}$ means ``for all, but finitely many''.
  Thus: 
  \begin{center}
    $\infpath$ is $(L,+\infty)$-bounded \ \ iff \ \
    $\liminf\limits_{n \to \infty} \wgt(\prefix{\infpath}{n})\geqslant L$\enskip.
  \end{center}
  An infinite path $\infpath$ is 
  \emph{bounded from below} if there is some integer $L$ such that
  $\infpath$ is $(L,+\infty)$-bounded.  
  Clearly, this is equivalent to the requirement that
  $\liminf\limits_{n \to \infty} \wgt(\prefix{\infpath}{n})\in \Integer$.

  Scheduler $\sched$ for a strongly connected MDP $\cM$ 
  is said to be \emph{almost-surely $(L,U)$-bounded} from state $s$,
  or briefly \emph{$(L,U)$-bounded} from $s$,
  if for almost all infinite $\sched$-path $\infpath$ starting in $s$
  are $(L,U)$-bounded.
  Scheduler $\sched$ is said to be \emph{probably $(L,U)$-bounded} 
  from state $s$
  if
  $$
     \Pr^{\sched}_{\cM,s}\bigl\{ \
         \infpath \in \InfPaths \ : \ 
         \stackrel{\infty}{\forall} n \in \Nat. \
         L \leqslant \wgt(\prefix{\infpath}{n}) \leqslant U \
     \bigr\} \ > \ 0
  $$
  We say $\sched$ is \emph{probably bounded} (resp.~\emph{almost-surely bounded})
  from $s$ if there 
  exist $L,U\in \Integer$ with $L \leqslant U$ such that 
  $\sched$ is probably (resp.~almost-surely) $(L,U)$-bounded from $s$.
  Likewise, $\sched$ is called
  \emph{probably bounded from below} 
  (resp.~\emph{almost-surely bounded from below}
  or briefly \emph{bounded from below})
  from state $s$ if there
  exists an integer $L$ such that
  $\sched$ is probably (resp.~almost-surely) $(L,+\infty)$-bounded from $s_0$.

\begin{lemma}
\label{from-bounded-to-L-bounded}
   If $\sched$ is a scheduler that is
   probably bounded from below from some state $s$
   then there is some integer $L$ such that
   $\sched$ is probably $(L,+\infty)$-bounded
   from $s$. 
\end{lemma}

\begin{proof}
 Let $\Pi$ denote the set of all $\sched$-paths $\infpath$ from $s$
 that are bounded from below.
 Likewise, for $L\in \Integer$, let
 $\Pi_L$ denote the set of all $\sched$-paths $\infpath$ from $s$
 that are $(L,+\infty)$-bounded.
 Then, $\Pi=\bigcup_{L\in \Integer} \Pi_L$.
 Hence, $\Pr^{\sched}_{\cM,s}(\Pi)>0$ iff
 there exists $L\in \Integer$ with $\Pr^{\sched}_{\cM,s}(\Pi_L)>0$.
\end{proof}

\begin{figure}[ht]
	\centering%

\begin{tikzpicture}
	\node[state] (s) {$s$};
	\node[state, right=25mm of s] (t) {$t$};
	\node[bullet, above right=5mm and 15mm of s] (b) {};
	\path%
		(t) edge[ptran,loop, min distance=9mm, out=-30, in=30] node[right]{$\alpha$/\nilw} (t)
		(t) edge[ptran, bend left] node[below]{$\beta$/\nilw} (s)
		(s) edge[ntran,bend left] node[above] {$\alpha$/\negw{1}} (b)
		(b) edge[ptran,bend left] coordinate[pos=0.3] (bt) (t)
		(b) edge[ptran, bend left] coordinate[pos=0.3] (bs) (s);
	\draw%
		(bt) to[bend left] (bs);				
\end{tikzpicture}
 	\caption{Lemma \ref{from-bounded-to-L-bounded} does not hold for almost-sure boundedness}
\label{fig:probably-bounded-scheduler}
\end{figure}
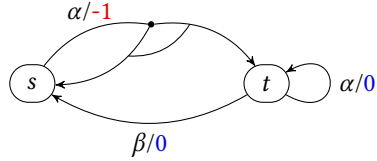

Note that the analogous statement of Lemma \ref{from-bounded-to-L-bounded}
does not hold for almost-sure boundedness.
An example is the strongly connected MDP $\cM$
depicted in Figure~\ref{fig:probably-bounded-scheduler}
consisting of the state-action pairs
$(s,\alpha)$, $(t,\alpha)$ and $(t,\beta)$.
Then, the MD-scheduler $\sched$ that choses $\alpha$ for states $s$ and $t$
is almost-surely bounded from below, but there is no integer $L$
such that $\cM$ has an almost-surely $(L,+\infty)$-bounded scheduler.

Obviously, if $\cM$ is strongly connected then 
the existence of a probably bounded scheduler 
does not depend on the starting state $s$. 
To see this, we suppose that
$\cM$ has a probably $(L,U)$-bounded scheduler from $s$. 
For each state $t$ in $\cM$ we pick a finite path
$\fpath_t$ from $t$ to $s$.
Then, for each state $t$, $\cM$ has a probably $(L_t,U_t)$-bounded scheduler
from $t$ where $L_t=L+\wgt(\fpath_t)$ and $U_t=U+\wgt(\fpath_t)$.
The analogous statement does not hold for almost-surely bounded schedulers.
For example, considering again the MDP $\cM$ 
illustrated by Figure~\ref{fig:probably-bounded-scheduler}.
Then $\cM$ has an almost-surely $(0,0)$-bounded scheduler from $t$, but 
there is no scheduler that is almost-surely bounded from below 
from state $s$.

\begin{lemma}[From probably to almost-surely boundedness]
\label{from-prob-bounded-to-surely-bounded}
   Let $\cM$ be a strongly connected MDP with $\ExpRew{\max}{\cM}(\MP)=0$
   and let $s_0$ be a state in $\cM$, 
   $L\in \Integer$ with $U\in \Integer \cup \{+\infty\}$,
   and $\sched$ a scheduler that is probably $(L,U)$-bounded from $s_0$.
   Then, there exists an MD-scheduler $\tsched$ such that
   the Markov chain $\cC_{\tsched}$ induced by $\tsched$ contains a 0-BSCC
   $\cB$. In particular, 
   $\tsched$ is bounded from below.
\end{lemma}

\begin{proof}
   Recall that for an infinite paths 
   $\infpath=s_0\, \alpha_0\, s_1\, \alpha_1 \ldots$, 
   the limit of $\infpath$, denoted $\lim(\infpath)$,
   is the set of all state-action pairs $(s,\alpha)$ 
   such that $(s,\alpha)=(s_n,\alpha_n)$ for infinitely many 
   indices $n\in \Nat$.
   Given an end component $\cE$ of $\cM$, we define:
   $$
     \Pi_{\cE} \ \ = \ \ 
     \bigl\{ \ \infpath \in \InfPaths \ : \ 
         \lim(\infpath)=\cE \ \text{ and } \ 
         \stackrel{\infty}{\forall} n \in \Nat. \
         L \leqslant \wgt(\prefix{\infpath}{n}) \leqslant U \
     \bigr\}\enskip.
   $$
   As $\sched$ is probably $(L,U)$-bounded, 
   there exists an end component $\cE$ of $\cM$ such that
   $\Pr^{\sched}_{\cM,s_0}(\Pi_{\cE})$ is positive.

   For each state $s$ in $\cE$ and each integer 
   $k\in \{\ell \in \Integer: L \leqslant \ell \leqslant U\}$,
   let $\Pi_{s,k}$ denote the following set:
   $$
     \Pi_{s,k} \ \ = \ \
     \bigl\{ \ \infpath \in \Pi_{\cE} \ : \ 
               \stackrel{\infty}{\exists} n \in \Nat. \ 
               (\ \state{\infpath}{n}=s \ \wedge \ 
                  \wgt(\prefix{\infpath}{n}) = k \ ) \
     \bigr\}
   $$
   where $\state{\infpath}{n}=s_n$ denotes the $(n{+}1)$-st state of
   $\infpath$.

   For each state $s$ in $\cE$, let $K_s$ denote the set of integers 
   $k \in \{\ell \in \Integer: L \leqslant \ell \leqslant U\}$
   such that $\Pr^{\sched}_{\cM,s_0}(\Pi_{s,k})>0$.
   Then, $K_s$ is nonempty, as $\Pi_{\cE}$ agrees with the union of the sets
   $\Pi_{s,k}$ when $k$ ranges over all integers $\ell$ with
   $L \leqslant \ell \leqslant U$.
   Let $k_s=\min K_s$ and $\Pi_s=\Pi_{s,k_s}$.
   Thus, $\Pr^{\sched}_{\cM,s_0}(\Pi_s)>0$.

   For each state-action pair $(s,\alpha)$ in $\cE$, let
   $\Pi_{s,\alpha}$ denote the set
   of all paths $\infpath \in \Pi_s$ such that the following
   condition holds:
   $$
               \stackrel{\infty}{\exists} n \in \Nat. \ 
               (\ \state{\infpath}{n}=s \ \wedge \ 
                  \wgt(\prefix{\infpath}{n}) = k_s \ \wedge \  
                  \sched(\prefix{\infpath}{n}) = \alpha \ )\enskip.
   $$
   The above condition assumes that $\sched$ is deterministic.
   If $\sched$ is randomized then we replace the condition
   ``$\sched(\prefix{\infpath}{n}) = \alpha$'' with
   ``$\sched(\prefix{\infpath}{n})(\alpha) \geqslant 1/|\Act(s)|$''.
   As there are only finitely many actions $\alpha$ with
   $(s,\alpha)\in \cE$ 
   and $\Pi_s$ is the union of the sets $\Pi_{s,\alpha}$
   there is some action $\alpha_s \in \Act_{\cE}(s)$ with
   $\Pr^{\sched}_{\cM,s_0}(\Pi_{s,\alpha_s})>0$.
   But then 
   $$
     L \ \ \leqslant \ \ k_s+\wgt(s,\alpha_s) \ \ \leqslant \ \ U
   $$
   and for each state $t$ with $P(s,\alpha_s,t)>0$ we have that
   $t$ belongs to $\cE$ and 
   $$
      k_t \ \ \leqslant \ \ k_s+\wgt(s,\alpha_s)\enskip.
   $$
   Let $R = L +\max_{s\in \cE} k_s$.
   We now consider the MD-scheduler $\tsched$ that schedules $\alpha_s$
   for each state $s$ in $\cE$ and 
   satisfies $\Pr^{\tsched}_{\cM,u}(\Diamond \cE)=1$
   for all states $u$ in $\cM$ that do not belong to $\cE$.
   (Such a scheduler exists as $\cM$ is strongly connected.)

   By induction on the length $|\fpath|$ of 
   finite $\tsched$-paths starting in some state of $\cE$, we obtain
   $\wgt(\fpath) \geqslant L-R+k_{\last(\fpath)}$ 
   if $\fpath$ is a finite $\tsched$-path with
   $\first(\fpath)\in \cE$.
   \begin{itemize}
   \item
     Basis of induction:
     If $|\fpath|=0$, say $\fpath=s\in \cE$, then
     $\wgt(\fpath)=0 \geqslant L-R+k_s$ as $R \geqslant L+k_s$ by 
     the choice of $R$.
   \item
     Step of induction:
     If $\fpath$ is a path of length $n{+}1$ and its last transition
     is $s \move{\alpha_s} t$ then we apply the induction hypothesis to
     its prefix of length $n$ and obtain:
   \begin{eqnarray*}
     \wgt(\fpath) & \ = \ & 
     \wgt(\prefix{\fpath}{n}) \ + \ \wgt(s,\alpha_s)
     \\[1ex]
     & \ \geqslant \ & (L-R+k_s) \ + \ \wgt(s,\alpha_s)
     \\[1ex]
     & \ \geqslant \ & 
     L-R \ + \ \underbrace{(k_s + \wgt(s,\alpha_s))}_{\geqslant k_t}
     \ \ \geqslant \ \ 
     L-R +k_t
   \end{eqnarray*}
   \end{itemize}

   In particular, $\wgt(\fpath) \geqslant 
                      L-R +\min_{s\in \cE} k_s \eqdef L^*$
   for all finite $\tsched$-paths starting in some state of $\cE$.
   
   Let now $\cB$ be a BSCC of $\tsched$. 
   Then, $\cB$ is a sub-component of $\cE$.
   As the weight of all finite
   paths in $\cB$ is bounded by $L^*$ from below, $\cB$ does not have
   negative cycles. Hence, $\ExpRew{}{\cB}(\MP)\geqslant 0$.
   On the other hand,
   $\ExpRew{}{\cB}(\MP)\leqslant \ExpRew{\max}{\cM}(\MP)=0$.
   This yields that $\ExpRew{}{\cB}(\MP)=0$ and that $\cB$ is a 0-BSCC
   (Lemma \ref{weight-div-MC}).
\end{proof}

The following lemma restates part (a) of Theorem~\ref{universal-neg-wgt-div}.

\begin{lemma}%
 \label{lemma:0-EC}
  Let $\cM$ be a strongly connected MDP with $\ExpRew{\max}{\cM}(\MP) =0$. 
  Then, the following statements are equivalent
\begin{enumerate}
\item [(a)] 
    $\cM$ has a 0-EC.
\item [(b)]
    There exists an MD-scheduler $\sched$ such that 
    the Markov chain induced by $\sched$ contains a 0-BSCC.
\item [(c)]
   $\cM$ has a scheduler that is almost-surely bounded from below 
  from each state.
\item [(d)]
   There exist integers $L, U$ such that
   $\cM$ has a probably $(L,U)$-bounded scheduler
   from some state.
\item [(e)]
   There exists an integer $L$ such that
   $\cM$ has a probably $(L,+\infty)$-bounded scheduler
   from some state.
\end{enumerate}
\end{lemma}

\begin{proof}
 We first show the equivalence of (a) and (b).
 The implication ``(b) $\Longrightarrow$ (a)'' is trivial.
 For the proof of ``(a) $\Longrightarrow$ (b)'', we suppose we are
 given a 0-EC $\cE$ of $\cM$. We pick an MD-scheduler $\sched$ for $\cM$
 such that the state-action pairs $(s,\sched(s))$ belong to $\cE$ whenever $s$
 is a state of $\cE$ and $\Pr^{\sched}_{\cM,s}(\Diamond \cE)=1$ for all
 states $s$ in $\cM$ with $s \notin \cE$.
 Then, each BSCC of the Markov chain induced for $\sched$ is a 0-BSCC.
 The equivalence of statements (c), (d) and (e) is a consequence of
 Lemma \ref{from-prob-bounded-to-surely-bounded}, which shows
 ``(e) $\Longrightarrow$ (c)'', while
 ``(c) $\Longrightarrow$ (d)'' and ``(d) $\Longrightarrow$ (e)'' 
 are trivial.
 We finally check the equivalence of statements (a)/(b) and (c)/(d).
 The implication ``(b) $\Longrightarrow$ (d)'' is obvious,
 while ``(c) $\Longrightarrow$ (b)'' has been shown in the proof of
 Lemma \ref{from-prob-bounded-to-surely-bounded}.
\end{proof}

The next lemma is part (b) in Theorem~\ref{universal-neg-wgt-div}.

\begin{lemma}%
 \label{appendix:universal-neg-wgt-div}
  Let $\cM$ be a strongly connected MDP with
  $\Exp{\max}{\cM}(\MP)=0$.
  Then, the following statements are equivalent:
  \begin{enumerate}
  \item [(a)]
     $\cM$ has no 0-EC.
  \item [(b)]
     $\cM$ has no scheduler that is bounded from below.
  \item [(c)]
     Each scheduler for $\cM$ is negatively weight-divergent.
  \end{enumerate}
\end{lemma}

\begin{proof}
The implications ``(c) $\Longrightarrow$ (b)'' 
and ``(b) $\Longrightarrow$ (c)'' 
are trivial as 
schedulers realizing a 0-EC $\cE$ are bounded from below
and as no negatively weight-divergent scheduler is bounded from below.
To prove ``(a) $\Longrightarrow$ (c)'' 
we suppose that $\cM$ has no 0-EC.
We suppose by contraction that $\cM$ has a scheduler $\sched$ that is
not negatively weight-divergent. That is, there is a state $s$
such that
  \[ 
   \Pr^{\sched}_{\cE,s}
   \bigl\{ \ \infpath \in \InfPaths \ : \ 
         \liminf_{n \to \infty} \ \wgt(\prefix{\infpath}{n}) \ > \ -\infty
   \bigr\} \ \ > \ \ 0\enspace.
  \]
  But then, there exists $L\in \Integer$ such that $\sched$ is
  probably $(L,+\infty)$-bounded from $s$ (see
  Lemma~\ref{from-bounded-to-L-bounded}). We derive, by
  Lemma~\ref{lemma:0-EC} that $\cM$ has a 0-EC, a contradiction.
\end{proof}

\subsection{Checking the Gambling Property}

\label{sec:checking-gambling}

\begin{figure}[ht]
	\centering%

\begin{tikzpicture}
	\node[state] (s0) {$0$};
	\node[state, right=16mm of s0] (s1) {$1$};
	\node[state, right=16mm of s1] (s2) {$2$};
	\node[istate, right=16mm of s2] (dots) {$\ldots$};
	\node[state, right=16mm of dots] (sn1) {$n{-}1$};
	\node[state, right=16mm of sn1] (sn) {$n$};
	\node[state, right=16mm of sn] (snp1) {$n{+}1$};
	\node[bullet, above=10mm of s0] (b1) {};
	\node[bullet, below=10mm of snp1] (b2) {};
	\path%
		(s0) edge[ntran, bend left] node[left]{$\alpha$/\posw{1}} (b1)
		(b1) edge[ptran, bend left] coordinate[pos=0.3](b1s0) (s0)
		(b1) edge[ptran, bend left] coordinate[pos=0.3](b1s1) (s1)
		(snp1) edge[ntran, bend left] node[right]{$\alpha$/\negw{1}} (b2)
		(b2) edge[ptran, bend left] coordinate[pos=0.3](b2snp1) (snp1)
		(s1) edge[ptran, bend left] node[above] {$\alpha$/\nilw} (s2)
		(s2) edge[ptran, bend left] node[above] {$\alpha$/\nilw} (dots)
		(dots) edge[ptran, bend left] node[above] {$\alpha$/\nilw} (sn1) 
		(sn1) edge[ptran, bend left] node[above] {$\alpha$/\nilw} (sn)
		(sn) edge[ptran, bend left] node[above] {$\alpha$/\negw{$b$}} (snp1)
		(s1) edge[ptran, bend right] node[below] {$\beta$/\posw{$a_1$}} (s2)
		(s2) edge[ptran, bend right] node[below] {$\beta$/\posw{$a_2$}} (dots)
		(dots) edge[ptran, bend right] node[below] {$\beta$/\posw{$a_{n-2}$}} (sn1) 
		(sn1) edge[ptran, bend right] node[below] {$\beta$/\posw{$a_{n-1}$}} (sn);
	\coordinate[below=12mm of s0] (bs0);
	\coordinate[below=12mm of sn] (bsn);
	\coordinate[below=12mm of sn1] (bsn1);
	\draw[-] 
	  	(b2) .. coordinate[pos=0.1](b2snp2) controls (bsn) .. (bsn1);
	\draw[ptran]	
		(bsn1) .. controls (bs0) .. (s0.south);
	\draw%
		(b1s0) to[bend right] (b1s1);				
	\draw%
		(b2snp1) to[bend right] (b2snp2);
\end{tikzpicture}
 	\caption{Reduction from subset sum for the \NP-hardness of checking the gambling property.}
\label{fig:complexity-checking-existence-gampling-md-scheduler}
\end{figure}
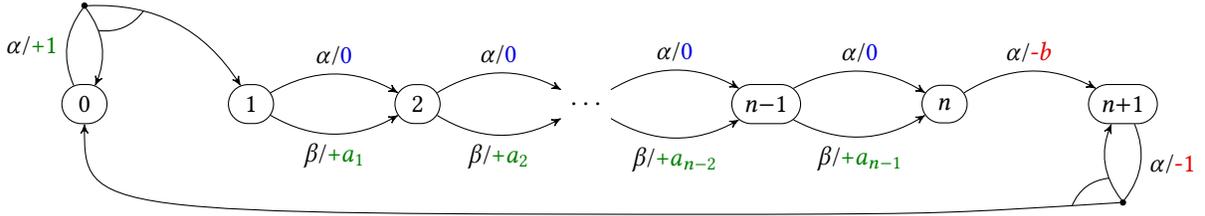

\gambling*

\begin{proof}
  We first observe that if
  $\ExpRew{\max}{\cM}(\MP) = 0$,
  then an end component is gambling iff it is weight-divergent.
  Hence, 
  statement (a) follows from Theorem~\ref{weight-divergence-algorithm}.

  We show the \NP-completeness in the general case (statement (b)).
    A nondeterministic polynomially time-bounded algorithm 
  is obtained by first guessing nondeterministically an
  MD-scheduler $\sched$ and then checking deterministically whether
  $\ExpRew{\sched}{\cM}(\MP)=0$ and 
  the Markov chain induced by $\sched$ has a positive cycle in some its BSCCs.
  \NP-hardness is achieved by a polynomial reduction from the 
  subset sum problem, which takes as input a finite nonempty sequence
  $a_1,a_2,\ldots,a_{n-1},b$ of positive integers 
  and asks for a subset
  $I$ of $\{1,\ldots,n{-}1\}$ such that $\sum_{i\in I} a_i=b$.
  Let $\cM$ be the MDP 
  illustrated by Figure~\ref{fig:complexity-checking-existence-gampling-md-scheduler}.
  Clearly, $\cM$ is strongly connected.
  There is a one-to-one correspondence between the MD-schedulers
  for $\cM$ and the subsets $I$ of $\{1,\ldots,n{-}1\}$.
  Given $I \subseteq \{1,\ldots,n{-}1\}$, we define
  $\sched_I$ as the MD-scheduler that picks $\beta$ for the states
  $i\in I$ and action $\alpha$ for all other states.
  Then, the Markov chain $\cC_I$ induced by $\sched_I$ consists of a single
  BSCC, and $\sched_I$ is gambling iff 
  $\ExpRew{\sched_I}{\cM}(\MP)=0$ iff $\sum_{i\in I} a_i=b$.
  Thus, $\cM$ has a gambling MD-scheduler iff there exists 
  $I\subseteq \{1,\ldots,n{-}1\}$ with $\sum_{i\in I} a_i=b$.
\end{proof}

\subsection{Computing the set $\ZeroEC$ and the Recurrence Values}
 
\label{section:app-minimal-credits}

We now provide the proof for Lemma~\ref{mincredit-ZeroEC}.
So, we suppose that we are given a strongly connected MDP $\cM$
with $\ExpRew{\max}{\cM}(\MP)=0$.
Section \ref{sec:compute-max-0-EC} explains 
how to compute the maximal 0-ECs. 
In Section \ref{sec:min-credit}, we will explain
how to compute the recurrence values of the states belonging
to a 0-EC.

\subsubsection{Computing the Maximal 0-ECs}

\label{sec:compute-max-0-EC}

We now turn to the computation of the maximal 0-ECs.
(Recall the notion of maximal 0-ECs from Section \ref{sec:max-0-EC}.)
Let $\ZeroEC$ denote the set of all states that belong to some 0-EC. 
Thus, $\ZeroEC$ is the union of the state spaces of all maximal 0-ECs.

\begin{lemma}[First part of Lemma~\ref{mincredit-ZeroEC}]
\label{appendix:ZeroEC}
If $\cM$ is strongly connected and
$\ExpRew{\max}{\cM}(\MP)=0$ 
then the maximal 0-ECs (and the set $\ZeroEC$) 
are computable in polynomial time.
\end{lemma}

\begin{proof}
  To compute the maximal 0-ECs we present an algorithm that identifies all
  state-action pairs $(s,\alpha)$ that are contained in some 0-EC.
  For this, we combine the polynomial-time algorithm
  to check the existence of 0-BSCCs (Section \ref{sec:algo-checking-0-EC})
  and the polynomial-time algorithm
  to flatten 0-BSCCs (spider construction).

  Obviously, all state-action pairs $(s,\alpha)$ with
  $P(s,\alpha,s)=1$ and $\wgt(s,\alpha)=0$ constitute a 0-EC.
  In the sequel, we suppose that such ``trivial'' state-action pairs have been
  removed from $\cM$.
  We first rename the actions in $\cM$ to ensure that
  $\Act_{\cM}(s)\not= \Act_{\cM}(s')$
  for all states
  $s$, $s'$ with $s\not= s'$. That is, for each action name $\alpha$ 
  in $\cM$ there is a unique state-action pair $(s,\alpha)\in \cM$.
  Thus, it suffices to identify all actions that belong to some 0-EC.

  The algorithm works as follows.
  We first run the polynomial-time algorithm 
  to check the existence of a 0-BSCC
  (see Section \ref{sec:algo-checking-0-EC}).
  If a 0-BSCC $\cB$ of $\cM$ is found then we apply the spider construction
  to transform the original
  $\cM$ into an equivalent MDP $\cN=\spider{\cM}{\cB}$ 
  with the same state space and 
  where $\cB$ has been flattened.
  Recall that $\cN$ contains fewer state-action
  pairs than $\cM$ (see (S1) in Lemma \ref{spider-construction}).
  By property (S4) stated in Lemma \ref{spider-construction}
  (see also Lemma \ref{spider-construction-S4})
  we get that $\ZeroEC$ equals the union of the state space of
  $\cB$ and the set of states belonging to some 0-EC of $\cN$,
  and the analogous statement for the actions that are contained
  in some 0-EC of $\cM$ resp.~$\cN$.
  (Of course, the extra $\tau$-actions of $\cN$ from all states
  $s$ in $\cB \setminus \{s_0\}$ to $s_0$ have to be ignored.)
  Hence, we can repeat the same procedure to the new MDP $\cN$.
  In this way we encounter all states and actions that belong to a 0-EC.
  The number of iterations is bounded by the number of state-action pairs
  that belong to some 0-EC of $\cM$. Thus, due to Lemma~\ref{size-iterative-spider-construction},
  the time complexity is polynomial in the size of $\cM$.
\end{proof}

\subsubsection{Recurrence Values in Maximal 0-ECs}

\label{sec:min-credit}

Given a strongly connected MDP $\cM$ with $\ExpRew{\max}{\cM}(\MP) =0$
and a state $s$ that belongs to some maximal 0-EC $\cZ$ of $\cM$, the
long-run weight and the recurrence value 
of state $s$ are defined by:
$$
\begin{array}{lcr@{\hspace*{0.1cm}}l}
  \lgr(s) & =  & 
  \max & 
  \bigl\{ \, 
      w \in \Integer \, : \, 
      \exists \sched. \ 
         \Pr^{\sched}_{\cZ,s}
           \big(\, \Diamond \Box (\accwgt  \geqslant w) \, \big) =1 \, 
  \bigr\}
  \\[2ex]

  \rec(s)
  & =  &
  \max & 
  \bigl\{ \, 
      w \in \Integer \, : \, 
      \exists \sched. \ 
      \Pr^{\sched}_{\cZ,s}
           \big(\, \Box (\accwgt  \geqslant w) \, \wedge \, \Box \Diamond s \, \big)
       =1 \, 
  \bigr\}
\end{array}
$$
where the existential quantifier
$\exists \sched$ ranges over the schedulers for $\cZ$ (rather than $\cM$).

Whenever $s$ and $t$ are states that belong to the same maximal 0-EC
then $\lgr(s)=w(s,t)+\lgr(t)$.
On the other hand,
$\rec(s) \leqslant 0$ for all states $s$ in a 0-EC
and $\rec(s)\not= w(s,t)+\rec(t)$ is possible if $s$ and $t$ belong to the
same maximal 0-EC.

\begin{example}
\label{example:values-in-ZeroEC}
{\rm 
Consider the MDP $\cM$ 
with the following deterministic (\ie with probability 1) transitions:
$$
  s \stackrel{3}{\longrightarrow} t
    \stackrel{-2}{\longrightarrow} u
    \stackrel{-1}{\longrightarrow} s
 \qquad \text{and} \qquad
 v \stackrel{4}{\longrightarrow} s
   \stackrel{-4}{\longrightarrow} v
$$
For simplicity, we dropped here the action names and attached the weights to 
the transitions. Then, $\cM$ constitutes a maximal 0-EC with
$\rec(s)=\lgr(s)=0$, $\rec(t)=\lgr(t)=-3$, $\rec(u)=\lgr(u)=-1$, while
$\rec(v)=0$ and $\lgr(v)=4$.
\Ende
 }
\end{example}

\begin{lemma}[Second part of Lemma~\ref{mincredit-ZeroEC}]
\label{appendix:values-in-ZeroEC}
If $\cM$ is strongly connected and $\ExpRew{\max}{\cM}(\MP)=0$ 
then the long-run weights $\lgr(s)$ and the recurrence values $\rec(s)$
for the states $s \in \ZeroEC$
are computable in polynomial time.
\end{lemma}

\begin{proof}
Let $\cZ$ be a maximal 0-EC of $\cM$. As before, $w(s,t) = \wgt(\fpath)$
for some/each path from $s$ to $t$ in $\cZ$. 
Then, $w(s,s)=0$ and $w(s,t)= w(s,u)+w(u,t)$
for all states $s,t,u\in \cE$.
Obviously, the values $w(s,t)$ for all states $s,t \in \cZ$, 
can be obtained in polynomial time,
\eg using a BFS or DFS from every state.
Let 
$$
  W \ \ = \ \ 
  \bigl\{ \, w(s,t) \, : \, \text{$s,t$ are states in $\cZ$}\, \bigr\} 
   \ \ \subseteq \ \ \Integer\enskip.
$$
Note that $W$ contains at most
$n^2$ elements when $n$ is the number of states in $\cZ$.
Moreover, the absolute value of the elements in $W$ is bounded by
$(n{-}1) \cdot \max_{s,\alpha} |\wgt(s,\alpha)|$
where $(s,\alpha)$ ranges over all state-action pairs in $\cZ$.
Thus, the logarithmic lengths of the elements in 
$W$ is polynomially bounded in the size of $\cZ$.
Recall that the 
size of an MDP is defined as the number of states plus the total sum of
the logarithmic lengths of its transition probabilities and weights.

Let $s$ be a state in $\cZ$. 
Let $w_1,w_2,\ldots,w_k$ an enumeration
of the elements in $\{ w(s,t) : t\in \cZ \}$ such that 
$w_1 > w_2 > \ldots > w_k$. 
For $j=1,2,\ldots,k$, let 
$$
  T_{s,j} \ \ = \ \ 
  \bigl\{\, t \, : \ \text{$t$ is a state in $\cZ$ with 
                           $w(s,t) \geqslant w_j$} \, \bigr\}\enskip.
$$
For each state $t\in T_{s,j}$, let 
$$
  \Act_{s,j}(t) \ \ = \ \ 
  \bigl\{\ 
     \beta \in \Act_{\cZ}(t) \ : \ \Post(t,\beta)\subseteq T_{s,j} \ 
  \bigr\}
$$
where $\Post(t,\beta)$ denotes the set of states $u$ with $P(t,\beta,u)>0$.
Consider the sub-MDP $\cZ_{s,j}$ consisting of all
state-action pairs $(t,\beta) \in \cZ$ with $t\in T_{s,j}$
and $\beta \in \Act_{s,j}(t)$.
Let $\textit{MEC}_{s,j}$ be the set of states in 
$\cZ_{s,j}$ that are contained
in some maximal end component of $\cZ_{s,j}$.

Some comments are in order. 
First, as $w(s,s)=0$ we have $s \notin T_{s,j}$ iff $w_j$ is positive. 
In particular, 
$s \in \textit{MEC}_{s,j}$ is only possible if $w_j \leqslant 0$.
Second, for each path $\fpath$ in $\cZ$ with $\first(\fpath)=s$ and
$\last(\fpath)=t\in T_{s,j}$ we have $\wgt(\fpath)=w(s,t)\geqslant w_j$.
Third, the sets $\Act_{s,j}(t)$ can be empty, in which case $t$ is a trap
in $\cZ_{s,j}$.

Let $j$ be the smallest index in $\{1,\ldots,k\}$
such that $\Pr^{\max}_{\cZ,s}(\Diamond \textit{MEC}_{s,j})=1$.
Note that $\cZ_{s,k}=\cZ$ and $s\in \MEC_{s,k}$, and therefore
$\Pr^{\max}_{\cZ,s}(\Diamond \textit{MEC}_{s,k})=1$.
This ensures the existence of such an index $j$.
We now show that $\lgr(s)=w_j$.
\begin{itemize}
\item
To prove $\lgr(s) \leqslant w_j$, 
we pick a scheduler $\sched$ for $\cZ$ with
$\Pr^{\sched}_{\cZ,s}\big(\Diamond \Box (\accwgt \geqslant \lgr(s))\big)=1$.
Let $m$ be the largest value such that $w_m \geqslant \lgr(s)$.
Then,
$\Pr^{\sched}_{\cZ,s}(\Diamond \textit{MEC}_{s,m})=1$.
But then $j \leqslant m$, and therefore $w_j \geqslant w_m \geqslant \lgr(s)$.
\item
To see why $\lgr(s) \geqslant w_j$,
we pick an MD-scheduler $\sched$ 
for $\cZ$ such that
$\Pr^{\sched}_{\cZ,s}(\Diamond \textit{MEC}_{s,j})=1$.
But then $\wgt(\fpath)=w(s,t)\geqslant w_j$ for all
$\sched$-paths $\fpath$ from $s$ to a state $t$ that belongs
to a BSCC of $\sched$.
Hence, $\Pr^{\sched}_{\cZ,s}(\Diamond \Box (\accwgt \geqslant w_j))=1$.
By the definition of $\lgr(s)$, we get $\lgr(s) \geqslant w_j$.
\end{itemize}
With similar arguments, we get
$\rec(s)=w_i$ where $i$ is the smallest index such that $s \in \MEC_{s,i}$.
Clearly, these indices $i$ and $j$ can be computed in polynomial time using
standard algorithms to compute the maximal end components in MDPs
and maximal reachability probabilities.
Note that $k \leqslant |W|\leqslant n^2$ where $n$ is the
number of states in $\cZ$.
\end{proof}

\begin{example}
\label{example:mincredit-2}
{\rm 
  Let us revisit the MDP of Example \ref{example:values-in-ZeroEC}.
  The values $w(x,y)$ are as follows:
  $$
    \begin{array}{c|rrrr}
        & s & t & u & v 
        \\
        \hline
        s & 0 & \phantom{-}3 & 1 & -4 \\
        t & -3 & 0 & -2 & -7 \\
        u & -1 & 2 & 0 & -5 \\
        v & 4 & 7 & 5 & 0
    \end{array}
  $$
  Hence, $W = \{7,5,4,3,2,1,0,-1,-2,-3,-4,-5,-7\}$.

  For instance, for state $s$, we consider the values
  in the row for $s$:
  $w_1 = 3$, $w_2=1$, $w_3=0$ and $w_4=-4$, 
  and look for the smallest index $j\in \{1,\ldots,4\}$ such that
  $\Pr^{max}_{\cZ,s}(\Diamond \textit{MEC}_{s,j})=1$.
  The MDP $\cZ_{s,1}$ consists of state $t$, which is a trap in $\cZ_{s,1}$.
  Hence, $\textit{MEC}_{s,1}=\varnothing$.
  The MDP $\cZ_{s,2}$ consists of the transition
  $t \to u$ and $u$ is a trap in $\cZ_{s,2}$.
  Again, we have $\textit{MEC}_{s,2}=\varnothing$.
  The MDP $\cZ_{s,3}$ consists of the cycle
  $s\to t \to u \to s$. Hence, 
  $\textit{MEC}_{s,2}=\{s,t,u\}$. This yields
  $\lgr(s)=\rec(s)=w_3=0$.

  Let us look now for state $v$.
  Here, we deal with the values $w_1=7$, $w_2=5$, $w_3=4$, 
  and $w_4=0$.
  We have $\textit{MEC}_{v,1}=\varnothing$ as the 
  MDP $\cZ_{v,1}$ consists of the trap state $t$.
  The MDP $\cZ_{v,2}$ consists of states $t$ and $u$. 
  Both are traps in $\cZ_{v,2}$. Hence, $\textit{MEC}_{v,2}=\varnothing$.
  The MDP $\cZ_{v,3}$ consists of the cycle $s \to t \to u \to v$ plus
  the trap states $v$. Hence,
  $\textit{MEC}_{v,3}=\{s,t,u\}$ and
  $\lgr(v)=w_3=4$.
  The MDP $\cZ_{v,4}$ equals $\cM$. Therefore,
  $v\in \MEC_{v,4}=\{s,t,u,v\}$, which yields
  $\rec(v)=w_4=0$.
\Ende
  }
\end{example}

The proof of Lemma \ref{appendix:values-in-ZeroEC} 
also shows the existence of
MD-schedulers that achieve the long-run weights and recurrence values.

\begin{corollary}
  For each maximal 0-EC $\cZ$
  there exist MD-schedulers $\sched$ and $\tsched$ such that
  $\Pr^{\sched}_{\cZ,s}(\Diamond \Box \big(\accwgt \geqslant \lgr(s))\big)=1$
  and
  $\Pr^{\tsched}_{\cZ,s}
  \big(\Box (\accwgt \geqslant \rec(s)) \wedge \Box \Diamond s\big)=1$
  for each state $s$ in $\cZ$.
\end{corollary}

\begin{remark}[Minimal credits in energy-MDPs]
\label{remark:minimal-credit}
{\rm
Let $\cM$ be an MDP, $F$ a set of states in $\cM$ and $s$ a state in $\cM$.
The value
$$
  \mincredit_{\cM}(s,F)
  \ \ = \ \ 
  \min \, \bigl\{
    w \in \Integer \ : \ 
      \exists \sched. \ 
      \Pr^{\sched}_{\cZ,s}
           \big(\, \Box (\accwgt {+}w \geqslant 0) \, \wedge \, 
            \Box \Diamond F \, \big)
       =1 \, 
  \bigr\}
$$    
is the minimal initial weight budget required in state $s$ to ensure that
the accumulated weight is always nonnegative and that
the B\"uchi condition $\Box \Diamond F$ holds almost surely
(under some scheduler).
Following the literature on energy-MDPs with B\"uchi objectives,
this value is called the minimal credit for state $s$ in $\cM$.
It relates to the recurrence values as follows.
If $\cZ$ is a maximal 0-EC of $\cM$ and $s$ a state in $\cZ$ then
$$
  \rec(s) \ \ = \ \ -\mincredit_{\cZ}(s,\{s\})\enskip.
$$
Pseudo-polynomial algorithms for computing the minimal credits
for all states in an arbitrary MDPs $\cM$ are known from the
literature on energy games and energy-MDPs with B\"uchi objectives
\cite{ChatDoy11} and \citeapx{ChatDoy12,MaySchTozWoj17}.
If $\cM$ is a strongly connected MDP with
$\wgt(\cycle)=0$ for all cycles in $\cM$
then $\mincredit_{\cM}(s,F)$ can be computed in polynomial time
by an algorithm similar to the one for the long-run weights
and recurrence values.
Using the notations introduced in the proof of 
Lemma \ref{appendix:values-in-ZeroEC},
$\mincredit_{\cM}(s,F)=-w_j$ where $j$ is the largest index
such that $\Pr^{\max}_{\cZ_{s,j},s}(\Diamond \MEC^F_{s,j})=1$
where $\MEC^F_{s,j}$ denotes the set of all states in $\cZ_{s,j}$
that belong to a maximal end component $\cE$ of $\cZ_{s,j}$
where  $\cE$ contains at least one state in $F$.

Thus, our results show that these algorithms can be 
improved for MDPs that constitute a 0-EC
as then the values $\mincredit_{\cM}(s,F)$
are computable in polynomial time.
This is not an interesting instance of energy-MDPs.
However, the efficient computability of the 
recurrence values of states belonging to 0-ECs
will be crucial for the 
weight-bounded repeated reachability problems
(see Section \ref{sec:appendix-buechi}).
\Ende
  }
\end{remark}

\section{Proofs of Section~\ref{sec:SSP}}
\setcounter{theorem}{0}
\label{sec:appendix-min-exp-wgt}
We will interpret the assumption $\mathsf{(BT)}$ as follows:
if $\Pr^{\sched}_{\cM,s}(\Diamond \goal) <1$ then
the set of pumping $\sched$-paths from state $s$ has positive measure
under $\sched$.
If $\mathsf{(BT)}$ holds then
$\Exp{\inf}{\M,\sinit}(\accdiaplus \goal)$ is achieved by an
MD-scheduler and such an MD-scheduler is computable in polynomial time
using linear-programming techniques~\cite{BerTsi91}.

\begin{example}[Incompleteness of condition $\mathsf{(BT)}$]
Consider the following MDP
\begin{center}
\begin{tikzpicture}
	\node[state] (s) {$s$};
	\node[state, right=15mm of s] (t) {$t$};
	\node[state,left=15mm of s] (goal) {$\goal$};
	\path%
        (s) edge[ptran,bend right] node[below]{$\alpha$/\posw{1}} (t)
        (t) edge[ptran,bend right] node[above]{$\beta$/\negw{-1}} (s)
        (s) edge[ptran] node[above]{$\gamma$/\nilw} (goal);
\end{tikzpicture}
\end{center}
Then, $\Expected^{\sched}_{\M,s}(\accdiaplus \goal)=0$,
$\Expected^{\sched}_{\M,t}(\accdiaplus \goal)=-1$ for each proper
scheduler $\sched$. 
However, $\cM$ violates condition 
$\mathsf{(BT)}$ as the expected total weight of
the improper MD-scheduler $\sched$ with $\sched(s)=\alpha$ is not $+\infty$.
\Ende
\end{example}

\ssp*

The implication ``$\Longrightarrow$'' is shown in 
Lemma \ref{fin-exp-acc-wgt-implies-no-wgt-div-EC},
and the last statement in
Lemma \ref{BT-hold-if-no-0-EC}.
The proof of the implication ``$\Longleftarrow$'' in
Lemma \ref{th:finiteness-min-exp-accwgt} will be presented
afterwards together with an explanation how the iterative
application of the spider construction to flatten 0-ECs
can be used to generate a new MDP that satisfies condition $\mathsf{(BT)}$
and has the same minimal expected accumulated weight.

Intuitively, if there are negatively weight-divergent components, then
one can define a family of schedulers that decrease the weight arbitrarily low
before moving to the goal state, hence showing that the minimal expectation is $-\infty$.
We prove this formally in the following lemma.
\begin{lemma}%
  [Implication ``$\Longrightarrow$'' of
              Lemma~\ref{th:finiteness-min-exp-accwgt}]
 \label{fin-exp-acc-wgt-implies-no-wgt-div-EC}
  Let $\cM$ be as in Lemma~\ref{th:finiteness-min-exp-accwgt}.
  If $\cM$ has negatively weight-divergent end components then
   $\Expected^{\inf}_{\M,\sinit}(\accdiaplus \goal)=-\infty$.
\end{lemma}
\begin{proof}
Assume thus that $\cM$ has a negatively weight-divergent end component, and let us show that $\ExpRew{\inf}{\cM,\sinit}(\accdiaplus \goal) = -\infty$ by exhibiting a sequence $(\sched_R)_{R \in \Nat}$ of proper schedulers for $\cM$ with $\inf_{R \in \Nat} \ExpRew{\sched_R}{\cM,\sinit}(\accdiaplus \goal) = -\infty$.

Let $\cE$ be a negatively weight-divergent end component of $\M$. To define $\sched_R$, we combine three natural schedulers. Let first $\sched$ be an MD-scheduler for $\cM$ such that $\sched$ is proper, \ie for every $s \in S$, $\Pr^{\sched}_{\cM,s}(\Diamond \goal)=1$. Let then $\tsched$ be an MD-scheduler that reaches $\cE$ with positive probability and agrees with $\sched$ in states from which $\cE$ is no longer reachable:   $\Pr^{\tsched}_{\cM,\sinit}(\Diamond \cE) > 0$, and  $\tsched(t)=\sched(t)$ for each state $t \in T$ where
$T = \{t\in S: \Pr^{\tsched}_{\cM,t}(\Diamond \cE) = 0\}$. Therefore, for every $t \in T$, $\Pr^{\sched}_{\cM,t}(\Diamond \goal)=1$. Last, let $\vsched$ be a negatively weight-divergent scheduler for $\cE$.

From $\sched$, $\tsched$, $\vsched$ and $R \in \Nat$, we define $\sched_R$ as follows. Initially, $\sched_R$ mimics $\tsched$ until reaching a state, say $u$, belonging to $\cE$ or $T$.
\begin{itemize}
\item If $u \in \cE$, and the accumulated weight is $r$, then $\sched_R$ switches mode and  simulates $\vsched$ until the accumulated weight is at most $-R$ (this happens almost surely because $\cE$ is negatively weight-divergent); then $\sched_R$ again switches mode and behaves as $\sched$ until reaching $\goal$.
\item If $u \in T$, then $\sched_R$ switches mode and behaves as $\sched$ until reaching $\goal$.
\end{itemize}

We claim that $\inf_{R \in \Nat} \ExpRew{\sched_R}{\cM,\sinit}(\accdiaplus \goal) = -\infty$. To prove it, we provide an upper bound to $\ExpRew{\sched_R}{\cM,\sinit}(\accdiaplus \goal)$ for fixed $R \in \Nat$. For each state $u \in S$, we define $p_u$ as the probability under $\tsched$ to reach $u$ before traversing $\cE$.  Let $E = \max_{s \in \cE} \ \ExpRew{\sched}{\cM,s}(\accdiaplus \goal)$ be the maximum of the expected accumulated weights until reaching $\goal$, taken over all paths starting in a state of $\cE$. Then, 
the expected accumulated weight until reaching $\goal$ under $\sched_R$ is   
\begin{eqnarray*}
     \ExpRew{\sched_R}{\cM,\sinit}(\accdiaplus \goal)
     & \ \leqslant \ &
    \sum_{s\in \cE} p_s \cdot (-R + E)
     \ + \ 
     \sum_{t\in T} p_t \cdot \ExpRew{\sched}{\cM,t}(\accdiaplus \goal)
     \\
     & \ = \ &
      \Pr^{\tsched}_{\cM,\sinit}(\Diamond \cE) \cdot (-R  + E ) \ \ + \ \ 
     \sum_{t\in T} p_u \cdot 
          \ExpRew{\sched}{\cM,t}(\accdiaplus \goal)\enskip.
   \end{eqnarray*}
From $\Pr^{\tsched}_{\cM,\sinit}(\Diamond \cE) > 0$, 
and the fact that $E$ and the sum over $T$ are constants, 
we derive the desired limit: 
$$
  \inf_{R \in \Nat} \ExpRew{\sched_R}{\cM,\sinit}(\accdiaplus \goal) 
  \ \ =  \ \ -\infty \enskip.
$$
This completes the proof of Lemma \ref{fin-exp-acc-wgt-implies-no-wgt-div-EC}.
\end{proof}

\begin{lemma}[Last statement of 
              Lemma \ref{th:finiteness-min-exp-accwgt}]
\label{BT-hold-if-no-0-EC}
  If $\cM$ has no negatively weight-divergent end component
  then condition $\mathsf{(BT)}$ holds if and only if
  $\cM$ has no 0-ECs.
\end{lemma}

\begin{proof}
  ``$\Longrightarrow$'' is trivial
   as the expected total weight of each scheduler that realizes a 0-EC
   is bounded.
  To prove ``$\Longleftarrow$'' we suppose that
  $\cM$ has no negatively weight-divergent end component and no 0-EC.
  But then $\Exp{\inf}{\cE}(\MP)>0$ for each end component $\cE$ of
  $\cM$.   
  By Lemma \ref{universally-pumping}, all end components of $\cM$ are
  universally pumping.
  Hence, condition $\mathsf{(BT)}$ holds.
\end{proof}

To complete the proof of Lemma~\ref{th:finiteness-min-exp-accwgt}
we need to show the finiteness of
$\Expected^{\inf}_{\M,\sinit}(\accdiaplus \goal)$ if
$\cM$ has no negatively weight-divergent end components.

For a better fit of the notations used in 
the previous section, we switch here to the
maximal expected accumulated weight until reaching the goal state:
\[
  \Exp{\max}{\cM,s}(\accdiaplus \goal)
  \ \ = \ \ 
  \sup \ 
  \bigl\{
      \Exp{\sched}{\cM,s}(\accdiaplus \goal) \ : \ 
      \text{$\sched$ is a proper scheduler} \ 
  \bigr\}.
\]
The aim is to show that
$\Expected^{\max}_{\M,\sinit}(\accdiaplus \goal)$ is finite if
$\cM$ has no positively weight-divergent end components.
The corresponding result for the minimal expected 
accumulated weight until reaching the goal state
is then obtained by multiplying
all weights by $-1$.

Rephrased for maximal expected weights, 
the assumption of Bertsekas and Tsitsiklis \cite{BerTsi91} 
asserts that
the expected total weight of each improper scheduler is $-\infty$.
More precisely:
\begin{equation}
   \label{BTmax}
   \text{If $\Pr^{\sched}_{\cM,s}(\Diamond \goal) <1$ 
   then 
   $\Pr^{\sched}_{\cM,s}\{\infpath \in \InfPaths :
          \text{$\infpath$ is negatively pumping} \} >0$.}
   \tag{$\mathsf{BTmax}$}
\end{equation}

\begin{lemma}[\cite{BerTsi91}]
\label{BerTsi-max-exp}
  Under assumption \eqref{BTmax}, $\Exp{\max}{\cM,s}(\accdiaplus \goal)$ is finite,
  and 
  there is an MD-scheduler $\sched$
  with $\Exp{\max}{\cM,s}(\accdiaplus \goal)=
        \Exp{\sched}{\cM,s}(\accdiaplus \goal)$ for all states~$s$.
\end{lemma}

Rephrased for maximal expectations, 
Lemma \ref{fin-exp-acc-wgt-implies-no-wgt-div-EC} states that
$\Expected^{\sup}_{\M,\sinit}(\accdiaplus \goal)=+\infty$
if $\cM$ has positively weight-divergent end components.
Likewise, Lemma \ref{BT-hold-if-no-0-EC} yields that if
$\cM$ has no positively weight-divergent end component then
\eqref{BTmax} is equivalent to the nonexistence of 0-ECs.

\begin{lemma}
 \label{Expmax-BT}
   If $\cM$ is an MDP with $\Exp{\max}{\cE}(\MP) <0$ 
   for each end component $\cE$ of $\cM$
   then $\cM$ satisfies condition \eqref{BTmax}.
\end{lemma}

\begin{proof}
 If $\Pr^{\sched}_{\cM,s}(\Diamond \goal) <1$ then
 there exists an end component $\cE$ of $\cM$ such that
 $\Pr^{\sched}_{\cM,s}\{ \infpath \in \InfPaths : \lim(\infpath)=\cE \}>0$.
 (Recall that $\lim(\infpath)$ denotes the set of state-action pairs
  that are taken infinitely often in $\infpath$.)
 As $\Exp{\max}{\cE}(\MP)<0$, all schedulers for $\cE$ are
 negatively pumping.
 Hence,
 the set of negatively pumping $\sched$-paths starting in $s$ 
 has positive measure. 
\end{proof}

\tudparagraph{1ex}{Construction of the new MDP $\cN$.}
  Suppose now that $\cM$ is 
  an MDP that has no positively weight-divergent end component.
  We now generate from $\cM$ a new MDP $\cN$ 
  with the same state space
  that satisfies \eqref{BTmax}
  and is equivalent to $\cM$ with respect to the 
  maximal expected accumulated weight until reaching the goal state.
  For this, we apply the iterative spider construction to $\cM$
  of Section \ref{sec:iterative-application-spider}.
  The resulting MDP $\cN$ has the following properties:

\begin{lemma}
\label{Expmax-spider}
  Let $\cM$ be an MDP that has no weight-divergent end component,
  and let $\cN$ be the MDP resulting from $\cM$ by flattening the 0-ECs
  using the iterative spider construction
  of Section \ref{sec:iterative-application-spider}. Then:
  \begin{enumerate}
  \item [(E1)]
     $\|\cN\|\leqslant \|\cM\|$ and the size of $\cN$ is polynomially
     bounded by the size of $\cM$.
  \item [(E2)]
     $\cN$ satisfies condition \eqref{BTmax}. 
  \item [(E3)]
     For proper scheduler $\tsched$ for $\cN$ there is a 
     proper scheduler
     $\sched$ for $\cM$ with
     $\Exp{\sched}{\cM,s}(\accdiaplus \goal) =
      \Exp{\tsched}{\cN,s}(\accdiaplus \goal)$
     for all states $s$.
     If $\tsched$ is an MD-scheduler, then $\sched$ can be chosen as an
     MD-scheduler.
  \item [(E4)]
     For each proper scheduler $\sched$ for $\cM$ there is a 
     proper scheduler $\tsched$ for $\cN$ with
     $\Exp{\sched}{\cM,s}(\accdiaplus \goal) =
      \Exp{\tsched}{\cN,s}(\accdiaplus \goal)$
     for all states $s$.
  \end{enumerate} 
\end{lemma}

\begin{proof}
  Statement (E1) follows from property (S1) of Lemma \ref{spider-construction}
  and Lemma \ref{size-iterative-spider-construction}.
  Statement (E2) is a consequence of
  (W2) in Theorem \ref{appendix:wgt-div-algo} which yields that
  $\Exp{\max}{\cE}(\MP) <0$ 
  for all end components $\cE$ of $\cN$.
  By Lemma \ref{Expmax-BT}, we obtain that $\cN$ satisfies \eqref{BTmax}. 

  We now turn to the proof of the scheduler transformations as stated
  in (E3) and (E4).
  For this, we can rely on the equivalence of $\cM$ and $\cN$
  with respect to the class of all 0-EC-invariant properties
  as stated in Lemma \ref{lemma:equivalence-M-and-Mi}.
  To apply Lemma \ref{lemma:equivalence-M-and-Mi} we use the following
  facts:
  \begin{itemize}
  \item
    Whenever $\sched$ is a proper scheduler for $\cM$ then
    $\Pr^{\sched}_{\cM,s}(\Limit{\cE})=0$ for each end component
    $\cE$ of $\cM$. 
    (As before, $\Limit{\cE}=\{\infpath \in \InfPaths : \lim(\infpath)=\cE\}$.)
  \item 
    For each $K\in \Integer$, the property
    $\psi_K=\Diamond (\goal \wedge (\accwgt=K))$ is measurable and
    0-EC-invariant.
    (Recall that $\goal$ is a trap. Therefore, there is no end component
     containing $\goal$.)
  \item
    For the proper schedulers $\sched$ for $\cM$ resp.~proper schedulers
    $\tsched$ for $\cN$ we have:
    $$
      \Exp{\sched}{\cM,s}(\accdiaplus \goal) \ \ = \ \ 
      \sum_{K=-\infty}^{+\infty} 
         K \cdot \Pr^{\sched}_{\cM,s}(\psi_K)
      \hspace*{1.5cm}
      \Exp{\tsched}{\cN,s}(\accdiaplus \goal) \ \ = \ \ 
      \sum_{K=-\infty}^{+\infty} 
         K \cdot \Pr^{\tsched}_{\cN,s}(\psi_K)\enskip.
    $$
  \end{itemize}
  Thus, statement (E3) follows directly from
  part (a) of Lemma \ref{lemma:equivalence-M-and-Mi}, 
  while (E4) follows from
  part (b) of Lemma \ref{lemma:equivalence-M-and-Mi}.
\end{proof}

By (E2) in Lemma \ref{Expmax-spider}
and Lemma \ref{BerTsi-max-exp}, 
$\Exp{\max}{\cN,s}(\accdiaplus \goal)$ is finite for all
states $s$ and $\cN$ has a proper MD-scheduler that maximizes the expected
accumulated weight from each state.
Moreover, we have:

\begin{lemma}
\label{Exp-agrees-spider}
  Let $\cM$ and $\cN$ be as in Lemma \ref{Expmax-spider}.
  We have
  $\Exp{\max}{\cN,s}(\accdiaplus \goal) =
   \Exp{\max}{\cM,s}(\accdiaplus \goal)$
  for each state $s$, and
  $\cM$ has an MD-scheduler $\sched$ with
  $\Exp{\sched}{\cM,s}(\accdiaplus \goal) =
   \Exp{\max}{\cM,s}(\accdiaplus \goal)$.
\end{lemma}
\begin{proof}
  Follows  from  Lemma~\ref{BerTsi-max-exp}
  and (E3) and (E4) in Lemma~\ref{Expmax-spider}.
\end{proof}

\begin{corollary}
  [Implication ``$\Longleftarrow$'' of
              Lemma~\ref{th:finiteness-min-exp-accwgt}
   for maximal expectations]
  If $\cM$ has no positively weight-divergent end component
  then $\Exp{\max}{\cN,s}(\accdiaplus \goal)$ is finite for all
  states $s$.
  Moreover, one can construct in polynomial time an MDP $\cN$
  with the same state space such that 
  $\cN$ satisfies condition \eqref{BTmax} and
  $\Exp{\max}{\cM,s}(\accdiaplus \goal)=
   \Exp{\max}{\cN,s}(\accdiaplus \goal)$ 
  for all states $s$.
\end{corollary}

 By multiplying all weights by $-1$, we obtain that
 $\Exp{\inf}{\cN,s}(\accdiaplus \goal)$ is finite for all
 states $s$ in an MDP $\cM$ without negatively weight-divergent end components 
 (assuming the existence of proper schedulers).
This completes the proof of Lemma \ref{th:finiteness-min-exp-accwgt}.

\section{Proofs of Section~\ref{sec:DWR}}
\setcounter{theorem}{0}

\label{sec:appendix-eventually}

We prove here the statements of Section~\ref{sec:DWR}.
Our results for weight-bounded reachability and
B\"uchi  constraints are summarized in 
Figure \ref{table:summary} where $\cM=(S,\Act,P,\wgt)$ is an MDP, $\sinit$
a state of $\cM$, $T$ and $F$ are set of states in $\cM$.
Moreover, $K\in \Integer$ and $K_t\in \Integer \cup \{-\infty\}$ for
$t\in T$.

\begin{figure}[ht]
\begin{center}
   \begin{tabular}{l|l||l|l}
    \multicolumn{2}{c||}{solvable in polynomial time} 
    & \multicolumn{2}{c}{%
         \begin{tabular}{c}
            in $\textrm{NP}\cap\textrm{coNP}$, 
            solvable in pseudo-polynomial time \\
            hard for non-stochastic two-player mean-payoff games \\[0.8ex]
         \end{tabular}}
    \\
    \hline
    &  \\[-2ex]
 
    \Eposdwr{} &
    $\exists \sched. \
      \Pr^{\sched}_{\cM,\sinit}
       \bigl(\, 
                \bigvee_{t\in T} 
                   \Diamond (t \wedge (\accwgt \geqslant K_t))
       \bigr) \, > \, 0$ ?
    &
    \Easdwr{} &
    $\exists \sched. \
      \Pr^{\sched}_{\cM,\sinit}
       \bigl(\, 
                \bigvee_{t\in T} 
                   \Diamond (t \wedge (\accwgt \geqslant K_t))
       \bigr) \, = \, 1$ ?
    \\[1ex]

    \Uasdwr{} &
    $\forall \sched. \
      \Pr^{\sched}_{\cM,\sinit}
       \bigl(\, 
                \bigvee_{t\in T} 
                   \Diamond (t \wedge (\accwgt \geqslant K_t))
       \bigr) \, = \, 1$ ?
    &
    \Uposdwr{} &
    $\forall \sched. \
      \Pr^{\sched}_{\cM,\sinit}
       \bigl(\, 
                \bigvee_{t\in T} 
                   \Diamond (t \wedge (\accwgt \geqslant K_t))
       \bigr) \, = \, 1$ ?
    \\[1ex]

    \EposwB &
    $\exists \sched. \
      \Pr^{\sched}_{\cM,\sinit}
       \bigl(\, \Box \Diamond (\accwgt \geqslant K)
             \, \wedge \, \Box \Diamond F \, 
       \bigr) \, > \, 0$ ?
    &
    \EaswB &
    $\exists \sched. \
      \Pr^{\sched}_{\cM,\sinit}
       \bigl(\, \Box \Diamond (\accwgt \geqslant K)
             \, \wedge \, \Box \Diamond F \, 
       \bigr) \, = \, 1$ ?
    \\[1ex]

    \UaswB &
    $\forall \sched. \
      \Pr^{\sched}_{\cM,\sinit}
       \bigl(\, \Box \Diamond (\accwgt \geqslant K)
             \, \wedge \, \Box \Diamond F \, 
       \bigr) \, = \, 1$ ?
    &
    \UposwB &
    $\forall \sched. \
      \Pr^{\sched}_{\cM,\sinit}
       \bigl(\, \Box \Diamond (\accwgt \geqslant K)
             \, \wedge \, \Box \Diamond F \, 
       \bigr) \, > \, 0$ ?
 \end{tabular}
\end{center}
\caption{Results for weight-bounded reachability and B\"uchi constraints}
 \label{table:summary}
\end{figure}

We start with
the four cases of disjunctive weight-bounded reachability properties
(Sections \ref{Eposdwr}, \ref{Uposdwr}, \ref{Uasdwr} and \ref{Easdwr})
and then consider weight-bounded B\"uchi properties in
Section \ref{sec:appendix-buechi}.
For weight constraints not to be trivial, we safely assume that
$T\setminus T^*$ is nonempty where $T^* = \{t\in T : K_t =-\infty\}$.

\subsection{Positive Reachability Under Some Scheduler}

\label{Eposdwr}

\begin{theorem}
\label{thm:Eposdwr}
Problem \Eposdwr{} belongs to \PTIME\ and the value $\valueEpos{\cM,s}$ 
is computable in polynomial time.
\end{theorem}

\begin{proof}
  The existence of a scheduler satisfying a \dwr-property $\varphi$
  with positive probability is equivalent to the existence of a path
  from the initial state $s$ to one of the targets $t \in T$ with
  accumulated weight at least $K_t$.  To decide the latter in
  polynomial time, one can rely on shortest-path algorithms for
  weighted graphs, such as the Bellman-Ford algorithm. More precisely,
  consider the weighted graph obtained from $\cM$ by ignoring action
  names and probabilities, and switching the weight function from
  $\wgt$ to $-\wgt$. Then $\valueEpos{\cM,s}$ is the weight of a
  shortest path from $s$ to $\goal$. To decide \Eposdwr, we apply a
  shortest-path algorithm for each $t \in T$ with source $s$ and
  target $t$, and compare the obtained value with $K_t$.
\end{proof}

\subsection{Positive Reachability Under All schedulers}

\label{Uposdwr}
\label{sec:Uposdwr-incoNP}

\begin{lemma}
\label{lem:Uposdwr-part-case}
The problem \Uposdwr{} for general MDPs can be reduced in polynomial
time to the case where $\cM$ has a single goal state which is a trap.
\end{lemma}
\begin{proof}
  Let $\varphi(T,(K_t)_{t \in T})$ denote the weight-bounded reachability constraint
  $\bigvee_{t\in T} \Diamond (t \wedge (\accwgt \geqslant K_t))$ and
  let us first show that one can assume that~$K_t > -\infty$ for all~$t \in T$.
  In fact, for all states~$t \in T^* = \{ t \mid K_t =-\infty\}$ (which can be assumed to be trap states), one can 
  add one action~$\alpha$ with~$P(t,\alpha,t) = \frac{1}{2}$ and~$P(t,\alpha,t') = \frac{1}{2}$ for some arbitrary
  $t' \in T\setminus T^*$ with~$\wgt(t,\alpha) = +1$.
  Then, for any scheduler, if state~$t$ is reached, then with positive probability so will be~$t'$ with weight at least~$K_{t'}$.
  The converse is also true: 
  if~$\varphi(T\setminus T^*, (K_t)_{t \in T\setminus T^*})$ 
  holds with positive probability for all schedulers for the new MDP,
  so does~$\varphi(T,(K_t)_{t \in T})$ on the original one.

  Let us now show how to make sure all target 
  states are trap states by defining~$\cM'$.
  For each non-trap~$t \in T$, we add a new trap goal state~$g_t$ in~$\cM'$.
  Moreover, for each state-action pair $(t,\alpha)$ in $\cM$,
  the new MDP $\cM'$ has a fresh state $g_{t,\alpha}$.
  Let $G$ denote the set consisting 
  of all states $t\in T$ that are traps in $\cM$
  and all states $g_t$ where $t \in T$ is not a trap in $\cM$
  and~$K_{g_t} = K_t$.
  The new MDP $\cM'$ is obtained from $\cM$ 
  by adding deterministic transitions from $g_{t,\alpha}$ to $g_t$
  with action label $\tau$ and weight $-\wgt(t,\alpha)$ and
  by modifying the transition probabilities of 
  each state-action pair~$(t,\alpha)$ where $t\in T$ is not a trap in $\cM$
  is as follows:
  $P_{\cM'}(t,\alpha,s) = \frac{1}{2}P_{\cM}(t,\alpha,s)$ for all~$s\in S$,
  and~$P_{\cM'}(t,\alpha,g_t) = \frac{1}{2}$. 
  Thus, whenever $\cM'$ visits $t$ it moves to $g_t$   
  with probability~$\frac{1}{2}$, otherwise $\cM'$ continues as in~$\cM$.
  Now, any scheduler~$\sched$ for $\cM'$ 
  with $\Pr^{\sched}_{\cM',\sinit}\big(\varphi( G,(K_g)_{g \in G})\big)>0$
  also satisfies~$\Pr^{\sched}_{\cM,\sinit}\big(\varphi(T,(K_t)_{t\in T})\big)>0$ 
  as any $\sched$-path reaching~$g_t$ in $\cM'$
  has a prefix ending in $t$
  with the same accumulated weight.
  Conversely, if scheduler 
  $\sched$ for $\cM$ 
  satisfies $\Pr^{\sched}_{\cM}\big(\varphi(T,(K_t)_{t\in T})\big)>0$,
  then there exists a $\sched$-path~$\fpath$ from $\sinit$
  that ends in some state~$t \in T$ with accumulated weight at least~$K_t$.
  If $t$ is a trap in $\cM$ then $t$ is also a goal state in $\cM'$.
  Otherwise, \ie if $t$ is not a trap in $\cM$, then
  there is a $\sched$-path
  $\fpath'=\fpath \, \alpha \, g_{t,\alpha} \, \tau \, g_t$
  in $\cM'$ with $\wgt(\fpath')=\wgt(\fpath)+\wgt(t,\alpha)-\wgt(t,\alpha)=
  \wgt(\fpath)\geqslant K_t = K_{g_t}$. 
  Here, $\alpha$ is any action that $\sched$ schedules with
  positive probability for the input path $\fpath$.
  This shows $\Pr^{\sched}_{\cM',\sinit}\big(\varphi( G,(K_g)_{g \in G})\big)>0$.

  Suppose now that all goal states~$t\in T$ are trap states in $\cM$. 
  It is now easy to reduce them to a single trap state. 
  In fact, the MDP~$\cM$ can be modified by adding 
  a fresh goal state~$g$, and from each~$t \in T$, 
  a single action that deterministically
  leads to~$g$ with weight~$-K_t$. 
  If~$\cM''$ denotes this new MDP, then satisfying
  $\varphi(T,(K_t)_{t \in T})$ in~$\cM$ 
  is equivalent to satisfying~$\varphi(g,0)$ in~$\cM''$.
\end{proof}

The following two lemmas establish the complexity of the problem \Uposdwr{} 
(the first part of Theorem~\ref{thm:DWR-U0}).
The algorithm for the computation of the values will be given afterwards.

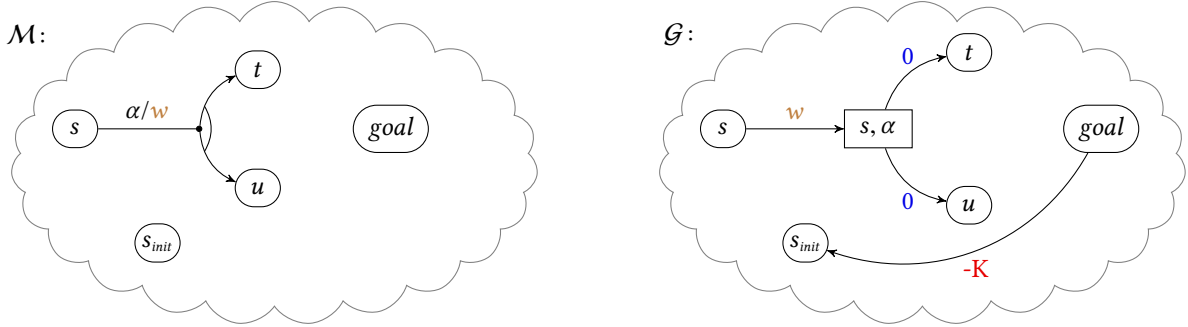
\begin{figure}[ht]
	\centering%

\begin{tikzpicture}
	\node[state] (s) {$s$};
	\node[istate, above left=7mm and 1mm of s] {$\cM\colon$};%
	\node[bullet, right=13mm of s] (b) {};
	\node[state, above right=5mm and 7mm of b] (t) {$t$};
	\node[state, below right=5mm and 7mm of b] (u) {$u$};
	\node[state, below right=10mm and 10mm of s] (init) {$s_\init$};
	\node[state, right=20mm of b] (goal) {$\goal$};
	\path%
		(s) edge[ntran] node[above] {$\alpha$/$\neuw{w}$} (b)
		(b) edge[ptran, bend left] coordinate[pos=0.3](bt) (t)
		(b) edge[ptran, bend right] coordinate[pos=0.3](bu) (u);	
	\draw%
		(bt) to[bend left] (bu);	
	\node[state, right=80mm of s] (gs) {$s$};
	\node[istate, above left=7mm and 1mm of gs] {$\cG\colon$};%
	\node[rstate, right=13mm of gs] (gb) {$s,\alpha$};
	\node[state, above right=5mm and 7mm of gb] (gt) {$t$};
	\node[state, below right=5mm and 7mm of gb] (gu) {$u$};
	\node[state, below right=10mm and 10mm of gs] (ginit){$s_\init$};
	\node[state, right=20mm of gb] (ggoal) {$\goal$};
	\path%
		(gs) edge[ptran] node[above] {$\neuw{w}$} (gb)
		(gb) edge[ptran, bend left] node[above]{\nilw} (gt)
		(gb) edge[ptran, bend right] node[below]{\nilw} (gu)
		(ggoal) edge[ptran, bend left=40] node[below]{\negw{K}} (ginit);	
	\coordinate[below right=2mm and 26mm of s] (c);
	\node[cloud, cloud puffs=23,cloud puff arc=120, cloud ignores aspect, draw=gray,
			minimum width=70mm, minimum height=43mm](cloud) at (c) {};
	\coordinate[below right=2mm and 26mm of gs] (gc);
	\node[cloud, cloud puffs=23,cloud puff arc=120, cloud ignores aspect, draw=gray,
			minimum width=70mm, minimum height=43mm](cloud) at (gc) {};
\end{tikzpicture}
 	\caption{Construction of a two-player game $\cG$ from an MPD $\cM$.}
\label{fig:DWRU0toMPBG}
\end{figure}

\begin{lemma}
\label{lem:Uposdwr-reduc}
  Let $\cM$ be an MDP, $\goal \in \cM$ a trap state, and
  $K \in \Integer$. Checking whether for all schedulers $\sched$,
  $\Pr^{\sched}_{\cM,\sinit}\big(\Diamond (\goal \wedge (\accwgt \geqslant K) )\big)>0$ reduces to the
  resolution of a two-player mean-payoff Büchi game.
  The problem \Uposdwr{} is thus in $\NP \cap \coNP$.
\end{lemma}
\begin{proof}
  From $\cM$, we construct a two-player game $\cG$ intuitively
  as follows: player~1 is responsible for choosing the actions, and
  player~2 resolves the probabilistic choices; moreover, from state
  $\goal$, controlled by player~1, we add an action leading back to
  the initial state with weight $-K$.  Formally, $\cG$ has set of
  vertices $V = S \cup S \times \Act$, partitionned into $V_1 = S$ and
  $V_2 = S \times \Act$, for each player. For every state $s \in S$ in
  the MDP $\cM$ and every action $\alpha$ enabled in $s$ there
  exists a transition in $\cG$ from $s \in V_1$ to
  $(s,\alpha) \in V_2$ with weight $\wgt(s,\alpha)$. Now, for all
  states $s,t \in S$ in the MDP and actions $\alpha$ satisfying
  $P(s,\alpha,t) >0$ there exists a transition in $\cG$ from
  $(s,\alpha) \in V_2$ to $t \in V_1$ with weight $0$. Finally, there
  is a transition from $\goal$ to the initial state $\sinit$ with
  weight $-K$. This transformation is represented in
  Figure~\ref{fig:DWRU0toMPBG}.

  In the sequel, $\stratone$ denotes a (pure) strategy for player~1
  and $\strattwo$ a (pure) strategy for player~2, and we write
  $\play{\cG}{\stratone}{\strattwo}$ for the play in $\cG$ yield by
  $\stratone$ and $\strattwo$.  The above transformation satifies
\begin{align*}
  \forall \sched,\ \Pr^{\sched}_{\cM,\sinit}(\Diamond (\goal \wedge (\accwgt \geqslant K)) )>0 & \Longleftrightarrow \forall \stratone \exists \strattwo,\ \play{\cG,\sinit}{\stratone}{\strattwo} \models (\neg \goal) \Until (\goal \wedge (\accwgt \geqslant K))  \\
& \Longleftrightarrow \forall \stratone \exists \strattwo,\ \play{\cG,\sinit}{\stratone}{\strattwo} \models \Diamond  (\goal \wedge (\accwgt \geqslant K))  \\
                                                                                               & \Longleftrightarrow \forall \stratone \exists \strattwo,\ \play{\cG,\sinit}{\stratone}{\strattwo} \models (\Box \Diamond \goal \wedge \MP \geq 0)\\
& \Longleftrightarrow  \exists \strattwo \forall \stratone,\ \play{\cG,\sinit}{\stratone}{\strattwo} \models (\Box \Diamond \goal \wedge \MP \geq 0)
\end{align*}
The first equivalence is immediate from the transformation since the
positive probability of the eventually property corresponds to the
existence of a path in the MDP. 

\noindent For the second equivalence, the left-to-right implication is
obvious; let us prove the right-to-left one. Let $\stratone$ be a
strategy for player~1, and fix $\strattwo$ a strategy for player~2,
such that
$\play{\cG,\sinit}{\stratone}{\strattwo} \models \Diamond (\goal \wedge
(\accwgt \geqslant K))$. The play $\play{\cG,\sinit}{\stratone}{\strattwo}$
until $\goal$ is reached with accumulated weight at least $K$, can be
decomposed into factors from $\sinit$ to $\goal$, alternated with
transitions from $\goal$ to $\sinit$:
$\pi_1 (\goal, \alpha , \sinit) \pi_2 \cdots (\goal, \alpha, \sinit)
\pi_m$ where the $\pi_i$'s do not visit $\goal$. We let $K_i$ be the
accumulated weight along $\pi_i$. Then the accumulated weight of this
prefix play is $\sum_{i=1}^{m-1} (K_i -K) + K_m$, and by assumption,
it is greater than $K$. We derive that
$\sum_{i=1}^m K_i \geq m \cdot K$, and thus there exists $i$ with
$K_i \geq K$. This fragment thus satisfies the property
$(\neg \goal) \Until (\goal \wedge (\accwgt \geqslant K))$. To
conclude, it suffices to observe that the strategy $\stratone$ can be
arbitrary on each of these fragments.

\noindent The third equivalence is relatively simple. First of all,
from left to right, given a strategy $\stratone$ for player~1, we aim
at building a strategy $\strattwo'$ for player~2 ensuring
$(\Box \Diamond \goal \wedge \MP \geq 0)$. To do so, the idea is to
apply the counterstrategy $\strattwo$ until $\goal$ is reached with
accumulated weight at least $K$; then $\strattwo'$ takes the $\alpha$
transition from $\goal$ to $\sinit$ with weight $-K$, so that the
accumulated weight is nonnegative; and we iterate the reasoning from
$\sinit$ again. Doing so, $\strattwo'$ guarantees infinitely many
visits to $\goal$ with accumulated weight at least $K$, and infinitely
many visits to $\sinit$ with nonnegative accumulated weight. The
mean-payoff of $\play{\cG,\sinit}{\stratone}{\strattwo}$ is thus
nonnegative.

\noindent The last equivalence is a consequence of the determinacy of
two-player turn-based games with mean-payoff and Büchi objectives, a
consequence of Martin's general determinacy theorem~\citeapx{Martin98}. Mean-payoff Büchi games are even finite-memory determined~\cite{ChatDoy11}.

\medskip
The complexity of the problem \Uposdwr{} then follows directly,
as determining the winner in a turn-based game with mean-payoff Büchi winning condition
is in $\NP \cap \coNP$~\cite{ChatDoy11}.
\end{proof}

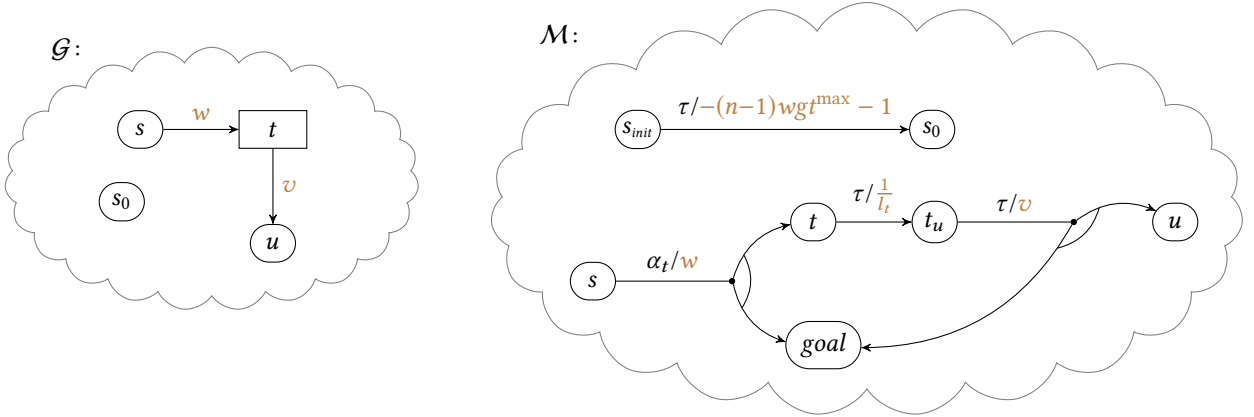
\begin{figure}[ht]
	\centering%

\begin{tikzpicture}
	\node[state] (gs) {$s$};
	\node[istate, above left=5mm and 5mm of gs] {$\cG\colon$};%
	\node[state, below left=7mm and 2mm] (gs0) {$s_0$};
	\node[rstate,right=10mm of gs] (gt) {$t$};
	\node[state,below=10mm of gt] (gu) {$u$};
	\path%
		(gs) edge[ptran] node[above] {$\neuw{w}$} (gt)
		(gt) edge[ptran] node[right] {$\neuw{v}$} (gu);
	\node[state,right=60mm of gs] (sinit) {$s_\init$};
	\node[istate, above left=7mm and 5mm of sinit] {$\cM\colon$};%
	\node[state, right=33mm of sinit] (s0) {$s_0$};
	\node[state, below left=15mm and 5mm of sinit] (s) {$s$};
	\node[bullet, right=15mm of s] (b1) {}; 
	\node[state, above right=5mm and 10mm of b1] (t) {$t$};
	\node[state, below right=5mm and 10mm of b1] (goal) {$\goal$}; 
	\node[state, right= 10mm of t] (tu) {$t_u$};
	\node[bullet, right=15mm of tu] (b2) {}; 
	\node[state, right=10mm of b2] (u) {$u$};
	\path%
		(sinit) edge[ptran] node[above] {$\tau$/$\neuw{-(n{-}1)\wgt^{\max} -1}$} (s0)
		(s) edge[ntran] node[above]{$\alpha_t$/$\neuw{w}$} (b1)
		(b1) edge[ptran, bend left] coordinate[pos=0.3](b1t) (t)
		(b1) edge[ptran, bend right] coordinate[pos=0.3](b1goal) (goal)
		(t) edge[ptran] node[above] {$\tau$/$\neuw{\frac{1}{l_t}}$} (tu)
		(tu) edge[ntran] node[above] {$\tau$/$\neuw{v}$} (b2)
		(b2) edge[ptran, bend left] coordinate[pos=0.3](b2u) (u)
		(b2) edge[ptran, bend left] coordinate[pos=0.1](b2goal) (goal);
	\draw%
		(b1t) to[bend left] (b1goal);	
	\draw%
		(b2u) to[bend left] (b2goal);	
	\coordinate[below right=4mm and 9mm of gs] (gc);
	\node[cloud, cloud puffs=23,cloud puff arc=120, cloud ignores aspect, draw=gray,
			minimum width=55mm, minimum height=35mm](cloud) at (gc) {};
	\coordinate[below right=8mm and 30mm of sinit] (c);
	\node[cloud, cloud puffs=23,cloud puff arc=120, cloud ignores aspect, draw=gray,
			minimum width=100mm, minimum height=55mm](cloud) at (c) {};
\end{tikzpicture}
 	\caption{Construction of an MDP $\cM$ from a two-player game $\cG$.}
\label{fig:uposdwr-hard-mpg}
\end{figure}

\begin{lemma}
 \label{MPG-hardness-Uposdwr}
The problem \Uposdwr{} is hard for
(non-stochastic) two-player mean-payoff games.
\end{lemma}
\begin{proof}
   We now prove the lower bound, that is,
   checking whether player 1 of a (non-stochastic) mean-payoff game
   has a winning strategy
   is polynomially reducible to the complement of~\Uposdwr.

   More precisely, we provide a polynomial reduction to the problem to decide
   whether
   $\Pr^{\sched}_{\cM,\sinit}
               \big(\Diamond (\goal \wedge (\accwgt \geqslant 0))\big)=0$ holds
   for all schedulers $\sched$ for a given MDP $\cM$ with 
   distinguished states $\goal$ and $\sinit$.

   Consider a mean-payoff game $\cG$ with
   starting state $s_0$. Let $\cM$ be the MDP obtained from $\cG$ by
   performing the following steps (see also Figure~\ref{fig:uposdwr-hard-mpg}).

   \begin{itemize}
   \item 
     Add a new initial state $\sinit$ and a trap state
     $\goal$. 
   \item
     For each player-1 state~$s$
     and edge
     $s \stackrel{w}{\longrightarrow} t$ in $\cG$,
     state $s$ in $\cM$ has an enabled action $\alpha_t$
     with $P(s,\alpha_t,t)=P(s,\alpha_t,\goal)=\frac{1}{2}$ and
     $\wgt(s,\alpha_t) = w$.
   \item
     For each player-2 state~$s$ in~$\cG$,
     we add states~$s$ and
     $s_t$ for all successors $t$ of $s$ to~$\cM$.
     State $s$ in $\cM$ has a single enabled action $\tau$
     with $P(s,\tau,s_t)=\frac{1}{\ell_s}$ 
     where $\ell_s$ denotes the number of successors of $s$ in $\cG$
     and where $t$ ranges over all successors of $s$ in $\cG$.
     The states $s_t$ have a single enabled action $\tau$
     with $P(s_t,\tau,\goal)=P(s_t,\tau,t)=\frac{1}{2}$
     and $\wgt(s_t,\tau)$ equals the weight of the edge from $s$ to $t$
     in $\cG$. 

   \item
     State $\sinit$ has a single action 
     with $P(\sinit,\tau,s_0)=1$ and 
     $\wgt(\sinit,\tau)= -(n{-}1)\wgt^{\max} -1$ where $\wgt^{\max}$ is 
     the maximal
     weight attached to the edges in $\cG$ and $n$ is the number of
     states in $\cG$.
     (We suppose $\wgt^{\max} > 0$. If this is not the case we put
     $\wgt(\sinit,\tau)=-1$.)
   \end{itemize}
   Then, $\cM$ is contracting in the sense
   $\Pr^{\min}_{\cM,s}(\Diamond \goal)=1$. 
   In particular, $\cM$ has no end components.
   We have for all schedulers $\sched$ for $\cM$ that
     \(\Pr^{\sched}_{\cM,\sinit}\big(\Diamond \goal \wedge (\accwgt <0)\big)=1
     \text{ iff }
     \Pr^{\sched}_{\cM,\sinit}
          \big(\Diamond \goal \wedge (\accwgt \geqslant 0)\big)=0\).
   Moreover, there is a one-to-one
   correspondence between the schedulers for $\cM$ and the strategies
   for player 1 in $\cG$.

   If $\sched$ is an MD-strategy for player 1 in $\cG$ such that 
   the mean payoff of all
   $\sched$-plays is nonpositive, then 
   $\sched$ has no positive cycles and
   $$
      \Pr^{\sched}_{\cM,\sinit}
          \bigl(\, \Diamond \goal \wedge (\accwgt <0)\, \bigr) 
      \ \ = \ \ 1.
   $$
   In fact $\wgt(\fpath) < -\wgt(\sinit,s_0)$
    for all simple $\sched$-paths $\fpath$ starting in state $s_0$,
    and since there are no positive cycles under~$\sched$, any non-simple
    path has also negative weight.
    That is,
    $\wgt(\fpath) < 0$ for all $\sched$-paths starting in $\sinit$.

   Conversely, if $\sched$ is a scheduler for $\cM$ with
   $\Pr^{\sched}_{\cM,\sinit}\big(\Diamond (\goal \wedge (\accwgt <0))\big)=1$
   then there is an MD-scheduler $\tsched$ for $\cM$ with
   $\Pr^{\tsched}_{\cM,\sinit}\big(\Diamond (\goal \wedge (\accwgt <0))\big)=1$
   (see Lemma \ref{MD-sufficient-eventually-almostsure-max-equals-1})
   and the Markov chain induced by $\tsched$ has no positive cycles.
   Thus, the mean payoff of all $\tsched$-plays starting is nonpositive.
\end{proof}

Finally, we explain how to compute 
$\valueUpos{\cM,s}$ in 
pseudo-polynomial time.
Note that this implies that the decision problem is solvable in pseudo-polynomial time as well.

We may assume w.l.o.g. that $T\setminus T^*$ is a singleton 
(following the argumentation provided in Lemma~\ref{lem:Uposdwr-part-case}
later on). That is,
$T$ contains a single trap state with finite~$K_t$.
Additionally, we make the following assumption \textbf{(A)}:
\begin{description}
\item [(A)]
  $\Pr^{\min}_{\cM,\sinit}(\Diamond T)>0$.
\end{description}
  This assumption is justified as $\Pr^{\min}_{\cM,\sinit}(\Diamond T)=0$
  implies
  $\Pr^{\sched}_{\cM,\sinit}(\varphi)=0$ for some scheduler $\sched$
  regardless the value of~$K_t$.

\tudparagraph{1ex}{Preprocessing.} 
In what follows, let $T =\{\goal\}$ and suppose that 
$\cM$ satisfies assumption \textbf{(A)}.
We now define a new MDP $\cN$ that arises from $\cM$
by adding %
  a fresh action symbol $\tau$ 
and a new trap state %
$\fail$ and by performing the following steps:
\begin{enumerate}
\item [1.]
  Collapse all states $s$ with
  $\Pr^{\min}_{\cM,s}(\Diamond T)=0$ to the single trap state $\fail$.
\item [2.]
  Remove all states $s$ with $\sinit \not\models \exists \Diamond s$.
\end{enumerate}
As all states $s$ that belong to some end component $\cE$ of
$\cM$ %
are collapsed to $\fail$ (see step 1), the MDP $\cN$ has no end
component. %
Hence, under all schedulers for $\cN$, almost surely one of its two
trap states $\fail$ or $\goal$ will be reached:
\[
  \Pr^{\min}_{\cN,s}
    \bigl(\, \Diamond (\goal \vee \fail )\, \bigr) \ =\ 1
  \quad
  \text{for all states $s$ of $\cN$}\enskip.
\]
Assumption \textbf{(A)} yields
$\Pr^{\min}_{\cN,\sinit}\bigl(\, \Diamond (\goal)\, \bigr) 
  \, > \, 0$.

For $K\in \Integer$, let
\begin{equation*}
  \label{phi-prime}
  \psi_K \ \ = \ \ 
  \Diamond (\goal \wedge (\accwgt \geqslant K))\enskip.
\end{equation*}
Problem \Uposdwr{} rephrased for $\cN$ asks whether
$\Pr^{\sched}_{\cN,\sinit}(\psi_K) >0$ for all schedulers
$\sched$ for $\cN$ where $K\in \Integer$ is fixed. 
The corresponding optimization problem
asks to compute for the states $s$ in $\cN$ the values
\[
  \valueUpos{\cN,s}
 \ \ = \ \ 
  \sup \, \bigl\{ \,
     K\mid \forall \sched . \, 
       \Pr^{\sched}_{\cN,s}(\psi_K) >0 \ 
    \bigr\}\enskip.
\]
We have
$\valueUpos{\cM,s} = \valueUpos{\cN,s}$ %
for all states $s$ in $\cM$ with
$\Pr^{\min}_{\cM,s}(\Diamond \goal)>0$,
while $\valueUpos{\cM,s} = -\infty$ if
$\Pr^{\min}_{\cM,s}(\Diamond T)=0$.

\tudparagraph{1ex}{Assumptions after the preprocessing.}
We now have the following assumptions
\begin{description}
\item [(C1)]
   $\cM$ has no end components and two traps states
   $\goal$ and $\fail$. 
\item [(C2)]
   $\Pr^{\min}_{\cM,s}\big(\Diamond (\goal \vee \fail)\big)=1$
   for all states $s$ of $\cM$.
\item [(C3)]
   $\Pr^{\min}_{\cM,s}\big(\Diamond (\goal)\big) >0$
   for all states $s$ of $\cM$ with $s \not= \fail$.
\item [(C4)]
   All states in $\cM$ are reachable from $\sinit$.
\end{description}
The values for the trap states are trivial as we have
$\valueUpos{\cM,\fail} =-\infty$ and
$\valueUpos{\cM,\goal}=0$.

\begin{lemma}
  If $\cM$ satisfies the above assumptions \textbf{(C1)} to \textbf{(C3)}, then
  $\valueUpos{\cM,s}
  \in \Integer \cup\{+\infty\}$
  for all non-trap states $s$ in $\cM$.
\end{lemma}
\begin{proof}
 Let $s$ be a non-trap state in $\cM$.
 Let $E_s$ denote the maximal conditional expected number of steps for reaching
 $\goal$ from $s$ in $\cM$,
 under the condition $\Diamond (\goal)$.
 By assumption \textbf{(C3)} and
 the results of~\citeapx{BKKW17}, $E_s$ is finite for all states $s$ in $\cM$.
 Let $k_s = \lceil E_s \rceil$. Then, for each scheduler $\sched$ there is
 at least one path from $s$ to $\goal$ of length at most $k_s$,
 which yields $\valueUpos{\cM,s} \in \Integer \cup \{\infty\}$.
\end{proof}

\begin{lemma}
Assumptions and notations as before. 
For each $K\in \Integer$ and each state $s$ in $\cN$:
\begin{center}
  $\Pr^{\sched}_{\cM,s}(\varphi_K)>0$ for all schedulers
  $\sched$ for $\cM$\enskip
  iff \enskip
  $\Pr^{\sched}_{\cN,s}(\varphi_K)>0$
  for all schedulers
  $\sched$ for $\cN$\enskip.
\end{center}
\end{lemma}

Thus, the value $\valueUpos{\cM,s}$
is the maximal value $K\in \Integer \cup \{\infty\}$ such that
$\Pr^{\sched}_{\cN,s}(\varphi_K)>0$
for all schedulers $\sched$ for $\cN$.

\begin{corollary}
$\valueUpos{\cM,s}= \valueUpos{\cN,s}$
  for all states $s$ in $\cN$.
\end{corollary}

Given a scheduler $\sched$ for $\cN$ we define
\[
  K_{\sched,s}^{0} \ \ = \ \ 
   \sup \, 
   \bigl\{ \, K \in \Integer \, : \,
                 \Pr^{\sched}_{\cN,s}(\varphi_K)>0 \, 
   \bigr\}\enskip.
\]
If $\sched$ is an MD-scheduler then let
$\cN_{\sched}$ denote the Markov chain induced by $\sched$.

\begin{lemma}
 \label{AllPos-Reach-MC}
  Let $\sched$ be an MD-scheduler for $\cN$. 
  Then, for all non-trap states $s$ in
  $\cN$:
  \begin{enumerate}
  \item [(a)]
    $K_{\sched,s}^0 =+\infty$ iff $\cN_{\sched}$ has a positive cycle 
    that is reachable from $s$.
  \item [(b)]
    If $\cN_{\sched}$ 
    does not contain any positive cycle that is reachable from $s$,
    then 
    \[
       K_{\sched,s}^0 \ \ = \ \ 
       \max \ 
       \bigl\{ \ \wgt(\fpath) \ : \ 
           \text{$\fpath$ is a path from $s$ to $\goal$ in $\cN_{\sched}$} \ 
       \bigr\}\enskip.
    \]
  \end{enumerate}
  The values $K_{\sched,s}^{0}$ are computable in polynomial time.
\end{lemma}
\begin{proof}
Statements (a) and (b) are obvious. To check the existence of
positive cycles and to compute the values $K_{\sched,s}^{0}$ 
in (b) we can apply
standard shortest-path algorithms to the weighted graph that arises
from $\cN_{\sched}$ by ignoring the transition probabilities and multiplying
all weights with $-1$.
\end{proof}

Let $S$ be the state space of $\cN$ without state $\fail$.
If $(s,\alpha)$ is a state-action pair in $\cN$ then
$\Post(s,\alpha)=\{t\in S \cup \{\fail\}: P(s,\alpha,t)>0\}$.

\begin{lemma}
\label{K-U0-fin-value-no-pos-cycle}
  For each state $s\in S$ we have that
  $\valueUpos{\cN,s} \in \Integer$ iff
  there is at least one MD-scheduler $\sched$ for $\cN$
  such that the Markov chain $\cN_{\sched}$ induced by $\sched$ has no
  positive cycle that is reachable from $s$.
  In this case, $\valueUpos{\cN,s} = \min_{\sched} K_{\sched,s}^0$
  where the minimum ranges over all MD-schedulers $\sched$ 
  for $\cN$.
\end{lemma}
\begin{proof}
``$\Longleftarrow$'': 
If there is an MD-scheduler $\sched$ without positive
cycles, then %
$\valueUpos{\cN,s}$
is bounded from above by the maximal weight of the
$\sched$-paths from $s$ to $\goal$. This value is finite.

``$\Longrightarrow$'': 
Suppose $K\eqdef \valueUpos{\cN,s} \in \Integer$.
Because~$K$ is a maximum,
there is some scheduler $\sched$ such that $\wgt(\fpath) \leqslant K$
for all $\sched$-paths from $s$ to $\goal$.
But then:
\[
  \Pr^{\sched}_{\cN,s}
   \bigl(\ \Diamond \fail \ \vee \ 
           \Diamond (\goal \wedge (\accwgt \leqslant K)) \ \bigr)
  \ \ = \ \ 1\enskip.
\]
As $\cN$ has no end components, $\cN$ has no (positively or negatively)
weight-divergent scheduler.
Hence, we may apply 
Lemma \ref{MD-sufficient-eventually-almostsure-max-equals-1}
to obtain the existence of an MD-scheduler $\tsched$ such that
\[
  \Pr^{\tsched}_{\cN,s}
   \bigl(\ \Diamond \fail \ \vee \ 
           \Diamond (\goal \wedge (\accwgt \leqslant K)) \ \bigr)
  \ \ = \ \ 1\enskip.
\]
But then the weight of all $\tsched$-paths from $s$ to $\goal$ is bounded
by $K$. Lemma \ref{AllPos-Reach-MC} yields that $\tsched$ has no
positive cycle that is reachable from $s$.
The last part is obvious from Lemma \ref{AllPos-Reach-MC}.
\end{proof}

Note that the previous lemma is sufficient to derive an exponential-time algorithm
to compute the values: one can enumerate all MD-schedulers and pick the one with
the best value. In the remaining of this section, we will show how to compute
these values in pseudo-polynomial time.

\begin{lemma}
\label{S-infty}
   Let $S_{\infty}=\{s\in S : \valueUpos{\cN,s}=\infty\}$. 
   Then for each state $s\in S$ the following statements are equivalent:
   \begin{enumerate}
   \item [(a)]
      $s\in S_{\infty}$
   \item [(b)]
      For each $w\in \Integer$ and each scheduler $\sched$ there is
      an $\sched$-path $\fpath$ from $s$ to $\goal$ with
      $\wgt(\fpath) \geqslant w$.
   \item [(c)]
      $\Pr^{\min}_{\cN,s}(\Diamond S_{\infty}) >0$
   \end{enumerate}
\end{lemma}

\begin{proof}
``(a) $\Longleftrightarrow$ (b)'' 
and ``(a) $\Longrightarrow$ (c)'' are trivial.
To prove ``(c) $\Longrightarrow$ (b)'' we suppose 
$\Pr^{\min}_{\cN,s}(\Diamond S_{\infty}) >0$.
Let $w\in \Integer$ and $\sched$ be a scheduler.
As $\Pr^{\sched}_{\cN,s}(\Diamond S_{\infty}) >0$ there is
a state $t\in S_{\infty}$ and an $\sched$-path $\fpath'$ from $s$ to $t$.
Let $w'=\wgt(\fpath')$.
We now consider the residual scheduler $\sched'=\residual{\sched}{\fpath'}$.
As (a) and (b) are equivalent and $t\in S_{\infty}$, 
there is an $\sched'$-path $\fpath''$ from $t$ to $\goal$ with
$\wgt(\fpath'')\geqslant w-w'$.
But then $\fpath \eqdef \fpath';\fpath''$ is an $\sched$-path from
$s$ to $\goal$ with 
\[
  \wgt(\fpath)\ \ = \ \ \wgt(\fpath') + \wgt(\fpath'') \ \ \geqslant \ \ 
  w' \ + \ (w-w') \ \ = \ \ w\enskip.
\]
This completes the proof of Lemma \ref{S-infty}.
\end{proof}

\begin{remark}[Reduction to mean-payoff games]
\label{infty-via-MP-games}
{\rm
   Checking whether $\valueUpos{\cM,s} =+\infty$
   is polynomially reducible to non-stochastic two-player mean-payoff games.
   For this, we regard MDPs as non-stochastic
   two-player games (action player against probabilistic player). 
   The objective of the action player is
   to ensure that the mean payoff is nonpositive.
   Then,
   ``all MD-schedulers of an MDP have a positive cycle''
   is equivalent to
   ``there is no winning strategy for the action player''.
   Thus, Lemma \ref{K-U0-fin-value-no-pos-cycle} 
   yields a polynomial reduction 
   to the complement of non-stochastic mean-payoff games with threshold 0.
 }
\end{remark}

The previous remark allows us to compute $S_\infty$ in pseudo-polynomial time
since mean-payoff games can be solved in pseudo-polynomial time.
In the rest of this section, we will assume that~$S_\infty$ is given, and show how to compute the values in polynomial time. The overall complexity will thus be pseudo-polynomial time.

\tudparagraph{1ex}{Computing the values in $\cN$.}
Suppose we have an oracle to compute $S_{\infty}$.
Let
$$
   S_{\text{fin}} \ \ \eqdef \ \ \ 
   S\setminus S_{\infty}
   \ \ = \ \
   \bigl\{ \ s \in S \ : \
      \valueUpos{\cN,v} \in \Integer \ \bigr\}\enskip.
$$
For each state $s\in S_{\text{fin}}$ we define
$\Act_{\text{fin}}(s)$ as the set of actions $\alpha \in \Act(s)$
such that $P(s,\alpha,s')>0$ implies $s'\in S_{\text{fin}}$.
 Note that $\Act_{\text{fin}}(s)$ is nonempty if
 $s \in S_{\text{fin}}\setminus \{\goal\}$.
We have $\Pr^{\min}_{\cN,s}(\Diamond S_{\infty})=0$ for all
states $s \in S_{\text{fin}}$.

The values $\valueUpos{\cN,s}$ for
the states $s \in S_{\text{fin}}\setminus \{\goal\}$
satisfy the following equation:
 \[
   \valueUpos{\cN,s}\ \ = \ \
   \min \ 
   \bigl\{ \ K_{s,\alpha} \ : \ \alpha \in \Act_{\text{fin}}(s) \
   \bigr\}
 \]
 where for $(s,\alpha) \in \cN$
 \[
      K_{s,\alpha} \ \ = \ \
      \wgt(s,\alpha)  \ + \
      \max \ 
        \bigl\{ \valueUpos{\cN,v} \ : \ 
                  v\in \Post(s,\alpha) \setminus \{\fail\} \ \bigr\}\enskip.
\]
Recall that $\valueUpos{\cN,\goal}=0$.

We now provide a polynomial-time algorithm for the computation of the
values $\valueUpos{\cN,s}$ for 
$s \in S_{\text{fin}}$.
Let $n = |S_{\text{fin}}|$ denote the number of states in $S_{\text{fin}}$.
\begin{description}
\item [{\it Initialization.}] 
  Let $K_{\goal}^{(j)}=0$ for $j=0,1,\ldots,n{-}1$. 
  For all states $s\in S_{\text{fin}} \setminus \{\goal\}$ we start with
  $K_s^{(0)} \ = \ -\infty$.
\item [{\it Iteration.}] 
  For $j=1,\ldots,n{-}1$ we compute the following values
  for all states $s\in S\setminus \{\goal\}$ 
  and all actions $\alpha \in \Act_{\text{fin}}(s)$:
  \[
     K_{s,\alpha}^{(j)} \ \ = \ \ 
     \wgt(s,\alpha) \ + \!\!
     \max_{t\in \Post(s,\alpha)} \!\! K_t^{(j-1)}
   \qquad \text{and} 
   \qquad
     K_s^{(j)} \ = \ 
     \min \, 
        \bigl\{ \, K_{s,\alpha}^{(j)} \, : \, 
                                \alpha \in \Act_{\text{fin}}(s) \, \bigr\}.
  \]
\end{description}

\begin{lemma}[Soundness]
    The above algorithm correctly computes the values 
    $K_s^{(n-1)}=\valueUpos{\cN,s}$ for all states 
    $s\in S_{\text{fin}}$.
\end{lemma}

\begin{proof}
 Let $\cN_{\text{fin}}$ 
 be the largest sub-MDP of $\cN$ that does not contain
 any state of $S_{\infty}$. 
 That is, the state space of $\cN_{\text{fin}}$ is
 $S_{\text{fin}}\cup\{\fail\}$ and
 $\cN_{\text{fin}}$ results from $\cN$
 by removing the states $t\in S_{\infty}$ and all state-action pairs
 $(s,\alpha)$ with $P(s,\alpha,t)>0$ for some $t\in S_{\infty}$.
 Thus, the action set of each state $s \in S_{\text{fin}}$ is
 $\Act_{\text{fin}}(s)$.
 Then, $\cN_{\text{fin}}$ has no positive cycle
 (Lemma \ref{K-U0-fin-value-no-pos-cycle}).

  By induction on $j$, we get
  for all states $s\in S_{\text{fin}}$ 
  and all actions $\alpha \in \Act_{\text{fin}}(s)$:
    $$
      K_{s,\alpha}^{(j)} \ \leqslant \ 
      K_{s,\alpha}^{(j+1)} \ \leqslant \ 
      K_{s,\alpha}
      \ \ \ \text{and} \ \ \
      K_s^{(j)} \ \leqslant \ 
      K_s^{(j+1)} \ \leqslant \ \valueUpos{\cN,s}\enskip.
    $$
  Given a scheduler $\sched$ for $\cN_{\text{fin}}$, let
  \[
    \maxwgt{(j)}{s}{\sched}
    \ \ = \ \
    \max
      \bigl\{ \ \wgt(\fpath) \ :
                \text{$\fpath$ is a $\sched$-path from $s$ to $\goal$
                      with $|\fpath|\leqslant j$} \ 
      \bigr\}\enskip.
  \]
  Then, by induction on $j$ we get:
  \[
    K_s^{(j)} \ \ = \ \ 
    \min_{\sched} \ \maxwgt{(j)}{s}{\sched}
  \]
  where $\sched$ ranges over all schedulers for  $\cN_{\text{fin}}$.
  Moreover, there exists a scheduler $\sched_j$ for  $\cN_{\text{fin}}$
  such that
  $K_s^{(j)} = \maxwgt{(j)}{s}{\sched_j}$.
  
  Let now $\sched=\sched_{n-1}$.
  As $\cN_{\text{fin}}$ has no positive cycles, we have: 
  If $\fpath$ is a finite path of length at least $n$, then 
  $\wgt(\fpath) \leqslant \wgt(\fpath')$ where $\fpath'$ results from
  $\fpath$ by removing all cycles.
  Thus, $\maxwgt{(n-1)}{s}{\sched}$ 
  is the maximal weight of a $\sched$-path from
  $s$ to $\goal$. But then $K_s^{(n-1)} =  \valueUpos{\cN,v}$ 
  for all states $s \in S_{\text{fin}}$.
\end{proof}

\begin{corollary}
  The values $\valueUpos{\cM,s}$ for the states 
  $s \in S_{\text{fin}}$ are
  computable in polynomial time, 
  assuming an oracle to compute the set~$S_{\infty}$.
\end{corollary}

By Lemma \ref{K-U0-fin-value-no-pos-cycle},
if $\cM$ has no positive cycles then $S_{\infty}$ is empty.
Hence:

\begin{corollary}[Complexity of \Uposdwr{} for MDPs without positive cycles]
   If $\cM$ has no positive cycles then
   problem \Uposdwr{} is in P and the values
   $\valueUpos{\cM,s}$ are computable in polynomial time.
\end{corollary}

For the general case we have:

\begin{theorem}[Complexity of \Uposdwr]
  The decision problem \Uposdwr{} is in coNP and
  the values $\valueUpos{\cM,s}$ for the states
  $s$ in $\cM$ are computable in pseudo-polynomial time.
\end{theorem}

\begin{proof}
 To prove membership to coNP we rely on the statements of
 Lemma \ref{K-U0-fin-value-no-pos-cycle}, which yields that the answer
 to question \Uposdwr{} is ``no'' iff there is an MD-scheduler
 $\sched$ for $\cN$ such that the Markov chain induced by $\sched$ contains
 no positive cycle and $K_{\sched,\sinit}^0 < K$.
 So, a nondeterministic polynomially time-bounded algorithm for the
 complement of \Uposdwr{} is obtained by guessing an MD-scheduler for $\cN$,
 computing the value $K_{\sched,\sinit}^0$ in polynomial time
 (see Lemma \ref{AllPos-Reach-MC}) and finally checking whether
 $K_{\sched,\sinit}^0 < K$.

 To compute the values $\valueUpos{\cM,s}$ in pseudo-polynomial time,
 we compute $S_{\infty}$ in pseudo-polynomial time by Remark~\ref{infty-via-MP-games},
 and apply the above algorithm to compute
 the values $\valueUpos{\cM,s}$ for the states
 $s\in S_{\text{fin}}$.
\end{proof}

\subsection{Almost-Sure Reachability Under All Schedulers}

\label{Uasdwr}
\label{sec:Uasdwr-appendix}

In this section, we prove the following theorem.

\dwruas*

From $\cM$ we construct a weighted directed graph $G  = (V,\to,\wgt)$. The set of vertices is $V = \{s \in S : \Pr_{\cM,s}^{\min}(\Diamond T^*)<1\}$. There is an edge in $G$ from $s$ to $s'$ iff there exists an action $\alpha \in \Act$  with $\mdpP(s,\alpha,s')>0$. The weight associated with edge $s \to s'$ is the minimum among actions that can lead from $s$ to $s'$: $\wgt(s \to s') = \min\{\wgt(s,\alpha) \mid P(s,\alpha,s') >0\}$. Finally, for $s \in V$, $G_s$ denotes the subgraph of $G$ reachable from $s$.

In the case where all states in $T$ are traps, Theorem \ref{thm:DWRU1}
derives from the following characterization of positive instances of \Uasdwr:

\begin{lemma}
\label{lm:charac-Uasdwr-traps} 
Let $\varphi$ be a \dwr-property with all states in $T$ being traps. Then $\Pr_{\cM,s}^{\min}(\varphi)=1$ iff the following two conditions hold:
\begin{enumerate}[(i)]
\item $\Pr_{\cM,s}^{\min}(\Diamond T)=1$, and
\item if\ \ $\Pr_{\cM,s}^{\min}(\Diamond T^*)<1$, then the weighted graph $G_s$ does not contain any negative cycle, and for each path $\fpath$ in $G_s$ that starts in $s$ and ends in some $t \in T\setminus T^*$ we have 
$\wgt(\fpath) \geqslant K_t$.
\end{enumerate}
In this case and if\ \ $T \setminus T^*=\{\goal\}$ is a singleton, then
$K^{\forall,=1}_{\cM,s}$ is the minimal weight of a path from $s$ to $\goal$
in $G_s$.
\end{lemma}

\begin{proof}
``$\Longrightarrow$'': $\Pr_{\cM,s}^{\min}(\Diamond T)=1$ is clearly a necessary condition for $\Pr_{\cM,s}^{\min}(\varphi)=1$. Assume now that $\Pr_{\cM,s}^{\min}(\Diamond T^*) < 1$ and that either $G_s$ contains a negative cycle, or there is a path $\fpath$ from $s$ to some $t \in T\setminus T^*$ with $\wgt(\fpath) < K_t$. In both cases there is a scheduler $\sched$ such that $\Pr_{\cM,s}^\sched(\Diamond (t \wedge \accwgt < K_t)) >0$ and hence $\Pr_{\cM,s}^{\min}(\varphi)<1$.

``$\Longleftarrow$'': $\Pr_{\cM,s}^{\min}(\Diamond T^*)=1$ clearly implies $\Pr_{\cM,s}^{\min}(\varphi)=1$. Assume now that  $\Pr_{\cM,s}^{\min}(\Diamond T)=1$, $\Pr_{\cM,s}^{\min}(\Diamond T^*)<1$, $G_s$ does not contain any negative cycle, and for each path $\fpath$ in $G_s$ that starts in $s$ and ends in some $t \in T\setminus T^*$ we have $\wgt(\fpath) \geq K_t$. Then under all schedulers and for every path from $s$ to a target state $t \in T\setminus T^*$ the accumulated weight necessarily is at least $K_t$. We thus derive $\Pr_{\cM,s}^{\min}(\varphi)=1$.
\end{proof}

Condition (i) from the characterization of Lemma~\ref{lm:charac-Uasdwr-traps} is a classical verification question for MDP and can be solved in \PTIME. For condition (ii), the weighted graph can be constructed in polynomial time, and using standard shortest-path algorithms in weighted graphs one can check for the nonexistence of a negative cycle and compute the minimal weight of paths
from $s$ to a target state $t\in T \setminus T^*$.
Thus, in case all $T$-states are traps, \Uasdwr{} can be solved in polynomial time.

\medskip

We now address the general case. Intuitively, this case is harder since a target state $t \in T \setminus T^*$ might be visited several times before $t$ is actually visited with the constraint $\accwgt \geqslant K_t$. However, let us explain how to reduce the general case to the case where all states in $T$ are traps. To ease the presentation, we consider a simple  \dwr-property of the form $\Diamond (\goal \wedge (\accwgt \geqslant K))$.

Clearly, $\Pr_{\cM,s}^{\min}(\Diamond \goal)=1$ is a necessary condition for $\Pr_{\cM,s}^{\min}(\varphi)=1$. We thus check first whether $\Pr_{\cM,s}^{\min}(\Diamond \goal)=1$ holds. Then, without loss of generality, we assume that all states are reachable from the initial state $s$, and that $\Pr_{\cM,t}^{\min}(\Diamond \goal)=1$ for all states $t$.
Under these assumptions, all end components of $\cM$ must contain $\goal$. 
If $\cM$ has no end components then $\goal$ is a trap, and we are back to the special case (Lemma \ref{lm:charac-Uasdwr-traps}). Suppose now that $\goal$ is not a trap. Then, $\cM$ has a unique maximal end component $\cE$ and $\cE$ contains $\goal$.
\begin{itemize}
\item 
  If $\Exp{\min}{\cE}(\MP) >0$ then
  all end components of $\cM$ are pumping.
  Hence, $\Pr_{\cM,s}^{\min}\big(\Diamond (\goal \wedge (\accwgt \geqslant K))\big)=1$
  for all $K \in \Integer$, and therefore
  $K^{\forall,=1}_{\cM,s}=+\infty$.
\item
  If $\Exp{\min}{\cE}(\MP) <0$ then $\cM$ has negatively weight-divergent 
  end components.
  In this case, 
  $\Pr^{\sched}_{\cE,\goal}\big(\neXt \Diamond (\goal \wedge (\accwgt <0))\big)>0$
  where $\sched$ is an MD-scheduler for $\cE$ with
  $\Exp{\sched}{\cE}(\MP) = \Exp{\min}{\cE}(\MP) <0$.
  (This follows from the quotient representation of
   the expected mean payoff in the Markov chain induced by $\sched$,
   see Lemma \ref{lemma:scMC-exp-mp-quotient}.)
  Hence,
  $K^{\forall,=1}_{\cM,s}=K^{\forall,=1}_{\cM',s}$
  where $\cM'$ is the MDP resulting from $\cM$ when turning $\goal$ into 
  a trap (\ie removing all state-action pairs $(\goal,\alpha)$).
  For $\cM'$ we can then rely on Lemma~\ref{lm:charac-Uasdwr-traps}.
\item 
  Suppose now that $\ExpRew{\min}{\cE}(\MP) = 0$.
  \begin{itemize}
  \item
    If $\cE$ does not contain any 0-EC then
    $\cE$ is universally weight-divergent, \ie
    all end components of $\cE$ are positively weight-divergent.
    This is a consequence of 
    Theorem \ref{universal-neg-wgt-div}
    applied to the
    MDP resulting from $\cM$ by multiplying all weights with $-1$.
    Hence, in this case we have $K^{\forall,=1}_{\cM,s}=+\infty$.
 \item
    Suppose now that $\cE$ contains at least one 0-EC.
    In this case, we can treat $\goal$ as a trap and rely on
    Lemma~\ref{lm:charac-Uasdwr-traps}.
    In fact, let~$\cF$ be a 0-EC. If~$\goal \not \in \cF$, then
    from $\goal$ under some scheduler, $\cF$ is reached with positive probability
    and the run remains almost surely in~$\cF$. So if the first visit to~$\goal$
    does not satisfy the weight constraint, no further visits might be possible under some schedulers.
    If~$\goal \in \cF$, for some schedulers that remain in~$\cF$ the accumulated 
    weight will be identical at each visit to~$\goal$ since~$\cF$ is a 0-EC.
 \end{itemize}
\end{itemize}
To check which of the above cases applies we can use standard polynomial-time
algorithms to compute the minimal expected mean payoff in strongly connected
MDPs and the algorithms to check the existence of 0-ECs
presented in Section \ref{sec:0-ECs} 
(see also Section \ref{sec:algo-checking-0-EC}).
Recall from part (b) of Theorem \ref{universal-neg-wgt-div} 
and Corollary \ref{checking-universal-wgt-div}
that the latter can also be used to check 
universal weight-divergence. Putting things together, this proves Theorem~\ref{thm:DWRU1}.

\medskip

We finish this section with a very similar result applied to the particular special case of Markov chains, which will be useful in the next section.

\begin{lemma}
\label{lem:no-neg-cycle-MC-as}
Let $\varphi$ be a \dwr{}-property and $s$ a state of a Markov chain
$\cC$. Then 
$\Pr_{\cC,s}(\varphi)=1$ if and only if:
\begin{enumerate}
\item [(i)]
  $\Pr^{}_{\cC,s}(\Diamond T)=1$ and
\item [(ii)]
  for each $t\in T \setminus T^*$, 
  there is no path $\fpath$ from $s$ to $t$ 
  containing a state that belongs to a negative cycle; moreover
   the minimal weight of a path from $s$ to $t$ is
  at least $K_t$. 
\end{enumerate}
\end{lemma}

Observe that conditions (i) and (ii) in
Lemma~\ref{lem:no-neg-cycle-MC-as} can be checked in polynomial time:
in particular (ii) reduces to standard shortest-path algorithms in
weighted graphs.  As a consequence, for Markov chains we obtain that
$\Pr_{\cC,s}\big(\Diamond T^* \vee \bigvee_{t\in T \setminus T^*} t \wedge
(\accwgt \geqslant K_t)\big) =1$ can be decided in polynomial time.

\subsection{Almost-Sure Reachability Under Some Scheduler}

\label{Easdwr}
\label{sec:Easdwr-inNP}

\dwreas*

To establish the upper complexity bound of Theorem~\ref{thm:DWR-E1},
we first justify that we can assume without loss of generality that $\cM$ has certain
properties.

\begin{lemma}
\label{lem:Uposdwr-part-case}
The problem \Easdwr{} for general MDPs can be reduced in polynomial time
to the case where instances of \Easdwr{} enjoy the following properties:
\begin{description}
\item [(A1)]  $T^*=\{\good\}$ and
  $T\setminus T^*=\{\goal\}$ with $\good$ and $\goal$ are both traps.
\item [(A2)]
  For all states $s\in S \setminus T^*$, $\Pr^{\max}_{\cM,s}(\Diamond T^*)<1$.
\item [(A3)]
  For all states $s\in S$, $\Pr^{\max}_{\cM,s}(\Diamond T)=1$.
\end{description}
\end{lemma}

\begin{proof}
  Let $\cM$ be an MDP, and
  $\varphi = \bigvee_{t\in T} \Diamond (t \wedge (\accwgt \geqslant
  K_t))$.
  We start by proving that~\textbf{(A1)} is not a real restriction. %
From $\cM$, we build $\cM'$ that extends $\cM$ 
with copies $t'$ of the states $t \in T^*$, an additional state
$\goal$, and new state-action pairs $(t,\tau)$ such that: in case
$t \in T^*$, $P_{\cM'}(t,\tau,t')=1$ and $\wgt_{\cM'}(t,\tau)=0$; and
for $t \in T \setminus T^*$, $P_{\cM'}(t,\tau,\goal)=1$ and
$\wgt_{\cM'}(t,\tau)=-K_t$. The new states ($\goal$ and each $t'$ for
$t \in T^*$) are traps.  We then let
\[
  \varphi' \ \ = \ \  \Diamond (\goal \wedge (\accwgt \geqslant 0) \vee 
  \bigvee_{t\in T^*} \Diamond (t' \wedge (\accwgt \geqslant K_t)) \enspace.
\]
This construction in particular ensures
$\Pr^{\max}_{\cM,\sinit}(\varphi)=1$ iff
$\Pr^{\max}_{\cM',\sinit}(\varphi')=1$.

To justify assumption \textbf{(A2)}, observe that
$\valueEas{\cM,s}=+\infty$ for all states $s$ with
$\Pr^{\max}_{\cM,s}(\Diamond T^*)=1$, so that \Easdwr{} is trivial for
such states.  Moreover, setting
$\widetilde{T}^*=\{s\in S : \Pr^{\max}_{\cM,s}(\Diamond T^*)=1\}$ and
$$
  \widetilde{\varphi}_K \ \ = \ \ 
  \Diamond \widetilde{T}^* \vee \Diamond (\goal \wedge (\accwgt \geqslant K))
$$
we have that for all states $s$ in $\cM$,
$\Pr^{\max}_{\cM,s}\bigl(\varphi_K \bigr) \ \ = \ \
\Pr^{\max}_{\cM,s}\bigl(\widetilde{\varphi}_K \bigr)$. We can thus
safely assume $\widetilde{T}^*=T^*$.

Finally, it is no loss of generality to assume \textbf{(A3)} because
$\Pr^{\max}_{\cM,s}(\Diamond T)<1$ implies
$\Pr^{\max}_{\cM,s}(\varphi_K) <1$ for all $K\in \Integer$ and
therefore $\valueEas{\cM,s}=-\infty$.  Transforming  
$\cM$ into the largest sub-MDP $\widetilde{\cM}$ where the state space is
$\widetilde{S}=\{s\in S : \Pr^{\max}_{\cM,s}(\Diamond T)=1\}$ we get
\[
  \forall K\in \Integer,\ \forall s\in \widetilde{S},\
  \Pr^{\max}_{\cM,s}(\varphi_K) =1 \Longleftrightarrow
  \Pr^{\max}_{\widetilde{\cM},s}(\varphi_K) =1\enspace.
\] 
Therefore, for all states $s\in \widetilde{S}$,
$\valueEas{\cM,s}=\valueEas{\widetilde{\cM},s}$.
\end{proof}
In the rest of this section, we hence assume
$\varphi = \Diamond \good \vee \Diamond (\goal \wedge \accwgt
\geqslant K)$ for trap states $\good$ and $\goal$.

\subsubsection*{Case of MDPs Without Positively Weight-Divergent End Components}
We first address the special case where $\cM$ has no positively
weight-divergent end component.  Thanks to the spider construction
from Section~\ref{sec:spider-construction}, we may assume that $\cM$
has no 0-EC. As a consequence, for all end components $\cE$ of $\cM$,
$\ExpRew{\max}{\cE}(\MP)\leqslant 0$. Let us prove that MD-schedulers
are sufficient for \Easdwr, assuming $\cM$ has no positively
weight-divergent end component.

\begin{lemma}[MD-scheduler suffice if no weight-divergent EC]
 \label{MD-sufficient-eventually-almostsure-max-equals-1}
 Let $\cM$ be an MDP such that $\cM$ has no positively
 weight-divergent end component. Let
 $\varphi = \Diamond \good \vee \Diamond (\goal \wedge \accwgt
 \geqslant K)$ where $\good$ and $\goal$ are traps. If there exists a
 scheduler $\sched$ with $\Pr^{\sched}_{\cM,\sinit}(\varphi)=1$, then
 there exists an MD-scheduler $\tsched$ with
 $\Pr^{\tsched}_{\cM,\sinit}(\varphi)=1$.
    \end{lemma}

\begin{proof}
  Suppose we are given a scheduler $\sched$ with
  $\Pr^{\sched}_{\cM,\sinit}(\varphi)=1$.  
First, we consider the case where
  $\varphi = \Diamond (\goal \wedge \accwgt \geqslant K)$,
  \ie $T^*=\varnothing$. Later we will explain how to adapt the
  proof for the case where $T = \{\goal,\good\}$ with $K_\good=-\infty$.

 Let $\sim$ denote the following equivalence relation on 
 the state space of $\cM$:
 \begin{center}
     $s \sim t$ \qquad iff \qquad
        $s$ and $t$ belong to the same %
        maximal end component of $\cM$ \enspace.
 \end{center}
 A state-action pair $(s,\alpha)$ is called a \emph{progress move} if
 $\alpha \in \Act(s)$ and $P(s,\alpha,t)>0$ for at least one state
 $t$ with $s\not\sim t$. Note that if state $s$ does not belong to any
 end component then all state-action pairs $(s,\alpha)$ with
 $\alpha \in \Act(s)$ are progress moves.
 Moreover, if $(s,\alpha)$ is a progress move then
 there is at least one state $t$ such that
 $\Pr^{\max}_{\cM,t}(\Diamond s) <1$.

 Let $X$ denote the set of state-weight pairs
 $(s,w)\in (S \setminus \{\goal\}) \times \Integer$ such that there is
 some $\sched$-path $\fpath$ from $\sinit$ to $s$ with
 $\wgt(\fpath)=w$ and $s\notin T$.  Let
 $X(s) = \{w\in \Integer : (s,w)\in X\}$. 

 \tudparagraph{1ex}{{\it Claim 1:}} 
   If $X(s) \not= \varnothing$ then $\min X(s)$ exists.
 
   To prove the Claim, suppose by contradiction that
   $\inf X(s)=-\infty$.  By assumption, for each $R\in \Nat$ there
   exists an $\sched$-path $\fpath_R$ from $\sinit$ to $s$ with
   $\wgt(\fpath_R) \leqslant -R$. Let $\cN$ be the MDP that extends
   $\cM$ by a trap state $\goal'$ and the state-action pair
   $(\goal,\tau)$ with $\wgt(\goal,\tau)=-K$ and
   $P(\goal,\tau,\goal')=1$.
 Obviously, the residual scheduler $\residual{\sched}{\fpath_R}$
 can be extended to a scheduler $\vsched_R$ for $\cN$ 
 such that
 \[
   \Pr^{\vsched_R}_{\cN,s}
     \bigl(\, \Diamond (\goal' \wedge (\accwgt \geqslant R)) \, \bigr)
   \ \ = \ \ 1 \enspace.
 \]
 This holds for every $R \in \Nat$,
 thus %
 $\cN$ has a positively weight-divergent end component $\cE$; a
 contradiction. %
 This completes the proof of Claim~1.

\smallskip

 Let $U=\{s\in S : X(s)\not= \varnothing\}$.  For $s\in U$ we define
 $w_s =\min X(s) \in \Integer$.  Let $A(s)$ denote the set of
 actions $\alpha \in \Act(s)$ such that $\sched(\fpath) = \alpha$ for
 at least one $\sched$-path $\fpath$ from $\sinit$ to $s$ with
 $\wgt(\fpath)=w_s$.
  
 Let now $Y_0$ be the set of all states $s\in U$ such that
 $A(s)$ contains at least one action $\alpha_s$ where 
 $(s,\alpha_s)$ is a progress move.
 We then inductively define $Y_{i+1}$ as the set of states
 $s\in U \setminus (Y_0\cup \ldots \cup Y_i)$ where
 $A(s)$ contains at least one action $\alpha_s$ such that
 $P(s,\alpha_s,t)>0$ for at least one state $t\in Y_i$.
 Clearly, there is some $j$ such that $Y_{j+1}$ is empty.
 Let 
 $Y  \ =  \ Y_0 \cup Y_1 \cup \ldots \cup Y_j$.
 Then, $Y \subseteq U$.

 \tudparagraph{1ex}{{\it Claim 2:}} $Y = U$.
 
 To prove this claim, we again suppose by contradiction that $Y$ is a
 proper subset of $U$.  Consider the sub-MDP $\widetilde{\cM}$ of
 $\cM$ induced by the state-action pairs $(s,\alpha)$ with
 $s\in U\setminus Y$ and $\alpha \in A(s)$.  Note that
 $s\in U\setminus Y$ and $\alpha \in A(s)$ implies $t\in U\setminus Y$
 for all states $t$ with $P(s,\alpha,t)>0$.  The residual schedulers
 $\residual{\sched}{\fpath}$ for the paths $\fpath$ from $\sinit$ to
 $s$ with $\wgt(\fpath)=w_s$ and $s\in U\setminus Y$ can be viewed as
 schedulers for $\widetilde{\cM}$.  Hence:
  $$
    \Pr^{\residual{\sched}{\fpath}}_{\widetilde{\cM},s}
       \bigl(\ \Box (U\setminus Y) \ \bigr)
    \ \ = \ \ 1\enskip.
  $$
  As $\goal \notin U$ (recall that $U$ consists of non-trap states,
  whereas $\goal$ is a trap) we get
  $\Pr^{\sched}_{\cM,\sinit}(\varphi)<1$, a contradiction, so that
  $U = Y$.  

\smallskip

To conclude, for the states $s\in U$ we can now pick some action
$\alpha_s \in A(s)$ such that either $s\in Y_0$ and $(s,\alpha_s)$ is
a progress move or $s\in Y_{i+1}$ and $P(s,\alpha_s,t)>0$ for some
state $t\in Y_i$.  Let $\tsched$ be a memoryless scheduler such that
$\tsched(s)=\alpha_s$ for $s\in U$ (and defined arbitrary from other
states).  Then, $ \Pr^{\tsched}_{\cM,\sinit}(\varphi) = 1$ and hence, 
for the case $T^* = \emptyset$ we are done with the proof of
Lemma~\ref{MD-sufficient-eventually-almostsure-max-equals-1}.

\bigskip

Suppose now that $T^*=\{\good\}$, and
$\varphi = \Diamond \good \vee \Diamond (\goal \wedge \accwgt \geq
K)$.  We pick an MD-scheduler $\usched$ such that
$\Pr^{\usched}_{\cM,s}(\Diamond \good)=p_s$ where for each state $s$,
$p_s=\Pr^{\max}_{\cM,s}(\Diamond \good)$.  We may assume w.l.o.g.~that
$\sched$ behaves as $\usched$ whenever a state $s$ with $p_s=1$ has
been reached.

The definition of the sets $X$, $X(s)$ is as before.  In Claim 1, we
show that $\min X(s)$ exists for all states $s$ where $p_s <1$.  The
argument is again by contradiction using paths $\fpath_R$ as above.
As
$\Pr^{\residual{\sched}{\fpath_R}}_{\cM,s}(\Diamond \good)\leqslant
p_s <1$ we have:
  \[
    \Pr^{\residual{\sched}{\fpath_R}}_{\cM,s} \Bigl( \bigvee_{t\in T
      \setminus T^*} (\Diamond (t \wedge (\accwgt \geqslant K_t))) \
    \Bigr) \ \ \geqslant \ \ 1-p_s \enspace.
  \] 
  The MDP $\cN=\cN_s$ is defined as above, with a $\tau$-transition
  from $\goal$ to the fresh state $\goal'$, while $\good$ has a
  $\tau$-transition to $s$, with weight small enough to ensure that
  $\cN$ has no positively weight-divergent end component.
  \begin{enumerate}
  \item [] Let us briefly explain how to find the value
    $\wgt(\good,\tau) \in \Integer$ so that the constructed MDP $\cN$
    has no positively weight-divergent end component.  Suppose $\cN$
    has an end component that contains $\good$.  Let $\cE$ be the
    maximal end component of $\cN$ that contains $\good$, and view it as an MDP.
    Obviously, we have $\Pr^{\max}_{\cE,s}(\Diamond \good)=1$ and
    $\cE$ is a sub-MDP of $\cM$.  (Notice that none of the new
    state-action pairs $(t,\tau)$ for $t\in \{\goal,\good\}$ is
    contained in $\cE$.)  As $\cM$ has no positively weight-divergent
    end component, so does $\cE$. Thus, 
    the maximal expected accumulated weight until reaching $\good$ is
    finite.  Let $E= \ExpRew{\max}{\cE,s}(\text{``$\wgt$ until
      $\good$''})$.  We then may define $\wgt(t,\tau)$ as any value
    that is smaller than $-E$.  This ensures that
    $\ExpRew{\usched}{\cE,s}(\MP) <0$ for each MD-scheduler $\usched$
    for $\cE$ with a single BSCC $\cB$ and $(\good,\tau)\in \cB$
    (Lemma \ref{lemma:scMC-exp-mp-quotient}).  But then
    $\ExpRew{\max}{\cE}(\MP)<0$.  Hence, $\cE$ is not positively
    weight-divergent (Corollary \ref{wgt-div-implies-non-neg-mp}).
  \end{enumerate} 
  Scheduler $\vsched_R$ for $\cN$ is now defined as follows.  Starting
  in state $s$, $\vsched_R$ first behaves as
  $\residual{\sched}{\fpath_R}$ until reaching a state
  $t\in \{\good,\goal\}$.
  \begin{itemize}
  \item If $t = \goal $ then $\vsched_R$ schedules $\tau$ and moves to
    $\goal'$.
  \item If $t = \good$ then $\vsched_R$ takes the $\tau$-transition
    back to state $s$. If $R'$ is the weight of the (complete) path
    that $\vsched_R$  has generated then $\vsched_R$ behaves now
    as $\residual{\sched}{\fpath_{R'}}$ until reaching again a state
    $t\in \{\good,\goal\}$.
  \end{itemize}
  As the probablities to generate a path from $s$ to $\good$
under all residual schedulers 
  $\residual{\sched}{\fpath_R}$ is at least $1{-}p_s$
  we obtain:
\[
  \Pr^{\vsched_R}_{\cM,s} \bigl(\, \Diamond (\goal' \wedge (\accwgt
  \geqslant R)) \, \bigr) \ \ = \ \ 1 \enspace.
\]
This is impossible as $\cN$ has no positively weight-divergent end
component.  The remaining argument and proof of Claim 2 is the same as
for the case $T^*=\varnothing$. This concludes the proof of
Lemma~\ref{MD-sufficient-eventually-almostsure-max-equals-1} in the
general case.
\end{proof}

As a consequence of
Lemma~\ref{MD-sufficient-eventually-almostsure-max-equals-1}, one can
decide in \NP\ the problem \Easdwr{} by guessing an MD-scheduler
$\sched$ and checking that it ensures
$\Pr^\sched(\Diamond \good \vee \Diamond (\goal \wedge \accwgt
\geqslant K)) =1$. This would also yield an \EXPTIME-algorithm to
compute the values $\valueEas{\cM,s}$. We now give a better
alternative in terms of complexity, both to answer the decision
problem \Easdwr{} and to compute the values $\valueEas{\cM,s}$, by using a
reduction to mean-payoff games.

\begin{theorem}
 \label{complexity-DWR-E1-no-wgt-div-EC}
 If $\cM$ has no weight-divergent end components then the problem
 \Easdwr{} is in $\NP \cap \coNP$.  The values $\valueEas{\cM,s}$ can be
 computed in pseudo-polynomial time.
\end{theorem}

\begin{proof}
  The complexity upper bound for the decision problem is established
  through a polynomial-time reduction to mean-payoff games,
  illustrated on Figure~\ref{fig:Easdwr-partcase-red-mp}. 
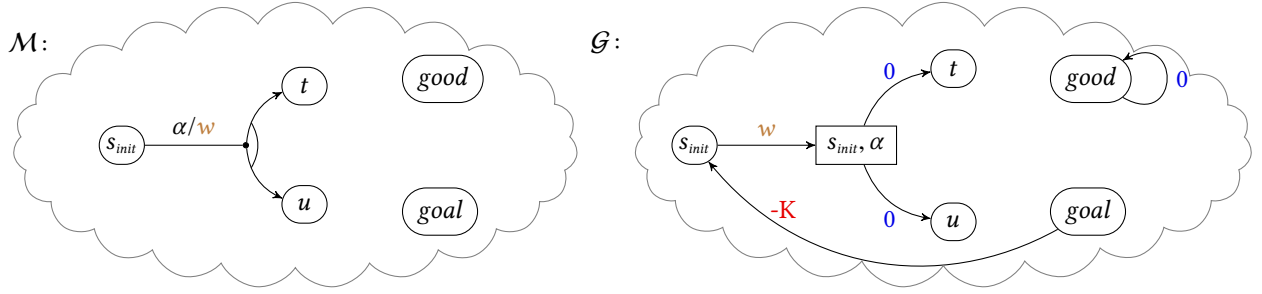
\begin{figure}

\begin{tikzpicture}
	\node[state] (sinit) {$s_\init$};
	\node[istate, above left=8mm and 7mm of sinit] {$\cM\colon$};%
	\node[bullet, right=13mm of sinit] (b) {};
	\node[state, above right=5mm and 7mm of b] (t) {$t$};
	\node[state, below right=5mm and 7mm of b] (u) {$u$};

	\node[state, above right=3mm and 40mm of sinit] (good) {$\good$};
	\node[state, below right=3mm and 40mm of sinit] (goal) {$\goal$};

	\path%
		(sinit) edge[ntran] node[above] {$\alpha$/$\neuw{w}$} (b)
		(b) edge[ptran, bend left] coordinate[pos=0.3](bt) (t)
		(b) edge[ptran, bend right] coordinate[pos=0.3](bu) (u);	
	\draw%
		(bt) to[bend left] (bu);

	\node[state, right=70mm of sinit] (gsinit) {$\sinit$};
	\node[rstate, right=13mm of gsinit] (gb) {$\sinit,\alpha$};
	\node[state, above right=5mm and 7mm of gb] (gt) {$t$};
	\node[state, below right=5mm and 7mm of gb] (gu) {$u$};

	\node[istate, above left=8mm and 7mm of gsinit] {$\cG\colon$};%
	\node[state, above right=3mm and 50mm of gsinit] (ggood) {$\good$};
	\node[state, below right=3mm and 50mm of gsinit] (ggoal) {$\goal$};

	\path%
		(ggood) edge[ptran,loop, min distance=9mm, out=-30, in=30] node[right]{\nilw} (ggood)
		(ggoal) edge[ptran, bend left=40] node[above,pos=.75]{\negw{K}} (gsinit)	
		(gsinit) edge[ptran] node[above] {$\neuw{w}$} (gb)
		(gb) edge[ptran, bend left] node[above]{\nilw} (gt)
		(gb) edge[ptran, bend right] node[below]{\nilw} (gu);

	\coordinate[right=20mm of sinit] (c);
	\node[cloud, cloud puffs=23,cloud puff arc=120, cloud ignores aspect, draw=gray,
			minimum width=75mm, minimum height=38mm](cloud) at (c) {};
	\coordinate[right=30mm of gsinit] (gc);
	\node[cloud, cloud puffs=23,cloud puff arc=120, cloud ignores aspect, draw=gray,
			minimum width=82mm, minimum height=38mm](cloud) at (gc) {};
\end{tikzpicture}
 \caption{Reduction of \Easdwr{} to mean-payoff games, assuming no
  weight-divergent EC.}
\label{fig:Easdwr-partcase-red-mp}
\end{figure}

From $\cM$, we build a two-player mean-payoff game $\cG$, in which
player~1 simulates the choices of the scheduler, player~2 is
responsible for the probabilistic choices, and the weight of a
transition by player~1 coincides with the weight in $\cM$. Moreover,
$\cG$ is extended with two transitions: a self-loop on $\good$ with
weight $0$, and an edge from $\goal$ to $\sinit$ with weight $-K$.

The construction ensures
\[
  \exists \sched,\ \Pr^\sched_{\cM,\sinit}(\Diamond \good \vee
  \Diamond \goal \wedge \accwgt \geqslant K) = 1 \ \
  \Longleftrightarrow \ \ \exists \stratone,\ \forall \strattwo,\
  \play{\cG}{\stratone}{\strattwo} \models \MP \geq 0 \enspace.
\]

$(\Longrightarrow)$ To prove the left-to-right direction, we pick
 a scheduler $\sched$ that satisfies
$\Pr^\sched_{\cM,\sinit}(\Diamond \good \vee \Diamond \goal \wedge
\accwgt \geqslant K) =1$. By
Lemma~\ref{MD-sufficient-eventually-almostsure-max-equals-1}, we may
assume that $\sched$ is an MD-scheduler. It trivially induces an MD-strategy 
$\stratone$ in $\cG$, which mimics $\sched$ in all states
except $\goal$, and moves back to $\sinit$ from $\goal$.

\noindent Thanks to Lemma~\ref{lem:no-neg-cycle-MC-as}, along all paths from
$\sinit$ to $\goal$ in the Markov chain $\cC$ induced by $\sched$, no
state belongs to a negative cycle of $\cC$. From the hypothesis that
for all $s \neq \good$, $\Pr^{\max}_{\cM,s}(\Diamond \good) <1$, we
derive that all states (except $\good$) have a path to $\goal$ in
$\cC$. Therefore, $\cC$ has no negative cycle on the way to $\goal$.
Moreover, the weight $-K$ of the transition from $\goal$ to $\sinit$
was chosen so that all cycles around $\sinit$ have nonnegative
weight.  Thus, for all strategy $\strattwo$ of player~2, the play
induced by $\stratone$ and $\strattwo$ has nonnegative mean payoff.

$(\Longleftarrow)$ For the other direction, we let $\stratone$ be an
MD-strategy winning for player~1 in $\cG$. Here it also trivially
defines a scheduler $\sched$ for $\cM$. We first show that $\sched$
ensures
$\Pr^\sched_{\cM,\sinit}(\Diamond \good \vee \Diamond \goal)=1$, by
contradiction. Then, there exists an EC $\cE$ such that
$\Pr^\sched_{\cM,\sinit}(\Diamond \Box \cE )>0$. Since all EC of $\cM$
have negative expected mean payoff, this holds in particular for
$\cE$, and thus $\sched$ induces at least a path with negative
mean-payoff. Now, assume that the weight constraint in $\goal$ is not
met, \ie there is an $\sched$-path reaching $\goal$ with
accumulated weight at most $K{-}1$. Iterations of this path followed
by the transition from $\goal$ to $\sinit$ yields a play in $\cG$
under $\stratone$ which had negative mean-payoff, a contradiction with
the fact that $\stratone$ is winning. Thus, all $\sched$-paths to
$\goal$ reach it with accumulated weight at least $K$. All in all,
$\sched$ ensures that $\varphi$ holds almost surely.

For the computation of the values, $\valueEas{\cM,s} = \infty$ if
$\Pr_{\cM,s}^{\max}(\Diamond \good) = 1$, 
and $\valueEas{\cM,s} = -\infty$ if
$\Pr_{\cM,s}^{\max}(\Diamond (\good \lor \goal)) < 1$.
Otherwise,
for each state~$s\in S$ we have that 
$\valueEas{\cM,s} \in \{-\infty\} \cup [|S|\cdot W_{\min}, |S|\cdot W_{\max}]$,
where~$W_{\min}$ is the minimal weight and~$W_{\max}$ is the maximal weight
that appears in~$\cM$.
In fact, MD-schedulers suffice to achieve a given threshold~$K$, and
an MD-scheduler~$\sched$ satisfying
$\Pr^\sched_{\cM,s}\big(\Diamond \good \vee
  \Diamond (\goal \wedge \accwgt \geqslant K)\big) = 1$ cannot induce
any negative cycle in~$\cM$. Thus, any threshold~$K$ that can be ensured
by~$\sched$ must be equal to or greater than the weight of a simple shortest path
from~$s$ to~$\goal$, which is at least~$|S|\cdot W_{\min}$.
Conversely, under any MD-scheduler~$\sched$, $\goal$ (which is a trap state) is reachable following a simple path
so a threshold~$K$ with $\Pr^\sched_{\cM,s}\big(\Diamond \good \vee
\Diamond (\goal \wedge \accwgt \geqslant K)\big) = 1$ cannot be larger than $|S|\cdot W_{\max}$.
To compute the values, we can thus run a binary search in this interval 
with $\log(|S|\cdot W)$ calls to
a pseudo-polynomial mean-payoff solver, where~$W=W_{\max} - W_{\min}$.
The binary search either determines
a finite value, or it returns that the value must be less than~$|S|\cdot W_{\min}$ in which case
the value is~$-\infty$.
\end{proof}

To prepare the proof of the general case (when $\cM$ may have
weight-divergent end components), we provide a characterization of the
different cases that arise for $\valueEas{\cM,s}$: whether the value
is $-\infty$, finite, or $+\infty$. To do so, we introduce a weighted
directed graph associated with $\cM$ and an MD-scheduler $\sched$.
Let $G^{\sched}_s$ denote the weighted directed graph where the
vertex set consists of all states $u$ in $\cM$ that belong to
$\sched$-path from $s$ to $\goal$.  The edge relation in
$G^{\sched}_s$ is given by $u \to u'$ if $P_{\cM}(u,\sched(u),u')>0$.
The weight of the edge $u \to u'$ is $\wgt_{\cM}(u,\sched(u))$.

\begin{lemma}
 \label{KsE1-pm-infty-or-integer}
 Let $\cM$ be an MDP with no positively weight-divergent end
 components that satisfies assumptions \textbf{(A1)}, \textbf{(A2)}
 and \textbf{(A3)} and let $s$ be a state of $\cM$. Then:
  \begin{enumerate}
  \item [(a)] $\valueEas{\cM,s}=+\infty$ iff $s = \good$
  \item [(b)] $\valueEas{\cM,s} \in \Integer$ iff $s\neq \good$ and
    there exists an MD-scheduler $\sched$ such that the graph
    $G^{\sched}_s$ does not contain any negative cycle.
  \item [(c)] $\valueEas{\cM,s}=-\infty$ iff $s\neq \good$ and there
    is no MD-scheduler $\sched$ satisfying the condition stated in
    (b).
  \end{enumerate}
\end{lemma}

\begin{proof}
   Statement (a) is trivial.  
   Statements (b) and (c) are easy consequences
   of Lemma \ref{MD-sufficient-eventually-almostsure-max-equals-1}
   and Lemma \ref{lem:no-neg-cycle-MC-as}.
\end{proof}

\subsubsection*{General Case}
We now consider the general case where $\cM$ might have positively
weight-divergent end components, still assuming \textbf{(A1)},
\textbf{(A2)} and \textbf{(A3)} to hold.

We first observe that states that belong to the same positively
weight-divergent end component have the same truth values for \Easdwr{}
and the same values for the corresponding optimization problem. More
precisely:

\begin{lemma}[Same values for states in weight-divergent ECs]
 \label{same-value-wgt-div-EC}
  Let $\cE$ be a positively weight-divergent end component of $\cM$.
  Then:
  \begin{enumerate}
  \item [(a)]
    For all states $s,s'\in \cE$, 
    $\Pr^{\max}_{\cM,s}(\varphi)=1$ iff\ \  $\Pr^{\max}_{\cM,s'}(\varphi)=1$ 
    and
    $\Pr^{\max}_{\cM,s}(\varphi)>0$ iff\ \  $\Pr^{\max}_{\cM,s'}(\varphi)>0$.
  \item [(b)] There exists
    $\valueEas{\cM,\cE} \in \{+\infty,-\infty\}$ such that for all
    states $s \in \cE$ we have $\valueEas{\cM,s} = \valueEas{\cM,\cE}$.
  \item [(c)] If $\valueEas{\cM,\cE}=+\infty$ and $s$ is a state with
    $\Pr^{\max}_{\cM,s}(\Diamond \cE)=1$, then
    $\valueEas{\cM,s}=+\infty$.
  \end{enumerate}
\end{lemma}

\begin{proof}
 Given two states $s,s'\in \cE$ and a natural number $R\in \Nat$,
  since $\cE$ is positively weight-divergent 
  there exists a scheduler $\vsched_R$ for $\cE$ such that
  $\Pr^{\vsched_R}_{\cE,s}\big(\Diamond (s' \wedge (\accwgt \geqslant
  R))\big)=1$.  Hence, for each scheduler $\sched$ for $\cM$ and each
  $K\in \Integer$ there is a scheduler $\sched_R$ for $\cM$ such that
  $\Pr^{\sched}_{\cM,s}(\varphi_K)=\Pr^{\sched_R}_{\cM,s'}(\varphi_{K+R})$.
  With $R=0$, we obtain statement (a) and
  $\valueEas{\cM,s}=\valueEas{\cM,s'}$. %
  With $s=s'$ and letting $R$ tend to $\infty$, we obtain that the
  value $\valueEas{\cM,s}$ %
  cannot be finite.  Thus,
  $\valueEas{\cM,s}, \valueEpos{\cM,s} \in \{\pm \infty\}$.  This
  yields statement (b).

  To prove statement (c),
  we pick a family of schedulers 
  $(\sched_K)_{K\in \Integer}$ such that
  $\Pr^{\sched_K}_{\cM,u}(\varphi_K)=1$ for all 
  states $u$ in $\cE$ and all integers $K$.
  Let now $s$ be a state such that
  $\Pr^{\sched}_{\cM,s}(\Diamond \cE)=1$.
  Given $R\in \Integer$, we
  now compose $\sched$ and the schedulers $\sched_K$ to obtain
  a scheduler $\tsched_R$ with $\Pr^{\tsched_R}_{\cM,s}(\varphi_R)=1$.
  $\tsched_R$ first behaves as $\sched$ until reaching a state in $\cE$.
  As soon as some state $u$ of $\cE$ is reached, say along a path
  $\fpath$ from $s$ to $u$ with $\wgt(\fpath)=w$, $\tsched$ switches
  mode and behaves as $\sched_{R-w}$ from then on.
\end{proof}

For solving the general case, let us now explain how to use the particular case of MDPs 
without positively weight-divergent end components. From $\cM$ we construct another $\cN$ by intuitively
replacing each maximal weight-divergent end component $\cE$ with an
entry state $\entrystate{\cE}$ and an exit state $\exitstate{\cE}$
and a transition from $\entrystate{\cE}$ to $\exitstate{\cE}$ with
``sufficiently high'' weight. $\cN$ has no positively weight-divergent
end components, however the values $\valueEas{\cN,s}$ are only a lower
bound on $\valueEas{\cM,s}$. In particular, it may happen that
$\valueEas{\cN,s} = -\infty$ and $\valueEas{\cM,s} = +\infty$. To
remedy this problem, we will see how to identify end components with
value $+\infty$. This is performed by a fixed-point computation of the
``good'' end components. All details of these steps are provided in
the remainder of this section.

\bigskip

To formally define $\cN$ we introduce some notations.  Let
$\cE_1,\ldots,\cE_k$ be the maximal end components of $\cM$ that are
positively weight-divergent. Recall that these can be computed in
polynomial time by first computing the maximal end components of
$\cM$ and then checking whether each of them is weight-divergent
thanks to Theorem~\ref{weight-divergence-algorithm} from
Section~\ref{sec:wgt-div}.

Let $\WDMEC$ consist of all states in $\cM$ that are contained in one
of the weight-divergent maximal end components $\cE_1,\ldots,\cE_k$.
Since $\good$ and $\goal$ are traps, the maximal end components of
$\cM$ do not contain them, and
$\{\good,\goal\} \cap \WDMEC = \varnothing$. For simplicity, we assume
that the action sets $\Act_{\cM}(s)$ are pairwise disjoint, and %
we write $\wgt^{\min}_{\cM} \in \Integer$ for the minimal weight
assigned to some state-action pair in $\M$. %
We have now all ingredients to precisely define $\cN$.
\begin{description}
\item[state space]
  $S_{\cN}  = 
  \bigl(S_{\cM} \setminus \WDMEC \bigr)
  \ \cup \ 
  \bigl\{\, \entrystate{\cE},\exitstate{\cE} \ : \ 
            \cE \in \{\cE_1,\ldots,\cE_k\} \ 
  \bigr\}$;
\item[action set] $\Act_{\cN} = \Act \, \cup \, \bigl\{\tau \bigr\} $,
  with $\tau \notin \Act$ a fresh action symbol;
\item[transitions]~
\begin{itemize}
\item If $(s,\alpha)$ is a state-action pair in $\cM$ where
  $s\in S_\cM \setminus \WDMEC$, then $(s,\alpha)$ is a state-action
  pair of $\cN$ with the same weight; moreover for each state
  $s'\in S_{\cM} \setminus \WDMEC$,
  $P_{\cN}(s,\alpha,s')=P_{\cM}(s,\alpha,s')$ and for each end
  component $\cE \in \{\cE_1,\ldots,\cE_k\}$ we define
  $P_{\cN}(s,\alpha,\entrystate{\cE})
  =P_{\cM}(s,\alpha,\cE)$.\footnote{The notation
    $P_{\cM}(s,\alpha,\cE)$ stands for the probability from $s$ to
    reach any state of $\cE$.}
\item Each state-action pair $(s,\alpha)$ in $\cM$ where
  $s \in \cE \subseteq \WDMEC$%
  and such that for some state $s'$ outside $\cE$,
  $P_{\cM}(s,\alpha,s')>0$ is turned into a state-action pair
  $(\exitstate{\cE},\alpha)$ in $\cN$; the weight of the state-action
  pair $(\cE,\alpha)$ in $\cN$ coincides with the weight of
  $(s,\alpha)$ in $\cM$; moreover the transition probabilities are
  defined as follows: for each state $s'\in S_{\cM} \setminus \WDMEC$ we set
  $ P_{\cN}(\exitstate{\cE},\alpha,s') =
  \frac{P_{\cM}(s,\alpha,s')}{1{-}P_{\cM}(s,\alpha,\cE)}$, for each
  maximal weight-divergent end component $\cF\neq \cE$ we set
  $ P_{\cN}(\exitstate{\cE},\alpha,\entrystate{\cF}) = \sum_{s'\in
    \cF} \frac{P_{\cM}(s,\alpha,s')}{1{-}P_{\cM}(s,\alpha,\cE)} $, and
  $P_{\cN}(\cE,\alpha,\cE)=0$.
\item Last, for each maximal weight-divergent end component $\cE$ in
  $\cM$, $\cN$ contains a state-action pair $(\entrystate{\cE},\tau)$
  for some fresh action symbol $\tau$ with
  $P_{\cN}(\entrystate{\cE},\tau,\exitstate{\cE})=1$ and
  $\wgt_{\cN}(\entrystate{\cE},\tau) = \omega$, with
  $ \omega = \max\, \{ \, 0, \, - (|S_{\cN}|{-}1
  ) \cdot \wgt^{\min}_{\cM} \, \} $.
\end{itemize}
\end{description}

An example of transformation from $\cM$ to $\cN$ is provided by 
Example~\ref{example:dwr} in the core of the paper.

The representation of all states belonging to the same maximal
weight-divergent end component $\cE$ by the states $\entrystate{\cE}$
and $\exitstate{\cE}$ is motivated by part (a) of Lemma
\ref{same-value-wgt-div-EC}, which expresses that states in the same
WDMEC have the same truth value for \Easdwr. Intuitively, the
$\tau$-transition from $\entrystate{\cE}$ to $\exitstate{\cE}$ serves
to mimic all state-action pairs $(s,\alpha)$ with $s\in \cE$ and
$P_{\cM}(s,\alpha,\cE)=1$.

\begin{lemma}[Simple properties of the new MDP]
\label{simple-properties-MDP-N} 
$\cN$ satisfies assumptions \textbf{(A1)}, \textbf{(A2)} and
\textbf{(A3)}. %
Moreover:
  \begin{enumerate}
  \item [(a)] None of the states $\entrystate{\cE}$, $\exitstate{\cE}$
    for $\cE\in \{\cE_1,\ldots,\cE_k\}$ belong to an end component of
    $\cN$.
  \item [(b)] 
     $\cN$ has no positively weight-divergent end component. 
   \item [(c)] None of the states $\entrystate{\cE}$,
     $\exitstate{\cE}$ for $\cE \in \{\cE_1,\ldots,\cE_k\}$ belongs to
     a simple negative cycle in $\cN$.
  \end{enumerate} 
\end{lemma}

\begin{proof}
  Given that $\cM$ satisfies assumptions \textbf{(A1)}, \textbf{(A2)}
  and \textbf{(A3)}, it is easy to see that $\cN$ satisfies these assumptions, too.

  Statements (a) and (b) follow from the fact that each end component
  $\cE$ of $\cN$ is an end component of $\cM$ that is not positively
  weight-divergent.
  Under the assumption that $\wgt^{\min}_{\cM}\geqslant 0$, statement
  (c) is trivial, since in this case all weights in $\cM$ %
  are nonnegative, and also in $\cN$ because all $\tau$ transitions
  then have weight $\omega = 0$. 

  \noindent Assume now that $\wgt^{\min}_{\cM} < 0$, in which case
  the weight of $\tau$-transitions satisfies
  $\omega= -(|S_{\cN}|{-}1) \cdot \wgt^{\min}_{\cM} <0$. Let $\cycle$
  be a simple cycle in $\cN$ that contains one of the states
  $\entrystate{\cE}$ or $\exitstate{\cE}$ for some
  $\cE \in \{\cE_1,\ldots,\cE_k\}$. Necessarily, $\cycle$ then also
  contains the $\tau$-transition from $\entrystate{\cE}$ to
  $\exitstate{\cE}$ (otherwise $\cycle$ could not enter
  $\exitstate{\cE}$ and could not leave $\entrystate{\cE}$).

  Regarding $\cycle$ as a sequence of state-action pairs, we let
  $\rho$ denote the sequence $(s_1,\alpha_1)\ldots (s_m,\alpha_m)$ of
  state-action pairs in $\cN$ that results from $\cycle$ by removing
  the state-action pairs $(\entrystate{\cF},\tau)$ for
  $\cF\in \{\cE_1,\ldots,\cE_k\}$. %
  In particular, $\{s_1,\ldots,s_m\}\subseteq S_{\cM}\setminus \WDMEC$
  and all state-action pairs $(s_i,\alpha_i)$ in $\rho$ belong to
  $\cM$.  Since $\cycle$ is a simple cycle, we derive
  $m \leqslant |S_{\cN}|{-}2$.  As a consequence, we get
  \[
   \wgt_{\cN}(\cycle) 
    \geqslant 
   \sum_{i=1}^m \wgt_{\cM}(s_i,\alpha_i) \ + \ 
   \wgt_{\cN}(\entrystate{\cE},\tau)
 \geqslant 
   m \cdot \wgt^{\min}_{\cM} \ + \ 
   \wgt_{\cN}(\entrystate{\cE},\tau)
    \geqslant 
    (N{-}2) \cdot \wgt^{\min}_{\cM} \ - \ 
    (N{-}1) \cdot \wgt^{\min}_{\cM} 
    \ \ = \ \  - \wgt^{\min}_{\cM} \ \ > \ \ 0 \enspace.
\]
Thus $\cycle$ is not a negative cycle, and the proof of statement (c) is complete.
\end{proof}

To establish a relation between values in $\cM$ and in $\cN$, we first
assign states of $\cM$ a corresponding state in $\cN$. For
$s \in \cM$, $s_\cN$ is defined as $s$ if
$s \in S_\cM \setminus \WDMEC$, and otherwise $s = \exitstate{\cE}$ if
$s$ belongs to the MEC $\cE$ that is positively weight-divergent. %
Furthermore, we define
$ \WDMEC_{\cN} = \bigl\{\entrystate{\cE},\exitstate{\cE} : \cE \in
\{\cE_1,\ldots,\cE_k\}\bigr\}$, as the set of entry and exit states in $\cN$.

The values for states in $\cN$ is defined analogously to values in
$\cM$: for $u \in \cN$, $\valueEas{\cN,u}$ is the supremum over all
integers $K$ such that $\Pr^{\max}_{\cN,u}(\varphi_K)=1$. On the one
hand recall from Lemma~\ref{same-value-wgt-div-EC} that for each
$\cE \in \{\cE_1,\ldots,\cE_k\}$ all states in $\cE$ share the same
value $\valueEas{\cM,\cE} \in \{+\infty,-\infty\}$. On the other hand,
as $\cN$ satisfies assumptions \textbf{(A1)}, \textbf{(A2)} and
\textbf{(A3)}, has no positively weight-divergent end components,
and (see Lemma \ref{simple-properties-MDP-N}) we can apply
Lemma~\ref{KsE1-pm-infty-or-integer}, we obtain that
$\valueEas{\cN,\exitstate{\cE}} \in \{-\infty\} \cup \Integer$. Thus,
we cannot expect that $\valueEas{\cM,s}$ and
$\valueEas{\cN,\stateN{s}{\cN}}$ agree.  Nevertheless, we will prove
that the values $\valueEas{\cM,s}$ can be derived from the values
$\valueEas{\cN,\stateN{s}{\cN}}$.

We start with the following observation that relates the values
in $\cM$ and $\cN$ 
where we switch to a different objective for $\cN$.

\begin{lemma}
\label{Tstar-N}
Let $s$ and $K \in \Integer$ be such that
$\Pr^{\max}_{\cM,s}(\varphi_K)=1$. Then, for
$T^*_{\cN}\ = \ T^* \, \cup \, \bigl\{\, \exitstate{\cE} \, : \,
 \valueEas{\cM,\cE}=+\infty \, \bigr\}$ it holds
  \[
   \Pr^{\max}_{\cN,\stateN{s}{\cN}}
     \bigl(\, 
       \Diamond T^*_{\cN} \, \vee \, 
       \Diamond (\goal \wedge (\accwgt \geqslant K))\,
     \bigr)\ = \ 1 \enspace.
  \]
\end{lemma}

\begin{proof}
  If the given state $s$ of $\cM$ belongs to a maximal
  weight-divergent end component of $\cM$, then
  $\stateN{s}{\cN}\in T^*_{\cN}$ and the claim is obvious.  Suppose
  now that $s \notin \WDMEC$, in which case $s$ is also a state of
  $\cN$.

  We pick an MD-scheduler $\usched$ for $\cN$ that maximizes the
  probabilities to reach $T^*$ from every state in $\cN$.  Let
  $V=\{v\in S_{\cN}:\Pr^{\max}_{\cN,v}(\Diamond T^*)=1\}$.  Clearly,
  we have $\Pr^{\sched}_{\cN,v}(\Diamond T^*)=1$ for every state
  $v\in V$.  Consider now any scheduler $\tsched$ for $\cN$ such that:
 \begin{itemize}
 \item
   $\tsched(\fpath)=\sched(\fpath)$
   for each finite path $\fpath$ in $\cN$ with
   $\last(\fpath)\notin V$ consisting of states that do not belong to
   $\WDMEC_{\cN}$.
 \item
   $\tsched(\fpath)=\usched(\last(\fpath))$ if
   $\last(\fpath) \in V$.
 \end{itemize}
(The behavior of $\tsched$ for the paths $\fpath$ that contain
 a state in $\WDMEC_{\cN}$ is irrelevant.)

 Obviously,
 $\Pr^{\tsched}_{\cN,\stateN{s}{\cN}}(\Diamond \entrystate{\cE})>0$
 implies $\Pr^{\sched}_{\cM,s}(\Diamond \cE)>0$.  But then there is an
 $\sched$-path $\fpath$ from $s$ to some state $u$ in $\cE$.  For the
 residual scheduler, we have
 $ \Pr^{\residual{\sched}{\fpath}}_{\cM,u} \bigl(\,
 \varphi_{K-\wgt(\fpath)} \, \bigr) = 1 $.  In particular,
 $\Pr^{\max}_{\cM,u}(\varphi_{K-\wgt(\fpath)})=1$.  Therefore,
 $\valueEas{\cM,\cE}=+\infty$ by part (b) of
 Lemma~\ref{same-value-wgt-div-EC}.  But then
 $\exitstate{\cE}\in T^*_{\cN}$.
 Moreover, each $\tsched$-path $\fpath$ from $\stateN{s}{\cN}$ to $\goal$
 that does not enter a state in $\WDMEC_{\cN}$
 is a $\sched$-path from $s$ to $\goal$.
 Hence, $\wgt(\fpath)\geqslant K$.
\end{proof} 

To simplify our notations, we write
$\valueEas{\cN,\cE}$ rather than $\valueEas{\cN,\exitstate{\cE}}$
for a given $\cE\in \{\cE_1,\ldots,\cE_k\}$.
Recall that by assumption \textbf{(A2)} we have
$\valueEas{\cN,\exitstate{\cE}} \in \Integer \cup \{-\infty\}$.  The
idea is now to identify the end components
$\cE\in \{\cE_1,\ldots,\cE_k\}$ of $\cM$ where
$\valueEas{\cM,\cE}=+\infty$.  For this, we observe:

\begin{lemma}
 \label{K_N-2-K-M}
 For each scheduler $\tsched$ for $\cN$ and each $K\in \Integer$,
 there is a scheduler $\sched$ for $\cM$ such that for all states $s$
 in $\cM$:
  \[
    \Pr^{\tsched}_{\cN,\stateN{s}{\cN}}(\varphi_K)
    \ \ \leqslant \ \ 
    \Pr^{\sched}_{\cM,s}(\varphi_K) \enspace.
  \]
\end{lemma}

\begin{proof}
  The proof is an easy verification. It relies on the fact that any
  scheduler $\tsched$ for $\cN$ naturally induces a scheduler $\sched$
  for $\cM$ that mimics the behavior of $\tsched$ and uses a
  weight-divergent scheduler for the behavior inside the end
  components $\cE\in \{\cE_1,\ldots,\cE_k\}$ to ensure that the
  accumulated weight inside $\cE$ is at least $\omega$.
\end{proof}

As a consequence, values in $\cM$ are at least as large as values in $\cN$.
\begin{corollary}[Values in $\cN$ are lower bounds for the values in $\cM$]
 \label{values-N-lower-bounds}
~
  \begin{enumerate}
  \item [(a)]
     For every  state $s$ of $\cM$, $\valueEas{\cN,s}\leqslant \valueEas{\cM,s}$.
      
    \item [(b)] If $\cE\in \{\cE_1,\ldots,\cE_k\}$ and
      $\valueEas{\cN,\exitstate{\cE}}\in \Integer$ then
      $\valueEas{\cM,\cE}=+\infty$.
  \end{enumerate}
\end{corollary}

Still there can be end components $\cE \in \{\cE_1,\ldots,\cE_k\}$
with $\valueEas{\cM,\cE}=+\infty$ while
$\valueEas{\cN,\exitstate{\cE}}=-\infty$.  The following example
illustrates this phenomenon.

\begin{figure}
\pgfdeclarelayer{background}%
\pgfsetlayers{background,main}%
\begin{tikzpicture}
	\node[state] (u) {$u$};
	\node[bullet, below=15mm of u] (b1) {};
	\node[state, left=13mm of b1] (s1) {$s_1$};
	\node[state, right=13mm of b1] (s2) {$s_2$};
	\node[state, below=25mm of b1] (goal) {$\goal$};
	\node[bullet, left=15mm of s1] (b2) {};
	\node[bullet, right=15mm of s2] (b3) {};
	\node[istate, left=30mm of u] {$\cM\colon$};%
	\path%
		(u) edge[ntran, bend right] node[left,pos=0.2] {$\gamma$/\negw{1}} (b1)
		(b1) edge[ptran, bend right] coordinate[pos=0.3](b1u) (u)
		(b1) edge[ptran, bend left] coordinate[pos=0.3](b1s1) (s1)
		(b1) edge[ptran, bend right] (s2)
		(s1) edge[ptran,loop, min distance=9mm, out=-60, in=-120] node[below]{$\beta$/\posw{1}} (s1)
		(s1) edge[ntran] node[above] {$\alpha$/\nilw} (b2)
		(b2) edge[ptran, bend left=50] coordinate[pos=0.1](b2u) (u)
		(b2) edge[ptran, bend right=50] coordinate[pos=0.1](b2goal) (goal)
		(s2) edge[ptran,loop, min distance=9mm, out=-60, in=-120] node[below]{$\beta$/\posw{1}} (s2)
		(s2) edge[ntran] node[above] {$\alpha$/\nilw} (b3)
		(b3) edge[ptran, bend right=50] coordinate[pos=0.1](b3u) (u)
		(b3) edge[ptran, bend left=50] coordinate[pos=0.1](b3goal) (goal);
	\draw%
		(b1s1) to[bend right=90] (b1u)
		(b2u) to[bend right=80] (b2goal)
		(b3u) to[bend left=80] (b3goal);	
	\coordinate[above left=4mm and 3mm of s1] (e1s1);
	\coordinate[above right=4mm and 3mm of s1] (e2s1);
	\coordinate[below right=10mm and 3mm of s1] (e3s1);
	\coordinate[below left=10mm and 3mm of s1] (e4s1);
	\coordinate[above left=4mm and 3mm of s2] (e1s2);
	\coordinate[above right=4mm and 3mm of s2] (e2s2);
	\coordinate[below right=10mm and 3mm of s2] (e3s2);
	\coordinate[below left=10mm and 3mm of s2] (e4s2);
	\begin{pgfonlayer}{background}%
        		\draw[rounded corners=1em,line width=2em,black!10,fill=black!10]
	        		(e1s1) -- (e2s1) -- (e3s1) -- (e4s1) -- cycle;
		\node[above=1mm of s1,white]{$\mathcal{E}_1$};
        		\draw[rounded corners=1em,line width=2em,black!10,fill=black!10]
	        		(e1s2) -- (e2s2) -- (e3s2) -- (e4s2) -- cycle;
		\node[above=1mm of s2,white]{$\mathcal{E}_2$};
	\end{pgfonlayer}
	\node[state,right=50mm of g] (gu) {$u$};
	\node[bullet, below=15mm of gu] (gb1) {};
	\node[state, left=13mm of gb1] (gs1) {$\entrystate{\cE_1}$};
	\node[state, below= of gs1] (gs1ex) {$\exitstate{\cE_1}$};
	\node[state, right=13mm of gb1] (gs2) {$\entrystate{\cE_2}$};
	\node[state, below= of gs2] (gs2ex) {$\exitstate{\cE_2}$};
	\node[state, below=25mm of gb1] (ggoal) {$\goal$};
	\node[bullet, left=15mm of gs1] (gb2) {};
	\node[bullet, right=15mm of gs2] (gb3) {};
	\node[istate, left=30mm of gu] {$\cN\colon$};%
	\path%
		(gu) edge[ntran, bend right] node[left,pos=.2] {$\gamma$/\negw{1}} (gb1)
		(gb1) edge[ptran, bend right] coordinate[pos=0.3](gb1gu) (gu)
		(gb1) edge[ptran, bend left] coordinate[pos=0.3](gb1gs1) (gs1)
		(gb1) edge[ptran, bend right] (gs2)
		(gs1) edge[ptran] node[right]{$\tau$/$\neuw{\omega}$} (gs1ex)
		(gs1ex) edge[ntran] node[right=0.1,pos=0.6] {$\alpha$/\nilw} (gb2)
		(gb2) edge[ptran, bend left=50] coordinate[pos=0.1](gb2gu) (gu)
		(gb2) edge[ptran, bend right=50] coordinate[pos=0.1](gb2ggoal) (ggoal)
		(gs2) edge[ptran] node[left]{$\tau$/$\neuw{\omega}$} (gs2ex)
		(gs2ex) edge[ntran] node[left=0.1,pos=0.6] {$\alpha$/\nilw} (gb3)
		(gb3) edge[ptran, bend right=50] coordinate[pos=0.1](gb3gu) (gu)
		(gb3) edge[ptran, bend left=50] coordinate[pos=0.1](gb3ggoal) (ggoal);
	\draw%
		(gb1gs1) to[bend right=90] (gb1gu)
		(gb2gu) to[bend right=80] (gb2ggoal)
		(gb3gu) to[bend left=80] (gb3ggoal);	
\end{tikzpicture}
 \caption{We have $\valueEas{\cM,s_1}=\valueEas{\cM,s_1}=+\infty$
while $\valueEas{\cN,\exitstate{\cE_1}}=\valueEas{\cN,\exitstate{\cE_2}}=-\infty$.}
\label{fig:valueeas-infty-vs-minus-infty}
\end{figure}

\begin{example}
{\rm
Let $\cM$ be the MDP depicted in Figure~\ref{fig:valueeas-infty-vs-minus-infty} on the left.
Then, $\cM$ has two maximal end components $\cE_1$, $\cE_2$ where
$\cE_1$ consists of the state-action pair $(s_1,\beta)$
and $\cE_2$  of the state-action pair $(s_2,\beta)$.
State $u$ does not belong to any end component.
Hence, $\WDMEC=\{s_1,s_2\}$.
We have $\valueEas{\cM,s_1}=\valueEas{\cM,s_1}=+\infty$.
The new MDP $\cN$ 
illustrated by Figure~\ref{fig:valueeas-infty-vs-minus-infty}
on the right can be seen as a Markov chain.
As $\wgt_{\cN}(u,\gamma)=-1$, the five 
states $\entrystate{\cE_1},\exitstate{\cE_1}$,
$\entrystate{\cE_2},\exitstate{\cE_2}$ and $u$ consistute
a strongly connected component of $\cN$ that contains a negative
cycle. Hence, 
$\Pr_{\cN,\exitstate{\cE_1}}(\varphi_K)=
 \Pr_{\cN,\exitstate{\cE_2}}(\varphi_K)=0$
for each $K$, and therefore
$\valueEas{\cN,\exitstate{\cE_1}}=
 \valueEas{\cN,\exitstate{\cE_2}}=-\infty$.
  }

\end{example}

To detect end components $\cE$ with with $\valueEas{\cM,\cE}=+\infty$
and $\valueEas{\cN,\exitstate{\cE}}=-\infty$, we introduce the
notation of \emph{good states} and \emph{good end components}.

\begin{definition}[Good states and end components]
\label{def:good-states}
If  $X \subseteq \{\cE_1,\ldots,\cE_k\}$ then we define
$\entrystate{X}= \{\entrystate{\cE} : \cE \in X\}$ and
\[
  \varphi_K[X] \ \ = \ \ 
   \Diamond ( T^* \cup \entrystate{X} ) \ \vee \ 
   \Diamond (\goal \wedge (\accwgt \geqslant K)) \enspace.
\]
The set of \emph{good end components} $\GoodEC$ is the largest subset $X$ of
$\{\cE_1,\ldots,\cE_k\}$ such that:
\begin{equation}
      \label{Good}
        \forall \cE \in X \ \exists K \in \Integer
        \text{ s.t. }
        \Pr^{\max}_{\cN,\exitstate{\cE}}
                \bigl( \, \varphi_K[X] \, \bigr) \ = \ 1 \enspace.
      \tag{Good}
\end{equation}
The set of \emph{good states} is defined as $\Good =\bigcup_{\cE \in \GoodEC} S_{\cE}$
where $S_{\cE}$ denotes the set of states
of end component $\cE$.
\end{definition}

\begin{remark}[Greatest fixed-point characterization of the $\GoodEC$]
  By definition, $\GoodEC$ is the greatest fixed point of the operator
  $ \Omega\colon 2^{\{\cE_1,\ldots,\cE_k\}} \to
  2^{\{\cE_1,\ldots,\cE_k\}}$ that
  \label{fixed-point-Good}
maps a given $X \subseteq \{\cE_1,\ldots,\cE_k\}$ to
\[
  \Omega(X)
  \ \ = \ \ 
  \bigl\{ \ \cE \in \{\cE_1,\ldots,\cE_k\} \ : \ 
              \exists K\in \Integer \text{ s.t. }
               \Pr^{\max}_{\cN,\exitstate{\cN}}
                  \bigl(\, \varphi_K[X]\, \bigr)=1 \ 
  \bigr\} \enspace.
\]
The operator $\Omega$ is monotonic, %
thus Tarski's fixed-point theorem ensures the existence of a greatest
fixed point
that can be obtained
as the limit of the sequence $X_0=\{\cE_1,\ldots,\cE_k\}$,
$X_{i+1}=\Omega(X_i)$ for $i \geqslant 0$.  
\end{remark}

First we prove that good states have value $+\infty$ in $\cM$.
\begin{lemma}
 \label{Good-WDMEC-infty}
   If $s \in \Good$ then $\valueEas{\cM,s}=+\infty$.
\end{lemma}

\begin{proof}
Obviously, there exists $K\in \Integer$ such that
$\Pr^{\max}_{\cN,\exitstate{\cE}}\big(\varphi_K[\GoodEC]\big)=1$
for all $\cE \in \GoodEC$.
Using Lemma \ref{MD-sufficient-eventually-almostsure-max-equals-1}
one can show that 
there is an MD-scheduler $\tsched$ for $\cN$ such that
\begin{center}
  $\Pr^{\tsched}_{\cN,\exitstate{\cE}}\big(\varphi_K[\GoodEC]\big)=1$  
  \ \ for all $\cE \in \GoodEC$ \ .
\end{center}
Scheduler $\tsched$ enjoys the following properties:
\begin{enumerate}
\item [(1)]
  $\Pr^{\tsched}_{\cN,\exitstate{\cE}}(\Diamond \entrystate{\cE})
    \ <  \ 1$ for all $\cE \in \{\cE_1,\ldots,\cE_k\}$.
  
  This follows from the observation that 
  the exit and entry states of the end components 
  $\cE \in \{\cE_1,\ldots,\cE_k\}$ do not belong to any end component
  of $\cN$ (Lemma \ref{simple-properties-MDP-N}).

\item [(2)]
  If $\cE,\cF \in \{\cE_1,\ldots,\cE_k\}$ and
  $\cE$ is good then
  $\Pr^{\tsched}_{\cN,\exitstate{\cE}}
    ((\neg \Good_{\cN}) \Until \entrystate{\cF})
    >  0$
  implies $\cF$ is good.

  Suppose by contradiction that $\cF$ is not good.
  Let $X=\GoodEC \cup \{\cF\}$.
  We pick a $\tsched$-path $\fpath$ from
  $\exitstate{\cE}$ to $\entrystate{\cF}$
  that does not contain a state in $\Good_{\cN}$, the set of good states of $\cN$.
  Then, $\fpath' = \fpath \, \tau \, \exitstate{\cF}$
  is a $\tsched$-path, too, and
  $\Pr^{\tsched}_{\cN,\exitstate{\cF}}\big(\varphi_L[\GoodEC]\big)=1$
  where $L=K-\wgt_{\cN}(\fpath')$.
  With $H=\min\{K,L\}$ we get
  $\Pr^{\tsched}_{\cN,\exitstate{\cE}}\big(\varphi_H[X]\big)=1$
  for all $\cE \in X$.
  Thus, $X$ %
  is a fixed point of $\Omega$.  This contradicts that $\GoodEC$ is
  the greatest fixed point of $\Omega$.

\item [(3)]
 $\Pr^{\tsched}_{\cN,\exitstate{\cE}}
   \bigl(\, \Diamond (T \cup \Good_{\cN}) \, \bigr) \ =  \ 1$
 as $\Pr^{\tsched}_{\cN,\exitstate{\cE}}\big(\varphi_K[\GoodEC]\big)=1$. 
\end{enumerate}
To prove $\valueEas{\cM,s}=+\infty$ we pick
some $R\in \Integer$ and design a scheduler $\sched=\sched_R$ for $\cM$
such that
$\Pr^{\sched}_{\cM,s}(\varphi_R)=1$ for all good states $s$.
The behavior of $\sched$ is as follows.
\begin{itemize}
\item
  For all input paths $\fpath$ ending in a state $u$ of 
  $S_{\cM}\setminus \WDMEC_{\cN} \subseteq S_{\cN} \setminus \Good_{\cN}$, 
  scheduler $\sched$ behaves as $\tsched$, \ie
  $\sched(\fpath)=\tsched(u)$.
\item For all input paths $\fpath$ that end in the entry state
  $\entrystate{\cE}$ of some end component $\cE \in \GoodEC$,
  scheduler $\sched$ uses a weight-divergent scheduler for $\cE$ until
  it has generated a path $\fpath'$ where the total weight is at least
  $R{-}K$ and where $\tsched(\exitstate{\cE})$ is an action of state
  $\last(\fpath')$.  Scheduler $\sched$ then schedules action
  $\tsched(\exitstate{\cE})$ for $\fpath'$.
\item
  The behavior of $\sched$ for input paths $\fpath$ where
  $\last(\fpath)$ is the entry state $\entrystate{\cE}$ of some
  non-good end component is irrelevant.
\end{itemize}
Using properties (1), (2) and (3) of $\tsched$, we get that
none of the $\sched$-paths starting in a good state will visit
a non-good end component and that
$\Pr^{\sched}_{\cM,s}(\varphi_R)=1$ for all good states $s$.
\end{proof}

Conversely, we establish that states belonging to a weight-divergent
end component and with value $+\infty$ are good states.
\begin{lemma}
 \label{WDMEC-infty-Good}
   If $s\in \WDMEC$ and 
   $\valueEas{\cM,s}=+\infty$ then $s\in \Good$.
\end{lemma}

\begin{proof}
  Recall that all states that belong to the same end component
  $\cE \in \{\cE_1,\ldots,\cE_k\}$ have the same value
  $\valueEas{\cM,\cE}=\{+\infty,-\infty\}$.
\begin{description}
\item[Set $X$] Let $X$ denote the set of end components
  $\cE \in \{\cE_1,\ldots,\cE_k\}$ such that
  $\valueEas{\cM,\cE}=+\infty$.  It suffices to show that $X$ is a
  fixed point of $\Omega$, as then $X \subseteq \GoodEC$ will follow,
  and therefore all states belonging to an end component $\cE \in X$
  are good.
\item[Progress moves]
  As in the proof of
  Lemma~\ref{MD-sufficient-eventually-almostsure-max-equals-1},
  page~\pageref{MD-sufficient-eventually-almostsure-max-equals-1}, we
  use the notion of progress moves.  This time, we need the term
  progress move for maximal end components.  Given a maximal end
  component $\cE$, we refer to a state-action pair $(u,\alpha)$ with
  $u \in \cE$ as a \emph{progress move} for $\cE$ if there is some
  state $v$ with $P_{\cM}(u,\alpha,v)>0$ and $v$ does not belong to
  $\cE$.  By assumption \textbf{(A3)}, each maximal end component of
  $\cM$ has a progress move. Moreover, whenever $\sched$ is a
  scheduler for $\cM$ and $s$ a state in $\cM$ such that
  $\Pr^{\sched}_{\cM,s}(\Diamond T)=1$ then for each $\sched$-path
  $\fpath$ from $s$ to some maximal end component $\cE$ there is an
  $\sched$-path $\fpath'$ that extends $\fpath$ and where
  $(\last(\fpath'),\sched(\fpath'))$ is a progress move of $\cE$.
\item[Schedulers $\sched_{\cE}$]
For each $\cE \in X$, there is some $R\in \Integer$
and a scheduler
$\sched_{\cE}$ enjoying the following properties:
\begin{description}
\item [(P1)]
  $\Pr^{\sched_{\cE}}_{\cM,s}(\varphi_R)=1$
  for all states $s$ in $\cE$
\item [(P2)]
  There is a progress move $(u_{\cE},\alpha_{\cE})$ 
  of $\cE$ such that
  $\sched_{\cE}(\fpath)=\alpha_{\cE}$ for each
  $\sched$-path $\fpath$ starting in some state of $\cE$
  with $\wgt_{\cM}(\fpath)\geqslant 0$ and
  $\last(\fpath)=u_{\cE}$.
\item [(P3)]
  Whenever $\fpath$ is a $\sched_{\cE}$-path consisting of
  states and state-action pairs in $\cE$ 
  such that $\sched_{\cE}(\fpath)$ is a progress move then
  $\wgt_{\cM}(\fpath)\geqslant 0$,
  $\last(\fpath)=u_{\cE}$ and
  $\sched_{\cE}(\fpath)=\alpha_{\cE}$.
\end{description}
Conditions \textbf{(P2)} and \textbf{(P3)} can be ensured by using a weight-divergent scheduler
until having accumulated enough weight and the current state in
$u_{\cE}$.

Property \textbf{(P1)} implies:
\begin{description}
\item [(P4)]
  All states
  that are reachable from $\cE$ via an $\sched_{\cE}$-path
  are either not contained in $\WDMEC$ or belong to an end component
  $\cE \in X$.
\item [(P5)]
  Whenever $\fpath$ is an $\sched_{\cE}$-path from $\cE$ to $\goal$
  where only the first state is contained in $\WDMEC$,
  then $\fpath$ can be seen as a path in $\cN$ starting in
  $\exitstate{\cE}$ and $\wgt_{\cN}(\fpath)\geqslant R$.
\item [(P6)]
  $\Pr^{\sched_{\cE}}_{\cM,u_{\cE}}(\Diamond \cE) <1$
\end{description}
Let now $\tsched_{\cE}$ be a scheduler for $\cN$ that schedules
$\alpha_{\cE}$ for state $\exitstate{\cE}$
and behaves as $\sched$ afterwards (when identifying $u_{\cE}$ with
$\entrystate{\cE}$) until reachaing a trap state $t\in T$ or the entry state
of an end component $\cF \in \{\cE_1,\ldots,\cE_k\}$.
In the latter case, $\cF \in X$ (by \textbf{(P4)}) and the behavior of
$\tsched_{\cE}$ after having reached $\entrystate{\cF}$ is irrelevant for our
purposes.
By \textbf{(P5)} and \textbf{(P6)} we then have 
\[
  \Pr^{\tsched_{\cE}}_{\cN,\exitstate{\cE}}(\varphi_R[X])=1 \enspace.
\]
This shows that $X$ is indeed a fixed point of $\Omega$.
\end{description}
As a consequence, all states in $X$ are good states, showing the desired result.
\end{proof}

From Lemmas~\ref{Good-WDMEC-infty} and~\ref{WDMEC-infty-Good} we
obtain a characterization of good states and good end components.
\begin{corollary}
 \label{Good=WDMEC-infty}
  $\Good = 
   \bigl\{ s\in \WDMEC : \valueEas{\cM,s}=+\infty \bigr\}$
  and
  $
  \GoodEC = 
  \bigl\{ \cE \in \{\cE_1,\ldots,\cE_k\} : 
                      \valueEas{\cM,\cE}=+\infty \ 
  \bigr\}
  $.
\end{corollary}

Finally we characterize states of $\cM$ with value $+\infty$:
\begin{lemma}
  Let $s$ be a state in $\cM$. Then 
  $\valueEas{\cM,s}=+\infty$
  iff \ $\Pr^{\max}_{\cM,s}\big(\Diamond (T^*\cup \Good)\big)=1$.
\end{lemma}

\begin{proof}
The implication ``$\Longleftarrow$'' is an easy verification.
The task is to provide schedulers $\tsched_K$ with
$\Pr^{\sched_K}_{\cM,s}(\varphi_K)=1$ for each $K\in \Integer$.
The idea is to combine an MD-scheduler $\sched$ satisfying
$\Pr^{\sched}_{\cM,s}(\Diamond (T^*\cup \Good))=1$ with schedulers of
a family $(\sched_R)_{R\in \Integer}$ where
$\Pr^{\sched_R}_{\cM,u}(\varphi_R)=1$ for each state $u \in \Good$ and
each $R\in \Integer$.  For this, scheduler $\tsched_K$ first behaves
as $\sched$ until reaching the target state in $\good \in T^*$ or a
good state. In the latter case, if $w$ is the weight that has been
accumulated so far, $\tsched_K$ behaves as $\sched_{K-w}$ after having
reached a good state.

To prove ``$\Longrightarrow$'', we pick a state $s$ of $\cM$ such that
$\valueEas{\cM,s}=+\infty$.  The claim is trivial if
$s \in \{\good\} \cup \Good$.  Consider now the case where
$s \notin \{\good\}\cup \Good$.  Suppose by contradiction that
$\Pr^{\max}_{\cM,s}(\Diamond (\good \cup \Good))<1$.  Then, for each
scheduler $\sched$ with $\Pr^{\sched}_{\cM,s}(\Diamond T)=1$ we have
$\Pr^{\sched}_{\cM,s}((\neg \Good) \Until \goal)>0$.

Let $\cM'$ be the MDP that results from $\cM$ by (i) removing all
state-action pairs $(u,\alpha)$ with $u\in \WDMEC$ and (ii) adding the
state-action pairs $(u,\tau)$ with $P(u,\tau,s)=1$ and
$\wgt(u,\tau)=0$ for $u\in \Good \cup \{\good\}$.  Then, the end
components of $\cM'$ are exactly the end components of $\cM$ that do
not contain any state in $\WDMEC$.  In particular, $\cM'$ has no
(positively) weight-divergent end components.  This is because $s$ and
the states $u\in \Good \cup \{\good\}$ cannot belong to an end
component of $\cM'$ since
$\Pr^{\max}_{\cM',s}\big(\Diamond (\{\good\}\cup \Good)\big)=
\Pr^{\max}_{\cM,s}\big(\Diamond (\{\good\}\cup \Good)\big)<1$.

We now consider a sequence of schedulers $(\sched_K)_{K\in \Nat}$
with $\Pr^{\sched_K}_{\cM,s}(\varphi_K)=1$ for all $K\in \Nat$.
None of the states that are reachable from $s$ via a $\sched_K$-path
belongs to $\WDMEC \setminus \Good$
as no $\sched_K$-path from $s$ to $\goal$ can traverse
a state $u$ where $K_{\cM,u}^{\text{\rm E1}}=-\infty$.

Given $R\in \Integer$, we design a scheduler $\tsched_R$ for $\cM'$ as
follows.  Given an input path $\fpath$ for $\tsched$ where $\fpath$
does not contain a $\tau$-transition from a state
$u \in \Good \cup \{\good\}$ then $\tsched_R$ behaves as $\sched_R$.
Otherwise, $\fpath$ has the form $\fpath_1 ; \fpath_2$ where
$\fpath_1$ is a path from $s$ to $s$ where the last transition is a
$\tau$-transition from some state $u\in \Good \cup \{\good\}$ to $s$
and $\fpath_2$ is a path from $s$ that does not contain such a
$\tau$-transition. In this case, we define $w=\wgt(\fpath_1)$ and
$\tsched_R$ behaves for $\fpath$ in the same way as $\sched_{R-w}$
behaves for the path $\fpath_2$.  We then have
$\Pr^{\tsched_R}_{\cM',s}\big(\Diamond (\goal \wedge (\accwgt \geqslant
R))\big)=1$.  In particular,
$\Exp{\tsched_R}{\cM',s}(\accdiaplus \goal)\geqslant R$.  As this
holds for each $R \in \Integer$, we obtain:
$\Exp{\sup}{\cM',s}(\accdiaplus \goal)=+\infty$.  This is impossible
by 
Lemma \ref{th:finiteness-min-exp-accwgt}
(rephrased for maximal expected accumulated weights)
as $\cM'$ does not have positively weight-divergent end
components.
\end{proof}

We can finally relate values in $\cM$ and in $\cN$.
\begin{definition}[Values for the states in $\cN$]
\label{values-N}
  Given a state $u$ in $\cN$ the \emph{value of $u$ in $\cN$} is
  \[
     K_{\cN,u}
     \ \ = \ \ 
     \sup \ \bigl\{ \ K \in \Integer \ : \
                      \Pr^{\max}_{\cN,u}(\varphi_K[\GoodEC])=1
                     \ 
            \bigr\} \enspace.
 \]
\end{definition}

\begin{lemma}
 \label{K-N-E1=K-M-E1}
 For each state $s$ in $\cM$, $\valueEas{\cM,s} = K_{\cN,s_\cN}$.
\end{lemma}

\begin{proof}
   We have $\valueEas{\cM,s} \leqslant K_{\cN,\stateN{s}{\cN}}$
   by Lemma \ref{Tstar-N}
   and Corollary \ref{Good=WDMEC-infty}.
   To prove 
   $\valueEas{\cM,s} \geqslant K_{\cN,\stateN{s}{\cN}}$
   we show that
   $\Pr^{\max}_{\cN,\stateN{s}{\cN}}(\varphi_K[\GoodEC])=1$ implies
   $\Pr^{\max}_{\cM,s}(\varphi_K)=1$.
   For this, pick a scheduler
   $\tsched$ for $\cN$ with
   $\Pr^{\tsched}_{\cN,\stateN{s}{\cN}}(\varphi_K[\GoodEC])=1$
   and a family $(\sched_R)_{R \in \Integer}$ of schedulers
   for $\cM$ such that
   $\Pr^{\sched_R}_{\cM,u}(\varphi_R)=1$ for all good states
   $u\in \Good$.
   Let now $\sched$ be the following scheduler for $\cM$ that
   mimics $\tsched$ until reaching a good end component $\cE$.
   If $w$ is the weight that has been accumulated so far then
   $\tsched$ behaves as $\sched_{K-w}$ from then on.
   We then have $\Pr^{\tsched}_{\cM,s}(\varphi_K)=1$.
\end{proof}

\subsection*{Computation of the Values $\valueEas{\cM,s}$ in the
  General Case}

Relying on Lemma~\ref{K-N-E1=K-M-E1}, the values $\valueEas{\cM,s}$
for each $s$ can be computed as follows.
\begin{enumerate}
\item [1.] Construct from $\cM$ the MDP $\cN$.
\item [2.]
   Compute the set $\GoodEC$, iteratively starting with 
 $X_0=\{\cE_1,\ldots,\cE_k\}$, and with 
   \[
     X_{i+1} \ \ = \ \ 
     \bigl\{ \ \cE\in X_i \ : \
               \exists K \in \Integer \text{ s.t. }
               \Pr^{\max}_{\cN,\exitstate{\cE}}(\varphi_K[X_i])=1 \ 
     \bigr\} \enspace.
   \]
   When the sequence converges, the obtained set is
   $\GoodEC$. %
   \\
   Since $\cN$ has no weight-divergent end components, for the
   computation of the sets $X_1,X_2,\ldots$ we rely on the techniques
   for MDPs with this restriction, presented at the beginning of this
   section as a particular case.
 \item [3.]  Compute the values $K_{\cN,s}$ (see Definition
   \ref{values-N}), again using the techniques for MDPs without
   weight-divergent end components.
\end{enumerate}
For the third step, we can switch from $\cN$ to the sub-MDP that
arises by removing the entry and exit states for the end components
$\cE\in \{\cE_1,\ldots,\cE_k\} \setminus \GoodEC$.  These are the
maximal weight-divergent end components of $\cM$ where
$\valueEas{\cM,\cE}=-\infty$.  Furthermore, for $\cE \in \GoodEC$,
state $\exitstate{\cE}$ can be turned into a trap.

\smallskip

Let us analyse the complexity of the above procedure.  The number of
iterations in the second step is bounded by the number $k$ of maximal
weight-divergent end components.  The values $\valueEas{\cM,s}$ can be
computed in polynomial time, assuming an oracle that determines the
values $\valueEas{\cN,s}$ for MDPs without weight-divergent end
components. Recall (see Theorem~\ref{complexity-DWR-E1-no-wgt-div-EC})
that for MDPs without weight-divergent end components, the decision
problem lies in $\NP \cap \coNP$, and the values are computable in
pseudo-polynomial time. Since $\PTIME^{\NP \cap \coNP}$ agrees with
$\NP\cap \coNP$ \citeapx{Brassard79}, we conclude:
  
\begin{theorem}
 \label{complexity-DWR-E1-general-case}
 The decision problem \Easdwr{} belongs to $\NP \cap \coNP$.
 The values $\valueEas{\cM,s}$ can be computed in pseudo-polynomial time.
\end{theorem}

\subsection*{Mean-Payoff Game Hardness}

We now prove a lower complexity bound for the \Easdwr{} decision problem. 

\begin{lemma}
\label{MPG-hardness-Easdwr}
  \Easdwr{} is mean-payoff game hard.
\end{lemma}

\begin{proof}
  To establish mean-payoff hardness, we describe a polynomial-time
  reduction from the problem to decide whether player 1 of a
  (non-stochastic) two-player mean-payoff game has a winning strategy
  from a given game location $\sinit$.  This problem is known to be in
  $\NP \cap \coNP$ (even $\UP \cap \coUP$), but not
  known to be in $\PTIME$.

  Let $\cG = (V,V_1,V_2,E,\wgt)$ be a two-player mean-payoff game
  where $V$ is a finite set of game locations, disjointly partitioned
  into $V_1$ and $V_2$. The set $V_i$ stands for the set of game
  locations where player $i$ has to move.
  $E \subseteq V_1 \times V_2 \cup V_2 \times V_1$ is the edge
  relation, where we suppose that $E$ is total in the sense that each
  location has at least one outgoing edge, and
  $\wgt\colon E_1 \to \Integer$ is the weight function\footnote{%
    Note that the general case, where players do not strictly
    alternate, and weights are also attached to moves of
    player~2 %
    can easily be reduced to game
    structures of this form.%
  } where
  $E_1=E \cap (V_1 \times V)$.  The objective of player 1 is to ensure
  that the mean payoff of all plays is nonnegative.  More precisely,
  we consider the problem where we are given $\cG$ and a distinguished
  starting location $v_0\in V$ and where the task is to decide whether
  player 1 has a strategy $\sched$ such that the mean payoff of all
  $\sched$-plays from $v_0$ is nonnegative.%

  Let now $\cM$ be the following MDP, as illustrated on
  Figure~\ref{fig:mp-red-Easdwr}.
\begin{figure}

\begin{tikzpicture}
	\node[state] (gv0) {$v_0$};
	\node[rstate, above right=1mm and 15mm of gv0] (gv2) {$v_2$};
	\node[state, above=10mm of gv2] (gv1) {$v_1$};
	\node[state, right=10mm of gv2] (gv1p) {$v_1'$};
	\node[state, below=7mm of gv1p] (gv1pp) {$v_1''$};
	\node[istate, above left=17mm and 8mm of gv0] {$\cG\colon$};%
	\path%
		(gv1) edge[ptran] node[right] {$\neuw{w}$} (gv2)
		(gv2) edge[ptran] (gv1p)
		(gv2) edge[ptran] (gv1pp);
	\node[state, right=70mm of gv0] (sinit) {$s_\init$};
	\node[state, right= of sinit] (v0) {$v_0$};
	\node[bullet, right= of v0] (b) {};
	\node[state, above=7mm of b] (v2) {$v_2$};	
	\node[state, above=7mm of v2] (v1) {$v_1$};
	\node[state, right=15mm of b] (v1p) {$v_1'$};
	\node[state, below=7mm of v1p] (v1pp) {$v_1''$};
	\node[state, below=15mm of b] (goal) {$\goal$};
	\node[istate, above left=17mm and 4mm of sinit] {$\cM\colon$};%
	\path%
		(sinit) edge[ptran,loop, min distance=9mm, out=60, in=120] node[above]{$\alpha$/\posw{1}} (sinit)
		(sinit) edge[ptran] node[above]{$\tau$/\nilw} (v0)
		(v1) edge[ptran] node[right,pos=0.4] {$v_2$/$\neuw{w}$} (v2)
		(v2) edge[ntran] node[right,pos=0.3] {$\tau$/$\nilw$} (b)
		(b) edge[ptran, bend right] node[above]{$\frac{1}{4}$} coordinate[pos=0.3](bv1p) (v1p)
		(b) edge[ptran, bend right] node[above,pos=0.6]{$\frac{1}{4}$} (v1pp)
		(b) edge[ptran, bend right] node[right,pos=0.6]{$\frac{1}{2}$} coordinate[pos=0.3](bgoal) (goal);	
	\draw%
		(bgoal) to[bend right] (bv1p);
	\coordinate[above right=0mm and 17mm of gv0] (c);
	\node[cloud, cloud puffs=23,cloud puff arc=120, cloud ignores aspect, draw=gray,
			minimum width=70mm, minimum height=57mm](cloud) at (c) {};
	\coordinate[above right=0mm and 9mm of v0] (gc);
	\node[cloud, cloud puffs=23,cloud puff arc=120, cloud ignores aspect, draw=gray,
			minimum width=80mm, minimum height=57mm](cloud) at (gc) {};
\end{tikzpicture}
 \caption{Reduction from mean-payoff games to \Easdwr.}
\label{fig:mp-red-Easdwr}
\end{figure}
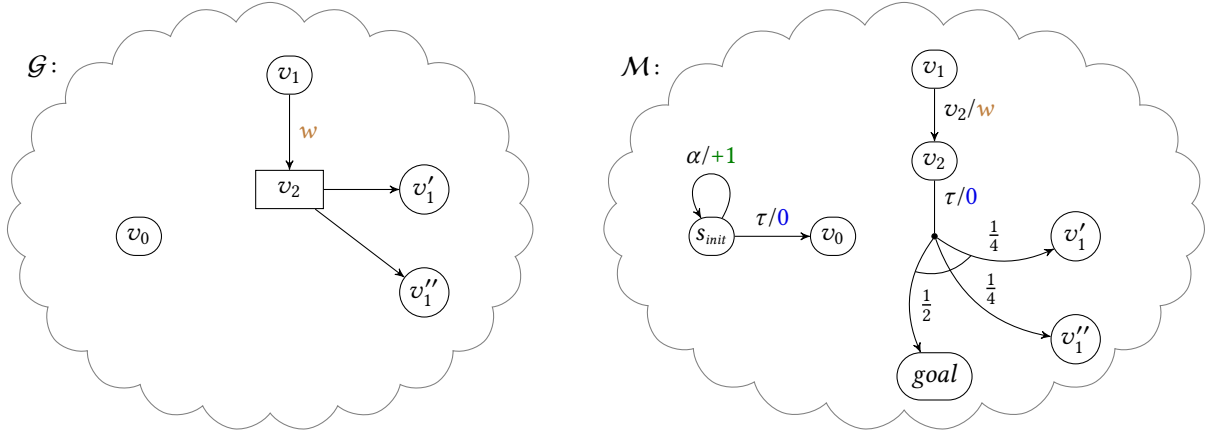
  The state space is
  $S_{\cM}=V \cup \{\sinit,\goal\}$ where $\sinit$ is the initial
  state of $\cM$.  The action set is $\Act = V\cup \{\alpha,\tau\}$.
  The transition probabilities and weights are defined as follows.  In
  $\sinit$, actions $\tau$ and $\alpha$ are enabled with
  $P(\sinit,\tau,v_0)=1$ and $P(\sinit,\alpha,\sinit)=1$ and
  $\wgt(\sinit,\tau)=0$, $\wgt(\sinit,\alpha)=1$.  If $(v_1,v)\in E$
  where $v_1\in V_1$ then $P(v_1,v,v)=1$. The weight of the
  state-action pair $(v_1,v)$ is the weight of the edge $(v_1,v)$ in
  $\cG$.  In all other cases, $P(v_1,\cdot)=0$. (That is, only the
  actions $v\in V$ where $(v_1,v)$ is an edge in $\cG$ are enabled in
  state $v_1$ of $\cM$.)  Let now $v_2 \in V_2$ and let
  $\Post(v_2)=\{v\in V: (v_2,v)\in E\}$.  The only enabled action in
  $v_2$ is $\tau$.  The transition probabilities are given by
  $P(v_2,\tau,\goal)=\frac{1}{2}$ and $P(v_2,\tau,v)=\frac{1}{2k}$ where
  $k=|\Post(v_2)|$.  The weight of the state-action pair $(v_2,\tau)$
  is $0$.  State $\goal$ is a trap in $\cM$.

  As the edge relation $E$ of $\cG$ is total, $\goal$ is the only trap
  state of $\cM$.  By construction, we have
  $\Pr^{\min}_{\cM,s}(\Diamond \goal)=1$ for all states $s$ in $\cM$
  with $s\not= \sinit$.  This also implies that $\cM$ has a single end
  component $\cE$ consisting of the state-action pair
  $(\sinit,\alpha)$.

  Obviously, $\cM$ can be constructed in polynomial time from
  $\cG$. We let
\[
  \varphi \ \ = \ \ \Diamond (\goal \wedge (\accwgt \geqslant 0))
\]
that is, $T=\{\goal\}$ and $K_{\goal}=0$.  We claim that player~1 has
a winning strategy $\sched$ in $\cG$ iff there exists a scheduler
$\sched$ in $\cM$ with
$\Pr^{\max}_{\cM,\sinit}(\varphi)=1$.

Suppose first that player~1 has a winning strategy $\stratone$ in the
mean-payoff game. Without loss of generality, this winning strategy can be assumed to be
an MD-strategy. Let $\cG'$ be the graph structure induced by $\stratone$
restricted to the states that are reachable from $v_0$ along finite
$\stratone$-plays.  As $\stratone$ is winning, $\cG'$ has no negative
cycle.  Let $w$ be the minimal weight of a path $\fpath$ from $v_0$ to
$\goal$ in $\cG'$, and let $k=\max\{0,-w\}$.  Consider now the
following scheduler $\sched$ for $\cM$.  It schedules $k$-times action
$\alpha$ in state $\sinit$, moves to state $v_0$ via the
$\tau$-transition afterwards and behaves as $\stratone$ from then
on. In particular $\sched$ is a finite-memory scheduler. The
underlying graph of the Markov chain $\cC$ induced by $\sched$
(restricted to the states reachable from $\sinit$) agrees with $\cG'$
extended by an initial phase
\[
  \underbrace{
   \sinit \move{\alpha} \sinit \move{\alpha} \ldots \move{\alpha} \sinit}_{
    \text{$k$ transitions}}
   \move{\tau} v_0 \enspace.
\]
As $\cG'$, also $\cC$ has no negative cycles.
Moreover, $\wgt(\fpath)\geqslant k+w \geqslant 0$ for all paths
$\fpath$ in $\cC$ from $\sinit$ to $\goal$.
Hence, $\Pr^{\sched}_{\cM,\sinit}(\varphi)=1$.

Vice versa, suppose $\sched$ is a scheduler for $\cM$ such that
$\Pr^{\sched}_{\cM,\sinit}(\varphi)=1$.
In particular, $\Pr^{\sched}_{\cM,\sinit}(\Diamond \goal)=1$.
Thus, $\sched$ schedules $\tau$ for
$\sinit$ after having generated a path of the form
$\fpath = 
 \sinit \, \alpha \, \sinit \alpha \sinit \alpha \ldots \alpha \, \sinit$.
With $k=|\fpath|$ we have $\wgt(\fpath)=k$.
The residual scheduler $\tsched=\residual{\sched}{\fpath\, \tau \, v_0}$ 
can be viewed as a strategy for player 1 in the game structure $\cG$.
As
\[
  \Pr^{\tsched}_{\cM,v_0}(\Diamond (\goal \wedge (\accwgt \geqslant -k)))
  \ \ = \ \ 1 \enspace,
\]
all $\tsched$-paths starting in $v_0$ and ending in $\goal$ have weight
at least $-k$.
The sub-MDP $\cM'$ resulting from $\cM$ by removing the state-action
pair $(\sinit,\alpha)$ has no end components.
As $\tsched$ can be viewed as a scheduler for $\cM'$, we can rely on
Lemma~\ref{MD-sufficient-eventually-almostsure-max-equals-1}, 
which ensures the existence of an MD-scheduler $\usched$ with
\[
  \Pr^{\usched}_{\cM,v_0}(\Diamond (\goal \wedge (\accwgt \geqslant -k)))
  \ \ = \ \ 1 \enspace.
\]
Lemma~\ref{lm:charac-Uasdwr-traps}
yields
that the Markov chain induced by $\usched$ has no negative cycle.
Hence, $\usched$ is a winning strategy for player~1 in the game $\cG$.
\end{proof}

\subsection{Weight-Bounded B\"uchi Constraints}

\label{sec:appendix-buechi}

We now provide the proofs for Section \ref{sec:Buechi}.
Throughout this section, we suppose that $\cM=(S,\Act,P,\wgt)$ is an MDP
without traps, which ensures
that all maximal paths are infinite.
Furthermore, let
$\sinit$ be a state in $\cM$, $F$ a set of states in $\cM$,
and $K\in \Integer$.

\begin{lemma}
  Problems \EaswB{} and \UposwB{} are
  hard for non-stochastic two-player mean-payoff games.
\end{lemma}

\begin{proof}
  As the corresponding result has been established for the DWR problems
  \Easdwr{} and \Uposdwr{} 
  (see Lemma \ref{MPG-hardness-Easdwr} and \ref{MPG-hardness-Uposdwr}), 
  it suffices to provide prolynomial reductions from
  them.
  Let $\cM$ be an MDP and $\varphi(T,(K_t)_{t\in T})$ be a 
  disjunctive weight-bounded reachability constraint.
  Let $\cM'$ be the MDP resulting from $\cM$ by 
  (i) discarding all state-action pairs $(t,\alpha)\in \cM$ with $t\in T$,
  (ii) adding state-action pairs $(t,\tau)$ with
   $P_{\cM'}(t,\tau,t)=1$ and $\wgt_{\cM'}(t,\tau)=0$
   for all states $t\in T$ and all traps $t$ of $\cM$.
  Then, 
  $(\cM,\sinit)$ satisfies \Easdwr{} iff
  $(\cM,\sinit)$ satisfies \EaswB,
  and
  $(\cM,\sinit)$ satisfies \Uposdwr{} iff
  $(\cM,\sinit)$ satisfies \UposwB.
\end{proof}

\subsubsection{The Existential Problems \EaswB{} and \EposwB}

\label{sec:buechi-existential}

We introduce notations for sets of states belonging to given end
components. 
Let $\PumpEC_F$ is the set of states that belong to a
pumping end component $\cE$ that contain at least one state in $F$.
Likewise, $\GambEC_F$ is the set of states that belong to a
gambling end component containing at least one $F$-state. 
We write $\ZeroEC_F$ for the set of states 
that belong to a 0-EC $\cZ$ such that $\cZ$ contains at least one $F$-state
and $\cZ$ is a sub-component of
an MEC $\cE$ with $\Exp{\max}{\cE}(\MP)= 0$.
Then, $\WDMEC_F= \PumpEC_F \cup \GambEC_F$ consists of
the states that belong to some weight-divergent end component
containing at least one $F$-state.
$\PumpEC_F  \cup \ZeroEC_F$ is the set of states that 
are contained in some end component intersecting with $F$
that has a scheduler where the accumulated weight is bounded from below
almost surely.

\begin{lemma}
 \label{lemma:divsched}
  $\PumpEC_F$, $\WDMEC_F$,
  $\PumpEC_F \cup \ZeroEC_F$ and
  $\WDMEC_F \cup \ZeroEC_F$ 
  are computable in polynomial time.
Moreover, there exist a
scheduler $\divsched$ such that:
\begin{itemize}
\item
  $\Pr^{\divsched}_{\cM,s}(\Box \Diamond F)=1$ for all 
  $s\in \WDMEC_F \cup \ZeroEC_F$,
\item
  $\divsched$ is pumping from all states 
  $s\in \PumpEC_F$,
\item
  $\divsched$ is gambling from all states 
  $s\in \GambEC_F$, and
\item
  from all states $t \in \ZeroEC_F$,
    $\divsched$ realizes a 0-EC 
    $\cE$ with $\cE \cap F\not= \varnothing$. 
  In particular,
  $\Pr^{\divsched}_{\cM,t}(\Box \Diamond (\accwgt \geqslant 0) \wedge
      \Box \Diamond F)=1$ for all $t \in \ZeroEC_F$.    
\end{itemize}
\end{lemma}

\begin{proof}
  $\PumpEC_F$ is the 
  union of the state spaces of the MEC $\cE$ of $\cM$ that contain at least
  one $F$-state and that enjoy the property $\Exp{\max}{\cE}(\MP)>0$
  (see Lemma \ref{lemma:pumping-ecs}).
  This yields the polynomial-time computability of $\PumpEC_F$.

  To compute $\WDMEC_F$ we can rely on the fact that
  each weight-divergent end component is contained in an MEC that
  is weight-divergent.
  Thus, we can apply standard techniques to compute the MECs
  of $\cM$. For each of the MECs $\cE$, we check whether $\cE$
  contains at least one $F$-state, and if so, we 
  check whether $\cE$ is weight-divergent 
  using the polynomial-time weight-divergence algorithm 
  presented in Section \ref{sec:wgt-div}.
  Then, $\WDMEC_F$ arises by the union of the state spaces of these MECs.

  The set $\PumpEC_F \cup \ZeroEC_F$ is the set of all states $s$
  that belong to a maximal end component $\cE$ with 
  $\cE \cap F \not= \varnothing$ and such either $\cE$ is pumping
  or $\Exp{\max}{\cE}(\MP)=0$ and $s$ belongs to a maximal 0-EC $\cZ$
  of $\cE$ with $\cZ \cap F\not=\varnothing$.
  Thus, to compute $\PumpEC_F \cup \ZeroEC_F$ in polynomial time
  we can rely on the pumping criterion presented in
  Lemma \ref{lemma:pumping-ecs} and the techniques
  of 
  Lemma \ref{mincredit-ZeroEC} 
  to determine maximal 0-ECs of in strongly connected
  MDPs with maximal expected mean payoff 0.
  The statement about the recurrence values $\rec(s)$ is immediate from
  Lemma \ref{mincredit-ZeroEC}.

  It remains to explain how to obtain scheduler $\divsched$.
  For each pumping resp. gambling MEC $\cE$ containing at least one $F$-state, 
  we pick a pumping resp. gambling
  scheduler $\vsched_{\cE}$ and an MD-scheduler
  $\usched_{\cE}$ for $\cE$ 
  such that $\Pr^{\usched_{\cE}}_{\cE,s}(\Diamond F)=1$
  for all states $s$ in $\cE$. 
  Obviously, $\usched_{\cE}$ and $\vsched_{\cE}$ can be combined to obtain
  a (possibly infinite-memory) 
  weight-divergent scheduler $\divsched_{\cE}$ for $\cE$ with
  $\Pr^{\divsched_{\cE}}_{\cE,s}(\Box \Diamond F)=1$ for all states
  $s$ in $\cE$.
  Composing the schedulers $\divsched_{\cE}$ with schedulers that realize
  maximal 0-ECs contained
  in MECs with maximal expected mean payoff 0 yields a scheduler
  $\divsched$ as stated in the lemma.
\end{proof}

We now show that problems
\EaswB{} and \EposwB{} are polynomially reducible
to \Easdwr{} and \Eposdwr, respectively 
(see Lemma \ref{lemma:Buechi-max-as} and \ref{lemma:Buechi-max-pos} below).
In both cases we use the
disjunctive weight-bounded reachability constraint
$\varphi=\varphi(T, (K_t)_{t\in T})$ where
 $T = \WDMEC_F \cup \ZeroEC_F$ 
 and $K_t=-\infty$ for $t\in \WDMEC_F$ and
 $K_t=K$ for $t\in \ZeroEC_F$.
 That is, $T^*=\WDMEC_F$ and
 \[
  \varphi \ \ = \ \ 
  \Diamond \WDMEC_F \vee 
  \bigvee_{t\in \ZeroEC_F}
    \Diamond (t \wedge (\accwgt \geqslant K)) \enspace.
\]

Given an end component $\cE$, let $\Limit{\cE}$ denote the set of
infinite paths $\infpath$ such that the limit of $\infpath$ equals $\cE$.
Recall that the limit of an infinite path $\infpath$ is the set of 
all state-action pairs $(s,\alpha)$ that occur infinitely often in
$\infpath$. 
In what follows, we often use de Alfaro's Theorem \cite{Alfaro98Thesis}
stating that under each scheduler, the probability
of the paths in $\bigcup_{\cE} \Limit{\cE}$ equals 1 when $\cE$ ranges
over all (possibly non-maximal) end components.

\begin{lemma}
 \label{lemma:Buechi-max-as}
 Let $\varphi$ be as above.
  Then
   $\Pr^{\max}_{\cM,\sinit}
      \bigl(\Box \Diamond (\accwgt \geqslant K) \wedge \Box \Diamond F 
      \bigr)=1$
  iff\ \  $\Pr^{\max}_{\cM,\sinit}(\varphi)=1$.
\end{lemma}

\begin{proof}
 The implication ``$\Longleftarrow$'' is an easy verification.
  Given a scheduler $\sched$ with $\Pr^{\sched}_{\cM,\sinit}(\varphi)=1$,
  combine $\sched$ with the scheduler $\divsched$
  of Lemma \ref{lemma:divsched}
  to obtain a new scheduler
  $\tsched$ with
  $\Pr^{\tsched}_{\cM,\sinit}
    \big(\Box \Diamond (\accwgt \geqslant K)\wedge \Box \Diamond F\big)=1$.

  To prove ``$\Longrightarrow$'',
  we suppose that we are given a scheduler $\tsched$ for $\cM$ with
  $\Pr^{\tsched}_{\cM,\sinit}
   \big(\Box \Diamond (\accwgt \geqslant K)\wedge \Box \Diamond F\big)=1$.
  Then, $\Pr^{\tsched}_{\cM,\sinit}(\Limit{\cE})>0$ implies
  $\ExpRew{\max}{\cE}(\MP) \geqslant 0$ for each end component $\cE$.
  Thus,
  $\cE$ is either weight-divergent or a 0-EC.
  As almost all $\tsched$-paths satisfy $\Box \Diamond F$ almost surely,
  $\cE$ must contain at least one $F$-state.
  This yields
  $$
    \Pr^{\tsched}_{\cM,\sinit}
     \bigl(\, \Diamond (\WDMEC_F \cup \ZeroEC_F) \, \bigr)\ \ = \ \ 1 \ .
  $$
  Moreover, if $\cE$ is a 0-EC with
  $\Pr^{\tsched}_{\cM,\sinit}(\Limit{\cE})>0$ and $\cE$ 
  is not a sub-component of
  a weight-divergent MEC  
  then all states of $\cE$ belong to $\ZeroEC_F$ and almost all
  $\tsched$-paths $\infpath \in \Limit{\cE}$
  have infinitely many prefixes $\fpath$ with
  $\wgt(\fpath) \geqslant K$.
  In particular, these paths $\infpath$ satisfy the
  formula
  $\bigvee_{t\in \ZeroEC_F} \!\!
      \Diamond (t \wedge (\accwgt \geqslant K))$.
  This yields $\Pr^{\tsched}_{\cM,\sinit}(\varphi)=1$. 
\end{proof}

\begin{lemma}
 \label{lemma:Buechi-max-pos}
 Let $\varphi$ be as in Lemma \ref{lemma:Buechi-max-as}.
 Then
  $\Pr^{\max}_{\cM,\sinit}
      \bigl(\Box \Diamond (\accwgt \geqslant K) \wedge \Box \Diamond F 
      \bigr)>0$ iff\ \ 
  $\Pr^{\max}_{\cM,\sinit}(\varphi)>0$.
\end{lemma}

\begin{proof}
  ``$\Longrightarrow$'':
  If $\sched$ is  a scheduler for $\cM$ 
  with $\Pr^{\sched}_{\cM,\sinit}(\varphi)>0$,
  then there is a finite $\sched$-path $\fpath$ from $\sinit$
  such that either $\last(\fpath)\in \WDMEC_F$ or
  $\last(\fpath)\in \ZeroEC_F$ and $\wgt(\fpath)\geqslant K$.
  Let now $\tsched$ be any scheduler for $\cM$ such that
  $\residual{\tsched}{\fpath}=\divsched$ where
  $\divsched$ is as in Lemma \ref{lemma:divsched}.
  Then,
  $\Pr^{\tsched}_{\cM,\sinit}
    (\Box \Diamond (\accwgt \geqslant K)\wedge \Box \Diamond F)>0$.

  ``$\Longleftarrow$'':
  We suppose we are given a scheduler $\tsched$
  for $\cM$ such that
  $\Pr^{\tsched}_{\cM,\sinit}
    \big(\Box \Diamond (\accwgt \geqslant K)\wedge \Box \Diamond F\big)$ is positive.
  There is an end component $\cE$ such that
  $\Pr^{\tsched}_{\cM,\sinit}(\Lambda_{\cE})>0$ where
  $\Lambda_{\cE}$ denotes the set of infinite paths $\infpath\in \Limit{\cE}$
  with $\infpath \models (\accwgt \geqslant K)\wedge \Box \Diamond F$.
  But then, $\cE$ contains an $F$-state and is probably bounded from below.
  Thus, $\Exp{\max}{\cE}(\MP)$ is nonnegative.

  If $\cE$ is weight-divergent then all states of $\cE$
  are contained in $\WDMEC_F$, in which case 
  $\Pr^{\sched}_{\cM,\sinit}(\Diamond \WDMEC_F)>0$
  and therefore $\Pr^{\sched}_{\cM,\sinit}(\varphi)>0$.

  Let us now consider the case where $\cE$ is not weight-divergent.
  Then, $\Exp{\max}{\cE}(\MP)=0$.
  Corollary \ref{Hilfslemma-Buechi-pos-max} yields
  that $\cE$ is a 0-EC.
  But then all states in $\cE$ are conatined in
  $\ZeroEC_F$. Hence, there is some state $t\in \ZeroEC$
  with $\Pr^{\sched}_{\cM,\sinit}
           \big(\Diamond (t \wedge (\accwgt \geqslant K))\big)>0$.
  But then $\Pr^{\sched}_{\cM,\sinit}(\varphi)>0$.
\end{proof}

\subsubsection{The Universal Problems \UaswB{} and \UposwB}

\label{sec:buechi-universal}

We now consider the universal variants \UaswB{} and \UposwB{} and show that
they are solvable using algorithms for the following two 
coB\"uchi problems:
\begin{center}
  \begin{tabular}{ll}
    \EposwcoB: &
    does there exist a scheduler $\sched$ s.t.
    $\Pr^{\sched}_{\cM,\sinit}\big(\Diamond \Box (\accwgt \geqslant K)\big)>0$ ?
    \\[0.5ex]
    \EaswcoB: &
    does there exist a scheduler $\sched$ s.t.
    $\Pr^{\sched}_{\cM,\sinit}\big(\Diamond \Box (\accwgt \geqslant K)\big)=1$ ?
  \end{tabular}
\end{center}

Property (S3) of the spider construction and 
the corresponding statement for the iterative application of
the spider construction (Lemma \ref{lemma:equivalence-M-and-Mi})
will serve as a useful vehicle to prove the following two lemmas,
which again will be used to reduce
\EposwcoB{} to \Eposdwr{} and \EaswcoB{} to \Easdwr
(see Lemma \ref{lemma:coBuechi} below).

\begin{lemma}
\label{Hilfslemma-Buechi-pos-max}
  Let $\cM$ be a strongly connected MDP where $\Exp{\max}{\cM}(\MP)=0$
  and where $\cM$ is not weight-divergent.
  If there exists a scheduler $\sched$ and some $K\in \Integer$ such that
  $\Pr^{\sched}_{\cM,s}\{\infpath \in \InfPaths :
     \infpath \models \Box \Diamond (\wgt \geqslant K)
     \wedge \lim(\infpath)=\cM\}$ is positive,
  then $\cM$ is a 0-EC. 
\end{lemma}

\begin{proof}
 Suppose by contradiction that $\cM$ is not a 0-EC.
 Let $\cN$ be the MDP resulting from applying the weight-divergence 
 algorithm to $\cM$ (iterative application
 of the spider construction).
  The set of infinite paths $\infpath$ 
  with $\limsup_{n \to \infty} \wgt(\prefix{\infpath}{n}) > -\infty$
  and where $\lim(\infpath)$ contains
  at least one state-action pair that is not contained in a 0-EC
  constitutes a 0-EC-invariant property
  with positive measure under $\sched$ from state $s$.
 Hence, Lemma \ref{lemma:equivalence-M-and-Mi}
 yields the existence of a scheduler $\tsched$ for
 $\cN$ such that
 the set of infinite paths $\infpath$ in $\cN$ 
 with $\limsup_{n \to \infty} \wgt(\prefix{\infpath}{n}) > -\infty$
 has positive measure from $s$.
 This, however, is impossible as 
 $\Exp{\max}{\cF}(\MP)<0$ for all end components $\cF$ of $\cN$
 (see Theorem \ref{weight-divergence-algorithm}),
 which yields that $\cN$ is universally negatively pumping.
\end{proof}

As a consequence of Lemma \ref{Hilfslemma-Buechi-pos-max}
we get the following corollary which can be seen as an add-on
for Lemma \ref{universal-neg-wgt-div} 
(but will not be used for the following
considerations on weight-bounded B\"uchi conditions).

\begin{corollary}
  Let $\cM$ be a strongly connected MDP where $\Exp{\max}{\cM}(\MP)=0$
  and where $\cM$ is not positively weight-divergent.
  Then, $\cM$ has no 0-ECs if and only if $\cM$ is universally
  negatively pumping.
\end{corollary}

\begin{proof}
  The implication ``$\Longleftarrow$'' is trivial as 0-ECs are obviously
  not negatively pumping.
  We now prove ``$\Longrightarrow$''. For this, we suppose that
  $\cM$ is not universally negativey pumping and show that
  $\cM$ has at least one 0-EC.
  Being not universally negatively pumping implies the
  existence of a scheduler $\sched$ and a state $s$ such that:
  $$
   \Pr^{\sched}_{\cM,s}
    \bigl\{\infpath \in \InfPaths : \liminf_{n \to \infty} 
     \wgt(\prefix{\infpath}{n}) \ > \ - \infty\bigr\} \ > \ 0\enskip.
  $$
  But then there is some $K\in \Integer$ and
  an end component $\cE$ such that 
  $$
    \Pr^{\sched}_{\cM,s}\bigl\{\infpath \in \InfPaths : 
           \infpath \models \Box \Diamond (\accwgt \geqslant K) 
           \wedge \lim(\infpath)=\cE 
    \bigr\} \ > \ 0 \enskip.
  $$
  As $\Exp{\max}{\cM}(\MP)=0$ we have
  $\Exp{\max}{\cE}(\MP)\leqslant 0$.
  As $\Exp{\max}{\cE}(\MP)<0$ would imply that $\cE$ is universally
  negatively pumping we get 
  $\Exp{\max}{\cE}(\MP)= 0$.
  Lemma \ref{Hilfslemma-Buechi-pos-max} 
  applied to $\cE$ yields that $\cE$ is a 0-EC.
\end{proof}

\begin{lemma}
\label{Hilfslemma-coBuechi-pos-max}
  Let $\cM$ be a strongly connected MDP where $\Exp{\max}{\cM}(\MP)=0$.
  If there exists a scheduler $\sched$ and some $K\in \Integer$ such that
  $\Pr^{\sched}_{\cM,s}\{\infpath \in \InfPaths :
     \infpath \models \Diamond \Box (\wgt \geqslant K)
     \wedge \lim(\infpath)=\cM\}$ is positive
  then $\cM$ is a 0-EC. 
\end{lemma}

\begin{proof}
  The argument is similar as for Lemma \ref{Hilfslemma-Buechi-pos-max}.
  Suppose by contradiction that $\cM$ is not a 0-EC.
  Let $\cN$ denote the MDP resulting from $\cM$ by successively applying
  the spider construction to flatten all 0-ECs of $\cM$. 
  For this we can rely on the algorithm to compute
  all maximal 0-ECs of $\cM$  presented
  in Section \ref{sec:compute-max-0-EC} 
  (see also Lemma \ref{mincredit-ZeroEC}) and then successively apply
  the spider construction to the BSCCs of the maximal 0-ECs.
  The final MDP $\cN$ has no 0-EC, but might still contain gambling
  end components.
  In any case, all end components of $\cN$ are negatively weight-divergent.
  The set of infinite paths $\infpath$ 
  with $\infpath \models \Diamond \Box (\wgt \geqslant K)$ 
  where $\lim(\infpath)$ contains at least one state-action pair
  that does not belong to a 0-EC 
  consitutes
  a measurable 0-EC-invariant property. 
  The probability for this property under
  $\sched$ from state $s$ is positive 
  (by assumption and as $\cM$ is not a 0-EC).
  Thanks to Lemma \ref{lemma:equivalence-M-and-Mi},
  scheduler $\sched$ for $\cM$ can be transformed into a scheduler
  $\tsched$ for $\cN$ such that 
  $\Diamond \Box (\wgt \geqslant K)$ holds with positive probability.
  But this is impossible as 
  all end components of $\cN$
  are negatively weight-divergent.
\end{proof}

\begin{lemma}
\label{lemma:coBuechi}
  Problem \EposwcoB{} is polynomially reducible to \Eposdwr, 
  while \EaswcoB{} is polynomially reducible to \Easdwr.
\end{lemma}

\begin{proof}
Let $\PumpEC=\PumpEC_S$ be the set of states that are contained in some
pumping end component and let $\ZeroEC=\ZeroEC_S$ 
denote the set of states that
belong to a maximal 0-EC $\cZ$ where $\cZ$ is a sub-component of a
maximal end component  of $\cM$ with maximal expected mean payoff 0.
Using the results of Section \ref{sec:classification} 
(see also Lemma \ref{lemma:divsched}), 
$\PumpEC$ and $\ZeroEC$ are computable in 
polynomial time.
Recall from Section \ref{sec:min-credit}
that for each state $t\in \ZeroEC$ the recurrence value
$\rec(t)$ defined as the maximal value $w\in \Integer$
such that
$\Pr^{\max}_{\cZ,s}
           \big(\, \Box (\accwgt  \geqslant w) \, \wedge \, \Box \Diamond t \, \big)
       =1$
is computable in polynomial time where $\cZ$ denotes the
unique maximal 0-EC that contains $t$.
Let now $\usched$ be a scheduler such that: 
\begin{itemize}
\item 
  $\usched$ is pumping
  from each state $t\in \PumpEC$,
\item
  $\Pr^{\usched}_{\cM,s}
    (\, \Box (\accwgt  \geqslant \rec(t)) \, \wedge \, \Box \Diamond t \, )
   =1$ for each state $t\in \ZeroEC$.
\end{itemize}
 Let $\varphi$ be the disjunctive weight-bounded reachability constraint
 $\varphi=\varphi(T,(K_t)_{t\in T})$ where
 $T = \PumpEC \cup \ZeroEC$ and $K_t=-\infty$ for $t\in \PumpEC$ and
 $K_t=K-\rec(t)$ for $t\in \ZeroEC \setminus \PumpEC$. 
 We now show:
 \begin{enumerate}
 \item [(a)] 
     $\Pr^{\max}_{\cM,\sinit}\big(\Diamond \Box (\accwgt \geqslant K)\big)>0$
      iff 
     $\Pr^{\max}_{\cM,\sinit}(\varphi)>0$
 \item [(b)]
     $\Pr^{\max}_{\cM,\sinit}\big(\Diamond \Box (\accwgt \geqslant K)\big)=1$
      iff 
     $\Pr^{\max}_{\cM,\sinit}(\varphi)=1$ 
 \end{enumerate}
 Clearly,
 statement (a) implies the polynomial reducibility of \EposwcoB{} to \Eposdwr,
 while the polynomial reducibility of \EaswcoB{} to \Easdwr{} follows from 
 statement (b).

 {\it Proof of statement (a).}
 The implication is ``$\Longleftarrow$'' is obvious any scheduler
 $\tsched$ with $\Pr^{\tsched}_{\cM,\sinit}(\varphi)>0$ 
 can be combined with the scheduler $\usched$ above
 to obtain a new scheduler $\sched$ with
 $\Pr^{\sched}_{\cM,\sinit}\big(\Diamond \Box (\accwgt \geqslant K)\big)>0$.

 To prove ``$\Longrightarrow$'', we suppose there is a  
 scheduler $\tsched$ with
 $\Pr^{\tsched}_{\cM,\sinit}\big(\Diamond \Box (\accwgt \geqslant K)\big)>0$.
 There exists an end component $\cE$ such that
 \[
   \Pr^{\tsched}_{\cM,\sinit}
    \bigl\{\ \infpath \in \Limit{\cE} \ : \ 
             \infpath \models \Diamond \Box (\accwgt \geqslant K) \ 
    \bigr\} \ \ > \ \ 0\enskip.
 \]
 If $\cE$ contains some state that belongs to a pumping end component
 (this covers the case where $\cE$ itself is pumping)
 then $\cE$ contains at least one state in 
 $\PumpEC$, which obviously yields $\Pr^{\tsched}_{\cM,\sinit}(\varphi)>0$.
 Suppose now that none of the states in $\cE$ belongs to a pumping
 end component. 
 Obviously we then have $\Exp{\max}{\cE}(\MP)=0$.
 Thanks to Lemma \ref{Hilfslemma-coBuechi-pos-max}
 we get that $\cE$ is a 0-EC.

Pick some state $t$ in $\cE$. 
The presented algorithm for computing the recurrence values 
(see Section \ref{sec:min-credit})
shows that $\rec(t) \geqslant w_{\cE}(t)$ 
where $w_{\cE}(t)=\min \{w(t,s):s\in \cE\}$.
(As before, $w(t,s)$ is the weight of all paths from $t$ to $s$ in 
the maximal 0-EC $\cZ$ that subsumes $\cE$.)
Thus, if $\infpath \in \Limit{\cE}$ with
$\infpath \models \Diamond \Box (\accwgt \geqslant K)$
then there exists an integer $L \geqslant K$ 
such that $\infpath$ contains infinitely
many finite prefixes $\fpath$ with $\last(\fpath)=t$ and $\wgt(\fpath)=L$.
We then have $L+w(t,s)\geqslant K$ for all states $s$ in $\cE$.
Thus, $L+w_{\cE}(t) \geqslant K$ and therefore
$$
  \rec(t) \ \ \geqslant \ \ w_{\cE}(t) \ \ \geqslant \ \ K-L \enskip.
$$
We get $L \geqslant K-\rec(t)=K_t$.
This yields $\infpath \models \Diamond (t \wedge (\accwgt \geqslant K_t))$.
But then $\Pr^{\tsched}_{\cM,\sinit}(\varphi)>0$.

 {\it Proof of statement (b).}  
 For the implication is ``$\Longleftarrow$'' we suppose that we are 
 given a scheduler $\tsched$ with
 $\Pr^{\max}_{\cM,\sinit}(\varphi)=1$.
 Composing $\tsched$ with the above scheduler $\usched$ we obtain a
 new scheduler $\sched$ with
 $\Pr^{\sched}_{\cM,\sinit}\big(\Diamond \Box (\accwgt \geqslant K)\big)=1$.

 To prove ``$\Longrightarrow$'' we suppose that we are given a scheduler
 $\sched$ with
 $\Pr^{\sched}_{\cM,\sinit}(\Diamond \Box (\accwgt \geqslant K))=1$.
 But then each end component $\cE$ where
 $\Pr^{\sched}_{\cM,\sinit}\{\infpath \in \Limit{\cE} :
    \infpath \models \Diamond \Box (\accwgt \geqslant K)\}>0$
 has nonnegative maximal expected mean payoff. 
 If $\Exp{\max}{\cE}(\MP)>0$ then $\cE$ is pumping and all states
 of $\cE$ belong to $\PumpEC$.
 If $\Exp{\max}{\cE}(\MP)=0$ then Lemma \ref{Hilfslemma-coBuechi-pos-max}
 implies that $\cE$ is a 0-EC.
 As in the proof of statement (a) we can pick an arbitrary state $t$ of
 $\cE$ and show 
 each $\sched$-paths $\infpath \in \Limit{\cE}$ with
 $\infpath \models \Diamond \Box (\accwgt \geqslant K)\}>0$
 has infinitely many prefixes $\fpath$ with $\last(\fpath)=t$
 and $\wgt(\fpath)\geqslant K_t$.
 This yields $\Pr^{\sched}_{\cM,\sinit}(\varphi)=1$.
\end{proof}

Combining Lemma \ref{lemma:coBuechi} 
with the results on \dwr{}-problems
yields:

\begin{corollary}
  \label{coBuechi-complexity} 
  \EposwcoB{} is solvable in polynomial time,
  while \EaswcoB{} is in $\textrm{NP}\cap \textrm{coNP}$ and
  solvable in pseudo-polynomial time.
\end{corollary}

We now return to weight-bounded B\"uchi constraints and show how
\UaswB{} and \UposwB{} are solvable using algorithms for (the complements of)
the coB\"uchi problems \EposwcoB{} and \EaswcoB, respectively.

\begin{lemma}
  Problem \UaswB{} is solvable in polynomial time.
\end{lemma}

\begin{proof}
Let $\cM^-$ denote the MDP resulting from $\cM$ 
by multiplying all weights with $-1$ and let $L=-(K{-}1)$.
\begin{center}
 \begin{tabular}{lcl}
   
   $(\cM,\sinit)$ satisfies \UaswB{} &
   iff \ \ &
   $\Pr^{\min}_{\cM,\sinit}(\Box \Diamond F)=1$ and
   $\forall \sched. \
     \Pr^{\sched}_{\cM,\sinit}
      \bigl(\, \Box \Diamond (\accwgt \geqslant K) \, \big)=1$
   \\[1ex]

 &  iff \ \ &
   $\Pr^{\min}_{\cM,\sinit}(\Box \Diamond F)=1$ and
   $\neg \exists \sched. \ 
     \Pr^{\sched}_{\cM,\sinit}
      \bigl(\, \Diamond \Box (\accwgt < K) \, \big)>0$
   
   \\[1ex]

  & iff \ \ &
   $\Pr^{\min}_{\cM,\sinit}(\Box \Diamond F)=1$ and
   $\neg \exists \sched. \
     \Pr^{\sched}_{\cM^-,\sinit}
      \bigl(\, \Diamond \Box (\accwgt \geqslant L) \, \big)>0$
\end{tabular}
\end{center}
Thus, \UaswB{} is solvable using known polynomial-time algorithms
to check whether $\Pr^{\min}_{\cM,s}(\Box \Diamond F)=1$
and an algorithm for the complement of \EposwcoB.
As \EposwcoB{} is solvable in polynomial time 
(Corollary \ref{coBuechi-complexity}),
so is \UaswB.
\end{proof}

\begin{lemma}
  \UposwB{} is in $\textrm{NP}\cap \textrm{coNP}$ and solvable
  in pseudo-polynomial time.
\end{lemma}

\begin{proof}
As before, let $\cM^-$ denote the MDP resulting from $\cM$ 
by multiplying all weights with $-1$ and let $L=-(K{-}1)$.
\begin{center}
 \begin{tabular}{ll}
   &
   $(\cM,s)$ satisfies \UposwB
   \\[1ex]

   iff \ \ &
   $\forall \sched. \
     \Pr^{\sched}_{\cM,s}
      \bigl(\, \Box \Diamond (\accwgt \geqslant K) \, \, \wedge \, 
        \Box \Diamond F \, \big)>0$
   \\[1ex]

   iff \ \ &
   $\neg \exists \sched. \ 
     \Pr^{\sched}_{\cM,s}
      \bigl(\, \Diamond \Box (\accwgt < K) \, \vee \
        \Diamond \Box \neg F \, \big)=1$
   \\[1ex]

   iff \ \ &
   $\neg \exists \sched. \ 
     \Pr^{\sched}_{\cM^-,s}
      \bigl(\, \Diamond \Box (\accwgt \geqslant L) \, \vee \
        \Diamond \Box \neg F \, \big)=1$
 \end{tabular}
\end{center}
Let $G$ denote the union of all states belonging to a (possibly non-maximal)
end component $\cE$ of $\cM$ (or $\cM^-$) 
such that $\cE \cap F=\varnothing$.
The set $G$ is computable in polynomial time using standard techniques.
Let now $\cN$ be the MDP resulting from $\cM^-$
by collapsing all states in $G$ into a fresh trap state $g$
and adding a state-action pair $(g,\alpha)$ where
$P_{\cN}(g,\alpha,g)=1$ and $\wgt_{\cN}(g,\alpha)=1$.
Then:
\[
  \exists \sched. \ 
     \Pr^{\sched}_{\cM^-,s}
      \bigl(\, \Diamond \Box (\accwgt \geqslant L) \, \vee \
        \Diamond \Box \neg F \, )=1
\quad \Longleftrightarrow \quad
   \exists \sched. \ 
     \Pr^{\sched}_{\cN,s}
      \bigl(\, \Diamond \Box (\accwgt \geqslant L) \, \bigr)\, = \, 1
\enspace.
\]

This yields $(\cM,s)$ satisfies \UposwB{} \ iff $(\cN,s)$ does not satisfy
\EaswcoB{} \ for the weight bound $L$.
Thus, the claim follows from Corollary \ref{coBuechi-complexity}.
\end{proof}

\subsubsection{Optimal Values for Weight-Bounded B\"uchi Constraints}

The values of the optimization problems \EposwB, \EaswB, \UposwB{} and
\UaswB{}  are computable using the above reductions 
to the DWR problems and the algorithms presented for the optimization
variants for \Easdwr{} and \Eposdwr.
Thus,
\begin{center}
   \begin{tabular}{lcl}
      $\valueEposB{\cM,s}$ 
      & = &
      $\max \, \bigl\{ \, K \in \Integer \ : \ 
       \exists \sched. \
          \Pr^{\sched}_{\cM,s}
            \bigl(\, \Box \Diamond (\accwgt \geqslant K)
                \, \wedge \, \Box \Diamond F \, 
             \bigr) \, > \, 0 \, \bigr\}$
      \\[1ex]

      $\valueUasB{\cM,s}$ 
      & = &
      $\max \, \bigl\{ \, K \in \Integer \ : \ 
       \forall \sched. \
          \Pr^{\sched}_{\cM,s}
            \bigl(\, \Diamond \Box (\accwgt \geqslant K)
                \, \wedge \, \Box \Diamond F \, 
             \bigr) \, = \, 1 \, \bigr\}$
      \\[1ex]
\end{tabular}
\end{center}
are computable in polynomial time, while the optimal weight bounds
for \EaswB{} and \UposwB{} are computable in pseudo-polynomial time.

\end{document}